\PassOptionsToPackage{table}{xcolor}
\documentclass[11pt]{article}

\RequirePackage{amssymb} %
\RequirePackage{amsmath} %
\RequirePackage{latexsym}
\RequirePackage{mathrsfs}
\RequirePackage{varioref} %

\DeclareMathOperator*{\cov}{cov}
\DeclareMathOperator*{\var}{var}

\DeclareMathOperator*{\supp}{supp}

\DeclareMathOperator*{\diag}{diag}

\newcommand{\R}{\ensuremath{\mathbb{R}}}

\newcommand{\Exp}{\ensuremath{\mathbb{E}}} %
\newcommand{\Prob}{\ensuremath{\mathbb{P}}} %

\newcommand{\0}{\ensuremath{\mathbf{0}}}

\usepackage{bbm}
\newcommand{\indicator}{\ensuremath{\mathbbm{1}}}

\DeclareMathAlphabet{\mathpzc}{OT1}{pzc}{m}{it}

\renewcommand{\asymp}{\ensuremath{\stackrel{\text{a}}{\sim}}}

\def\independenT#1#2{\mathrel{\rlap{$#1#2$}\mkern2mu{#1#2}}}

\makeatletter
\newcommand*{\indep}{%
  \mathbin{%
    \mathpalette{\@indep}{}%
  }%
}
\newcommand*{\nindep}{%
  \mathbin{%
    \mathpalette{\@indep}{\not}%
  }%
}
\newcommand*{\@indep}[2]{%
  \sbox0{$#1\perp\m@th$}%
  \sbox2{$#1=$}%
  \sbox4{$#1\vcenter{}$}%
  \rlap{\copy0}%
  \dimen@=\dimexpr\ht2-\ht4-.2pt\relax
  \kern\dimen@
  {#2}%
  \kern\dimen@
  \copy0 %
} 
\makeatother

\usepackage{color}

\usepackage{xr-hyper}
\pdfoutput=1
\usepackage{hyperref}
\hypersetup{
    colorlinks=true,
    linkcolor=black,
    citecolor=black,
    filecolor=black,
    urlcolor=black,
}

\usepackage[dvipsnames]{xcolor}
\colorlet{mylinkcolor}{RoyalBlue}%

\usepackage{enumitem}

\usepackage{wrapfig}
\usepackage{tikz-cd}
\usepackage{tikz}
\usepackage{pgfplots}

\usepackage{tkz-berge}
\usetikzlibrary{arrows}
\usetikzlibrary{positioning}

\RequirePackage{amsthm}

\theoremstyle{definition}

\newtheorem{proposition}{Proposition}
\newtheorem{lemma}{Lemma}
\newtheorem{remark}{Remark}
\newtheorem{definition}{Definition}
\newtheorem{theorem}{Theorem}
\newtheorem{corollary}{Corollary}

\newenvironment{assump}[2][Assumption]{\begin{trivlist}
\item[\hskip \labelsep {\bfseries #1}\hskip \labelsep {\bfseries #2}]}{\end{trivlist}}

\newcounter{homogSection}
\newcommand{\aAssump}{A\arabic{homogSection}}
\newcounter{het1section}
\newcommand{\bAssump}{B\arabic{het1section}}
\newcounter{het2section}
\newcommand{\cAssump}{C\arabic{het2section}}

\newtheoremstyle{theoremSuppressedNumber}{}{}{}{}{\bfseries}{.}{ }{\thmname{#1}\thmnote{ (\mdseries #3)}}
\theoremstyle{theoremSuppressedNumber}
\newtheorem{homogAssump}{Assumption \aAssump \addtocounter{homogSection}{1}}
\newtheorem{het1Assump}{Assumption \bAssump \addtocounter{het1section}{1}}
\newtheorem{het2Assump}{Assumption \cAssump \addtocounter{het2section}{1}}

\usepackage{setspace}
\usepackage[longnamesfirst]{natbib}
\usepackage{ulem}
\usepackage{subfigure}
\usepackage{wrapfig}
\usepackage{float}
\usepackage{afterpage}

\usepackage{epstopdf}

\usepackage{chngpage}
\usepackage{multirow}
\usepackage{tabulary}
\usepackage{tabularx}
\usepackage{booktabs}
\usepackage{textcomp}

\usepackage{calc}

\title{\textbf{Salvaging Falsified Instrumental Variable Models}\footnote{This paper was presented at Auburn University, UC San Diego, Texas A\&M, Duke, Columbia, Cornell, the Yale MacMillan-CSAP workshop, the University of Mannheim, the 2018 Incomplete Models conference at Northwestern University, the 2018 and 2019 Southern Economic Association Meetings, the 2019 IAAE conference, the 2019 Georgetown Center for Economic Research Mini-Conference on Non-Standard Methods in Econometrics, the 2019 CeMMAP UCL/Vanderbilt Conference on Advances in Econometrics, the 29th Annual Meeting of the Midwest Econometrics Group, and the 2019 Greater New York Area Econometrics Colloquium. We thank audiences at those seminars and conferences as well as Federico Bugni, Tim Christensen, Allan Collard-Wexler, Joachim Freyberger, Chuck Manski, Arnaud Maurel, Francesca Molinari, Adam Rosen, Pedro Sant'Anna, and Alex Torgovitsky for helpful conversations and comments. We thank Paul Diegert and Peiran Xiao for excellent research assistance.
\newline\indent Data Acknowledgements: Researchers own analyses calculated (or derived) based in part on data from The Nielsen Company (US), LLC and marketing databases provided through the Nielsen Datasets at the Kilts Center for Marketing Data Center at The University of Chicago Booth School of Business. The conclusions drawn from the Nielsen data are those of the authors and do not reflect the views of Nielsen. Nielsen is not responsible for, had no role in, and was not involved in analyzing and preparing the results reported herein.
}}

\author{Matthew A. Masten\thanks{Department of Economics, Duke University, \texttt{matt.masten@duke.edu}}
\qquad
Alexandre Poirier\thanks{Department of Economics, Georgetown University, \texttt{alexandre.poirier@georgetown.edu}}
}

\usepackage{titlesec}

\usepackage{titling}
\definecolor{lightgray}{gray}{0.95}
\newcolumntype{a}{>{\columncolor{lightgray}}c}

\begin{document}
\maketitle

\vspace{-2em}

\begin{abstract}
What should researchers do when their baseline model is refuted? We provide four constructive answers. First, researchers can measure the extent of falsification. To do this, we consider continuous relaxations of the baseline assumptions of concern. We then define the falsification frontier: The smallest relaxations of the baseline model which are not refuted. This frontier provides a quantitative measure of the extent of falsification. Second, researchers can present the identified set for the parameter of interest under the assumption that the true model lies somewhere on this frontier. We call this the \emph{falsification adaptive set}. This set generalizes the standard baseline estimand to account for possible falsification. Importantly, it does not require the researcher to select or calibrate sensitivity parameters. Third, researchers can present the identified set for a specific point on this frontier. Finally, as a sensitivity analysis, researchers can present identified sets for points beyond the frontier. To illustrate these four ways of salvaging falsified models, we study overidentifying restrictions in two instrumental variable models: a homogeneous effects linear model, and heterogeneous effect models with either binary or continuous outcomes. In the linear model, we consider the classical overidentifying restrictions implied when multiple instruments are observed. We generalize these conditions by considering continuous relaxations of the exclusion restrictions. By sufficiently weakening the assumptions, a falsified baseline model becomes non-falsified. This lets us derive the falsification adaptive set, which has a simple closed form expression that depends only on a few 2SLS regression coefficients. We obtain similar results in the heterogeneous effect models, where we derive identified sets for marginal distributions of potential outcomes, falsification frontiers, and falsification adaptive sets under continuous relaxations of the instrument exogeneity assumptions. We apply our results to four different empirical studies: \cite{DurantonMorrowTurner2014}, \cite{AlesinaGiulianoNunn2013}, \cite{AcemogluJohnsonRobinson2001}, and \cite{Nevo2001}. We show how our results, especially the falsification adaptive set, can help researchers go beyond the binary pass/fail outcome of a specification test and instead summarize the variation in estimates obtained from alternative non-falsified models.
\end{abstract}

\bigskip
\small
\noindent \textbf{JEL classification:}
C14; C18; C21; C26; C51

\bigskip
\noindent \textbf{Keywords:}
Instrumental Variables, Nonparametric Identification, Partial Identification, Sensitivity Analysis

\onehalfspacing
\normalsize

\newpage
\section{Introduction}

Many models used in empirical research are falsifiable, in the sense that there exists a population distribution of the observable data which is inconsistent with the model. We equivalently call these overidentified or refutable models.\footnote{Keep in mind that falsifiability is a property of a model, not a parameter. It is nonetheless common for researchers to discuss  ``overidentified parameters.'' These are parameters whose identified sets are either empty (when the model is falsified) or a singleton (when the model is not falsified). Thus such parameters are point identified when the model is not falsified. To avoid confusion between properties of a parameter and properties of a model, we only use the terms falsifiable and refutable when describing models.} With finite samples, researchers often use specification tests to check whether their baseline model is refuted. A well known example is the overidentifying restrictions test in linear instrumental variable models (\citealt{AndersonRubin1949}, \citealt{Sargan1958}, \citealt{Hansen1982}). Abstracting from sampling uncertainty\footnote{We study population level identification and falsification in this paper. We briefly discuss finite sample estimation and inference in sections \ref{subsec:estimationInference} and \ref{subsec:FASestimationLinearIV}, but this is not the focus of the paper.}, the population versions of such specification tests have a persistent problem: What should researchers do when their baseline model is refuted?

One option is to ignore the finding and report some point estimand, like the 2SLS estimand in an overidentified linear instrumental variable model. There are two problems with this approach. First, it is often justified by appealing to large sample sizes. For example, \cite{Nevo2001} refers to this justification on page 325: ``It is well known that with a large enough sample a chi-squared test will reject essentially any model.'' If we are willing to relax assumptions, however, this common justification is not true: There always exist assumptions which are sufficiently weak that they are not refuted. 
Second, such point estimands can be highly misleading. For example, \cite{HeckmanHotz1989} compare estimates from observational data with those from experimental data. In their data, they show that the observational estimates from models which fail specification tests are much farther from the experimental estimates than observational estimates from models which pass specification tests. Essentially, parameters like the 2SLS estimand are motivated by arguments that they deliver a causal effect of interest under a set of baseline assumptions. When those assumptions are known to be false, there is no reason for the corresponding estimand to be close to the causal effect of interest.
In appendix \ref{sec:AltResponses} we discuss five other responses to baseline falsification from the literature. We review the related literature in appendix \ref{sec:relatedLiterature}.

Instead of ignoring findings from specification tests, we provide four constructive ways for researchers to salvage a falsified baseline model. First, researchers can measure the extent of falsification. To do this, we consider continuous relaxations of the baseline assumptions of concern. We then define the falsification frontier: The smallest relaxations of the baseline model which are not refuted. 
This frontier provides a quantitative measure of the extent of falsification. Second, researchers can present the identified set for the parameter of interest under the assumption that the true model lies somewhere on this frontier. We call this the \emph{falsification adaptive set} (FAS). This set collapses to the baseline identified set or point estimand when the baseline model is not refuted. When the baseline model is refuted, this set expands to include all parameter values consistent with the data and a model which is relaxed just enough to make it non-refuted. This set generalizes the standard baseline estimand to account for possible falsification. 
Importantly, \emph{researchers do not need to select or calibrate sensitivity parameters to compute the falsification adaptive set.}
While this set is agnostic about the direction of relaxation, our third suggestion is that researchers present the identified set for a specific minimally relaxed model. 
Finally, as a further sensitivity analysis, researchers can present identified sets for points beyond the falsification frontier. We discuss a complementary Bayesian approach in appendix \ref{sec:BayesianApproach}.

To illustrate these four constructive ways to salvage a falsified baseline model, we study a classic source of overidentification: observation of several instrumental variables. We do this in two different models. The first is the classical constant coefficients linear model with multiple instruments. This model imposes homogeneous treatment effects, but allows for continuous treatments. We generalize the usual overidentifying restrictions by allowing the instruments to have some direct effect on outcomes. We use this result to characterize the falsification frontier, which trades off bounds on the magnitudes of these direct instrument effects. We then characterize the identified set along the falsification frontier. This leads to a particularly simple closed form expression for the falsification adaptive set, depending only on the value of a handful of 2SLS regression coefficients. We also show how to use these identified sets to do sensitivity analysis for non-falsified models beyond the frontier.

It is well known that the classical overidentifying conditions may not hold when treatment effects are heterogeneous, even if all instrument exogeneity and exclusion restrictions hold. We therefore also study a second model, which allows for heterogeneous treatment effects and multiple instruments. In this model we focus on binary treatments. We consider both binary and continuous outcomes. Although this model allows for heterogeneous treatment effects, it is nonetheless falsifiable. We relax statistical independence between each instrument and potential outcomes using a latent propensity score distance from our previous work, \cite{MastenPoirier2017}. Under these relaxations, we derive the identified set for the marginal distributions of potential outcomes. 
We show how to quickly compute the falsification frontier and the falsification adaptive set using convex optimization. We then discuss how to use these identified sets to do sensitivity analysis regardless of whether the baseline model is refuted.

We show how to use our results in four previously published empirical studies: \citet*[\emph{The Review of Economic Studies}]{DurantonMorrowTurner2014}, \citet*[\emph{The Quarterly Journal of Economics}]{AlesinaGiulianoNunn2013}, \citet*[\emph{The American Economic Review}]{AcemogluJohnsonRobinson2001}, and \citet[\emph{Econometrica}]{Nevo2001}. Each paper reports 2SLS estimates using multiple instrumental variables and each paper discusses concerns about instrument validity. All four papers run overidentification tests, which sometimes fail. Even when they do not fail, the authors sometimes express a concern that this could simply be due to low sample size. We show that the falsification adaptive set is an informative complement to these traditional tests: Rather than focusing on null hypothesis significance testing, the FAS summarizes the range of estimates obtained from alternative models which which are not falsified by the data. Thus the FAS reflects the model uncertainty that arises from a falsified baseline model: Relying on different instruments to different degrees yields different results. The FAS gives this range of results.

Finally, it is important to distinguish between two kinds of falsifiable models. The first kind is falsified because of an assumption made directly on observed random variables. For example, suppose we observe a scalar random variable. Consider the model which assumes that this variable is normally distributed. At the population level, we can simply check whether the observed variable is actually normally distributed. If not, the model is refuted. When the model is refuted, it can be salvaged by removing the normality assumption. A massive literature in statistics and econometrics on semi- and non-parametrics is largely concerned with salvaging refuted models by relaxing parametric assumptions on observed random variables. For example, see the survey by \cite{Spanos2018}.
The second kind of model is falsified because of an assumption made on unobserved random variables, or unobserved structural parameters. Salvaging these models is delicate because, even at the population level, the data themselves do not tell us the correct alternative assumptions. We study this second kind of falsifiable model in this paper.

\section{Salvaging Falsified Models}\label{sec:generalFF}

In this section we consider a general falsifiable model. We use this model to precisely define our four recommended responses to baseline falsification. In particular, we formally define the falsification frontier and falsification adaptive set. In sections \ref{sec:homogModel} and \ref{sec:hetModel} we illustrate these general concepts in two specific instrumental variable models. We discuss some interpretation issues in section \ref{subsec:FFinterp}. In section \ref{subsec:FFvsBF} we discuss the relationship between the falsification frontier and the breakdown frontier we studied in our previous work (\citealt{MastenPoirier2017BF}). We discuss estimation and inference in section \ref{subsec:estimationInference}.

\subsection{Measuring the Extent of Falsification}\label{subsec:generalFF}

Let $W$ be a vector of observed random variables. Let $\mathcal{F}$ denote the set of all cdfs on the support of $W$, $\supp(W)$. A \emph{model} is a set of underlying structural parameters which generate the observed distribution $F_W$ and restrictions on those structural parameters. This definition of a model suffices for our purposes. See section 2 of \cite{Matzkin2007}, for example, for a more formal definition. A given model $\mathcal{M}$ is \emph{falsifiable} if there are some observed distributions $F_W$ which could not have been generated by the model. If such a cdf is observed, we say the model is falsified (equivalently, refuted). For a given model, let $\mathcal{F}_\text{f}$ denote the set of cdfs $F_W$ which falsify the model. Let $\mathcal{F}_\text{nf}$ denote the set of cdfs $F_W$ which do not falsify the model. Falsifiable models are often said to have \emph{testable implications}. We prefer the terms `falsifiable' or `refutable', to distinguish this population level feature from issues arising in statistical hypothesis testing.\footnote{Informally, it is possible that for every non-falsified distribution of observables there exists a falsified distribution of observables which is arbitrarily close. Thus hypothesis tests in finite samples have a difficult time distinguishing the two cases, even though we can distinguish them at the population level. In this case we say the model is falsifiable but not testable. \cite{CanaySantosShaikh2013} study an example of this; also see \cite{Freyberger2017}.}

Suppose we begin with a falsifiable baseline model $\mathcal{M}(0_L)$. Let $\mathcal{F}_\text{nf}(0_L)$ denote the set of joint distributions of the data which are not falsified by this model. Hence $\mathcal{F}_\text{nf}(0_L)$ is a strict subset of $\mathcal{F}$. $L$ denotes the number of assumptions which we believe may be the reason the model is falsifiable. For each assumption $\ell \in \{1,\ldots,L\}$, we define a class of assumptions indexed by a parameter $\delta_\ell$ such that the assumption is imposed for $\delta_\ell = 0$, the assumption is not imposed for $\delta_\ell = 1$, and the assumption is partially imposed for $\delta_\ell \in (0,1)$.\footnote{More generally, the range of $\delta$ could be $[0,\delta_\text{max}]$ for some value $\delta_\text{max} \geq 0$, possibly $+\infty$.} These assumptions must be nested in the sense that for $\delta_\ell' \geq \delta_\ell$, assumption $\delta_\ell'$ is weaker than assumption $\delta_\ell$. Let $\mathcal{M}(\delta)$ denote the model which imposes assumptions $\delta = (\delta_1,\ldots,\delta_L)$. Let $\mathcal{F}_\text{nf}(\delta)$ denote the set of joint distributions of the data which are not falsified by this model. Suppose further that the model $\mathcal{M}(1_L)$ which does not impose any of the $L$ assumptions is not falsifiable.

Recall that $F_W$ denotes the observed distribution of the data. Suppose $F_W \notin \mathcal{F}_\text{nf}(0_L)$, so that the baseline model is falsified. We first consider the case where we only relax a single assumption.

\begin{definition}
Suppose $L =1$. The \emph{falsification point} is
\[
	\delta^* = \inf \{ \delta \in [0,1] : F_W \in \mathcal{F}_\text{nf}(\delta) \}.
\]
\end{definition}

That is, the falsification point is the smallest relaxation of the baseline assumption such that the model is not falsified by the observed data $F_W$. For any $\delta < \delta^*$, the model is falsified. For any $\delta > \delta^*$, the model is not falsified. Under mild conditions, we can show that the set $ \{ \delta \in [0,1] : F_W \in \mathcal{F}_\text{nf}(\delta) \}$ is closed, and hence the model is not falsified at $\delta^*$. %

The falsification frontier is the multi-dimensional version of the falsification point. To formally define it, partition $[0,1]^L$ into the set of all assumptions which are falsified,
\[
	\mathcal{D}_\text{f} = \{ \delta \in [0,1]^L : F_W \notin \mathcal{F}_\text{nf}(\delta) \}
\]
and the set of all assumptions which are not falsified,
\[
	\mathcal{D}_\text{nf} = \{ \delta \in [0,1]^L : F_W \in \mathcal{F}_\text{nf}(\delta) \}.
\]

\begin{definition}\label{def:abstractFF}
The \emph{falsification frontier} is the set
\[
	\text{FF} = \{\delta \in [0,1]^L: \delta \in \mathcal{D}_\text{nf} \text{ and for any other $\delta' < \delta$, we have $\delta' \in \mathcal{D}_\text{f}$ } \}
\]
where $\delta' < \delta$ means that $\delta_\ell' \leq \delta_\ell$ for all $\ell \in \{1,\ldots,L\}$ and $\delta_m' < \delta_m$ for some $m \in \{1,\ldots,L\}$.
\end{definition}

That is, the falsification frontier is the set of assumptions which are not falsified, but if strengthened in any component, leads to a refuted model. In the one dimensional case, this corresponds to the smallest relaxation $\delta$ such that the model is not refuted. For example, suppose $L=2$. For each $\delta_1 \in [0,1]$ define
\[
	\delta_2^*(\delta_1) = \inf \{ \delta_2 \in [0,1] : F_W \in \mathcal{F}_\text{nf}(\delta_1,\delta_2) \}.
\]
This function $\delta_2^*(\cdot)$ plots the falsification frontier in the two-dimensional case. 
Figure \ref{FF_illustration} shows an example of such a function. Our baseline model is refuted and we have picked two different assumptions to focus on as possible explanations. The horizontal axis measures the relaxation of the first assumption. The vertical axis measures the relaxation of the second assumption. The origin at the lower left represents the baseline model. The point at the top right of the box represents the model where neither of the two baseline assumptions are imposed. First, look at the top left and lower right corners. We see that if we completely drop one assumption while maintaining the other, the model is no longer falsified. However, we can learn more by studying the points at which the model becomes non-refutable---that is, by examining the falsification frontier. This is shown as the boundary between the two regions. Consider its vertical and horizontal intercepts. Supposing the units are comparable, then the model continues to be refuted even for large relaxations of assumption 1, while maintaining assumption 2. Conversely, the model is not refuted if we allow moderate relaxations of assumption 2, while maintaining assumption 1. Looking along the frontier, strengthening assumption 2 by some amount requires relaxing assumption 1 by even more, if we want to avoid falsification. Thus the falsification frontier gives the trade-off between different assumptions' ability to falsify the model.

\begin{figure}[t]%
\centering
\includegraphics[width=65mm]{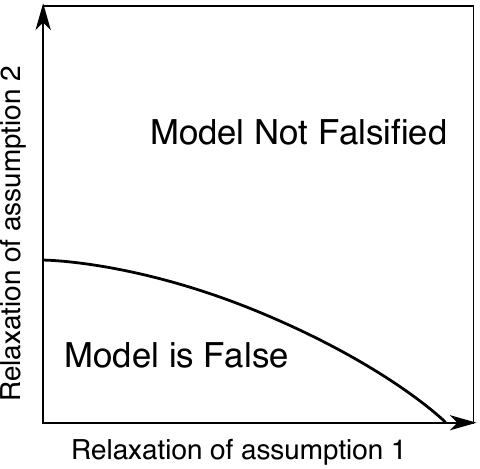}
\caption{An example falsification frontier, partitioning the space of assumptions into the set for which the model is falsified and the set for which the model is not falsified.}
\label{FF_illustration}
\end{figure}

\subsection{Interpreting Falsified and Non-Falsified Models}\label{subsec:FFinterp}

In this section we discuss two well-known points that are nevertheless important to keep in mind when discussing falsification. First, although a falsified model cannot be true, a non-falsified model is not necessarily true. This asymmetry underlies Popper's \citeyearpar{Popper1934} philosophy of falsification. In the classical linear instrumental variables model of section \ref{sec:homogModel}, this is discussed in \cite{KadaneAnderson1977}, \cite{Newey1985}, \cite{Breusch1986},  \cite{Small2007}, \cite{ParenteSantosSilva2012}, and \cite{Guggenberger2012}. They noted that the classical overidentifying conditions are really constraints that the instruments are consistent with each other. Thus it is possible that the instrument exogeneity assumptions fail in such a way that all instruments yield the same biased estimand for the true parameter. In this case, the baseline model is not refuted even though it is false. This is often explained by saying that the classical overidentification test is ``not consistent against all alternatives.'' This is why the falsification frontier is not a partition of true and false assumptions, but rather of falsified and non-falsified assumptions. We discuss this point further in appendix \ref{sec:GMMobjFun}.

Second, when the baseline model is falsified, it could be due to any combination of our baseline assumptions. Importantly, this is true \emph{regardless} of the shape of the falsification frontier. The falsification frontier does \emph{not} tell us why the baseline assumptions were refuted. For example, consider figure \ref{FF_illustration}. This figure does not tell us that `assumption 1 is probably the reason the baseline model was refuted'. Instead, it tells us about the \emph{robustness of falsification} to certain relaxations from the baseline assumptions. It tells us that, maintaining assumption 2, we must substantially relax assumption 1 in order to prevent falsification.

\subsection{The Falsification Adaptive Set}

In section \ref{subsec:generalFF}, we formalized our first recommendation: Measure the extent of falsification. When the falsification frontier is far from the origin in all directions, researchers may wish to stop there. In other cases, researchers may want to proceed and present identified sets for non-falsified relaxations of the baseline model.

Let $\Theta_I(\delta)$ denote the identified set for a parameter of interest $\theta \in \Theta$, given the model which imposes the assumptions $\delta$. When $\delta \in \mathcal{D}_\text{f}$, $\delta$ is below the falsification frontier. In this case, the identified set $\Theta_I(\delta)$ is empty. When $\delta \in \mathcal{D}_\text{nf}$, $\delta$ is on or above the falsification frontier. In this case, the identified set $\Theta_I(\delta)$ is nonempty.

The set of $\delta$'s on the falsification frontier are minimally non-falsified, in the sense that they are the closest to the baseline model while still not leading to a falsified model. These $\delta$'s, along with the $\delta$'s beyond the frontier, are consistent with the data. As discussed in section \ref{subsec:FFinterp}, the data alone cannot tell us which of these models are true. It only tells us that $\delta$'s below the falsification frontier are not true. Indeed, the same remark is true when the baseline model is not falsified. Yet in that case, it is common to present identified sets or point estimands under the baseline model. To generalize this practice to the case where the baseline model is falsified, consider the following definition.

\begin{definition}
Call
\[
	\bigcup_{\delta \in \text{FF}} \Theta_I(\delta)
\]
the \emph{falsification adaptive set}.
\end{definition}

The falsification adaptive set is the identified set for the parameter of interest under the assumption that the true model lies somewhere on the frontier. When the baseline model is not refuted, this set collapses to the $\Theta_I(0_L)$, the baseline identified set (which may be a singleton). This is what researchers typically report when their baseline model is not refuted. When the baseline model is refuted, however, the falsification adaptive set expands to account for uncertainty about which assumption along the frontier is true. Hence this set generalizes the standard baseline estimand to account for possible falsification. We recommend that researchers report this set.

In some cases, researchers may have a prior belief over the relative role of the various assumptions in falsification. In this case, researchers may also want to present $\Theta_I(\delta)$ for a specific $\delta$ on the falsification frontier. For example, consider figure \ref{FF_illustration}. Two natural relaxations of the baseline model are the horizontal intercept $(\delta_1^*,0)$ and the vertical intercept $(0,\delta_2^*)$. These correspond to fully maintaining one assumption while sufficiently relaxing the other. The corresponding identified sets are $\Theta_I(\delta_1^*,0)$ and $\Theta_I(0,\delta_2^*)$. Presenting sets like these for specific points on the frontier is our third recommendation.

\subsubsection*{Cautious Optimism}

Although the points on the falsification frontier are consistent with the data, this does not mean one of them is true. For this reason, the falsification adaptive set---although it represents uncertainty due to the nature of the misspecification---is nonetheless an optimistic set. It assumes that, although the model is misspecified, it is only minimally misspecified. Since this may not be true, our fourth and final recommendation is that researchers present identified sets $\Theta_I(\delta)$ for points $\delta$ beyond the falsification frontier, as a sensitivity analysis. We discuss this recommendation in the next subsection.

The widespread practice in existing empirical work is to first present results from a usually optimistic baseline model---which typically point identifies a parameter of interest---and then to present results from various sensitivity analyses. Our recommendations directly generalize this standard practice to falsified baseline models: First present the falsification adaptive set and then present identified sets for further relaxations.

\subsection{Breakdown Frontier Analysis}\label{subsec:FFvsBF}

Our fourth recommendation is to present identified sets $\Theta_I(\delta)$ for $\delta$'s beyond the falsification frontier. To guide this analysis, researchers can use the breakdown frontier concept which we previously studied in \cite{MastenPoirier2017BF}. In this subsection we briefly summarize this concept and relate it to the falsification frontier. Consider the model $\mathcal{M}(\delta)$ from section \ref{subsec:generalFF}. As above, let $\Theta_I(\delta) \subseteq \Theta$ be the identified set for some parameter of interest $\theta \in \Theta$. Let $\mathcal{C} \subseteq \Theta$. Suppose we are interested in the conclusion that $\theta \in \mathcal{C}$. For example, if our parameter is the average treatment effect, then we may be interested in the conclusion that the average treatment effect is positive, $\text{ATE} \geq 0$. Then $\mathcal{C} = [0,\infty)$. Define the \emph{robust region}
\[
	\text{RR} = \{ \delta \in [0,1]^L : \Theta_I(\delta) \subseteq \mathcal{C}, \Theta_I(\delta) \neq \emptyset \}.
\]
The conclusion of interest holds for all assumptions in this set. That is, for all assumptions such that all elements of the identified set are also elements of $\mathcal{C}$. We also restrict the set to assumptions such that the model is not refuted. The breakdown frontier is the boundary between the robust region and its complement:
\[
	\text{BF} = \overline{\text{RR}} \cap \overline{\text{RR}^c}.
\]
In the two dimensional case, we can write the breakdown frontier as the function
\[
	\text{BF}(\delta_1) = \sup \{ \delta_2 \in [0,1] : \Theta_I(\delta_1,\delta_2) \subseteq \mathcal{C}, \Theta_I(\delta_1,\delta_2) \neq \emptyset \}.
\]

It is possible for the robust region to be empty. This happens when the baseline model is falsified, but as soon as the model is not falsified, the identified set $\Theta_I(\delta)$ contains values of $\theta$ outside of $\mathcal{C}$.
When the robust region is not empty, so that the breakdown frontier is not empty, the breakdown frontier will always be weakly larger than the falsification frontier. This follows since the identified set is empty for all points below the falsification frontier.

Figure \ref{FF_illustration2} depicts an analysis combining both the falsification frontier and the breakdown frontier. Here the robust region is a band between the falsification frontier and the breakdown frontier. It is the region of assumptions which are weak enough that the model is not falsified, but strong enough that our conclusion of interest holds.

\begin{figure}[t]%
\centering
\includegraphics[width=65mm]{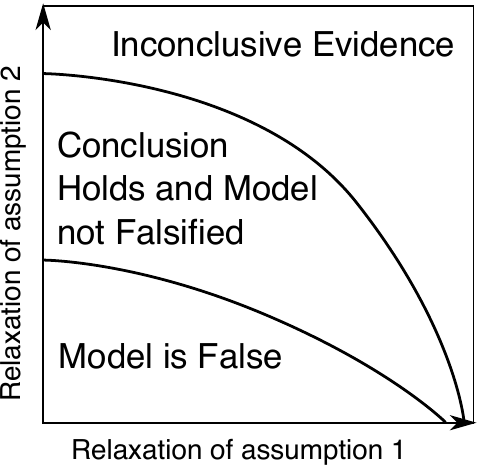}
\caption{An example of using the falsification frontier jointly with a breakdown frontier for a specific conclusion of interest. Here there is a band of assumptions which are sufficiently weak that the model is not falsified, but are sufficiently strong that the conclusion of interest holds. This band is the robust region.}
\label{FF_illustration2}
\end{figure}

A key difference between the breakdown and falsification frontiers is that the falsification frontier does not depend on (a) the parameter of interest or (b) the conclusion of interest. For a given class of relaxations from the baseline assumptions, there is only one falsification frontier. In contrast, there are many different breakdown frontiers, depending on the parameter and conclusion of interest. This difference is important since there have been extensive debates in the literature regarding what parameters researchers should study (for example, see \citealt{Imbens2010} and \citealt{HeckmanUrzua2010}). Thus one's opinions on this issue do not affect the problem of measuring the extent of falsification.

\subsection{Estimation and Inference}\label{subsec:estimationInference}

In this paper we focus on population level analysis. In this section we briefly discuss how to implement our recommendations with finite sample data. First, to measure the extent of falsification, one can estimate and do inference on the falsification frontier. The basic idea is similar to our analysis of inference on breakdown frontiers (\citealt{MastenPoirier2017BF}). For the breakdown frontier, we recommended constructing lower confidence bands. For the falsification frontier, we recommend constructing \emph{upper} confidence bands, since this provides an outer confidence set for the area under the falsification frontier---the set of models which are refuted. When doing a combined analysis, we recommend constructing an inner confidence set for the robust region; that is, for the area above the falsification frontier but below the breakdown frontier. This will require constructing two-sided confidence bands, whereas one only needs to construct one-sided confidence bands when considering each frontier separately.

When there are closed form expressions for the identified sets $\Theta_I(\delta)$---like in our analysis below---the falsification adaptive set can be estimated by sample analog:
\[
	\widehat{\text{FAS}} = \bigcup_{\delta \in \widehat{\text{FF}}} \widehat{\Theta}_I(\delta).
\]
In the linear instrumental variable model of section \ref{sec:homogModel}, we provide a particularly simple characterization of the falsification adaptive set which does not require pre-estimation of the falsification frontier. In that case, estimation and inference is straightforward. Recall that, by definition, the falsification adaptive set simplifies to the standard estimand when the baseline model is not refuted. Thus there is no need to do a specification pre-test---researchers can just always present the estimated falsification adaptive set with accompanying confidence intervals. However, note that in finite samples the estimated falsification adaptive set will generally always be an interval. This situation is similar to \citet[page 10]{HaileTamer2003}, who note that their population bounds collapse to a single point in a special case, although in finite samples their estimator still generally gives an interval.

\section{The Classical Linear Model with Multiple Instruments}\label{sec:homogModel}

In this section and section \ref{sec:hetModel}, we illustrate our recommendations using two variations of instrumental variable models. While many kinds of refutable assumptions have been considered in the literature, we focus on the classical case where variation from two or more instruments is used to falsify a model. In this section we consider the classical linear model with homogeneous treatment effects and continuous outcomes. In section \ref{sec:hetModel} we consider a model with heterogeneous treatment effects.

\subsection{Model and Identification}

We thus begin with the classical constant coefficient model
\begin{equation}\label{eq:constantCoeffOutcomeEqGen}
	Y(x,z) = x'\beta + z'\gamma + U
\end{equation}
where $Y(x,z)$ are potential outcomes defined for values $(x,z) \in \R^{K+L}$ and $U$ is an unobserved random variable. $\beta$ is an unknown constant $K$-vector. $\gamma$ is an unknown constant $L$-vector. Let $X$ be an observed $K$-vector of endogenous variables. Throughout we suppose $X$ does not contain a constant. Hence $U$ absorbs any nonzero constant intercept. Let $Z$ be an observed $L$-vector of potentially invalid instruments. For simplicity we have omitted any additional known exogenous covariates $W$ in equation \eqref{eq:constantCoeffOutcomeEqGen}. In appendix \ref{sec:covariatesInLinearModel} we show how to easily include them via partialling out.

We observe the outcome $Y = Y(X,Z)$. Thus our equation for observed outcomes is
\[
	Y = X'\beta + Z'\gamma + U.
\]
Equation \eqref{eq:constantCoeffOutcomeEqGen} imposes homogeneous treatment effects since the difference between two potential outcomes
\[
	Y(x_1,z) - Y(x_0,z) = (x_1 - x_0)' \beta
\]
is constant across the population. We consider a model with heterogeneous treatment effects in section \ref{sec:hetModel}.

\begin{remark}\label{remark:LinearHetTrtEffects}
In this section we present the linear model with homogeneous treatment effects. It is well known, however, that linear model analysis can be generalized to allow for heterogeneous treatment effects under some assumptions. For example, suppose $\beta$ is a random coefficient. \cite{HeckmanVytlacil1998} and \cite{Wooldridge1997, Wooldridge2003, Wooldridge2008} showed that if the causal effect of $Z$ on $X$ is homogeneous, then the 2SLS estimand equals $\Exp(\beta)$. Alternatively, we can allow for heterogeneous instrument effects if we assume these are independent of $\beta$; for example, see the discussion following equation (11) in \citet[page 1142]{Card2001}. Under these kinds of assumptions, all of our results in this section extend to allow for heterogeneous treatment effects. Our analysis of heterogeneous treatment effects in section \ref{sec:hetModel} applies when we are not willing to impose assumptions that restrict the nature of the heterogeneous instrument effects. \citet[page 261]{ConleyHansenRossi2012} make a similar observation.
\end{remark}

In addition to this homogeneous treatment effect assumption, we maintain the following relevance and sufficient variation assumptions throughout this section.

\begin{homogAssump}[Relevance]\label{assump:homog:relevance:gen}
The $L\times K$ matrix $\cov(Z,X)$ has rank $K$.
\end{homogAssump}

\begin{homogAssump}[Sufficient variation]\label{assump:homog:nonsing:gen}
The $L\times L$ matrix $\var(Z)$ is invertible.
\end{homogAssump}

A\ref{assump:homog:relevance:gen} implies the order condition $L \geq K$. When there is just one endogenous variable ($K=1$), A\ref{assump:homog:relevance:gen} only requires $\cov(X,Z_\ell) \neq 0$ for at least one instrument. Other instruments may have zero correlation. In this case, these other instruments provide additional falsifying power. We discuss this further below. A\ref{assump:homog:nonsing:gen} simply requires that the instruments are not degenerate and are not affine combinations of each other.

The classical model imposes two more assumptions:

\begin{homogAssump}[Exogeneity]\label{assump:exogeneity:gen}
$\cov(Z_\ell,U) = 0$ for all $\ell \in \{1,\ldots,L \}$.
\end{homogAssump}

\begin{homogAssump}[Exclusion]\label{assump:exclusion:gen}
$\gamma_\ell = 0$ for all $\ell \in \{1,\ldots,L \}$.
\end{homogAssump}

A\ref{assump:homog:relevance:gen}--A\ref{assump:exclusion:gen} imply that the coefficient vector $\beta$ is point identified and equals the two-stage least squares (2SLS) estimand. Furthermore, these assumptions imply well-known overidentifying conditions. The following proposition gives these conditions when there is just a single endogenous variable.

\begin{proposition}\label{prop:testableImplicationsOfLinearIV}
Consider model \eqref{eq:constantCoeffOutcomeEqGen}. Suppose A\ref{assump:homog:relevance:gen}--A\ref{assump:exclusion:gen} hold. Suppose $K=1$. Suppose the joint distribution of $(Y,X,Z)$ is known. Then the model is not refuted if and only if
\begin{equation}\label{eq:classicalSarganEqs}
	\cov(Y,Z_m) \cov(X,Z_\ell)
	= 
	\cov(Y,Z_\ell) \cov(X,Z_m)
\end{equation}
for all $m$ and $\ell$ in $\{1,\ldots,L\}$.
\end{proposition}

When all instruments are relevant, so that $\cov(X,Z_\ell) \neq 0$ for all $\ell \in \{1,\ldots,L\}$, equation \eqref{eq:classicalSarganEqs} can be written as
\[
	\frac{\cov(Y,Z_m)}{\cov(X,Z_m)} 
	= 
	\frac{\cov(Y,Z_\ell)}{\cov(X,Z_\ell)}.
\]
That is, the linear IV estimand must be the same for all instruments $Z_\ell$. This result is the basis for the classical test of overidentifying restrictions (\citealt{AndersonRubin1949}, \citealt{Sargan1958}, \citealt{Hansen1982}). Suppose the distribution of $(Y,X,Z)$ is such that the model is refuted. This happens when at least one of our model assumptions fails: (a) homogeneous treatment effects, (b) linearity in $X$, (c) instrument exogeneity, or (d) instrument exclusion.

In this section we maintain the homogeneous treatment effects assumption and focus on instrument exclusion or exogeneity as an explanation for refutation. We consider heterogeneous treatment effects in section \ref{sec:hetModel}. We also maintain linearity of potential outcomes in $X$. In principle our analysis can be extended to allow for relaxations of the linearity assumption, but we leave this to future work. Note that linearity holds automatically when $X$ is binary.

We thus focus on failure of (c) instrument exogeneity or (d) instrument exclusion as reasons for refuting the baseline model. These are two different substantive assumptions. Mathematically, however, we can use the same analysis to relax both assumptions. For simplicity, here we formally maintain the exogeneity assumption A\ref{assump:exogeneity:gen} and focus on failure of the exclusion assumption A\ref{assump:exclusion:gen}. In appendix \ref{sec:ExoExclu} we explain how to use our results to relax either or both assumptions.

We relax exclusion as follows.

\begin{assump}{A\ref{assump:exclusion:gen}$^\prime$}(Partial exclusion).
There are known constants $\delta_\ell \geq 0$ such that 
\[
	| \gamma_\ell | \leq \delta_\ell
\]
for all $\ell \in \{1,\ldots,L \}$.
\end{assump}

This kind of relaxation of the baseline instrumental variable assumptions was previously considered by \cite{Small2007} and \cite{ConleyHansenRossi2012}; also see \cite{AngristKrueger1994} and \cite{BoundJaegerBaker1995}. For continuous instruments one may want to standardize the instruments to have variance one, to make the magnitudes of the components in $\delta = (\delta_1,\ldots,\delta_L)$ comparable. We discuss interpretation of $\delta_\ell$ further in section 3.6 of \cite{MastenPoirierFFarxiv}; our empirical analysis in section \ref{sec:empirical} does not require us to select or calibrate this parameter. %

Under partial exclusion, the instruments may have a direct causal effect on outcomes. For sufficiently small values of the components in $\delta$, the model may nonetheless continue to be refuted. For sufficiently large values, however, the model will not be refuted. To characterize the falsification frontier, we begin by deriving the identified set for $\beta$ as a function of $\delta$. 

\begin{theorem}\label{thm:idset:homog:gen}
Suppose A\ref{assump:homog:relevance:gen}--A\ref{assump:exogeneity:gen} and A\ref{assump:exclusion:gen}$^\prime$ hold. Suppose the joint distribution of $(Y,X,Z)$ is known. Then
\begin{equation}
		\mathcal{B}(\delta) = \left\{ b \in \R^K : -\delta \leq \var(Z)^{-1}(\cov(Z,Y) - \cov(Z,X)b ) \leq \delta \right\}
\end{equation}
is the identified set for $\beta$. Here the inequalities are component-wise. The model is refuted if and only if this set is empty.
\end{theorem}

The identified set $\mathcal{B}(\delta)$ depends on the data via two terms:
\[\label{eq:defOfPsiPi}
	\underset{(L \times 1)}{\psi} \equiv \var(Z)^{-1}\cov(Z,Y)
	\qquad \text{and} \qquad
	\underset{(L \times K)}{\Pi} \equiv \var(Z)^{-1}\cov(Z,X).
\]
$\psi$ is the reduced form regression of $Y$ on $Z$. $\Pi$ is the first stage of $X$ on $Z$. If we demeaned $(Y,X,Z)$ then we would have $\psi = \Exp(ZZ')^{-1} \Exp(ZY)$ and $\Pi = \Exp(ZZ')^{-1} \Exp(ZX')$. Theorem \ref{thm:idset:homog:gen} shows that the identified set is the intersection of $L$ pairs of parallel half-spaces in $\R^K$. When $\delta = 0$, this identified set becomes the intersection of $L$ hyperplanes in $\R^K$. In this case, $\beta$ is point identified when 
\[
	\cov(Z,Y) = \cov(Z,X) b
\]
for a unique $b \in \R^K$. If $\cov(Z,Y) \neq \cov(Z,X)b$ for all $b \in \R^K$, then the baseline model $\delta = 0$ is refuted.

Increasing the components of $\delta$ leads to a weakly larger identified set. Furthermore, there always exists a $\delta$ with large enough components so that $\mathcal{B}(\delta)$ is nonempty. We characterize the set of such $\delta$ below. Before that, we show that the identified set can be written as simple intersection bounds when there is a single endogenous variable.

\begin{corollary}\label{corr:K1identBetaSetLinearIV}
Suppose the assumptions of theorem \ref{thm:idset:homog:gen} hold. Suppose $K=1$. Then
\[
	\mathcal{B}(\delta) = \bigcap_{\ell=1}^L B_\ell(\delta_\ell)
\]
is the identified set for $\beta$, where
\begin{equation}\label{eq:K1identBetaSetLinearIV}
	B_\ell(\delta_\ell)
	= \begin{cases}
	\left[\dfrac{\psi_\ell}{\pi_\ell} - \dfrac{\delta_\ell}{|\pi_\ell|}, \dfrac{\psi_\ell}{\pi_\ell} + \dfrac{\delta_\ell}{|\pi_\ell|} \right] 
		&\text{ if } \pi_\ell \neq 0\\
		\R 
		&\text{ if } \pi_\ell = 0 \text{ and } 0 \in [\psi_\ell - \delta_\ell, \psi_\ell + \delta_\ell] \\
		\emptyset 
		&\text{ if } \pi_\ell = 0 \text{ and } 0 \notin [\psi_\ell - \delta_\ell, \psi_\ell + \delta_\ell].
	\end{cases}
\end{equation}
Here $\pi_\ell$ is the $\ell$th component of $\Pi$, which is an $L$-vector.
\end{corollary}

To interpret this result, first consider an instrument $Z_\ell$ with a zero first stage coefficient, $\pi_\ell = 0$. If $Z_\ell$ has a sufficiently large covariance with the outcome, so that $\psi_\ell \pm \delta_\ell$ does not contain zero, then the model is refuted. Furthermore, in this case falsification can be solely attributed to the assumption that $| \gamma_\ell | \leq \delta_\ell$. This is similar to what is sometimes called the `zero first-stage test' (for example, see \citealt{Slichter2014} and the references therein). When this covariance with the outcome is sufficiently small, however, $Z_\ell$ unsurprisingly has no falsifying or identifying power for $\beta$.

Next consider a relevant instrument $Z_\ell$; so $\pi_\ell \neq 0$. To interpret corollary \ref{corr:K1identBetaSetLinearIV} in this case, we use the following lemma.

\begin{lemma}\label{lemma:interpretingPsiOverPi}
Suppose $K=1$. Let $\widetilde{X}_\ell = (Z_1,\ldots,Z_{\ell-1},X,Z_{\ell+1},\ldots,Z_L)$. Let $e_\ell$ be the $L \times 1$ vector of zeros with a one in the $\ell$th component. Suppose $\pi_\ell \neq 0$. Suppose $\cov(Z, \widetilde{X}_\ell)$ is invertible. Then
\[
	\frac{\psi_\ell}{\pi_\ell} = e_\ell'\cov(Z, \widetilde{X}_\ell)^{-1} \cov(Z,Y).
\]
\end{lemma}

This lemma shows that $\psi_\ell / \pi_\ell$ is the population 2SLS coefficient on $X$ using $Z_\ell$ as the excluded instrument and $Z_{-\ell}$ as controls. Thus the identified set $\mathcal{B}(\delta)$ is the intersection of intervals around these 2SLS coefficients using one relevant instrument at a time and controlling for the rest.

Finally, consider the baseline case where $\delta = 0$. Corollary \ref{corr:K1identBetaSetLinearIV} implies that $\mathcal{B}(0)$ is nonempty if and only if
\[
	\frac{\psi_m}{\pi_m} = \frac{\psi_\ell}{\pi_\ell}
\]
for any $m, \ell \in \{ 1,\ldots, L \}$ with $\pi_m, \pi_\ell \neq 0$. Moreover, in this case $\mathcal{B}(0)$ is a singleton equal to this common value. In this case---when the baseline model is not refuted---we also have
\begin{equation}\label{eq:2SLSwithAndWithoutControls}
	\frac{\psi_\ell}{\pi_\ell} = \frac{\cov(Y,Z_\ell)}{\cov(X,Z_\ell)}
\end{equation}
for all $\ell \in \{1,\ldots,L\}$. That is, $\psi_\ell / \pi_\ell$ equals the population 2SLS coefficient on $X$ using $Z_\ell$ as an instrument and \emph{not} including $Z_{-\ell}$ as controls. This equality of single instrument 2SLS coefficients with and without controls for the other instruments is an alternative characterization of the classic overidentifying conditions from proposition \ref{prop:testableImplicationsOfLinearIV}.

\subsection{The Falsification Frontier}\label{sec:linearModelFF}

So far we have characterized the identified set for $\beta$ given a fixed value of $\delta$, the upper bound on the violation of the exclusion restriction. We now consider the possibility that this identified set is empty when $\delta = 0$, so that the baseline model is falsified. Our next result characterizes the falsification frontier, the minimal set of $\delta$'s which lead to a non-empty identified set. Here we focus on the single endogenous regressor case. We extend this result to multiple endogenous regressors in section \ref{sec:linearModelGeneralK}.

\begin{proposition}\label{prop:K1FFhomogTrt}
Suppose A\ref{assump:homog:relevance:gen}--A\ref{assump:exogeneity:gen} hold. Suppose the joint distribution of $(Y,X,Z)$ is known. Suppose $K=1$. Then the falsification frontier is the set
\begin{equation}\label{eq:linearIVgeneralFF}
	\text{FF}
	= \left\{ \delta \in \R^L_{\geq 0}: \delta_\ell = | \psi_\ell - b \pi_\ell |, \ \ell=1,\ldots,L, \ b \in \left[\min_{\ell=1,\ldots,L:\pi_\ell \neq 0} \frac{\psi_\ell}{\pi_\ell}, \max_{\ell=1,\ldots,L:\pi_\ell \neq 0} \frac{\psi_\ell}{\pi_\ell}\right]\right\}.
\end{equation}
\end{proposition}

In the proof we show that this set satisfies our definition \ref{def:abstractFF} of the falsification frontier. Specifically: Any $\delta \in \text{FF}$ maps to a non-empty identified set, and strengthening any of the assumptions for a given $\delta \in \text{FF}$ leads to an empty identified set.

If we knew a priori that some instruments are valid, then we could obtain the falsification frontier for relaxing the potentially invalid instrument assumptions by simply setting $\delta_\ell = 0$ for the valid instruments $\ell$.

To better understand proposition \ref{prop:K1FFhomogTrt}, consider the two instrument case. In this case, the following corollary characterizes the falsification frontier.

\begin{corollary}\label{corr:K1L2FFhomogTrt}
Suppose A\ref{assump:homog:relevance:gen}--A\ref{assump:exogeneity:gen} hold. Suppose the joint distribution of $(Y,X,Z)$ is known. Suppose $K=1$ and $L=2$. Suppose $\pi_1 \neq 0$ and $\pi_2 \neq 0$. Then the falsification frontier is the set
\[
	\text{FF} = \left\{ (\delta_1,\delta_2) \in \R_{\geq 0}^2 : \delta_2 = \left| \frac{\psi_1}{\pi_1} - \frac{\psi_2}{\pi_2} \right| | \pi_2 | - \left| \frac{\pi_2}{\pi_1} \right| \delta_1 \right\}.
\]
\end{corollary}

Thus, in the two instrument case, the falsification frontier is simply the line with horizontal intercept $(\delta_1^*,0)$ and vertical intercept $(0,\delta_2^*)$, where
\[
	\delta_1^* = \left| \frac{\psi_1}{\pi_1} - \frac{\psi_2}{\pi_2} \right| | \pi_1 | 
	\qquad \text{and} \qquad
	\delta_2^* = \left| \frac{\psi_1}{\pi_1} - \frac{\psi_2}{\pi_2} \right| | \pi_2 | .
\]
The classical overidentifying restrictions (proposition \ref{prop:testableImplicationsOfLinearIV}) state that the baseline model is refuted if and only if $| \psi_1/\pi_1 - \psi_2/\pi_2 |$ is nonzero. A key aspect of our analysis is that not all refuted models are equivalent. Specifically, two dgps can have the same value of this difference but have different falsification frontiers. 
\begin{figure}[t]
\centering
\includegraphics[width=100mm]{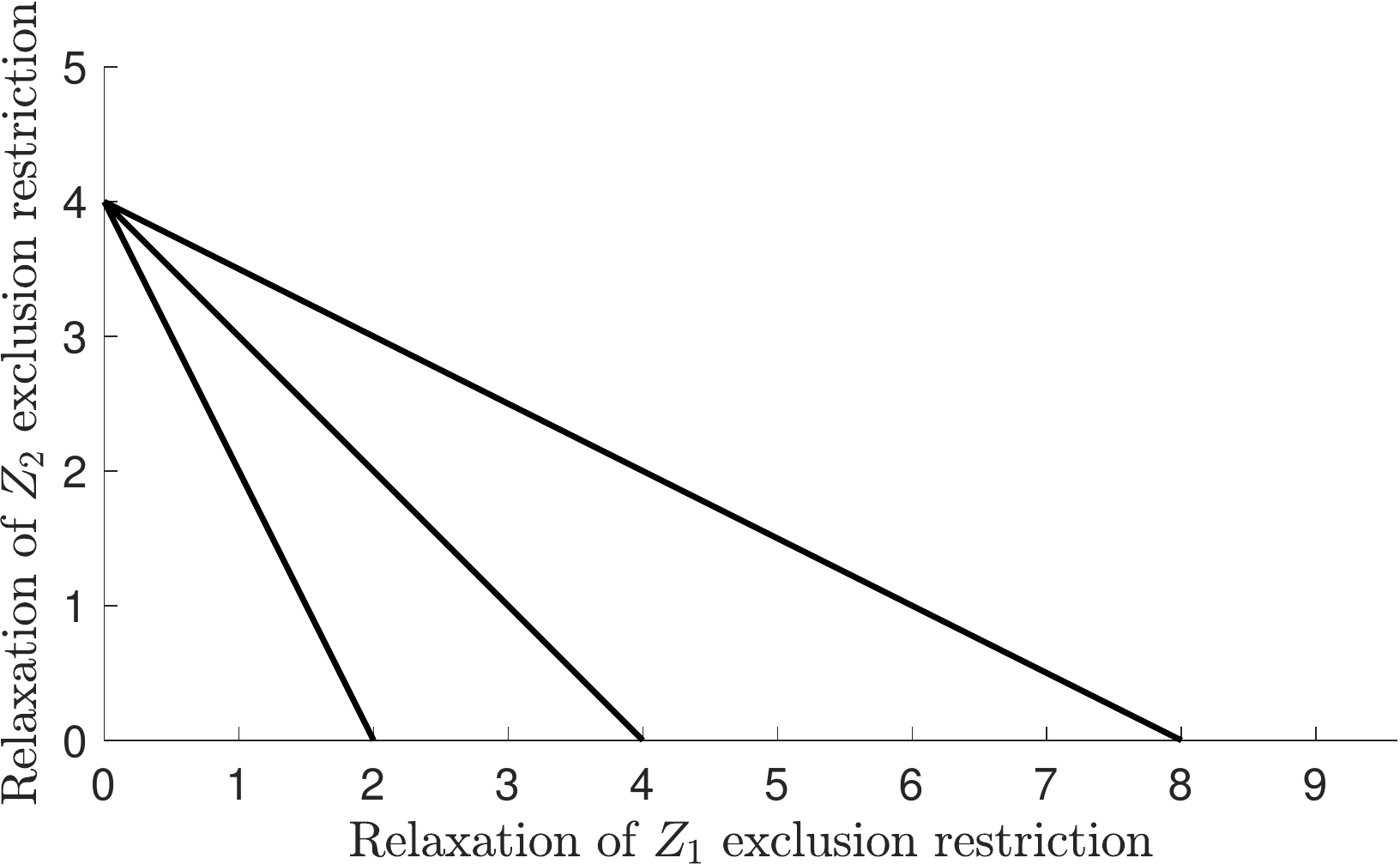}
\caption{Three example falsification frontiers. Middle: Instruments are equally strong. Right: $Z_1$ is two times stronger than $Z_2$. Left: $Z_1$ is half as strong as $Z_2$.}
\label{FFlinearIVexample}
\end{figure}
Figure \ref{FFlinearIVexample} illustrates this point. It shows three example falsification frontiers, corresponding to three different distributions of $(Y,X,Z_1,Z_2)$. All three dgps have
\[
	\left| \frac{\psi_1}{\pi_1} - \frac{\psi_2}{\pi_2} \right| = 4.
\]
The relative strength of the two instruments differs across the dgps, however. Specifically, we set $| \pi_1 / \pi_2 |$ to be either $0.5$, $1$, or $2$. That is, the middle dgp has two instruments with exactly the same first stage coefficient. Hence there is a one-to-one trade off in the falsification frontier. If we weaken $Z_1$ relative to $Z_2$, however, the falsification frontier rotates inwards as the horizontal intercept $\delta_1^*$ gets smaller. Conversely, if we strengthen $Z_1$ relative to $Z_2$, the falsification frontier rotates outwards, as the horizontal intercept $\delta_1^*$ gets larger. That is, a false baseline model is harder to save by relaxing instrument exclusion for a strong instrument, relative to a weak one.

\subsection{The Falsification Adaptive Set}

Thus far we have characterized the identified set for $\beta$ given any assumption $\delta \in \R_{\geq 0}^L$ as well as the falsification frontier. Next we characterize the falsification adaptive set, which is the identified set for $\beta$ under the assumption that one of the points on the falsification frontier is true. Here we again focus on the single endogenous regressor case. We generalize to multiple endogenous regressors in section \ref{sec:linearModelGeneralK}.

\begin{theorem}\label{thm:identSetOnFFhomogTrt}
Suppose A\ref{assump:homog:relevance:gen}--A\ref{assump:exogeneity:gen} hold. Suppose the joint distribution of $(Y,X,Z)$ is known. Suppose $K=1$. Then
\begin{equation}\label{eq:mainFAS}
	\bigcup_{\delta \in \text{FF}} \mathcal{B}(\delta) = \left[\min_{\ell=1,\ldots,L:\pi_\ell \neq 0} \frac{\psi_\ell}{\pi_\ell}, \ \max_{\ell=1,\ldots,L:\pi_\ell \neq 0} \frac{\psi_\ell}{\pi_\ell}\right]
\end{equation}
is the falsification adaptive set.
\end{theorem}

We first sketch the proof of this result and then discuss its implications. It follows from two main steps: First, the identified set $\mathcal{B}(\delta)$ is a singleton for any $\delta \in \text{FF}$ (see lemma \ref{lemma:FFisSingleton} in the appendix). Second, each of these singleton sets corresponds to an element in the interval on the right hand side of equation \eqref{eq:mainFAS} (follows using proposition \ref{prop:K1FFhomogTrt}). Thus we obtain the entire interval by taking the union over all of these singletons.

One of our main recommendations is that researchers report the falsification adaptive set. Theorem \ref{thm:identSetOnFFhomogTrt} shows that, in the classical linear model we consider here, this set has an exceptionally simple form. Most importantly, no $\delta$'s appear on the right hand side of equation \eqref{eq:mainFAS}. This implies that \emph{we can obtain the falsification adaptive set without pre-computing the falsification frontier or selecting any sensitivity parameters}. Furthermore, it is very simple to compute, since it just requires running $L$ different 2SLS regressions.
In appendix \ref{sec:2SLSvsFAS} we show that the baseline 2SLS estimand using all instruments does not have to be inside the falsification adaptive set. In fact, the baseline 2SLS estimand can be arbitrarily far from the FAS.

In finite samples, researchers can present sample analogs or bias-corrected versions of equation \eqref{eq:mainFAS}, along with corresponding confidence sets (e.g., \citealt{ChernozhukovLeeRosen2013}). We further discuss finite sample practice in section \ref{subsec:FASestimationLinearIV}.

In this model, we can also immediately see how this set adapts to falsification of the baseline model. When the baseline model is not false, $\psi_m / \pi_m = \psi_\ell / \pi_\ell$ for all $m, \ell \in \{1,\ldots,L\}$ with nonzero $\pi_m$ and $\pi_\ell$. In this case, the falsification adaptive set collapses to the singleton equal to the common value. This is the same point estimand researchers would usually present when their baseline model is not refuted. As the baseline model becomes more refuted, the values of $\psi_\ell/\pi_\ell$ become more different, and the falsification adaptive set expands. Thus the size of this set reflects the severity of baseline falsification.

\begin{figure}[t]
\centering
\includegraphics[width=50mm]{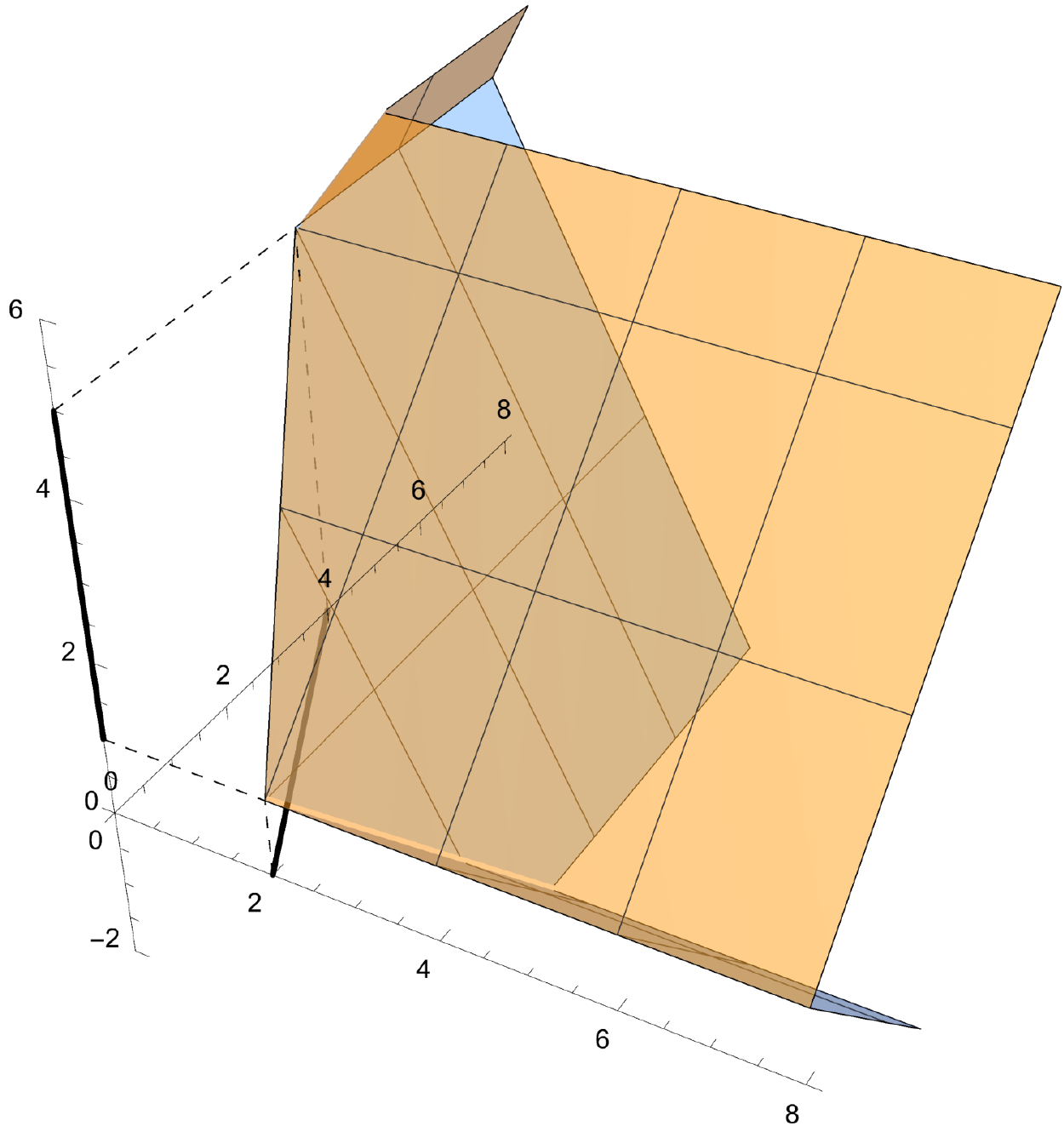}
\hspace{3mm}
\includegraphics[width=50mm]{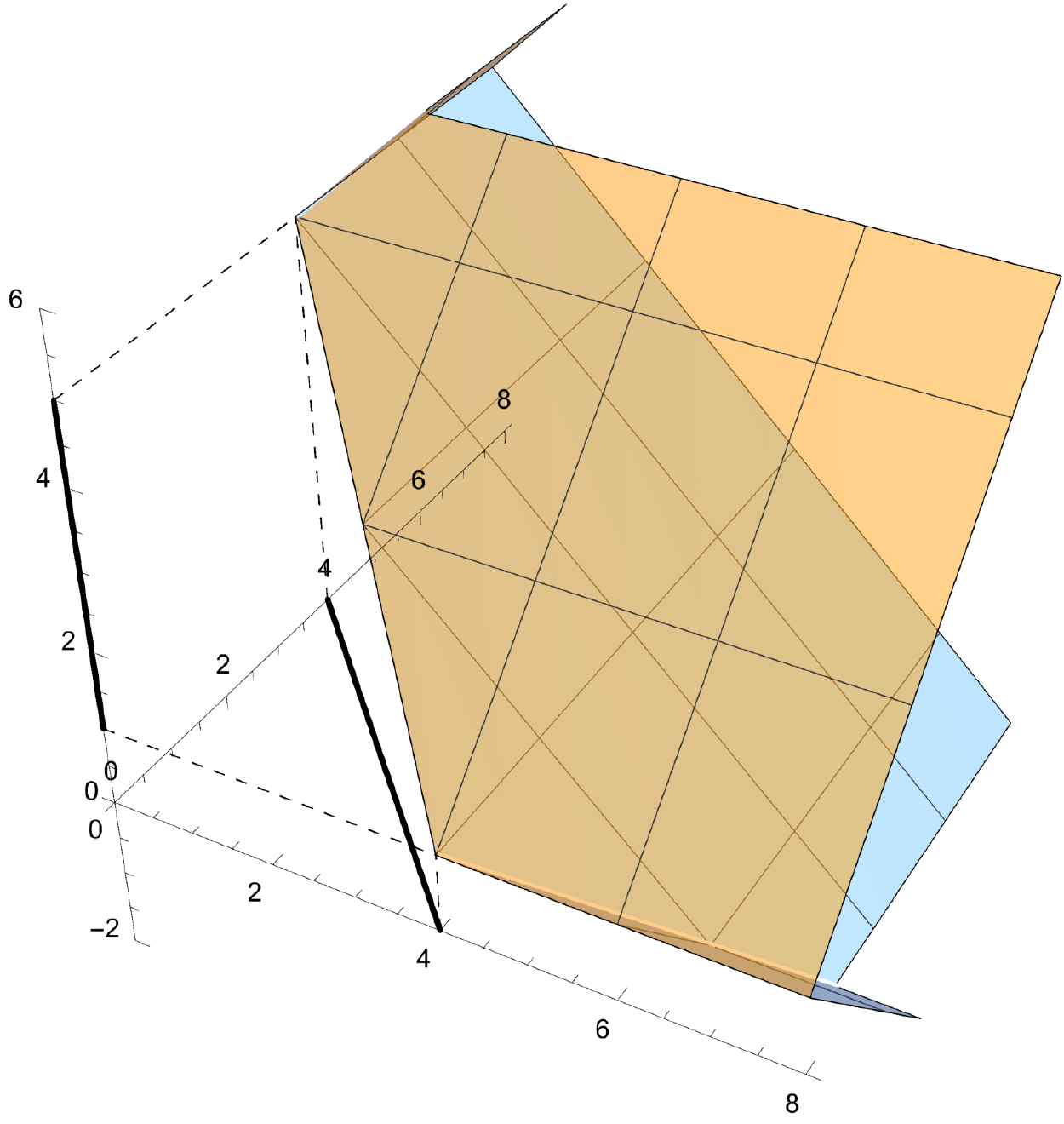}
\hspace{3mm}
\includegraphics[width=50mm]{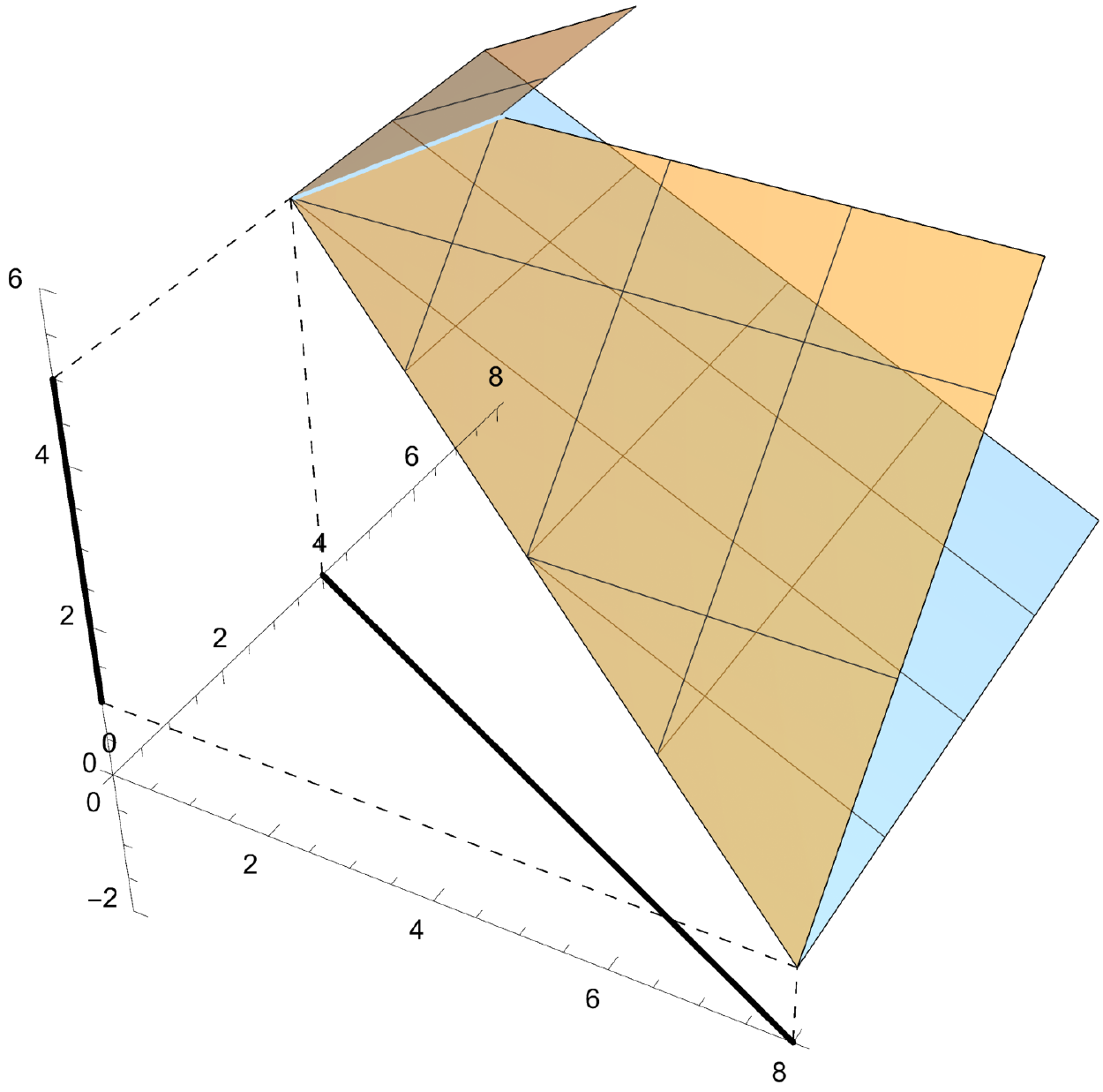} \\
\includegraphics[width=50mm]{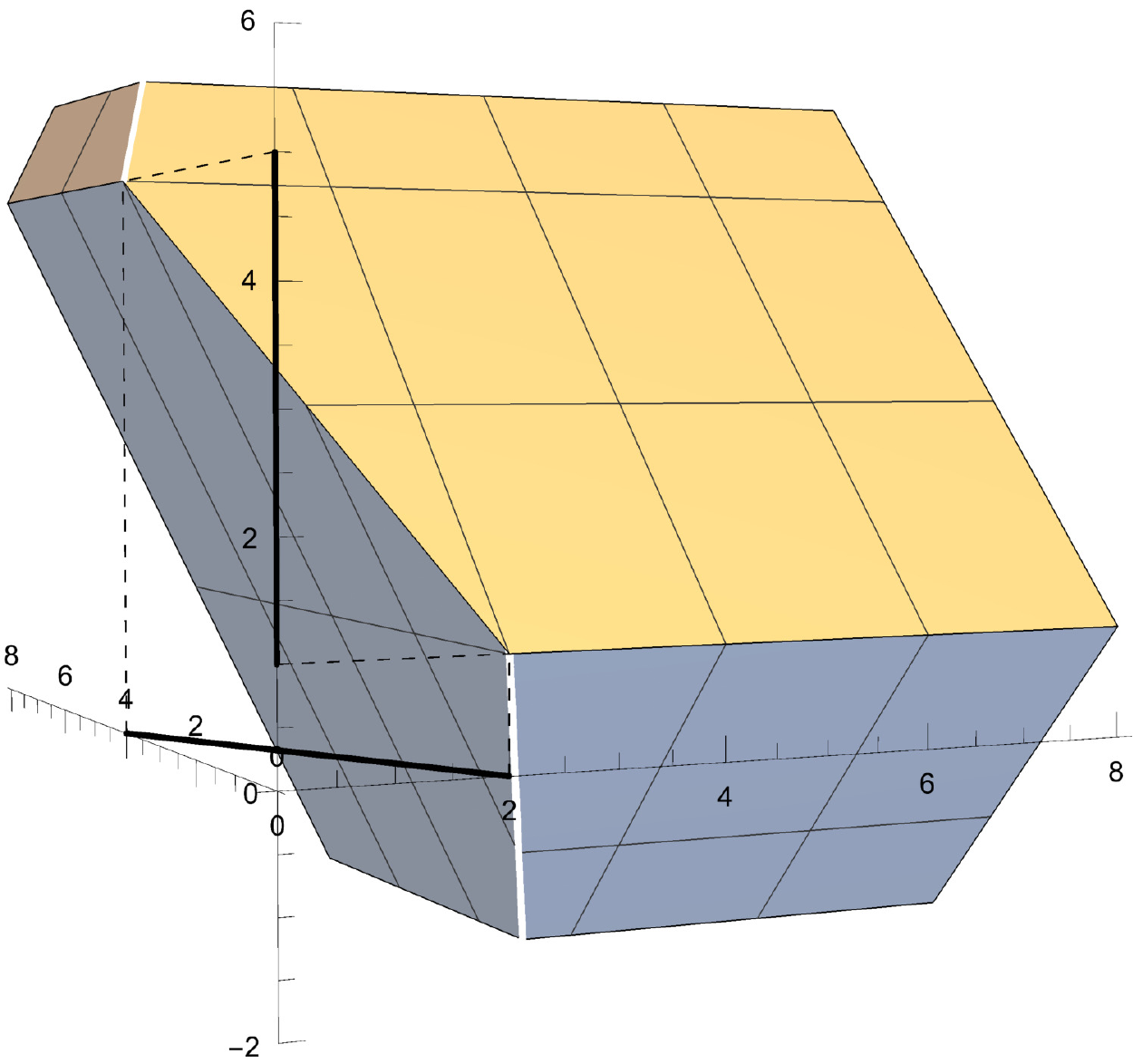}
\hspace{3mm}
\includegraphics[width=50mm]{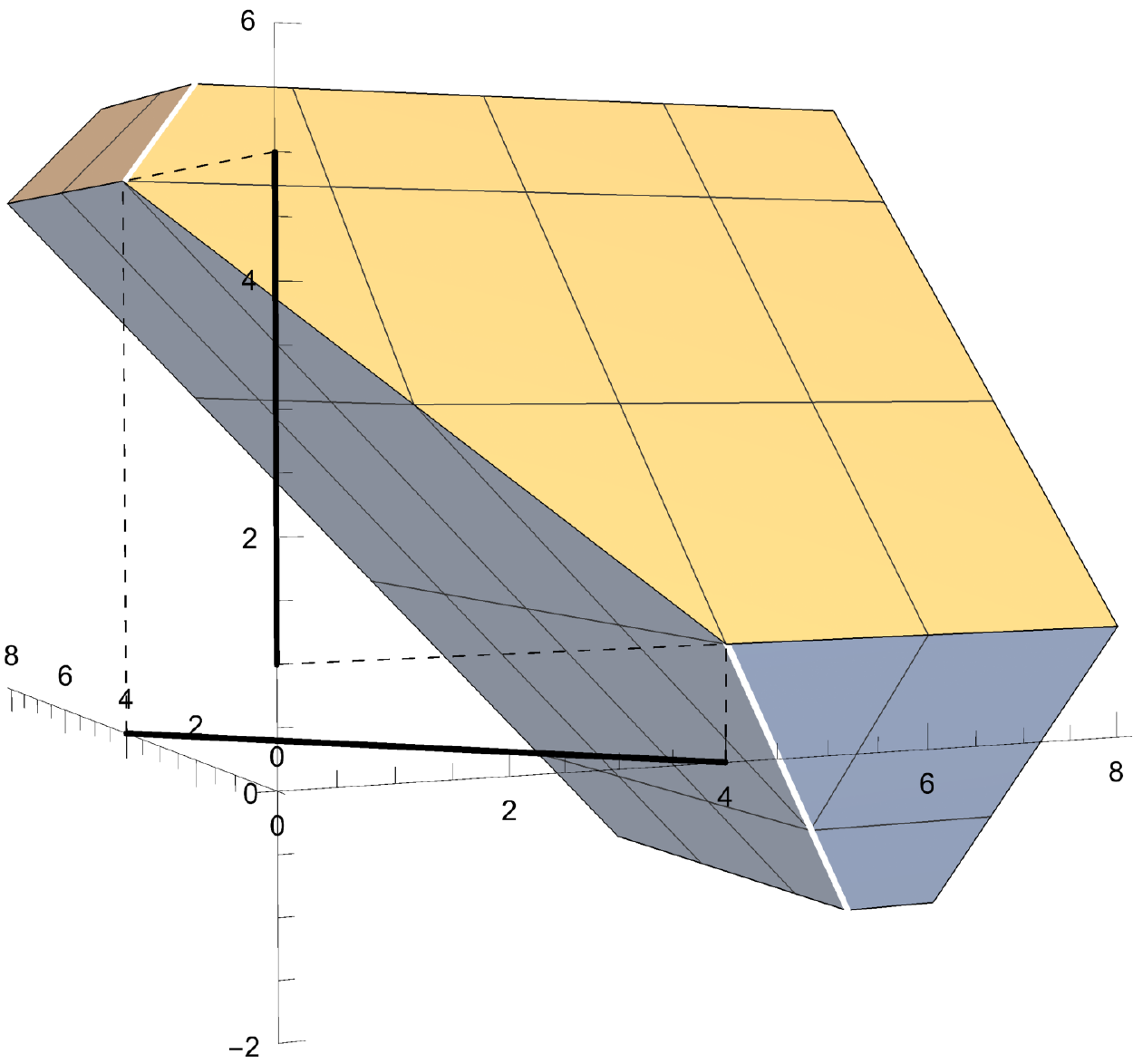}
\hspace{3mm}
\includegraphics[width=50mm]{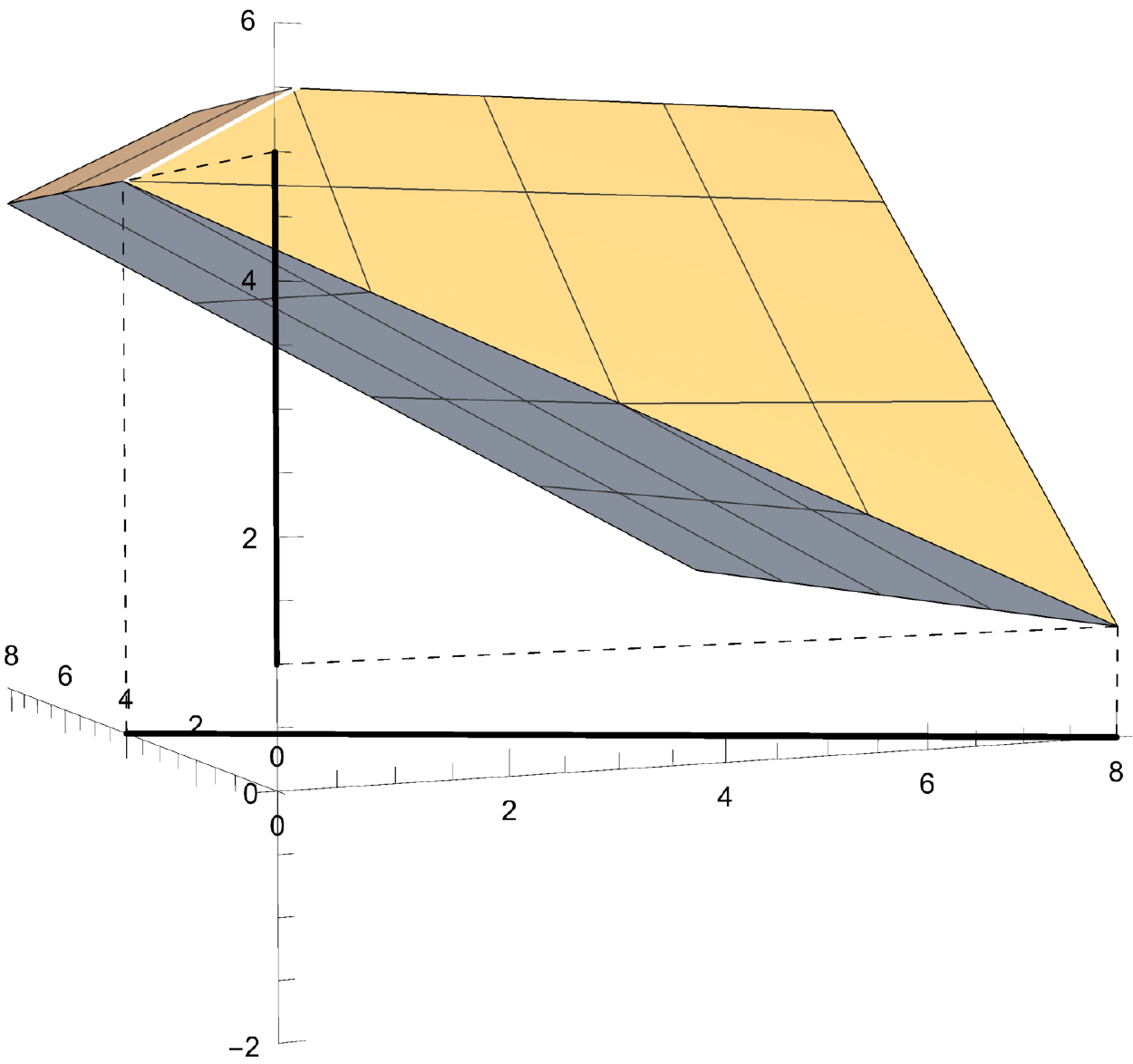}
\caption{These plots simultaneously show the three key objects in our analysis: The identified set $\mathcal{B}(\delta)$, the falsification frontier, and the falsification adaptive set. See the text for a full description.}
\label{fig:FFtwoIV_full_illustration}
\end{figure}

To show how all of the concepts we have discussed fit together, we return to the earlier two instrument numerical illustration discussed earlier. Figure \ref{fig:FFtwoIV_full_illustration} extends figure \ref{FFlinearIVexample} by adding a third dimension: Values of $\beta$. Specifically, each column in figure \ref{fig:FFtwoIV_full_illustration} is a different dgp. Within a column, the top and bottom row are just different viewing angles. For each plot, the horizontal plane shows different values of $(\delta_1,\delta_2)$. The solid line in that plane shows the falsification frontier, as in figure \ref{FFlinearIVexample}. The vertical axis shows the identified set $\mathcal{B}(\delta)$ as a function of $\delta$. For $\delta$ values sufficiently close to zero, this identified set is empty. But for $\delta$'s beyond the falsification frontier, this set is nonempty. For any point $\delta$ on the falsification frontier, this set is a singleton. But the set grows as both $\delta_1$ and $\delta_2$ increase. Finally, the falsification adaptive set is shown as the solid line segment on the vertical axis. It is the union over over the identified sets $\mathcal{B}(\delta)$ as $\delta$ varies over the falsification frontier. Since $\mathcal{B}(\delta_1^*,0) = \{ \psi_2/\pi_2 \}$, $\mathcal{B}(0,\delta_2^*) = \{ \psi_1/\pi_1 \}$, and $\mathcal{B}(\delta)$ varies monotonically as we traverse the falsification frontier, the falsification adaptive set is simply the interval between the points $\psi_1/\pi_1$ and $\psi_2/\pi_2$. In this illustration, these points are $1$ and $5$ for all three dgps and hence the falsification adaptive set is the interval $[1,5]$.

With more than two instruments, we cannot draw the space of $\delta$'s and $\beta$ at the same time. Nonetheless, theorem \ref{thm:identSetOnFFhomogTrt} shows that we can compute the falsification adaptive set by simply taking the convex hull of the $L$ different 2SLS estimands $\psi_\ell / \pi_\ell$.

\subsection{Estimation}\label{subsec:FASestimationLinearIV}

In section \ref{subsec:estimationInference} we discussed estimation and inference in general. Here we discuss the specific case given in theorem \ref{thm:identSetOnFFhomogTrt}. This characterization of the falsification adaptive set requires that we first screen for weak instruments. It is not clear how to best do this. We present a first pass approach, but leave a detailed analysis to future work. Let 
\[
	\mathcal{L} = \{ \ell \in \{1,\ldots,L \} : \pi_\ell \neq 0 \}
\]
be the set of indices corresponding to relevant instruments. Estimate this set by
\[
	\widehat{\mathcal{L}} = \{ \ell \in \{1,\ldots,L \} : F_\ell \geq C_n \}.
\]
$F_\ell$ is the first stage $F$-statistic when considering $Z_\ell$ as an instrument and $Z_{-\ell}$ as controls. $C_n$ is a cutoff that converges to zero as the sample size grows. Specifying the cutoff to shrink ensures that asymptotically we only discard instruments whose coefficients are exactly zero, so that $\widehat{\mathcal{L}}$ is consistent for $\mathcal{L}$.

We then estimate the falsification adaptive set by
\[
	\widehat{\text{FAS}} = \left[\min_{\ell \in \widehat{\mathcal{L}}} 
	\hspace{-1.5mm}
	\widehat{\phantom{\big(}\frac{\psi_\ell}{\pi_\ell}\phantom{\big)}}
	\hspace{-1.5mm}
	, \ 
	\max_{\ell \in \widehat{\mathcal{L}}}
	\hspace{-1.5mm}
	\widehat{\phantom{\big(}\frac{\psi_\ell}{\pi_\ell}\phantom{\big)}}
	\hspace{-1.5mm}
	\right]
\]
where $\widehat{\psi_\ell / \pi_\ell}$ is the estimated 2SLS coefficient on $X$ using $Z_\ell$ as the excluded instrument and $Z_{-\ell}$ as controls. We use this estimator in our empirical analysis of section \ref{sec:empirical}. There we use $C_n = 10$ as our default, although we sometimes consider other cutoffs, or a sequence of cutoffs.

\subsection{Directional Falsification Points}

Beyond presenting the falsification frontier and the falsification adaptive set, researchers may want to also present $\mathcal{B}(\delta)$ for specific choice of model relaxation. We consider the case where the researcher specifies the \emph{direction} of $\delta$ a priori and then chooses the smallest magnitude $\| \delta \|$ such that $\mathcal{B}(\delta) \neq \emptyset$.

Let $\delta = m \cdot d$ for $m \in \R$ and known $d = (d_1,\ldots,d_L)' \in (0,\infty)^L$. For example, when all the instruments are continuous, one reasonable direction may be $d = (\sqrt{ \var(Z_1)}, \ldots, \sqrt{ \var(Z_L)} )'$. 
Given this fixed direction $d$, our assumptions are now parameterized by a single scalar, $m$. The following proposition shows that there is a unique falsification point $m^*$.

\begin{proposition}\label{prop:directionalFalsificationPoint}
Suppose A\ref{assump:homog:relevance:gen}--A\ref{assump:exogeneity:gen} hold. Suppose the joint distribution of $(Y,X,Z)$ is known. Suppose $\delta = m \cdot d$ for known $d = (d_1,\ldots,d_L)' \in (0,\infty)^L$. Suppose $\pi_\ell \neq 0$ for all $\ell \in \{1,\ldots,L \}$. Then
\begin{equation}\label{eq:directionalFalsificationPointMStar}
	m^* 
	= 
	\max_{\ell,\ell' \in \{1,\ldots,L\}} \;
	\frac{\dfrac{\psi_\ell}{\pi_\ell} - \dfrac{\psi_{\ell'}}{\pi_{\ell'}}
	}{
	\dfrac{d_\ell}{|\pi_\ell|} + \dfrac{d_{\ell'}}{|\pi_{\ell'}|}
	}
\end{equation}
is the falsification point. That is, $\mathcal{B}(m^*\cdot d) \neq \emptyset$ and $\mathcal{B}(m\cdot d) = \emptyset$ for all $m < m^*$.
\end{proposition}

The following corollary characterizes the identified set at the point on the falsification frontier in the direction $d$.

\begin{corollary}\label{corr:identSetDirectionalFP}
Suppose the assumptions of proposition \ref{prop:directionalFalsificationPoint} hold.  Define $m^*$ by equation \eqref{eq:directionalFalsificationPointMStar} and let $(\ell^*, \ell^{\prime\ast})$ be the argmax in that equation. Let $\delta^* = m^* \cdot d$. Then $\mathcal{B}(\delta^*)$ is the singleton equal to
\begin{align*}
	\mathcal{B}(\delta^*)
	&= \left\{
	\frac{\psi_{\ell^*}}{\pi_{\ell^*}}  
	-  \frac{m^* \cdot d_{\ell^*}}{|\pi_{\ell^*}|}
	\right\} \\
	&= \left\{
	\frac{\psi_{\ell^{\prime\ast}}}{\pi_{\ell^{\prime\ast}}}  
	+ \frac{m^* \cdot d_{\ell^{\prime\ast}}}{|\pi_{\ell^{\prime\ast}}|} 
	\right\}.
\end{align*}
\end{corollary}

\subsection{The FF and FAS for $K$ Endogenous Variables}\label{sec:linearModelGeneralK}

Theorem \ref{thm:idset:homog:gen} characterizes the identified set for the vector of coefficients on the endogenous variables, as a function of the exclusion restriction relaxation. Our subsequent characterizations of the falsification frontier and the falsification adaptive set, however, restricted attention to the case with just one endogenous variable---see proposition \ref{prop:K1FFhomogTrt} and theorem \ref{thm:identSetOnFFhomogTrt}. In this section, we extend these two results to the general case with $K \geq 1$ endogenous variables. These results can be used in at least three different cases: (1) If there are multiple `basic' endogenous variables. For example, in our empirical application in section \ref{sec:NevoAnalysis}, one could allow both price and advertising spending to be endogenous. (2) If the outcome equation has interactions of one basic endogenous variable with covariates. (3) If the outcome equation is nonlinear in one basic endogenous variable. For example, it could be a quadratic function.

Recall our notation for the reduced form and first stage regressions:
\[
	\underset{(L \times 1)}{\psi} \equiv \var(Z)^{-1}\cov(Z,Y)
	\qquad \text{and} \qquad
	\underset{(L \times K)}{\Pi} \equiv \var(Z)^{-1}\cov(Z,X).
\]
$L$ denotes the number of instruments while $K$ denotes the number of endogenous variables. Let $\pi_\ell'$ denote the $\ell$th row of the matrix $\Pi$. When $K = 1$, this is a scalar and $\pi_\ell' = \pi_\ell$.

When there is a single endogenous variable, theorem \ref{thm:identSetOnFFhomogTrt} shows that the falsification adaptive set is the interval
\[
	\left[\min_{\ell=1,\ldots,L:\pi_\ell \neq 0} \frac{\psi_\ell}{\pi_\ell}, \ \max_{\ell=1,\ldots,L:\pi_\ell \neq 0} \frac{\psi_\ell}{\pi_\ell}\right].
\]
This interval can be interpreted as follows: Pick any set of $K$ instruments. Impose the exclusion restriction for those instruments. For the remaining instruments, completely drop the exclusion restriction. This yields a non-refutable model that is weaker than the original model which imposed exclusion for all of the instruments. But this weaker model is still strong enough that $\beta$ is point identified (ignoring the possibility of irrelevant instruments, which we discuss below). Moreover, in this weaker model $\beta$ equals the population 2SLS coefficient on $X$ using the chosen $K$ instruments as the excluded instruments and the remaining instruments as controls. Collect all of these different 2SLS estimands, as we vary which set of $K$ instruments we pick. Then take their convex hull.

When $K=1$, this convex hull is precisely the interval above: We pick a single relevant instrument to exclude, and include the others as controls. This gives us a single just-identified 2SLS point estimand. By cycling through which instrument we exclude, we obtain a variety of different point estimands. Since $K=1$, these estimands are scalars. So their convex hull is simply the interval between the smallest and largest values.

When $K > 1$ and $L = K + 1$, the falsification adaptive set is precisely the convex hull of just-identified 2SLS estimands that we've just described. For $L > K + 1$, it is somewhat more complicated. We discuss these general results next.

In the $K=1$ case we allowed for irrelevant instruments. For $K>1$ we focus on the case where all instruments are relevant for simplicity. Specifically, we impose assumption A\ref{assump:homog:relevance:gen}$^\prime$ below. To state this assumption, we consider submatrices of $\Pi$. Let $\mathcal{L} \subseteq \{ 1,\ldots, L \}$. Let $\Pi_\mathcal{L}$ be the $| \mathcal{L} | \times K$ submatrix of $\Pi$ formed by removing all rows $\ell \notin \mathcal{L}$. %

\begin{assump}{A\ref{assump:homog:relevance:gen}$^\prime$}(Relevance). The following hold:
\begin{enumerate}
\item For all $\mathcal{L} \subseteq \{ 1,\ldots, L \}$ with $| \mathcal{L} | = K$, $\Pi_\mathcal{L}$ has full rank.

\item For all $\mathcal{L}\subseteq \{1,\ldots,L\}$ such that $|\mathcal{L}| = K+1$,  $\{\pi_\ell :  \ell \in \mathcal{L}\}$ are affinely independent. That is,
\[
	\begin{pmatrix}
	1 & 1 & \ldots & 1\\
	\pi_{\ell_1} & \pi_{\ell_2} & \ldots & \pi_{\ell_{K+1}}
	\end{pmatrix}
\]
has full rank, where $\mathcal{L} = \{ \ell_1,\ldots, \ell_{K+1} \}$.
\end{enumerate}
\end{assump}

\noindent A\ref{assump:homog:relevance:gen}.1$^\prime$ is equivalent to the following:
\begin{itemize}
\item Pick any $K$ components of the instrument vector $Z = (Z_1,\ldots,Z_L)$. Let $Z_\mathcal{L}$ denote this subvector of instruments. Let $Z_{-\mathcal{L}}$ denote the remaining instruments. Then the population 2SLS coefficient on $X$ using $Z_\mathcal{L}$ as excluded instruments and $Z_{-\mathcal{L}}$ as controls is well-defined and unique. Moreover, it equals
\[
	\beta^\textsc{2sls}_\mathcal{L} = \Pi_\mathcal{L}^{-1} \psi_\mathcal{L}
\]
where $\psi_\mathcal{L}$ equals the subvector of $\psi$ after removing all components $\ell \notin \mathcal{L}$.
\end{itemize}
For example, suppose $L = K+1$. Then dropping the exclusion restriction for any single instrument returns us to a non-refutable model. Here we simply assume that $\beta$ is still point identified regardless of which exclusion restriction we choose to drop. Finally, note that A\ref{assump:homog:relevance:gen}.1$^\prime$ implies our earlier relevance assumption A\ref{assump:homog:relevance:gen}.
A1.2$^\prime$ means that there does not exist a hyperplane that passes through all of the $\pi_\ell$ vectors. It is equivalent to linear independence of $(\pi_L - \pi_1,\ldots, \pi_2 - \pi_1)$.

Let
\begin{equation}\label{eq:FASstar}
	\text{FAS}^* = \text{conv} \left( \left\{ \beta_\mathcal{L}^\textsc{2sls} : \mathcal{L} \subseteq \{ 1,\ldots, L \}, | \mathcal{L} |  = K \right\} \right).
\end{equation}
denote the polytope defined by the convex hull of the set of just-identified 2SLS estimands.

\begin{proposition}\label{prop:KisLplus1FF}
Suppose A\ref{assump:homog:relevance:gen}$^\prime$, A\ref{assump:homog:nonsing:gen}, and A\ref{assump:exogeneity:gen} hold. Suppose the joint distribution of $(Y,X,Z)$ is known. Suppose $L = K + 1$. Then the falsification frontier is the set
\begin{equation}\label{eq:KisLplus1FF}
	\text{FF} = \{\delta \in\R^L_{\geq 0}: \delta_\ell = | \psi_\ell - \pi_\ell' b |, \ \ell =1,\ldots,L, \ b \in \text{FAS}^* \}.
\end{equation}
\end{proposition}

\begin{theorem}\label{thm:LisKplus1FAS}
Suppose A\ref{assump:homog:relevance:gen}$^\prime$, A\ref{assump:homog:nonsing:gen}, and A\ref{assump:exogeneity:gen} hold. Suppose the joint distribution of $(Y,X,Z)$ is known. Suppose $L = K + 1$. Then $\text{FAS}^*$ is the falsification adaptive set.
\end{theorem}

\begin{figure}[!t]
\centering
\includegraphics[width=0.49\linewidth]{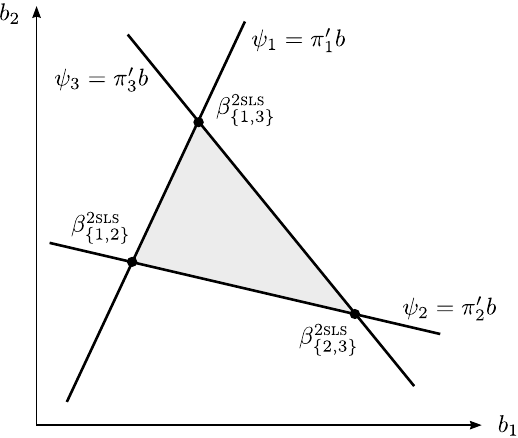}
\includegraphics[width=0.49\linewidth]{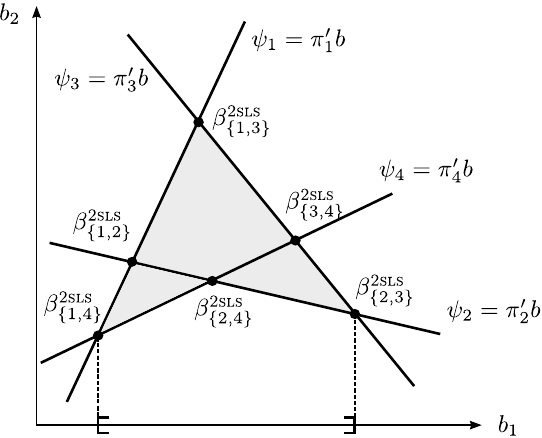}
\caption{Example with $K=2$ endogenous variables. Left: $L=3$ instruments. Right: $L=4$ instruments. In both plots, the falsification adaptive set for $(\beta_1,\beta_2)$ is the shaded region. In the right plot, the falsification adaptive set for $\beta_1$ is shown as the projection onto the first component. See text for additional explanation.}
\label{FAS_K2L3}
\end{figure}

To illustrate theorem \ref{thm:LisKplus1FAS}, consider the two endogenous variables ($K=2$) and three instruments ($L=3$) case. Since there are two endogenous variables, there are two coefficients we're interested in. Consider the left plot in figure \ref{FAS_K2L3}. This plot shows possible values $(b_1, b_2)$ for these two coefficients. The exclusion restriction from instrument $\ell$ imposes a single linear constraint $\psi_\ell = \pi_\ell' b$. These constraints are simply lines in $\R^2$. Since there are three instruments, there are three constraints. When these three lines do not intersect at a common point, the baseline model is refuted. This case is shown in the figure. Suppose we drop the exclusion restriction for instrument 1. Then two linear constraints remain, $\beta$ is point identified, and it equals the intersection point $\beta_{\{2,3\}}^\textsc{2sls}$ in the lower right of the figure. A similar interpretation applies to the other two intersection points $\beta_{\{1,3\}}^\textsc{2sls}$ and $\beta_{\{1,2\}}^\textsc{2sls}$. The falsification adaptive set is then simply the convex hull of these three points in $\R^2$, which is shown as the shaded triangular region.

Next consider the case with $K > 1$ and $L > K + 1$. We handle this case by reducing it to the case we just studied where $L = K + 1$. Let
\[
	 \mathcal{P}_{\mathcal{L}} = \text{conv} \big( \big\{ \beta_{\mathcal{L} \setminus \{ \ell \}}^\textsc{2sls} : \ell \in \mathcal{L} \big\} \big).
\]
In the following proposition, we consider subsets of indices $\mathcal{L} \subseteq \{ 1,\ldots, L \}$ such that $| \mathcal{L} | = K + 1$. Thus $| \mathcal{L} \setminus \{ \ell \} | = K$. So $\beta_{\mathcal{L} \setminus \{ \ell \}}^\textsc{2sls}$ is a just-identified 2SLS estimand. We form $K+1$ different estimands by cycling through which exclusion restriction to drop. We then take their convex hull. This set is similar to the set defined in equation \eqref{eq:FASstar}. The key difference is that here we only look at the estimands which use all but one index from a reference set $\mathcal{L}$.

\begin{proposition}\label{prop:generalLinearFF}
Suppose A\ref{assump:homog:relevance:gen}$^\prime$, A\ref{assump:homog:nonsing:gen}, and A\ref{assump:exogeneity:gen} hold. Suppose the joint distribution of $(Y,X,Z)$ is known. Then the falsification frontier is the set
\[
	\text{FF} = \left\{ \delta \in \R^L_{\geq 0}: 
	\delta_\ell = |\psi_\ell - \pi_\ell' b |, \
	b \in \mathcal{P}_{\mathcal{L}}, \
	\mathcal{L} \subseteq \{1,\ldots, L\}, \
	|\mathcal{L}| = K+1 \right\}.
\]
\end{proposition}

\begin{theorem}\label{thm:generalLinearFAS}
Suppose A\ref{assump:homog:relevance:gen}$^\prime$, A\ref{assump:homog:nonsing:gen}, and A\ref{assump:exogeneity:gen} hold. Suppose the joint distribution of $(Y,X,Z)$ is known. Let
\[
	\mathcal{P} = \bigcup_{\mathcal{L} \subseteq \{ 1,\ldots, L \} : | \mathcal{L} | = K+1}  \mathcal{P}_{\mathcal{L}}.
\]
Then $\mathcal{P}$ is the falsification adaptive set.
\end{theorem}

We use Farkas' lemma and Caratheodory's theorem from convex analysis to prove these results. As in the previous cases, $\mathcal{P}$ can be computed by just running a variety of 2SLS regressions. Next note that, although each $\mathcal{P}_\mathcal{L}$ is convex, their union generally is not. Nonetheless, we are often only interested in linear functionals of the coefficient vector $\beta$. For example, we often care about just one component of $\beta$. The following corollary shows that the falsification adaptive set for a linear functional of $\beta$ again has a simple form. For this result, recall the definition of $\text{FAS}^*$ from equation \eqref{eq:FASstar}. 

\begin{corollary}\label{corr:FASprojection}
Suppose A\ref{assump:homog:relevance:gen}$^\prime$, A\ref{assump:homog:nonsing:gen}, and A\ref{assump:exogeneity:gen} hold. Suppose the joint distribution of $(Y,X,Z)$ is known. Then $\text{FAS}^*$ contains the falsification adaptive set for $\beta$. Moreover, for any $\alpha \in \R^K$ the falsification adaptive set for $\alpha' \beta$ is
\[
	\left[\min_{\mathcal{L}\subseteq \{1,\ldots,L\}, |\mathcal{L}| = K} \alpha' \beta_{\mathcal{L}}^\textsc{2sls}, \
	\max_{\mathcal{L}\subseteq \{1,\ldots,L\}, |\mathcal{L}| = K} \alpha' \beta_{\mathcal{L}}^\textsc{2sls}\right].
\]
\end{corollary}

This result shows that we can simply cycle through all possible just identified models, compute the corresponding 2SLS estimand, take the convex hull, and project it onto one component to get the FAS for that component.

The right plot in figure \ref{FAS_K2L3} illustrates the $L > K+1$ case. Here we have $K=2$ and $L=4$. There are 6 different just-identified 2SLS estimands. The falsification adaptive set is no longer a convex set. Nonetheless, the projection of the convex hull of \emph{all} just-identified 2SLS estimands onto the first component still gives the falsification adaptive set for $\beta_1$. Moreover, this projection can be computed by simply taking the largest and smallest estimated values of $\beta_1$, among the just-identified 2SLS estimands. In this sense, corollary \ref{corr:FASprojection} directly generalizes theorem \ref{thm:identSetOnFFhomogTrt} to the case $K > 1$.

We conclude this subsection by briefly discussing estimation and inference. Corollary \ref{corr:FASprojection} shows that the FAS for a linear functional of $\beta$ has an intersection bounds form. As in the $K=1$ case, researchers can present sample analog estimators or bias-corrected versions, along with corresponding confidence sets (see \citealt{ChernozhukovLeeRosen2013}). Our discussion of weak instruments in section \ref{subsec:FASestimationLinearIV} applies here as well. More generally, researchers may want to do estimation and inference on the FAS for the entire vector $\beta$. Theorem \ref{thm:generalLinearFAS} shows that this set is defined by a finite set of unconditional linear moment inequalities. Note that here we have omitted covariates. Depending on how covariates enter the model, the FAS may continue to depend only on unconditional linear moment inequalities, or it may also depend on conditional linear moment inequalities. See appendix \ref{sec:covariatesInLinearModel} for details. While there is a large literature on inference in general moment inequality models---see \cite{CanayShaikh2017} and \cite{Molinari2019} for surveys---there is now a growing literature on the linear case. This includes \cite{HsiehShiShum2017}, \cite{AndrewsRothPakes2019}, \cite{ChoRussell2019}, and \cite{Gafarov2019}. We conjecture that some of these results can be applied to do inference on the FAS for $\beta$, the FAS for subvectors of $\beta$ with two or more components, and the FAS for nonlinear functionals of $\beta$. We leave a full analysis for future work, however.

\setlength{\heavyrulewidth}{0.03em}
\setlength{\lightrulewidth}{0.03em}

\section{Empirical Applications}\label{sec:empirical}

In this section we apply our results from section \ref{sec:homogModel} to four different empirical studies: \citet*[\emph{The Review of Economic Studies}]{DurantonMorrowTurner2014}, \citet*[\emph{The Quarterly Journal of Economics}]{AlesinaGiulianoNunn2013}, \citet*[\emph{The American Economic Review}]{AcemogluJohnsonRobinson2001}, and \citet[\emph{Econometrica}]{Nevo2001}. Each paper reports 2SLS estimates using multiple instrumental variables and each paper discusses concerns about instrument validity. In particular, all four papers run overidentification tests, which sometimes fail. Even when they do not fail, the authors sometimes express a concern that this could simply be due to low sample size. To address these concerns, we estimate falsification adaptive sets and compare them with the results reported in the original papers. We argue that the FAS is an informative complement to traditional overidentification test $p$-values: Rather than focusing on the null hypothesis that the instruments are consistent with each other, the FAS summarizes the range of estimates obtained from alternative models which which are not refuted by the data.

\subsection{Roads and Trade: \cite*{DurantonMorrowTurner2014}}\label{sec:DMT2014}

\def\mystrut{\rule{0pt}{1.25\normalbaselineskip}}

\cite{DurantonMorrowTurner2014} study the relationship between roads and trade. Specifically, they consider a dataset of 66 regions (`cities') in the United States. Their treatment variable is the log number of kilometers of interstate highways within a city, in 2007. This variable directly affects the cost of leaving a city, and therefore the cost of exporting from a city: It is easier to export from a city with many kilometers of interstate highways passing through it. Their outcome variable is a measure of how much that city exports. They consider two different ways of measuring exports: Weight (in tons) and value (in dollars). We focus on the weight measure for brevity. We discuss the value measure in appendix \ref{sec:additionalEmpirics}. They begin by estimating a gravity equation relating the weight of a city's exports to other cities with the highway distance between those cities, both measured in 2007. This equation includes a fixed effect for the exporting city. The estimate of this fixed effect is their main outcome variable. They call this variable the ``propensity to export weight.'' Thus their main goal is to estimate the causal effect of within city highways on the propensity to export weight.

We cannot learn this causal effect by simply regressing the propensity to export weight on within city highways since there is a classic simultaneity problem. We expect that building highways within the city will boost exports. But high export cities may also build more highways to facilitate their existing exports. The authors solve this problem by instrumenting for the number of kilometers of within city highways. They consider three different instruments:
\begin{enumerate}
\item \emph{Railroads}: The log number of kilometers of railroads in the city, in 1898.

\item \emph{Exploration}: A measure of the quantity of historical exploration routes that passed through the city. Specifically, they digitized five maps of exploration routes between 1528 and 1850. For each map, they counted the number of $1 \times 1$ km pixels crossed by an exploration route in the city, summed over all five maps, and took logs.

\item \emph{Plan}: The log number of kilometers of highway in the city, according to a planned highway construction map approved by the federal government in 1947. This map was written following the 1944 Federal Aid Highway Act. \cite{BaumSnow2007} had previously used this instrument, and provides a detailed history.

\end{enumerate}
We first summarize their arguments for validity of these instruments. We then review their analysis and present our new results.

\subsubsection*{Arguments for Instrument Validity}

Instrument relevance can be directly assessed from the first stage results in the data. We discuss this later. Still, there are a priori reasons to believe the instruments will be correlated with treatment. Building railroad tracks requires leveling ground. So does building roads. Thus old unused railroad tracks can be easily converted to highways. Exploration routes were easy to travel on foot, horseback, or wagon. Such routes are also likely to be good for cars. Finally, most of the highways on the 1947 plan were ultimately built, although additional highways not on the plan were built as well.

Next consider instrument exogeneity and exclusion. Exogeneity concerns the presence of omitted variables that affect the instrument and the outcome, or of simultaneity between the instrument and the outcome. Exclusion concerns the direct causal effect of the instrument on outcomes. We consider each instrument in turn:
\begin{enumerate}
\item \emph{Railroads}. Consider exogeneity. First they argue that simultaneity is not a concern, as railroads were ``constructed by private companies expecting to make a profit from railroad operations in a not too distant future,'' rather than over 100 years later. Moreover, in 1898 the U.S. was primarily agricultural, and railroads were built ``to transport grain, livestock and lumber as well as passengers over long distances''. This suggests that omitted variables may not be a concern, since modern manufactured goods trade in 2007 is quite different from agricultural and commodities trade in 1898. There are, however, three specific omitted variables they are worried about:
\begin{enumerate}
\item City population in 1898 affects the present of railroads in 1898. It also affects 2007 population, which affects trade in 2007. To address this, they control for population in 1920, the closest data they have to 1898.

\item Certain geographical features may make some cities more attractive for both railroad construction in 1898, as well as for trade in 2007. Thus they control for census region fixed effects, slope, and distance to the nearest body of water.

\item Time-invariant city productivity may affect 1898 railroads and trade in 2007. They argue that this can be controlled for by including population in 1950 and 2000 as controls.
\end{enumerate}
Next consider exclusion.
Here they are worried about two direct paths from 1898 railroads to trade in 2007:
\begin{enumerate}
\item Railroads may cause a city to specialize in manufacturing, which may then persist through the present. They address this by controlling for the log share of manufacturing employment in 1956 and in 2003. They also control for the weight of a city's exports in 1956.

\item Railroads may affect the socioeconomic characteristics of cities, which then affects patterns of trade in 2007. They address this by controlling for log income per capita and log share of the population with at least a college degree.
\end{enumerate}

\item \emph{Exploration}. Consider exogeneity. First they argue that simultaneity is not a concern, as exploration routes were forged for a variety of reasons, including the search for the fountain of youth, expanding U.S.\ territory to the Pacific Ocean, and finding emigration routes to Oregon. Moreover, this also suggests that omitted variables may not be a concern, since the variables affecting choice of exploration route are related to these motives like the search for the fountain of youth, rather than modern manufacturing. That said, there is one set of omitted variables they are concerned about: Geographical features. These could affect both exploration routes and trade patterns in 2007. They address this concern via the same set of geographic controls mentioned above in our discussion of the railroad instrument. Finally, they do not discuss any specific concerns regarding the exclusion restriction.

\item \emph{Plan}. Consider exogeneity. As with the other two instruments, they first argue that simultaneity is not a concern: The plan was explicitly drawn to ``connect population centers [in 1947], not to anticipate future population levels and trade patterns.'' However, because population is persistent, this could be an omitted variable. Hence they control for population in 1950 and 2007 as before. As with the exploration instrument, they do not discuss any specific concerns regarding the exclusion restriction.
\end{enumerate}
Overall, the authors raise concerns about validity of all three instruments. Although they address these concerns with various controls, these controls may still not perfectly fix failures of exogeneity, exclusion, or both. Hence the authors lean on overidentification, stating that
\begin{quote}
``Using different instruments, for which threats to validity differ, allows for informative over-identification tests.'' (page 700)
\end{quote}
With this motivation, we next present the results.

\subsubsection*{Results}

\begin{table}[!t]%
\scriptsize
\caption[]{\label{DMTtable1} Baseline 2SLS results for \cite{DurantonMorrowTurner2014}: The effect of highways on export weight. Non-highlighted parts reproduce results from their paper. Highlighted parts are new. Panel A reproduces columns 1--4 of their table 5. It also shows the estimated falsification adaptive set. Panel B uses only two of their instruments, controlling for the other. See text for discussion.}
\vspace{2mm}

\begin{adjustwidth}{-0.25in}{-0.25in}

\setlength{\linewidth}{.1cm}
\newcommand{\contents}{
\centering

\begin{tabular}{lcccc} \hline
\mystrut
& \multicolumn{4}{c}{Dependent variable: Export weight} \\[0.3em]
 & (1) & (2) & (3) & (4) \\[0.5em]
 \hline
\multicolumn{5}{l}{Panel A. Plan, exploration, and railroads used as instruments} \mystrut \\[0.5em]
\hline

\mystrut
log highway km & 1.13*** & 0.57*** & 0.47*** & 0.39***  \\
 & (0.14) & (0.16) & (0.14) & (0.12) \\[0.4em]
log employment &  & 0.52*** & 0.69* & 0.47  \\
 &  & (0.11) & (0.39) & (0.33)   \\[0.4em]
Market access (export) &  & -0.45*** & -0.65*** & -0.63***  \\
 &  & (0.14) & (0.14) & (0.11)  \\[0.4em]
log 1920 population &  &  & -0.38 & -0.29 \\
 &  &  & (0.25) & (0.23)  \\[0.4em]
log 1950 population &  &  & 1.00** & 0.65* \\
 &  &  & (0.39) & (0.38)  \\[0.4em]
log 2000 population &  &  & -0.74 & -0.20  \\
 &  &  & (0.48) & (0.45)  \\[0.4em]
log \% manuf. emp. &  &  &  & 0.64***  \\
 &  &  &  & (0.12)  \\[0.4em]
 First-stage $F$ stat. & 97.5 & 90.3 & 80 & 84.8  \\[0.4em]
Overid. $p$-value & 0.10 & 0.043 & 0.15 & 0.31  \\[0.4em]
\rowcolor{lightgray} FAS & [0.49, 0.86] & [-0.32, 0.28] & [-0.26, 0.31] & [0.18, 0.42]  \\[0.4em]
 \hline
\rowcolor{lightgray} \multicolumn{5}{l}{Panel B. Plan and exploration used as instruments, controlling for railroads} \mystrut \\[0.5em]
\hline
\rowcolor{lightgray} \mystrut
 log highway km & 0.79*** & 0.17 & 0.21 & 0.23 \\
\rowcolor{lightgray}  & (0.24) & (0.20) & (0.15) & (0.14)  \\[0.4em]
\rowcolor{lightgray} log 1898 railroad km & 0.33** & 0.33*** & 0.25** & 0.16  \\
\rowcolor{lightgray}  & (0.15) & (0.12) & (0.12) & (0.10)  \\[0.4em]
\rowcolor{lightgray}  First-stage $F$ stat. & 61.1 & 65.4 & 77.8 & 82.2 \\[0.4em]
\rowcolor{lightgray} Overid. $p$-value & 0.64 & 0.51 & 0.48 & 0.72  \\[0.4em]
 \hline
\multicolumn{5}{p{\linewidth}}{\emph{Notes}: 66 observations per column. All specifications include a constant. Heteroskedasticity robust standard errors in parentheses. ***, **, *: statistically significant at 1\%, 5\%, 10\%.}
\end{tabular}
}

\setbox0=\hbox{\contents}
\setlength{\linewidth}{\wd0-2\tabcolsep-.25em}
\contents

\end{adjustwidth}
\end{table}

First consider table \ref{DMTtable1}. In this and all other tables, the non-highlighted parts reproduce results from the original paper. The highlighted parts are new computations which we have added. Panel A reproduces columns 1--4 of table 5 in \cite{DurantonMorrowTurner2014}. These are their main results. In particular, they are interested in the coefficient on log highway km, the log number of highway kilometers within the city. This coefficient represents their estimate of the causal effect of roads on trade. Here it is estimated by 2SLS, using railroads, exploration, and plan as instruments. As mentioned above, we are concerned about possible invalidity of the instruments. Thus they implement the standard test of overidentifying restrictions. At conventional sizes, it easily passes in the longest specification, fails in the second specification, and marginally passes in the first specification. Also note that these specifications do not include all of the controls discussed above; the authors include those in separate analyses, which we discuss later (our table \ref{DMTappendixTable6}).

We add the falsification adaptive set to these baseline results. This is the last row of panel A. There are two things to notice: First, except for the last specification, none of the 2SLS estimates are within the FAS. We knew this was possible, given our theoretical results in appendix \ref{sec:comparing2SLSandFAS} (which relate the 2SLS estimand to the FAS). Second, the FAS magnitudes are all generally smaller than the 2SLS point estimates.

\begin{table}[!t]%
\scriptsize
\caption[]{\label{DMTtable2} The effect of controlling for unused instruments. Non-highlighted parts reproduce results from their paper. Highlighted parts are new. All columns have employment, market access, and past populations as controls. Columns 4-6 also have manufacturing share of employment as controls. Compare to table 6 of \cite{DurantonMorrowTurner2014}.}
\vspace{2mm}

\begin{adjustwidth}{-0.25in}{-0.25in}

\setlength{\linewidth}{.1cm}
\newcommand{\contents}{
\centering

\begin{tabular}{lccc | ccc} \hline
\mystrut
 & (1) & (2) & (3) & (4) & (5) & (6) \\
  & Plan & Railroad & Exploration & Plan & Railroad & Exploration \\[0.4em] 
\hline
\multicolumn{7}{l}{Panel A. Without controlling for other instruments} \mystrut \\[0.5em]
\hline
\mystrut
log highway km & 0.49*** & 0.83*** & 0.12 & 0.38*** & 0.64*** & 0.34* \\
 & (0.15) & (0.25) & (0.31) & (0.13) & (0.22) & (0.19) \\[0.4em] 
log \% manuf. emp. &  &  &  & 0.64*** & 0.60*** & 0.65*** \\
 &  &  &  & (0.12) & (0.13) & (0.13) \\[0.4em] 
 First stage $F$ stat. & 141 & 45.2 & 14.8 & 130 & 40.8 & 23.8 \\ [0.4em] 
 \hline
\rowcolor{lightgray} \multicolumn{7}{l}{Panel B. Controlling for other instruments} \mystrut \\[0.5em]
\hline
\rowcolor{lightgray}  \mystrut
log highway km & 0.31 & 4.09 & -0.26 & 0.18 & 3.65 & 0.42 \\
\rowcolor{lightgray}  & (0.22) & (4.09) & (0.71) & (0.21) & (4.16) & (0.52) \\[0.4em] 
\rowcolor{lightgray} log \% manuf. emp. &  &  &  & 0.63*** & 0.36 & 0.61*** \\
\rowcolor{lightgray}  &  &  &  & (0.12) & (0.38) & (0.12) \\[0.4em] 
\rowcolor{lightgray} log 1898 railroad km & 0.21* &  & 0.24** & 0.18 &  & 0.16 \\
\rowcolor{lightgray}  & (0.12) &  & (0.11) & (0.11) &  & (0.11) \\[0.4em] 
\rowcolor{lightgray} log 1528-1850 exploration & -0.053 & -0.40 &  & 0.025 & -0.32 &  \\
\rowcolor{lightgray}  & (0.077) & (0.36) &  & (0.065) & (0.40) &  \\[0.4em] 
\rowcolor{lightgray} log 1947 highway km &  & -2.09 & 0.32 &  & -1.90 & -0.13 \\
\rowcolor{lightgray}  &  & (2.50) & (0.46) &  & (2.48) & (0.36) \\[0.4em] 
\rowcolor{lightgray}  First stage $F$ stat. & 59.6 & 1.54 & 29.1 & 54.8 & 1.27 & 27 \\[0.4em]  
\hline
\rowcolor{lightgray} \mystrut
FAS for this specification & \multicolumn{3}{c}{[-0.26, 0.31]} & \multicolumn{3}{c}{[0.18, 0.42]} \\[0.4em]
\hline
\multicolumn{7}{p{\linewidth}}{\emph{Notes}: 66 observations per column. All specifications include a constant. Heteroskedasticity robust standard errors in parentheses. ***, **, *: statistically significant at 1\%, 5\%, 10\%.}
\end{tabular}
}

\setbox0=\hbox{\contents}
\setlength{\linewidth}{\wd0-2\tabcolsep-.25em}
\contents

\end{adjustwidth}
\end{table}

To better understand how we computed the FAS, and how to interpret it, next consider table \ref{DMTtable2}. Columns 1--3 include the same baseline controls as column 3 in table \ref{DMTtable1} while columns 4--6 include the same baseline controls as column 4 in table \ref{DMTtable1}. The only difference is that we no longer use all three variables (plan, railroad, exploration) as instruments. Instead, in panel A, we use only one of these variables as an instrument and we ignore the other two variables. Panel A reproduces columns 4--6 from table 6 in \cite{DurantonMorrowTurner2014}. The authors used these results as their main robustness check. They argue that the three estimates 0.38, 0.64, and 0.34 from columns 4--6, panel A, table \ref{DMTtable2} are consistent with their baseline estimates of 0.47 and 0.39 from columns 3 and 4, panel A, table \ref{DMTtable1}.

However, as we have discussed, omitting an invalid instrument can lead to omitted variable bias. In this application we are concerned that some of the instruments may be invalid. Thus the alternative models of interest are those where one of the instruments is valid but the others are not. When computing results in these alternative models, the invalid instruments should be included as controls. Panel B shows these results. Here we use one instrument while controlling for the other two. For example, in column 1 we use plan as an instrument and control for railroad and exploration.

For brevity, here we only describe the results in columns 4--6. These results use the full baseline specification. Consider column 5, panel B. This result uses railroad as an instrument, controlling for plan and highway. Unlike the uncontrolled result from panel A, railroad is a very weak instrument. Hence we ignore the result using railroad alone, as discussed in section \ref{subsec:FASestimationLinearIV}. Next consider column 4. Here we use plan as the instrument, controlling for the other two. Despite these controls, plan is still a strong instrument. The estimated effect 0.18 is roughly half as large as the estimate from panel A, 0.38. It is also no longer statistically significant at any conventional level. Next consider column 6. Here we use exploration as the instrument, controlling for the other two. Exploration continues to be a strong instrument with these controls. The estimated effect 0.42 in panel B is similar to the effect from panel A, 0.34. It is no longer statistically significant, however. 

Putting these coefficient estimates together gives us the estimated FAS, $[0.18, 0.42]$. The endpoints of this set correspond to point estimates from alternative models which maintain validity of only one instrument at a time. The interior of this set corresponds to alternative models which relax validity of all instruments at once, but just enough to avoid falsification. Thus the FAS reflects model uncertainty: Relying on different instruments to different degrees yields different results. This range of results is given by the FAS.

In panel B of table \ref{DMTtable2} we found that railroad is a weak instrument when controlling for the other two, and hence it yields the largest point estimates. Given this finding, one may also wonder how removing railroads as an instrument affects the baseline analysis. This is shown in panel B of table \ref{DMTtable1}. All of the coefficients on log highway km are smaller, to the point that they are no longer statistically significant for all but the shortest specification. Moreover, the standard overidentification tests are now all easily passed. (Note that these tests are only comparing results using plan and exploration as instruments.) However, the coefficients on railroads are statistically significant for all but the fourth column. This suggests that the full baseline model using all three instruments could be rejected, and also explains the source of the relatively small overidentification test $p$-values in panel A.

\renewcommand{\topfraction}{0.9}

\begin{table}[!t]%
\scriptsize
\caption[]{\label{DMTappendixTable6} The effect of controlling for unused instruments, continued. Non-highlighted parts reproduce results from their paper. Highlighted parts are new. This table extends the analysis of table \ref{DMTtable2} to consider specifications with additional control variables. Columns 1--3 reproduce the results in appendix table 6 of \cite{DurantonMorrowTurner2014}. Columns 4--6 just add controls for the other instruments.}
\vspace{2mm}

\begin{adjustwidth}{-0.25in}{-0.25in}

\setlength{\linewidth}{.1cm}
\newcommand{\contents}{
\centering
\begin{tabular}{c l ccc | aaa a }
&& \multicolumn{3}{c}{Without other IV controls} & \multicolumn{3}{c}{With other IV controls} \\[0.4em]
&& (1) & (2) & (3) & (4) & (5) & (6) & \\[0.4em]
Added variable && Plan & Railroad & Exploration & Plan & Railroad & Exploration & FAS \\[0.4em] 
\hline
\mystrut
\multirow{2}{*}{\parbox{0.08\linewidth}{\centering Water}} 
& log highway km & 0.34** & 0.66** & 0.24 & 0.13 & 3.96 & 0.30 & [0.13, 0.30]\\
&& (0.16) & (0.26) & (0.29) & (0.21) & (4.56) & (0.49) & \\[0.4em] 
& $F$ stat. & 126 & 25.3 & 10.9 & 67.6 & 1.17 & 26.2 & \\[0.4em]
\mystrut
\multirow{2}{*}{\parbox{0.08\linewidth}{\centering Slope}}
&\phantom{log highway km} & 0.39*** & 0.57*** & 0.46** & 0.20 & 3.86 & 0.66 & [0.20, 0.66]\\
 && (0.14) & (0.20) & (0.19) & (0.21) & (5.86) & (0.50)  & \\[0.4em]
& \phantom{First-stage Stat.} & 133 & 44.5 & 22.6 & 60.3 & 0.65 & 25.6 & \\[0.4em]
\mystrut
\multirow{2}{*}{\parbox{0.08\linewidth}{\centering Census regions}}
& \phantom{log highway km} & 0.32** & 0.62*** & 0.36* & -0.012 & 3.68 & 0.40 & [-0.012, 0.40] \\
& & (0.14) & (0.12) & (0.20) & (0.24) & (3.24) & (0.64) & \\[0.4em]
& \phantom{First-stage Stat.} & 122 & 58.3 & 22.6 & 39.4 & 1.59 & 10.7 & \\[0.4em]
\mystrut
\multirow{2}{*}{\parbox{0.08\linewidth}{\centering Percent college}}
&\phantom{log highway km} & 0.29** & 0.56** & 0.41** & 0.013 & 3.64 & 0.78 & [0.013, 0.78] \\
& & (0.13) & (0.23) & (0.18) & (0.19) & (4.13) & (0.49) & \\[0.4em]
& \phantom{First-stage Stat.} & 116 & 36.6 & 29.5 & 47.9 & 1.16 & 25.2 & \\[0.4em]
\mystrut
\multirow{2}{*}{\parbox{0.08\linewidth}{\centering Income per capita}}
&\phantom{log highway km} & 0.35** & 0.63*** & 0.35* & 0.079 & 3.42 & 0.54 & [0.079, 0.54] \\
& & (0.14) & (0.22) & (0.18) & (0.21) & (3.42) & (0.47) & \\[0.4em]
& \phantom{First-stage Stat.} & 123 & 36.8 & 26.4 & 47.8 & 1.70 & 24.9 & \\[0.4em]
\mystrut
\multirow{2}{*}{\parbox{0.08\linewidth}{\centering Percent wholesale}}
&\phantom{log highway km} & 0.41*** & 0.59*** & 0.49*** & 0.22 & 2.71 & 0.75 & [0.22, 0.75] \\
& & (0.12) & (0.21) & (0.12) & (0.19) & (3.29) & (0.49) & \\[0.4em]
& \phantom{First-stage Stat.} & 136 & 39.3 & 23.2 & 54.4 & 1.38 & 25.4 & \\[0.4em]
\mystrut
\multirow{2}{*}{\parbox{0.08\linewidth}{\centering Traffic}}
&\phantom{log highway km} & 0.42** & 1.00* & 0.34 & 0.23 & 5.57 & 0.40 & [0.23, 0.40] \\
& & (0.18) & (0.52) & (0.23) & (0.25) & (9.60) & (0.51)  &\\[0.4em]
& \phantom{First-stage Stat.} & 79 & 13.4 & 44.6 & 44.3 & 0.43 & 26 & \\ [0.4em]
\mystrut
\multirow{2}{*}{\parbox{0.08\linewidth}{\centering All}}
&\phantom{log highway km} & 0.39** & 0.73** & 0.58** & 0.18 & 2.43 & 0.69  & [0.18, 0.69] \\
& & (0.18) & (0.35) & (0.28) & (0.26) & (3.01) & (0.67) & \\[0.4em]
& \phantom{First-stage Stat.} & 52.6 & 19.9 & 15.5 & 32.1 & 0.88 & 6.53 & \\[0.4em]
\hline
\multicolumn{9}{p{\linewidth}}{66 observations per column. All specifications include a constant. Heteroskedasticity robust standard errors in parentheses. ***, **, *: statistically significant at 1\%, 5\%, 10\%.}
\end{tabular}
}

\setbox0=\hbox{\contents}
\setlength{\linewidth}{\wd0-2\tabcolsep-.25em}
\contents

\end{adjustwidth}
\end{table}

Thus far we have focused on the baseline results, which do not include all of the possible control variables discussed earlier. Table \ref{DMTappendixTable6} shows results with these controls. We begin with the full set of baseline control variables, as used in column 4 of table \ref{DMTtable1}. We then add just one control. Each row corresponds to a different control. The last row shows the results that add all controls at once. Unlike the main baseline result, column 4 of table \ref{DMTtable1}, here we only use one instrument at a time. Columns 1--3 use a single instrument, without controlling for the other two. These results reproduce appendix table 6 of \cite{DurantonMorrowTurner2014}. Based on these results, the authors argue that
\begin{quote}
``None of our main results is affected by these controls, even when we use our instruments individually.'' (page 708)
\end{quote}
They also argue that using one instrument at a time is an ``even more demanding exercise'' than examining the effect of additional controls when using all three instruments (not shown here; see their appendix table 5). As we've discussed, however, omitting the invalid instruments may cause omitted variable bias. So in columns 4--6, we replicate columns 1--3, except now controlling for the other two instruments.

There are three main differences between the results with the instrument controls and those without. First, the railroad instrument is again very weak, leading to large coefficients. This informs our understanding of the results from columns 1--3, since there we observed that the coefficients in column 2 are always larger than those in columns 1 and 3, and are often substantially larger. Second, none of the results are statistically significant at conventional levels. Finally, the coefficients using plan as the instrument all become smaller once the other instruments are controlled for (column 4 versus column 1), while the coefficients using exploration as the instrument all become larger once the other instruments are controlled for (column 6 versus column 3). Thus, ignoring the results using railroads, the overall range of point estimates is larger. This is reflected in the falsification adaptive sets, which are presented in the final column.

Overall, there are two main conclusions from our analysis: First, the evidence suggests that the railroad instrument is the most questionable, and should be used only as a control. Thus the estimates in panel B of table \ref{DMTtable1} are arguably the most appropriate baseline results. Second, there is substantially more uncertainty in the magnitude of the causal effect of roads on trade than suggested by the original results of \cite{DurantonMorrowTurner2014}. This is reflected in the various falsification adaptive sets we present. In particular, the FAS for the longest specification is $[0.18, 0.69]$; see table \ref{DMTappendixTable6}. Accounting for sampling uncertainty only increases this range. That said, these results do not change the overall qualitative conclusions of the paper: All points in the estimated FAS for the longest specification are still positive, suggesting that the number of within city highways appears to positively affect propensity to export weight.

\subsection{The Origin of Gender Roles: \cite*{AlesinaGiulianoNunn2013}}

\cite{AlesinaGiulianoNunn2013} study the relationship between traditional agricultural practices and modern gender norms. Specifically, following \cite{Boserup1970} they distinguish between two kinds of traditional agriculture: shifting cultivation, which ``is labor intensive and uses handheld tools like the hoe and the digging stick,'' and plough cultivation, which ``is much more capital intensive, using the plough to prepare the soil.'' They note that ``unlike the hoe or digging stick, the plough requires significant upper body strength, grip strength, and bursts of power.'' Hence ``when plough agriculture is practiced, men have an advantage in farming relative to women.'' Consequently, in societies that used plough agriculture, ``men tended to work outside the home in the fields, while women specialized in activities within the home.'' Boserup hypothesized that, as \cite{AlesinaGiulianoNunn2013} put it, ``this division of labor then generated norms about the appropriate role of women in society'' which persist today.

\cite{AlesinaGiulianoNunn2013} consider a variety of datasets and methods to test this hypothesis. We focus on their cross country instrumental variable analysis. Their dataset includes 160 countries. Their treatment variable is the ``estimated proportion of citizens with ancestors that traditionally used the plough in pre-industrial agriculture.'' They explain how they construct this variable in section 3. For this analysis they consider three different outcome variables: The female labor force participation rate in 2000, the share of a country's firms with some female ownership (measured using data between 2005--2011), and the share of political positions held by women in 2000. Their appendix A2 gives further details on the definition of these variables. They also consider a composite outcome variable, called the average effect size (AES). This is constructed by first dividing each outcome variable by its standard deviation and then averaging the three variables. %

Any correlation between treatment and outcomes is not necessarily causal for a few reasons. First, there may be reverse causality: Countries with less equal gender norms may have actively adopted innovations like the plough which conformed with these norms. Second, there may be omitted variables: Areas that were economically more developed may have been more likely to both adopt the plough and to have more equal gender norms today. This would tend to mask any causal effect of the plough on gender norms. To address these concerns, the authors use an instrumental variable approach.

The authors use properties of countries' ancestral land geography as instruments. Specifically, \cite{Pryor1985} classifies crops as either \emph{plough positive} or \emph{plough negative}. Plough positive crops benefit from the use of the plough while plough negative crops generally do not. The authors focus on cereal crops only, since they are similar except for how much they benefit from the plough. Given this classification, they construct instruments as follows: Consider a large piece of land. First measure the area of land that is suitable for growing crops in general. Pick a single plough positive cereal. Within this area of overall suitability, compute the area of land suitable for growing this specific plough positive cereal. Do this for each plough positive cereal. Note that some land is suitable for growing multiple crops, so there will be overlap in these areas. So take the average area across all plough positive cereals. Divide by the area that is suitable for growing crops overall. That is their measure of plough positive share. Next do the same for plough negative crops. Thus we have a procedure for defining suitability of a single piece of land. Using this procedure, the authors define the country level variable \emph{plough positive environment} as the ``average fraction of ancestral land that was suitable for growing [plough positive cereals] divided by the fraction that was suitable for any crops.'' They use the same procedure as in their definition of the treatment variable, described in their section 3. They define the country level variable \emph{plough negative environment} analogously.

\begin{table}[!t]
\scriptsize
\caption[]{\label{AGN2013table1} Country level 2SLS results from \cite{AlesinaGiulianoNunn2013}. Non-highlighted parts reproduce results from their paper. Highlighted parts are new. This table reproduces panels A and C of their table 8. It also shows the estimated falsification adaptive set, which here is a singleton since one of the two instruments is weak.}
\vspace{2mm}

\begin{adjustwidth}{-0.25in}{-0.25in}

\setlength{\linewidth}{.1cm}
\newcommand{\contents}{
\centering

\begin{tabular}{lcccccccc}
\hline
\mystrut
 & \multicolumn{2}{c}{\multirow{3}{*}{\parbox{0.17\linewidth}{\centering Female labor force participation in 2000}}} 
 & \multicolumn{2}{c}{\multirow{3}{*}{\parbox{0.17\linewidth}{\centering Share of firms with female ownership, 2005--2011}}} 
 & \multicolumn{2}{c}{\multirow{3}{*}{\parbox{0.17\linewidth}{\centering Share of political positions held by women in 2000}}} 
 & \multicolumn{2}{c}{\multirow{3}{*}{\parbox{0.17\linewidth}{\centering Average effect size (AES)}}} \\[3em]
 & (1) & (2) & (3) & (4) & (5) & (6) & (7) & (8) \\[0.4em]
 \hline
\multicolumn{7}{l}{Panel A. Second stage 2SLS estimates} \mystrut \\[0.5em]
\hline
\mystrut
Traditional plough use & -21.6*** & -25.0*** & -17.5*** & -22.7*** & -6.46*** & -9.73*** & -0.92*** & -1.31*** \\
 & (5.25) & (7.51) & (5.53) & (7.62) & (2.33) & (3.75) & (0.23) & (0.39) \\[0.4em]
Continent fixed effects &  & Y &  & Y &  & Y &  & Y \\[0.4em]
Overid. $p$-value & 0.00091 & 0.00012 & 0.41 & 0.25 & 0.72 & 0.86 & 0.048 & 0.027 \\[0.4em]
\rowcolor{lightgray} FAS &-14.3*** & -18.0*** & -15.3** & -18.9** & -7.08** & -9.99** & -0.73*** & -1.08***  \\
\rowcolor{lightgray}  & (5.41) & (6.78) & (6.18) & (8.11) & (3.00) & (4.11) & (0.23) & (0.38) \\[0.4em]
Observations & 160 & 160 & 122 & 122 & 140 & 140 & 104 & 104 \\[0.4em]
 \hline
\multicolumn{7}{l}{Panel B. First stage estimates. Dependent variable: Traditional plough use} \mystrut \\[0.5em]
\hline
\mystrut
Plough-pos. environment & 0.74*** & 0.63*** & 0.86*** & 0.67*** & 0.82*** & 0.68*** & 0.87*** & 0.72*** \\
 & (0.084) & (0.089) & (0.078) & (0.10) & (0.082) & (0.10) & (0.089) & (0.12) \\[0.4em]
Plough-neg. environment & 0.12 & 0.18 & 0.100 & 0.12 & 0.13 & 0.19 & 0.13 & 0.14 \\
 & (0.12) & (0.13) & (0.17) & (0.17) & (0.13) & (0.14) & (0.18) & (0.19) \\[0.4em]
$F$-stat (both) & 40.2 & 25.1 & 66.8 & 21.9 & 52 & 21.9 & 49.5 & 18.5 \\[0.4em]
\rowcolor{lightgray} $F$-stat (pos. only) & 80 & 57.7 & 123 & 42.8 & 101 & 43.6 & 95.6 & 36.6 \\[0.4em] 
\rowcolor{lightgray} $F$-stat (neg. only) & 0.58 & 1.55 & 0.36 & 0.45 & 1.03 & 1.76 & 0.51 & 0.57 \\[0.4em]
\hline
\multicolumn{9}{p{\linewidth}}{\emph{Notes}: These results include a variety of historical and contemporary controls. See the table 8 notes in \cite{AlesinaGiulianoNunn2013} along with their discussion in the text for details. Heteroskedasticity robust standard errors in parentheses. ***, **, *: statistically significant at 1\%, 5\%, 10\%.} 
\end{tabular}
}

\setbox0=\hbox{\contents}
\setlength{\linewidth}{\wd0-2\tabcolsep-.25em}\contents

\end{adjustwidth}
\end{table}

Their main concern with validity of these instruments is failure of exogeneity due to geographical variables which affect the instruments and also modern gender norms. In their baseline results (table 8) they include a variety of geographical controls to account for this. They consider more extensive geographical controls in appendix table A14. They consider many other controls in appendix table A15. Finally, they also use the presence of two instruments to conduct overidentification tests.

Table \ref{AGN2013table1} reproduces their main instrumental variable results. As in our analysis of \cite{DurantonMorrowTurner2014}, the non-highlighted parts reproduce results from the original paper. The highlighted parts are new computations which we have added. First focus on their results. The overidentification test fails at the conventional 5\% level for columns 1 and 2, as well as for the average effect size, columns 7 and 8.\footnote{The authors report overidentification test $p$-values of 0.31 in column 4 and 0.06 in column 8. These differ slightly from the $p$-values we report in table \ref{AGN2013table1}. The reason is that we use heteroskedasticity robust standard errors for all calculations, including columns 4 and 8. The authors used homoskedastic standard errors for their columns 4 and 8 overidentification test (but not for the standard errors reported under their point estimates).} The authors do not comment on this finding. To provide a constructive response to this rejection, we compute the falsification adaptive set. To do this, we must first check instrument relevance. As shown in panel B, plough positive environment satisfies the relevance assumption. Plough negative environment, however, does not. The authors commented on this, since they did present the first stage coefficients on both instruments, although they did not present the separate $F$-test statistics as we do. Despite this relevance failure, the authors used both instruments in their main results. 

In panel A, we present the estimated falsification adaptive set. Since there are only two instruments, one of which fails relevance, this set is simply a singleton. Specifically, this is the point estimate from the model using plough positive environment as an instrument, controlling for plough negative environment. Overall, the authors' qualitative findings all continue to hold using the FAS: All point estimates are still negative and statistically significant at conventional levels. Quantitatively, however, the FAS point estimates all have smaller magnitudes than the original baseline findings, except for columns 5 and 6. Excluding those columns, the point estimates are an average of 22\% smaller.
Finally, note that the columns where the estimates changed the least---columns 5 and 6---are those with the largest overidentification test $p$-values. Likewise, the columns where the estimates changed the most---columns 1 and 2---are the ones with the smallest overidentification test $p$-values. This is to be expected, given how the overidentification test is constructed. Our point is that the FAS is an informative complement to these $p$-values, since it directly summarizes the range of estimates corresponding to non-refuted alternative models. As in our analysis of \cite{DurantonMorrowTurner2014}, here it also highlighted the weakness of the plough negative instrument, suggesting that the FAS point estimates we present in panel A of table \ref{AGN2013table1} are more appropriate baselines than those originally presented by the authors.

\subsection{Colonial Origins of Development: \cite*{AcemogluJohnsonRobinson2001}}

\newcolumntype{K}[1]{>{\centering\arraybackslash}p{#1}}
\begin{table}
\scriptsize
\caption{\label{AJR2001table8} Results from \cite{AcemogluJohnsonRobinson2001}. Non-highlighted parts reproduce results from their paper. Highlighted parts are new. Panel B reproduces panel D of their table 8. Panel A is analogous to their panel A of table 8, except that we also use log settler mortality as an instrument. Columns  2, 3, 7, and 8 include a control for the number of years since independence. All other columns do not include controls, except as indicated in the table.}
\vspace{4mm}
\centering
\begin{adjustwidth}{-0.5in}{-0.5in}
\begin{tabular}{l K{1.1cm} K{1.1cm} K{1.1cm} K{1.1cm} K{1.2cm} | K{1.1cm} K{1.1cm} K{1.1cm} K{1.1cm} K{1.1cm} }
\hline
\mystrut
& \multicolumn{10}{c}{Dependent variable: Log GDP per capita in 1995} \\[0.3em]
 & (1) & (2) & (3) & (4) & (5) & (6) & (7) & (8) & (9) & (10) \\[0.4em]
 \hline
\multicolumn{11}{l}{Panel A. 2SLS estimates using two instruments at a time: log settler mortality paired with each of the other five instruments} \mystrut \\[0.5em]
\hline
\rowcolor{lightgray} \mystrut
 & 0.89*** & 0.81*** & 0.80*** & 0.67*** & 0.63*** & 0.95*** & 0.83*** & 0.82*** & 0.70*** & 0.65*** \\
\rowcolor{lightgray} 
\multirow{-2}{*}{\parbox{0.17\linewidth}{\raggedright Average protection against expropriation risk, 1985--1995}} & (0.13) & (0.13) & (0.12) & (0.11) & (0.11) & (0.17) & (0.17) & (0.16) & (0.14) & (0.12) \\[1.5em]
\rowcolor{lightgray} Latitude &  &  &  &  &  & -0.60 & -0.30 & -0.21 & -0.75 & -0.51 \\
\rowcolor{lightgray}  &  &  &  &  &  & (1.16) & (1.08) & (1.04) & (0.88) & (0.81) \\[0.4em]
\rowcolor{lightgray} Observations & 63 & 60 & 59 & 60 & 59 & 63 & 60 & 59 & 60 & 59 \\[0.4em]
\rowcolor{lightgray} First-stage $F$ stat. & 17.4 & 12.1 & 13.3 & 9.97 & 11.4 & 10.5 & 7.26 & 7.98 & 7.45 & 8.79 \\[0.4em]
\rowcolor{lightgray} Overid. $p$-value & 0.70 & 0.25 & 0.28 & 0.45 & 0.20 & 0.79 & 0.27 & 0.30 & 0.42 & 0.18 \\[0.4em]
\rowcolor{lightgray}  FAS & [0.81,0.99] & [0.45,1.03] & [0.51,1.03] & [0.48,0.77] & [0.40,0.85] & [0.88,1.02] & [0.42,1.06] & [0.48,1.04] & [0.49,0.84] & [0.41,0.93] \\[0.4em]
  \hline
\multicolumn{11}{l}{Panel B. 2SLS estimates using one of the other five instruments, controlling for log settler mortality} \mystrut \\[0.5em]
\hline
\mystrut
\multirow{2}{*}{\parbox{0.17\linewidth}{\raggedright Average protection against expropriation risk, 1985--1995}} & 0.81*** & 0.45* & 0.51** & 0.48** & 0.40** & 0.88*** & 0.42 & 0.48* & 0.49** & 0.41** \\
 & (0.22) & (0.24) & (0.22) & (0.22) & (0.17) & (0.28) & (0.29) & (0.27) & (0.24) & (0.18) \\[1.5em]
Log settler mortality & -0.067 & -0.25 & -0.21 & -0.14 & -0.19 & -0.050 & -0.26 & -0.22 & -0.14 & -0.19 \\
 & (0.16) & (0.16) & (0.15) & (0.15) & (0.12) & (0.18) & (0.16) & (0.16) & (0.14) & (0.12) \\[0.4em]
Latitude &  &  &  &  &  & -0.52 & 0.38 & 0.28 & -0.38 & -0.17 \\
 &  &  &  &  &  & (1.11) & (0.86) & (0.83) & (0.81) & (0.70) \\[0.4em]
\rowcolor{lightgray}  First-stage $F$ stat. & 9.34 & 3.66 & 4.87 & 3.56 & 5.88 & 6.58 & 2.50 & 3.21 & 3.25 & 5.36 \\[0.4em]
 \hline
\multicolumn{11}{l}{Panel C. 2SLS estimates using log settler mortality, controlling for one of the other five instruments} \mystrut \\[0.5em]
\hline
\rowcolor{lightgray} \mystrut
& 0.99*** & 1.03*** & 1.03*** & 0.77*** & 0.85*** & 1.02*** & 1.06*** & 1.04*** & 0.84*** & 0.93*** \\
\rowcolor{lightgray}  
\multirow{-2}{*}{\parbox{0.17\linewidth}{\raggedright Average protection against expropriation risk, 1985--1995}} & (0.32) & (0.29) & (0.31) & (0.19) & (0.25) & (0.34) & (0.34) & (0.33) & (0.25) & (0.32) \\[1.5em]
\rowcolor{lightgray}  & -0.0039 &  &  &  &  & -0.0029 &  &  &  &  \\
\rowcolor{lightgray}
\multirow{-2}{*}{\parbox{0.17\linewidth}{\raggedright European settlements in 1900}}
  & (0.011) &  &  &  &  & (0.012) &  &  &  &  \\[0.8em]
\rowcolor{lightgray}  &  & -0.10 &  &  &  &  & -0.099 &  &  &  \\
\rowcolor{lightgray}
\multirow{-2}{*}{\parbox{0.17\linewidth}{\raggedright Constraint on executive in 1900}}  &  & (0.11) &  &  &  &  & (0.11) &  &  &  \\[0.8em]
\rowcolor{lightgray} Democracy in 1900 &  &  & -0.074 &  &  &  &  & -0.071 &  &  \\
\rowcolor{lightgray}  &  &  & (0.086) &  &  &  &  & (0.085) &  &  \\[0.4em]
\rowcolor{lightgray} &  &  &  & -0.043 &  &  &  &  & -0.050 &  \\
\rowcolor{lightgray}
\multirow{-2}{*}{\parbox{0.17\linewidth}{\raggedright Constraint on executive in first year of independence}} &  &  &  & (0.064) &  &  &  &  & (0.071) &  \\[2em]
\rowcolor{lightgray} &  &  &  &  & -0.058 &  &  &  &  & -0.064 \\
\rowcolor{lightgray} 
\multirow{-2}{*}{\parbox{0.17\linewidth}{\raggedright Democracy in first year of independence}} &  &  &  &  & (0.057) &  &  &  &  & (0.064) \\[0.8em]
\rowcolor{lightgray} Latitude &  &  &  &  &  & -0.55 & -0.57 & -0.30 & -1.07 & -1.08 \\
\rowcolor{lightgray}  &  &  &  &  &  & (1.25) & (1.39) & (1.30) & (1.12) & (1.22) \\[0.4em]
\rowcolor{lightgray} First-stage $F$ stat. & 6.51 & 7.69 & 6.91 & 8.60 & 6.38 & 6.15 & 5.86 & 5.88 & 6.22 & 4.71 \\[0.4em]
 \hline
 \multicolumn{11}{l}{\emph{Notes}: Homoskedastic standard errors in parentheses. ***, **, *: statistically significant at 1\%, 5\%, 10\%.}
\end{tabular}
\end{adjustwidth}
\end{table}

\cite{AcemogluJohnsonRobinson2001} study the relationship between institutions and economic development. Since this paper is exceptionally well known, here we only briefly summarize it and then present our results. They consider a cross section of about 60 countries. Their treatment variable is the average protection against expropriation risk between 1985 and 1995, a measure of the strength of property rights in the country. Their outcome variable is log GDP per capita in 1995. Their primary instrument measures mortality rates faced by early settlers in the country. Their main results (table 4) use just this one instrument. In table 8, however, they consider five more instruments. We focus on that analysis.

They use the additional instruments as follows: First, they only ever consider pairs of instruments. As they say in footnote 29, they do this to avoid the many instruments problem. We follow them and also only consider two instruments at a time. For each of the five additional instruments, they estimate the coefficient on the treatment variable using either (a) that instrument alone (their panel A) or (b) that instrument alone but also including settler mortality as a control variable (their panel D). They test overidentification by using a Hausman test to compare these two coefficient estimates (their panel C). They also present the estimated coefficient on settler mortality when it is included as a control, noting that it is not statistically indistinguishable from zero at conventional levels (their panel D).

In our table \ref{AJR2001table8} we replicate and extend those results. Whereas \cite{AcemogluJohnsonRobinson2001} focused on using these additional instruments to perform overidentification tests, we instead focus on the point estimates themselves. Thus in our panel A we present a new set of baseline results, which use two instruments at a time. \cite{AcemogluJohnsonRobinson2001} did not present these results. Here we include the classical overidentification test $p$-value; the authors had instead presented a Hausman test result. Columns 1--5 present results with no control variables (except for 3 and 4, which include a single control: years since independence) while columns 6--10 add a control for latitude. 

The overidentification tests do not reject at conventional levels in each column. Nonetheless, it is informative to ask: What estimates would we obtain from alternative non-refutable models? The last line of panel A answers this question by presenting the FAS. This set is computed from the estimated coefficients on the treatment variable in panels B and C. The authors had originally computed one of the endpoints of this set, although they focused on the coefficient on the controlled instrument. Here we compute the other endpoint of the FAS. First note that none of the instruments are particularly strong. This is a well known problem in this specific empirical application (for example, see pages 3072--3073 of \citealt{albouy2012}). Since we have not formally studied inference on the FAS, we leave the question of how to accurately measure sampling uncertainty in this application to future work. Here we will simply focus on the point estimates. For these estimates, we use a small $F$-statistic cut-off for inclusion in the FAS, so that the FAS is an interval for all specifications. For every specification, the estimated FAS always lies above zero. Hence the FAS supports the qualitative conclusion of this paper: Increasing property rights causes an increase in GDP. That said, for some of the specifications the length of the FAS is quite large, indicating substantial variation in the estimated magnitude of this effect. 

As in our previous examples, the authors used multiple instruments to assess validity of the exogeneity and exclusion restrictions. They did this by focusing on a variety of statistical tests of overidentification. Here we illustrated the value of using the falsification adaptive set to go beyond the binary pass/fail outcome of a specification test and to instead summarize the substantive variation in estimates obtained from alternative non-refutable models.

\subsection{Discrete Choice Demand Estimation: \cite{Nevo2001}}\label{sec:NevoAnalysis}

In our final example, we study a classical application of instrumental variables: estimation of consumer demand. In a series of papers Aviv Nevo estimated the demand for ready-to-eat cereal as a first step to answering several important economic questions: How can we predict the welfare effects of a merger (\citealt{Nevo2000})? How can we test for collusion (\citealt{Nevo2001})? How can we construct price indices that account for new products and quality changes (\citealt{Nevo2003})? Answering these questions requires reliable estimates of price elasticities of demand. If one's baseline model of consumer demand is refuted, however, it is not clear whether the estimated elasticities from this known-to-be misspecified model are relevant for these questions. For this reason, we recommend computing falsification adaptive sets. These sets show the range of estimates consistent with non-refuted models.

Since all three of the papers \cite{Nevo2000, Nevo2001, Nevo2003} use essentially the same data and demand estimation, we focus on \cite{Nevo2001}. First we briefly summarize his analysis, focusing on falsification. We then discuss our results. Here we assume familiarity with standard discrete choice demand models. See \cite{Nevo2011} for a survey. 

\subsubsection*{Analysis from \cite{Nevo2001}}

Nevo begins by estimating a variety of logit demand models; see his table 5. In 9 out of the 10 specifications, the overidentification test rejects the model. Commenting on this, he says ``it is unclear whether the large number of observations is the reason for the rejection or that the IV's are not valid'' (page 325). He notes that the most flexible specification, column 10, passes the test. This specification added city fixed effects, and hence he concludes that it is important to flexibly control for demographics. In our results below, however, the specification with city fixed effects is strongly rejected.

In logit demand models estimated using overidentified 2SLS, the model can be refuted for two distinct reasons: (1) The instruments are inconsistent with each other or (2) the logit functional form is incorrect. The second issue is well known and is commonly addressed by using more flexible models of demand, like the random coefficients logit model. Nevo presents estimates from this model in table 6. Using these estimates he notes that the logit model ``is easily rejected'' (footnote 28 on page 332). 

The random coefficients logit model itself, however, is falsifiable. Nevo does not report the overidentification test result for his full model. He does report the GMM objective function value in table 6, but it is not clear if this objective function uses the optimal weighting matrix, as required to compute the $J$-test statistic. If it does, and assuming the statistic has not already been scaled by the sample size, then the overidentification test based on the reported value easily rejects the random coefficients logit model as well. Regardless, we also estimated the random coefficients logit model using our data (not reported here), and it is easily rejected.

As with the baseline logit model, this could be either because (1) the instruments are invalid or (2) the random coefficients logit functional form is incorrect. One response is to use even more flexible models of demand. For recent work along these lines, see \cite{Compiani2018} and \cite{TebaldiTorgovitskyYang2019}. An alternative approach is to maintain the random coefficients logit functional form, but relax the exogeneity and exclusion restrictions on the instruments. Or one could relax both assumptions simultaneously. Here we focus only on the instrument assumptions. Moreover, we also focus only on the baseline logit demand model. Our characterization of the falsification adaptive set (theorem \ref{thm:identSetOnFFhomogTrt}) can be extended to the random coefficients logit model, but since this is a substantial extension we leave it to separate work.

Finally, recall that \citet[pages 252--254]{Berry1994} shows how the nested logit model can be written as a linear instrumental variable model with multiple endogenous variables. This result combined with our multiple endogenous variable results in section \ref{sec:linearModelGeneralK} allows us to immediately study the nested logit model with endogenous price. Moreover, nested logit has been used as the baseline model in previous studies of the demand for ready-to-eat cereal---see \cite{CotterillHaller1997}. We omit a full empirical study of the nested logit baseline for brevity, however.

\subsubsection*{Data}

\begin{table}[tbp] \centering
\newcolumntype{C}{>{\centering\arraybackslash}X}
\caption{\label{NevoTable3}Brands used for estimating demand}
\vspace{2mm}
\scriptsize
\begin{tabularx}{\textwidth}{lCCC}
\toprule
{All Family / Basic Segment}&{Taste Enhanced Wholesome Segment}&{Simple Health Nutrition Segment}&{Kids Segment} \tabularnewline
\midrule\addlinespace[1.5ex]
K Rice Krispies & P Raisin Bran & P Grape Nuts &P Fruity Pebbles  \tabularnewline
K Corn Flakes & P Honey Bunches of Oats & K Special K Red Berry & K Frosted Flakes \tabularnewline
GM Cheerios & K Raisin Bran Crunch & K Special K & K Froot Loops \tabularnewline
& K Raisin Bran & GM Multigrain Cheerios &K Corn Pops  \tabularnewline
& K Frosted Mini-Wheats & & K Apple Jacks \tabularnewline
&&& Q Life \tabularnewline
&&& Q Cinnamon Life \tabularnewline
&&& Q Cap'n Crunch Crunch Berries \tabularnewline
&&& Q Cap'n Crunch \tabularnewline
&&&GM Reese's Puffs \tabularnewline
&&&GM Lucky Charms \tabularnewline
&&&GM Honey Nut Cheerios \tabularnewline
&&&GM Cinnamon Toast Crunch \tabularnewline
\end{tabularx}
\end{table}

\begin{table}[tbp] \centering
\newcolumntype{C}{>{\centering\arraybackslash}X}
\caption{\label{NevoTable1}Volume market shares}
\vspace{2mm}
\scriptsize
\begin{tabularx}{\textwidth}{lCCCCC}
\toprule
{Firm}&{2012 Q4}&{2013 Q4}&{2014 Q4}&{2015 Q4}&{2016 Q4} \tabularnewline
\midrule\addlinespace[1.5ex]
Kellogg&28.23&28.18&27.76&29.00&29.89 \tabularnewline
General Mills&31.03&31.03&30.48&30.74&28.92 \tabularnewline
Post&10.46&11.26&11.73&11.56&12.21 \tabularnewline
Quaker Oats&6.59&6.50&7.08&6.82&7.25 \tabularnewline
Malt-O-Meal&3.99&4.62&4.49&4.00&4.51 \tabularnewline
Kashi&2.69&2.65&2.09&2.12&1.92 \tabularnewline
C3&69.72&70.47&69.97&71.30&71.01 \tabularnewline
C6&83.00&84.24&83.64&84.24&84.70 \tabularnewline
Private Label&15.19&13.74&13.78&13.10&12.35 \tabularnewline
\bottomrule \addlinespace[1.5ex]

\end{tabularx}
\end{table}

Unlike our previous three applications, here the original data is not publicly available. We instead construct a new dataset, following the details of \cite{Nevo2001} as closely as possible. We consider a panel dataset of 53 cities between 2012 and 2016. Nevo's data covered 1988 to 1992. We obtained quarterly data on price and quantity of each cereal sold in these cities from Nielsen scanner data. We focus on the top 25 cereal brands, shown in table \ref{NevoTable3}. Table \ref{NevoTable1} shows the volume market shares (pounds sold) for the top 6 firms, for the last quarter of each year. Comparing this to table 1 from \cite{Nevo2001}, we see that firm market shares are relatively stable over time. The industry continues to be highly concentrated, with the top 3 firms having around 70\% of the market and the top 6 having around 84\%. The market share of private label cereals has increased somewhat, ranging from 12 to 15\% in our data, whereas it ranged from 3 to 8\% between 1988 and 1992.

{\renewcommand{\arraystretch}{1.2}
\begin{table}[t]
\centering
\caption{\label{NevoTable2}Prices and market shares of brands in sample}
\vspace{2mm}
\scriptsize
\begin{tabular}{lccccc | ccc}
\hline
\mystrut
 & &&&&& \multicolumn{3}{c}{Variation (percentage)} \\
Description                                   &Mean  &Median &Std  &Min  &Max   & Brand & City & Quarter \\[0.4em]
\hline
\mystrut
Prices                &10.38 &9.56   &2.45 &5.01 &18.02 &85                &6.7                 &0.8                    \\
\quad (\textcent\space per serving) & & & & & & & \\
Advertising &5.92  &4.98   &5.15 &0    &23.76 &53.6                &--                   &3.5                   \\
\quad (M\textdollar\space per quarter) & & & & & & & \\
Share within Cereal               &2.8   &2.47   &1.73 &0.22  &10.95 &77.6                &4.5                &1.2                   \\
\quad Market (\%) & & & & & & & \\[0.4em]
\hline
\end{tabular}
\end{table}
}

We obtained quarterly advertising spending from Kantar Media. Table \ref{NevoTable2} shows summary statistics for advertising, along with price and market shares. Compared to Nevo's table 4, we see that prices are generally substantially lower, with less variation over time. Advertising spending has increased substantially.

{
\renewcommand{\arraystretch}{1.2}
\begin{table}[t]
\caption{\label{NevoTable4}Sample statistics}
\vspace{2mm}
{\renewcommand\baselinestretch{1}\selectfont
{\small
\scriptsize
\begin{tabularx}{\textwidth}{@{\extracolsep{\fill}} lccccc}
\toprule
Description                             &Mean              &Median  &Std      &Min &Max       \\
\midrule
Total Calories                          &136               &120     &35.59    &100 &200       \\
Fat Calories (/100)                     &.11               &.1      &.08      &0   &.3        \\
Sodium (\% RDA/100)                      &.07               &.07     &.03      &0   &.12       \\
Fiber (\% RDA/100)                      &.1                &.08     &.09      &0   &.31       \\
Sugar (g/100)                           &.09               &.09     &.05      &.01 &.2        \\
Mushy (=1 if cereal gets soggy in milk) &.56 &1       &.51      &0   &1         \\
Serving weight (g)                      &35.96             &30      &11.62    &26  &61        \\
Income (\textdollar)                    &12,659          &9,052 &14,453 &0   &791,777 \\
Age (years)                             &34.92             &34      &21.96    &0   &85        \\
Child (=1 if age $<$ 16)                &.25               &0       &.43      &0   &1         \\
\bottomrule
\end{tabularx} }
\par}
\end{table}
}

We gathered nutrition data for each cereal by examining cereal boxes. For each city we obtained demographic information by sampling individuals from the March Current Population Survey. Table \ref{NevoTable4} provides summary statistics for nutrition and demographics. %

\subsubsection*{Results}

Nevo considered two kinds of instruments:
\begin{enumerate}
\item Cost shifter instruments: (1) annual average wages paid in the supermarket sector in each city and (2) city density, which is viewed as a proxy for the cost of rent. %
In our data, these instruments are very weak. We do not use them in our main analysis.

\item Hausman instruments: ``Regional quarterly average prices (excluding the city being instrumented) in all twenty quarters.'' Taken together, these instruments are very strong. Our falsification adaptive set analysis, however, requires us to consider the strength of each instrument conditional on the others. In this case, most of these instruments become very weak. For this reason we also consider aggregated versions of the Hausman instruments: (1) The regional average \emph{lagged} price, (2) the regional average \emph{contemporaneous} price, and (3) the regional average \emph{lead} price. For all three instruments we exclude the city being instrumented.

\end{enumerate}

\begin{table}[t]
\centering
\caption{\label{NevoTable5}Logit demand parameter estimates. The specifications here are similar to those in table 5 of \cite{Nevo2001}.}
\vspace{2mm}
\scriptsize
\begin{tabular}{l*{4}{c}}
                    &\multicolumn{1}{c}{(1)}&\multicolumn{1}{c}{(2)}&\multicolumn{1}{c}{(3)}&\multicolumn{1}{c}{(4)}\\
\midrule
\addlinespace[0.5em]
\multicolumn{5}{l}{Panel A. OLS} \\
\addlinespace[0.3em]
\midrule
Price               &       -2.45&       -1.49&       -2.41&      -14.97\\
                    &      (0.62)&      (0.59)&      (0.53)&      (0.34)\\
[0.4em]
Adjusted $R^2$        &        0.34&        0.48&        0.56&        0.92\\
\midrule
\addlinespace[0.5em]
\multicolumn{5}{l}{Panel B. 2SLS with per-quarter Hausman instruments} \\
\addlinespace[0.3em]
\midrule
Price               &       -6.67&       -8.85&       -6.89&      -15.15\\
                    &      (0.84)&      (0.93)&      (0.79)&      (1.12)\\
[0.4em]
First-stage $F$ stat. &     730.13&      391.52&      378.66&       65.49\\
[0.4em]
Overid. test stat.  &    2467.31&     1038.85&      874.62&     1085.28\\
[0.4em]
Overid. $p$-value   &       0.00&       0.00&       0.00&       0.00\\
\midrule
\addlinespace[0.5em]
\multicolumn{5}{l}{Panel C. 2SLS with aggregated Hausman instruments} \\
\addlinespace[0.3em]
\midrule
Price               &       -7.50&      -11.21&      -10.60&      -16.23\\
                    &      (0.76)&      (0.82)&      (0.73)&      (0.58)\\
[0.4em]
First-stage $F$ stat.  &     6574.61&     3715.17&     3578.36&     2134.16\\
[0.4em]
Overid. test stat. &      83.34&      143.45&       78.08&       30.13\\
[0.4em]
Overid. $p$-value    &       0.00&        0.00&        0.00&        0.00\\
[0.4em]
FAS ($F \geq 10$) &  -21.27 & -15.99 & -16.93 &  -17.97 \\
 & (1.71) & (1.59) & (1.48) &  (0.66) \\[0.4em]
FAS ($F \geq 5$) & -21.27 & [-332.10, -15.99] & [-182.92, -16.93]& [-58.96, -17.97] \\
\midrule
Advertising & Y & Y & Y & Y \\
Time fixed effects & Y & Y & Y & Y \\
Product characteristics & Y &  &  &  \\
Brand dummies &  & Y & Y & Y \\
Demographics &  &  & Y &  \\
City fixed effects &  &  &  & Y \\
\bottomrule
\multicolumn{5}{l}{\scriptsize Heteroskedasticity robust standard errors in parentheses.}\\
\end{tabular}
\end{table}

Table \ref{NevoTable5} shows our main results. We consider four specifications. All specifications include advertising and time fixed effects. Column 1 includes product characteristics as controls. Columns 2--4 replace these with brand dummies. Column 3 adds demographic controls: log of median income, log of median age, and median household size. Column 4 replaces this with city fixed effects. Panel A shows the OLS results. Panel B shows 2SLS results using the per-quarter Hausman instruments while panel C uses the aggregated Hausman instruments.\footnote{In panel A: Column 1 corresponds to column 1 of Nevo's table 5, Column 2 to his column 2, and column 3 to his column 3. He did not present the specification in our panel A, column 4. In panel B: Column 3 corresponds to Nevo's column 9, and column 4 to his column 10. He did not present the specifications in our panel B, columns 1 and 2. He also did not present any of the specifications in panel C.}

First consider panel A. In the first three columns, the OLS point estimates are much smaller in magnitude than the 2SLS point estimates. In column 4 the OLS and 2SLS estimates are quite similar. However, notice that the overidentification test easily rejects \emph{all} specifications (columns 1--4, panels B and C). This suggests that the instruments may not be valid. There has in fact been substantial debate over the validity of these instruments. See Bresnahan's comment on \cite{Hausman1996}. 

To address this concern, we next estimate the falsification adaptive set. As discussed above, we only do this for panel C. We consider two different cut-offs for the definition of a weak instrument: $F \geq 10$ and $F \geq 5$ (see section \ref{subsec:FASestimationLinearIV}). With the stricter cut-off, only one of the three 2SLS estimates is not weak, and thus the estimated FAS is a singleton. All the point estimates are larger in magnitude than the corresponding baseline estimates. The coefficients are also more stable across specifications. With the weaker cut-off, the FAS is an interval for columns 2--4. The length of this interval shrinks substantially as we move from column 2 to 3 to 4. Consider the FAS in column 4: $[-58.96, -17.97]$. All elements of this interval are larger in magnitude than the estimate $-16.23$ from the refuted baseline model. They are also larger than the OLS estimate of $-14.97$. In the logit model the price elasticities are proportional to the coefficient on price. So the lower bound on the FAS in column 4 suggests that the price elasticities could be as much as about three times larger than those suggested by the refuted baseline estimate. Thus ignoring the result of the overidentification test and using the baseline estimates anyway could lead to a large under-estimate of these elasticities. The falsification adaptive set is a constructive and informative tool researchers can use when their baseline model is rejected.

We conclude this section by briefly comparing our analysis with that of \cite{NevoRosen2012}, who also study discrete choice demand models. They derive identified sets under a discrete relaxation of the classical instrument exogeneity and exclusion restrictions. They do not discuss the general problem of salvaging a falsified model. In particular, their identified sets can be empty. They do not discuss what researchers should do in this case. Furthermore, we relax instrument validity in a different and non-nested way. Hence our analysis is complementary to theirs. Using our data, we estimated the Nevo and Rosen identified set for the coefficient on price. Under their assumptions 3 and 4, and for the specification in column 4 of table \ref{NevoTable5}, this set is $(-\infty, -16.727]$. This set is unbounded from below and strictly contains the estimated FAS $[-58.96, -17.97]$. In their proposition 5 they discuss an additional assumption which can sometimes be used to obtain a bounded identified set. They only consider the two instrument case, and hence we cannot apply their result here. Moreover, this result requires choosing a sensitivity parameter $\gamma$ a priori. By definition, the FAS does not require choosing any sensitivity parameters.

\section{Heterogeneous Treatment Effect Models}\label{sec:hetModel}

Even when both instrument exogeneity and exclusion assumptions hold, the classical overidentifying restrictions \eqref{eq:classicalSarganEqs} may fail when the homogeneous treatment effects assumption fails. Hence we next consider models with heterogeneous treatment effects. In section \ref{sec:hetTrtBinary} we consider binary outcomes while in section \ref{sec:hetTrtCts} we consider continuous outcomes. As earlier, we omit additional covariates for simplicity.

\subsection{Binary Outcomes}\label{sec:hetTrtBinary}

Unlike the baseline model in section \ref{sec:homogModel} using zero-correlation conditions, the baseline model we study here using statistical independence assumptions is falsifiable even with a single instrument. Thus we begin with the case where a single binary instrument is available. This allows us to explain the main ideas and results while keeping the notation simple. After this, we consider the general case where multiple discrete instruments are available.

Let $X \in \{ 0, 1 \}$ denote the observed treatment variable. Let $Y_1, Y_0 \in \{0,1\}$ denote binary potential outcomes. We observe
\begin{equation}\label{eq:generalOutcomeEquation}
	Y = Y_1 X + Y_0 (1-X).
\end{equation}
Let $Z \in \{0,1 \}$ denote an observed instrument. As mentioned above, we consider multiple discretely supported instruments later.
Thus we observe the joint distribution of $(Y,X,Z)$. Consider the following assumption.

\begin{het1Assump}[Instrument exogeneity]\label{assump:marginalIndepZ}
$Z \indep Y_1$ and $Z \indep Y_0$.
\end{het1Assump}

In this model, \cite{Manski1990} derived the identified set for $\Prob(Y_x=1)$ for $x \in \{0,1\}$ as well as the identified set for the average treatment effect
\[
	\text{ATE} = \Prob(Y_1=1) - \Prob(Y_0=1)
\]
under B\ref{assump:marginalIndepZ}.\footnote{Manski's \citeyearpar{Manski1990} analysis considered a general case which does not require outcomes, treatment, or instruments to be binary. In this general setting, he used a mean independence assumption. When outcomes are binary, mean independence of $Y_x$ from $Z$ is equivalent to statistical independence of $Y_x$ and $Z$.} \cite{BalkePearl1997} showed these identified sets can be empty, and hence that this model is falsifiable. Consequently, researchers whose model is falsified face the same question we discussed earlier: What should they do? We answer this question using the four recommendations provided in section \ref{sec:generalFF}. To do this, we must first define continuous relaxations of the instrument exogeneity assumption B\ref{assump:marginalIndepZ}. We do this using the following concept from \cite{MastenPoirier2017}.

\begin{definition}\label{def:c-dep}
Let $x \in \{ 0, 1 \}$. Let $c$ be a scalar between 0 and 1. Say $Z$ is \emph{$c$-dependent} with $Y_x$ if
\begin{equation}\label{eq:c-indep1}
	\sup_{y_x \in \supp(Y_x)} | \Prob(Z=1 \mid Y_x=y_x) - \Prob(Z=1) | \leq c.
\end{equation}
\end{definition}

When $c=0$, $c$-dependence is equivalent to $Z \indep Y_x$. When $c \geq \max\{ \Prob(Z=1), \Prob(Z=0) \}$, $c$-dependence does not constrain the stochastic relationship between $Z$ and $Y_x$. $c \in (0,1)$ partially constrains the stochastic relationship between $Z$ and $Y_x$. Specifically, while it allows the probability of $Z=1$ given $Y_x=y_x$ to vary with $y_x$, it must be within a band around the unconditional probability:
\[
	\Prob(Z=1 \mid Y_x=y_x) \in \big[ \Prob(Z=1) - c, \ \Prob(Z=1) + c \big] \cap [0,1]
\]
for all $y_x \in \supp(Y_x)$. \cite{MastenPoirier2017} give additional discussion of how to interpret $c$-dependence.

Thus we relax B\ref{assump:marginalIndepZ} as follows.

\begin{assump}{B\ref{assump:marginalIndepZ}$^\prime$}(Instrument partial exogeneity).
$Z$ is $c$-dependent with $Y_1$ and with $Y_0$. 
\end{assump}

Under this relaxation, we will derive identified sets for various parameters of interest. We use these identified sets to characterize the falsification point and the falsification adaptive set. We can also use these identified sets to do sensitivity analysis beyond the falsification point.

Let $p_Z = \Prob(Z=1)$. The following assumption rules out trivial cases.

\begin{het1Assump}\label{assump:NonTrivialinstrument}
Suppose $p_Z \in (0,1)$. Suppose $\Prob(X=1 \mid Z=z) \in (0,1)$ for $z \in \{0,1\}$.
\end{het1Assump}

First we set up some notation. Then we state and discuss the main theorem. For $z\in\{0,1\}$, define
\begin{equation}\label{eq:kandK}
	k_z(c) = \frac{\Prob(Z=z)\max\{\Prob(Z=1-z) - c,0\}}{\Prob(Z = 1-z)\min\{\Prob(Z=z)+c,1\}}.
\end{equation}
Note that $k_z(0) =1$, $k_z(1) = 0$, and $k_z(\cdot)$ is weakly decreasing on $c \in [0,1]$.
Using this function, define the set
\begin{equation}\label{eq:diamondSet} %
	\mathcal{D}(c)
	=
	\{(a_0,a_1) \in [0,1]^2 : a_1 \in [a_0 k_0(c), 1-k_0(c) + a_0 k_0(c)], a_0 \in [a_1 k_1(c), 1-k_1(c) + a_1 k_1 (c)] \},
\end{equation}
which depends only on $c$ and the marginal distribution of $Z$. For $x \in \{0,1\}$, define
\begin{align}
\label{eq:boxSet}
	\mathcal{H}_x
	&= [\Prob(Y=1,X=x \mid Z=0), \ \Prob(Y=1,X=x \mid Z=0) + \Prob(X=1-x \mid Z=0)] \\
	&\qquad \times [\Prob(Y=1,X=x \mid Z=1), \ \Prob(Y=1,X=x \mid Z=1) + \Prob(X=1-x \mid Z=1)], \notag
\end{align}
which depends on the joint distribution of $(Y,X,Z)$. Let
\[
	\Theta_x(c) = \mathcal{D}(c) \cap \mathcal{H}_x
\]
denote the intersection of these two sets.

\begin{theorem}\label{thm:BinaryYoneInstrumentCdep}
Consider the binary outcome, binary treatment model \eqref{eq:generalOutcomeEquation}. Suppose the joint distribution of $(Y,X,Z)$ is known, where $Z$ is binary. Suppose B\ref{assump:marginalIndepZ}$^\prime$ and B\ref{assump:NonTrivialinstrument} hold. Then the identified set for
\[
	\big( \Prob(Y_0=1 \mid Z=0), \Prob(Y_0=1 \mid Z=1), \Prob(Y_1=1 \mid Z=0), \Prob(Y_1=1 \mid Z=1) \big)
\]
is $\Theta_0(c) \times \Theta_1(c)$. Consequently, the model is falsified if and only if this set is empty. 
\end{theorem}

The identified sets in theorem \ref{thm:BinaryYoneInstrumentCdep} are defined via two kinds of constraints: the set $\mathcal{D}(c)$ and the sets $\mathcal{H}_x$ for $x\in\{0,1\}$. The set $\mathcal{H}_x$ is simply the identified set for
\[
	\big( \Prob(Y_x=1 \mid Z=0), \ \Prob(Y_x=1 \mid Z=1) \big)
\]
without imposing any assumptions on the stochastic relationship between $Y_x$ and $Z$, as derived by \cite{Manski1990}. This set can be thought of as the intersection of two pairs of parallel half-spaces. The shaded boxes in figure \ref{fig:FFbinaryYsingleIV} show one example of this set.

The set $\mathcal{D}(c)$ imposes the $c$-dependence constraint. For example, if $c=0$, it imposes that
\[
	\Prob(Y_x=1 \mid Z=0) = \Prob(Y_x=1 \mid Z=1)
\]
for $x \in \{0,1\}$. This is simply the statistical independence assumption B\ref{assump:marginalIndepZ} studied by \cite{Manski1990}. Thus the identified sets in theorem \ref{thm:BinaryYoneInstrumentCdep} simplify to those derived by Manski when $c = 0$. As $c$ gets larger, we no longer require the values $\Prob(Y_x=1 \mid Z=z)$ for $z \in \{0,1\}$ to exactly equal each other, but we do require them to be sufficiently close, as defined by the set $\mathcal{D}(c)$. Specifically, it constrains the pairs $(\Prob(Y_x=1 \mid Z=0), \Prob(Y_x=1 \mid Z=1))$ to lie inside a diamond-shaped parallelogram, as shown in figure \ref{fig:FFbinaryYsingleIV}. As $c$ gets larger, this parallelogram gets larger, eventually approaching the entire square $[0,1]^2$ as $c \rightarrow \max \{ p_Z, 1- p_Z \}$. Thus, for $c \geq \max \{ p_Z, 1 - p_Z \}$, theorem \ref{thm:BinaryYoneInstrumentCdep} simplifies to the no assumption bounds derived by \cite{Manski1990}. Thus theorem \ref{thm:BinaryYoneInstrumentCdep} continuously connects the no assumption bounds with the bounds under the statistical independence assumption B\ref{assump:marginalIndepZ}.

\begin{figure}[t]
\centering
\includegraphics[width=50mm]{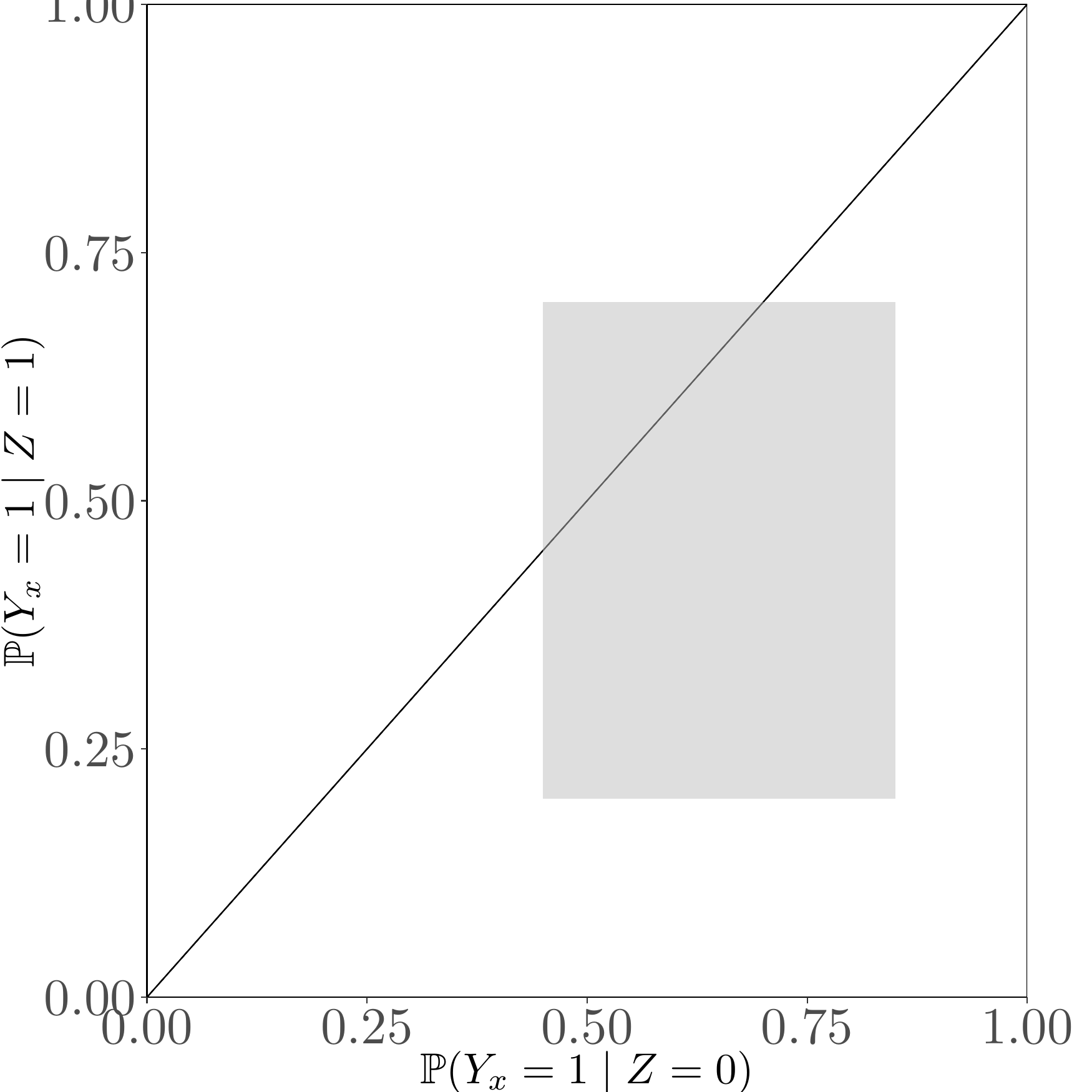}
\hspace{3mm}
\includegraphics[width=50mm]{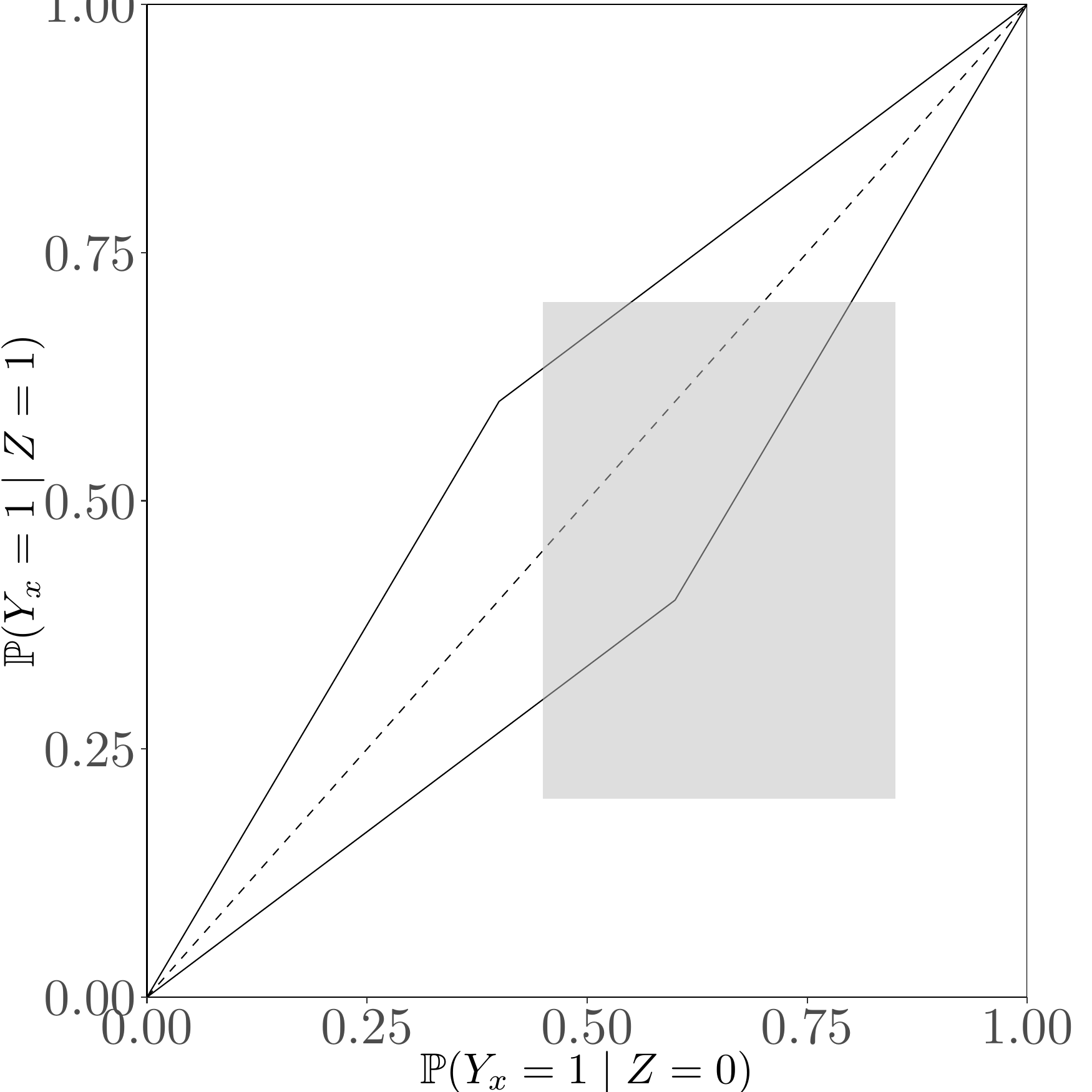}
\hspace{3mm}
\includegraphics[width=50mm]{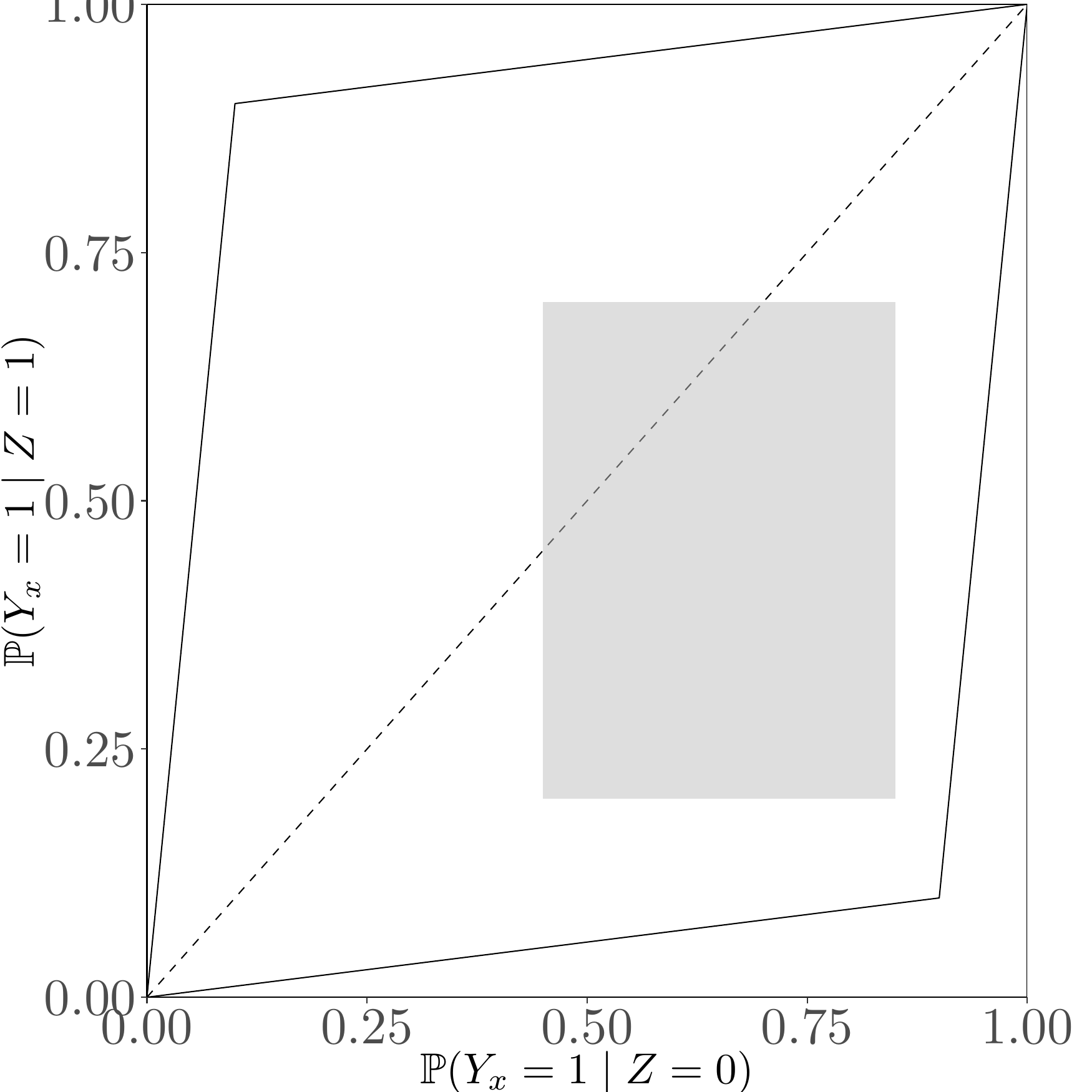}
\caption{Example identified sets for $(\Prob(Y_x=1 \mid Z=1), \Prob(Y_x=1 \mid Z=0))$. Here $p_Z = 0.5$. Left: $c=0$. Middle: $c=0.1$. Right: $c=0.4$.}
\label{fig:FFbinaryYsingleIV}
\end{figure}

While the no assumption bounds ($c \geq \max \{p_Z, 1-p_Z \}$) are never empty, the bounds under statistical independence ($c=0$) can be empty, and hence the baseline statistical independence assumption can be falsified. This happens when, for some $x \in \{0,1\}$, the no assumption bounds $\mathcal{H}_x$ have an empty intersection with the statistical independence constraint set $\mathcal{D}(0)$. Graphically, this happens when the box defined by the no assumption bounds does not intersect the 45-degree line. This is shown in figure \ref{fig:FFbinaryYsingleIV_falsification}. The falsification point is simply the smallest $c$ such that the parallelogram defined by $\mathcal{D}(c)$ has a nonempty intersection with the no assumption bounds $\mathcal{H}_x$ for each $x \in \{0,1\}$. This intersection is illustrated in the middle plot of figure \ref{fig:FFbinaryYsingleIV_falsification}. 

\begin{figure}[t]
\centering
\includegraphics[width=50mm]{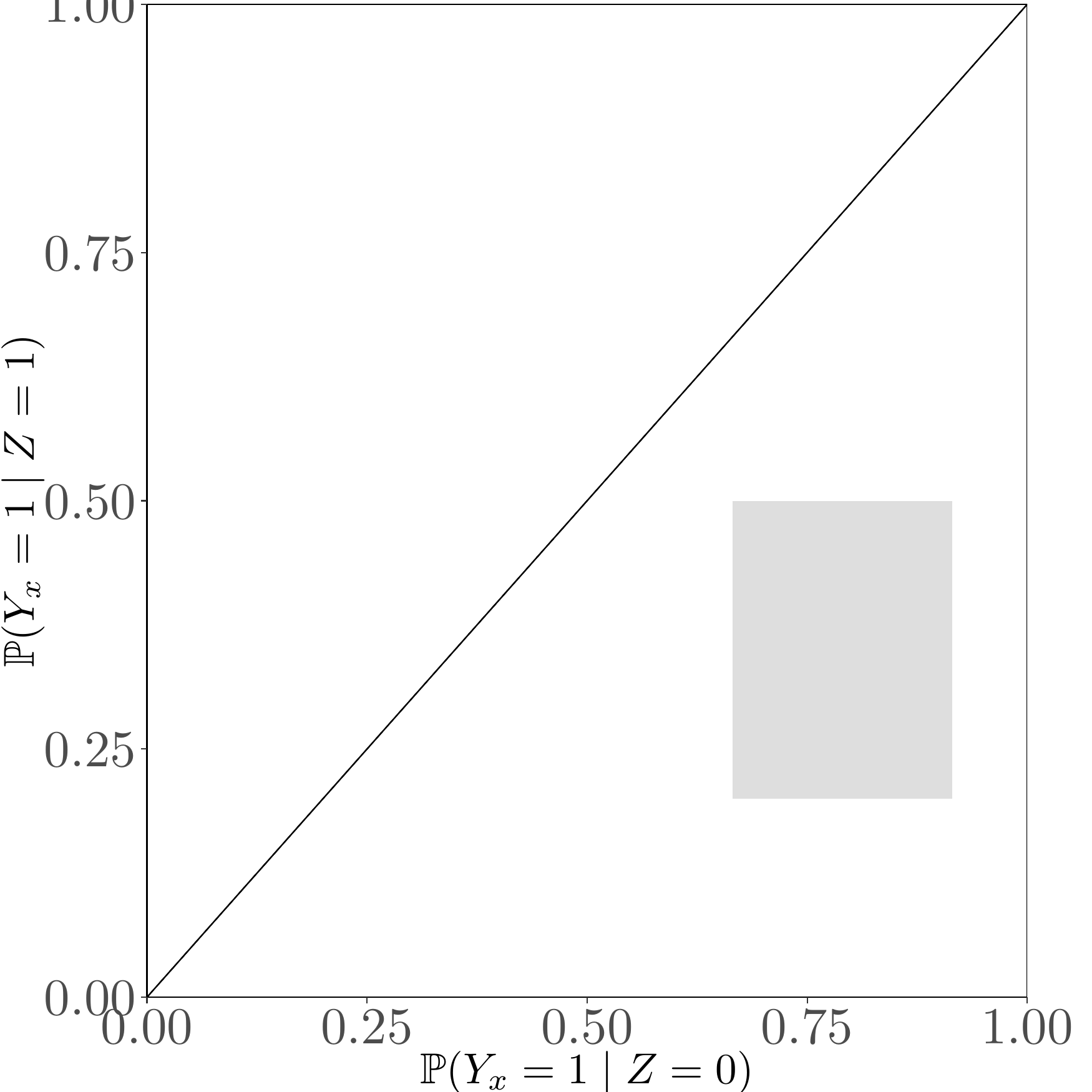}
\hspace{3mm}
\includegraphics[width=50mm]{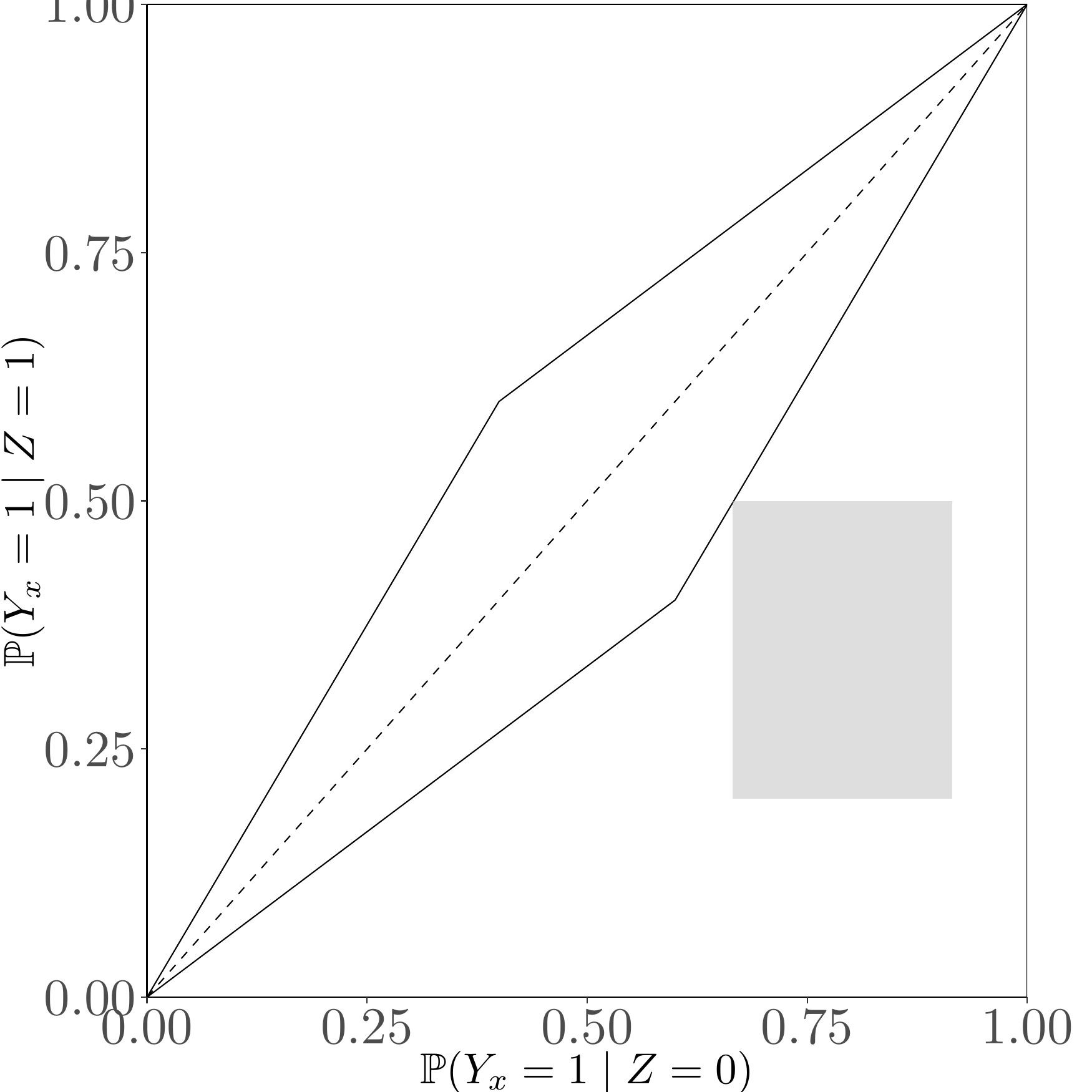}
\hspace{3mm}
\includegraphics[width=50mm]{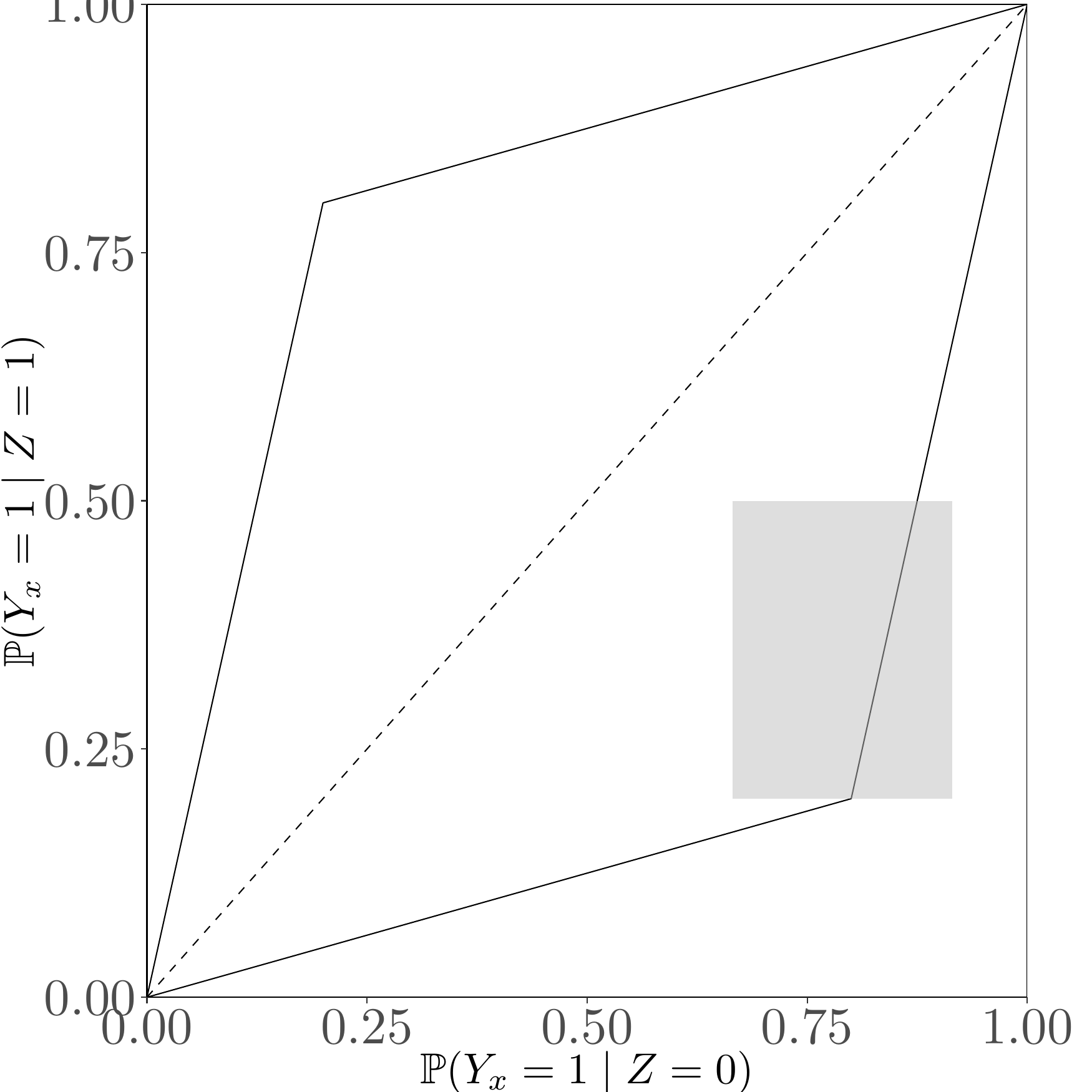}
\caption{Example identified set for $(\Prob(Y_x=1 \mid Z=1), \Prob(Y_x=1 \mid Z=0))$. Here $p_Z = 0.5$. Left: $c=0$. We see that the baseline model is falsified. Middle: $c = c^* = 0.1$, the falsification point. Here we have relaxed the assumption just enough that the model is no longer falsified. Right: $c = 0.3$. Relaxing independence even further increases the size of the identified set.}
\label{fig:FFbinaryYsingleIV_falsification}
\end{figure}

We next state some useful properties of the identified set given in theorem \ref{thm:BinaryYoneInstrumentCdep}.

\begin{proposition}\label{prop:PropertiesOneInstBounds}
Suppose the assumptions of theorem \ref{thm:BinaryYoneInstrumentCdep} hold. Then:
\begin{enumerate}
\item There is a value $c^* \in [0,1]$ such that $\Theta_0(c) \times \Theta_1(c)$ is empty for all $c \in [0,c^*)$ and is non-empty for all $c \in [c^*,1]$. The value $c^*$ is point identified.

\item For each $c \in [c^*,1]$, the set $\Theta_0(c) \times \Theta_1(c)$ is a closed, convex polytope in $[0,1]^4$.

\item The correspondence $\phi: [c^*,1] \rightrightarrows [0,1]^2$ defined by $\phi(c) = \Theta_0(c)\times \Theta_1(c)$ is continuous.
\end{enumerate}
\end{proposition}

The value $c^*$ is the falsification point. In the simple model with a single binary instrument it is possible to give a closed form expression for the falsification point, but we omit it for brevity. For part 2, recall that a polytope in $[0,1]^4$ is a set that can be written as the intersection of finitely many half-spaces. We use continuity of the correspondence $\phi$ below to show that bounds on functionals like ATE are continuous in $c$.

We next show how to use theorem \ref{thm:BinaryYoneInstrumentCdep} to get identified sets for the counterfactual probabilities $\Prob(Y_x=1)$ and for the average treatment effect. By the law of total probability,
\[
	\Prob(Y_x=1) = \Prob(Y_x=1 \mid Z=0) \Prob(Z=0) + \Prob(Y_x=1 \mid Z=1) \Prob(Z=1).
\]
The weights $\Prob(Z=0)$ and $\Prob(Z=1)$ are point identified. The joint identified set for $\Prob(Y_x=1 \mid Z=0)$ and $\Prob(Y_x=1 \mid Z=1)$ is given by $\Theta_x(c)$. Thus we can simply minimize and maximize the above convex combination over this set to obtain the identified set for $\Prob(Y_x=1)$. Hence we define
\[
	\overline{P}_x(c) = \sup_{(a_0,a_1) \in \Theta_x(c)} \big( a_0 \Prob(Z=0) + a_1 \Prob(Z=1) \big)
\]
and
\[
	\underline{P}_x(c) = \inf_{(a_0,a_1) \in \Theta_x(c)} \big( a_0 \Prob(Z=0) + a_1 \Prob(Z=1) \big).
\]
These are both finite dimensional linear programs and hence can be computed easily.

\begin{corollary}\label{corr:hetTrtBinATEbounds}
Suppose the assumptions of theorem \ref{thm:BinaryYoneInstrumentCdep} hold. Then:
\begin{enumerate}
\item Let $c\in [c^*,1]$. The identified set for $(\Prob(Y_0=1), \Prob(Y_1=1))$ is $I_0(c) \times I_1(c)$ where $I_x =[\underline{P}_x(c), \overline{P}_x(c)]$ for $x \in \{ 0, 1\}$. 

\item For $x \in \{0,1\}$, the functions $\underline{P}_x(c)$ and $\overline{P}_x(c)$ are continuous and monotonic over $c \in [c^*,1]$.

\item Let $c \in [c^*,1]$. The identified set for ATE is $[\underline{\text{ATE}}(c),\overline{\text{ATE}}(c)]$ where 
\[
	\underline{\text{ATE}}(c) = \underline{P}_1(c) - \overline{P}_0(c)
	\qquad \text{and} \qquad
	\overline{\text{ATE}}(c) = \overline{P}_1(c) - \underline{P}_0(c).
\]
\end{enumerate}
\end{corollary}

The falsification adaptive set for $\Prob(Y_x=1)$ is $[\underline{P}_x(c^*),\overline{P}_x(c^*)]$. This will be a singleton for one $x \in \{0,1\}$ and will typically be an interval with a nonempty interior for the other value of $x$. This can be seen in figure \ref{fig:FFbinaryYsingleIV_falsification}. In the middle plot, $\Prob(Y_x=1)$ is point identified since $\Theta_x(c^*)$ is a singleton. There is an analogous plot, however, for $\Prob(Y_{1-x}=1 \mid Z=z)$. In this plot, the box $\mathcal{H}_x$ will typically not also have a singleton intersection with the diamond $\mathcal{D}(c^*)$. Thus $\Prob(Y_{1-x}=1)$ will typically be partially identified. This discussion implies that ATE will typically be partially identified at the falsification point. That is: The falsification adaptive set for ATE, $[\underline{\text{ATE}}(c^*), \overline{\text{ATE}}(c^*)]$, will generally be an interval with a nonempty interior.

\subsubsection*{Generalization to Multiple Discrete Instruments}

We next consider the case where there are multiple discrete instruments. We relax each instrument exogeneity condition separately. We derive identified sets for a given relaxation of the instrument exogeneity conditions. As before, this allows us to derive the falsification frontier and the falsification adaptive set.

Suppose we observe $L \geq 1$ instruments, $Z = (Z_1,\ldots,Z_L) \in \R^L$. Let $\supp(Z_\ell) = \{ z_\ell^1, \ldots, z_\ell^{J_\ell} \}$ for some constant $J_\ell$. For notational simplicity we only consider the case $J_\ell = J$ for all $\ell \in \{1,\ldots,L \}$. As in theorem \ref{thm:BinaryYoneInstrumentCdep}, we will derive the joint identified set for the probabilities
\[
	\Prob(Y_x=1 \mid Z_\ell = z_\ell^j)
\]
for $x \in \{0,1\}$, $\ell \in \{1,\ldots,L\}$, and $j \in \{1,\ldots,J \}$. For each $x \in \{0,1\}$, there are $LJ$ conditional probabilities. As in the single binary instrument case, we make the following assumption to rule out trivial cases.

\begin{assump}{B\ref{assump:NonTrivialinstrument}$^\prime$}
For all $\ell \in \{1,\ldots,L\}$ and $j \in \{1,\ldots,J\}$,
\begin{enumerate}
\item $\Prob(Z_\ell = z_\ell^j) \in (0,1)$.

\item $\Prob(X=1 \mid Z_\ell=z_\ell^j) \in (0,1)$.
\end{enumerate}
\end{assump}

Let $Z_{-\ell}$ denote the vector of instruments without $Z_\ell$. $Z$ has $J^L$ support points while $Z_{-\ell}$ has $J^{L-1}$ support points. Let $\supp(Z_{-\ell}) = \{ z_{-\ell}^1,\ldots, z_{-\ell}^{J^{L-1}} \}$. By the law of total probability,
\begin{align}\label{eq:manyIVlawTotProb}
	&\Prob(Y_x = 1 \mid Z_\ell = z_\ell^j) \notag \\
	&\quad = \Prob(Y = 1 \mid X=x, Z_\ell = z_\ell^j) \Prob(X=x \mid Z_\ell = z_\ell^j) \notag \\
	&\qquad + \sum_{k=1}^{J^{L-1}} \Prob(Y_x = 1 \mid X=1-x, Z_\ell = z_\ell^j, Z_{-\ell} =z_{-\ell}^k) \Prob(X=1-x, Z_{-\ell} =z_{-\ell}^k \mid Z_\ell = z_\ell^j) \notag \\
	&\quad = \Prob(Y = 1,X=x \mid  Z_\ell = z_\ell^j) \notag \\
	&\qquad + \sum_{k=1}^{J^{L-1}} \Prob(Y_x = 1 \mid X=1-x, Z_\ell = z_\ell^j, Z_{-\ell} =z_{-\ell}^k) \frac{\Prob(X=1-x, Z_\ell = z_\ell^j, Z_{-\ell} = z_{-\ell}^k)}{\Prob(Z_\ell = z_\ell^j)}.
\end{align}
The terms
\[
	\Prob(Y = 1,X=x \mid  Z_\ell = z_\ell^j)
	\qquad \text{and} \qquad
	\frac{\Prob(X=1-x, Z_\ell = z_\ell^j, Z_{-\ell} = z_{-\ell}^k)}{\Prob(Z_\ell = z_\ell^j)}
\]
are observed directly in the data. In contrast, the probabilities
\[
	\Prob(Y_x = 1 \mid X=1-x, Z_\ell = z_\ell^j, Z_{-\ell} =z_{-\ell}^k)
\]
are not observed. We do know that they must lie in the interval $[0,1]$, however. For each $x \in \{0,1\}$, there are $J^L$ of these unknown probabilities. Thus for each $x \in \{0,1\}$ the $LJ$ probabilities $\Prob(Y_x=1 \mid Z_\ell = z_\ell^j)$ are determined by a system of $LJ$ known linear functions of $J^L$ unknowns which are all bounded between 0 and 1. Denote this set by
\begin{equation}\label{eq:discreteInstrumentBoxSet}
	\mathcal{H}_x = \left\{ \textbf{a} \in [0,1]^{LJ} : \textbf{a} = \textbf{b}_x + \textbf{A}_x \textbf{q} \ \text{ for some } \textbf{q}\in [0,1]^{J^L}\right\}
\end{equation}
where $\textbf{b}_x \in \R^{LJ}$ is defined by
\[
	[\textbf{b}_x]_{(\ell-1)J + j} = \Prob(Y=1,X=x \mid Z_\ell = z_\ell^j)
\]
for $\ell \in\{1,\ldots,L\}$, $j\in \{1,\ldots,J\}$ and $\textbf{A}_x \in \R^{LJ \times J^{L}}$ is defined by
\begin{align*}
	[\textbf{A}_x]_{(\ell-1)J + j, m}
	&= \frac{\Prob(X=1-x, Z = z^m)}{\Prob(Z_\ell = z_\ell^j)} \indicator \big( [z^m]_\ell = z_\ell^j \big)
\end{align*}
for $\ell \in\{1,\ldots,L\}$, $j\in \{1,\ldots,J\}$, $m \in \{1,\ldots J^{L}\}$. Here we use $z^m$ to denote the $m$th element of the support of $Z$. 

Although this notation is a bit complicated, it is simply the vector description of equation \eqref{eq:manyIVlawTotProb} above: The elements of $\textbf{b}_x$ are the intercepts, while the non-zero elements of the matrix $\textbf{A}_x$ are elements of the form
\[
	\Prob(X=1-x, Z_{-\ell} = z_{-\ell}^k \mid Z_\ell = z_\ell^j)
\]
for some $\ell \in \{1,\ldots,L\}$, $j \in \{1,\ldots,J\}$, and $k \in \{1,\ldots, J^{L-1} \}$.

As in the single instrument case, this set $\mathcal{H}_x$ is the identified set for the joint vector of $\Prob(Y_x=1 \mid Z_\ell = z_\ell^j)$ probabilities with no assumption on the dependence between $Y_x$ and each instrument $Z_\ell$. We next consider constraints on this dependence. As before, we relax independence using $c$-dependence. Specifically, we make the following assumption.

\begin{assump}{B\ref{assump:marginalIndepZ}$^{\prime \prime}$}(Instrument partial exogeneity).
For each $\ell \in \{1,\ldots, L \}$, let $c_\ell$ be a scalar between 0 and 1. Then
\[
	| \Prob(Z_\ell = z_\ell^j \mid Y_x = y_x) - \Prob(Z_\ell = z_\ell^j) | \leq c_\ell
\]
for each $j \in \{1,\ldots,J\}$, $\ell \in \{1,\ldots,L\}$, $y_x \in \supp(Y_x)$, and $x\in\{0,1\}$.
\end{assump}

This assumption generalizes $c$-dependence to allow for multiple discrete instruments. As in the single binary instrument case, $c_\ell = 0$ implies $Y_x \indep Z_\ell$, while $c_\ell > 0$ allows for some dependence. Importantly, we allow $c_\ell$ to vary across instruments. Let $c = (c_1,\ldots,c_L)$.

As shown in theorem \ref{thm:BinaryYManyInstrumentCdep} below, the set
\[
	\mathcal{D}(c) = \prod_{\ell = 1}^L \mathcal{D}_\ell(c_\ell)
\]
where
\begin{align}\label{eq:generalDiamondSet}
	\mathcal{D}_\ell(c_\ell) = \Bigg\{ a\in [0,1]^{J} : \ 
	&\text{For all $j=1,\ldots,J$}, \\
	&\Prob(Z_\ell = z_\ell^j) a_{j} - \min\{\Prob(Z_\ell = z_\ell^j) + c_\ell,1\} \sum_{k=1}^J \Prob(Z_\ell = z_\ell^k) a_{k} \leq 0, \notag \\
	 &\Prob(Z_\ell = z_\ell^j) a_{j} - \max\{\Prob(Z_\ell = z_\ell^j) - c_\ell,0\} \sum_{k=1}^J \Prob(Z_\ell = z_\ell^k) a_{k} \geq 0, \notag \\
	&\Prob(Z_\ell = z_\ell^j) (1-a_{j}) - \min\{\Prob(Z_\ell = z_\ell^j) + c_\ell,1\} \sum_{k=1}^J \Prob(Z_\ell = z_\ell^k) (1-a_{k}) \leq 0,\notag \\
	&\Prob(Z_\ell = z_\ell^j) (1-a_{j}) - \max\{\Prob(Z_\ell = z_\ell^j) - c_\ell,0\} \sum_{k=1}^J \Prob(Z_\ell = z_\ell^k) (1-a_{k}) \geq 0 \Bigg\} \notag
\end{align}
characterizes the constraints that B\ref{assump:marginalIndepZ}$^{\prime \prime}$ imposes on the probabilities $\Prob(Y_x=1 \mid Z_\ell = z_\ell^j)$. This set $\mathcal{D}(c)$ generalizes the parallelogram defined by equation \eqref{eq:diamondSet}. As in the single binary instrument case, it is defined by a finite number of linear inequalities.

Let
\[
	\Theta_x(c) = \mathcal{D}(c) \cap \mathcal{H}_x
\]
denote the intersection of these generalized diamond and box sets.

\begin{theorem}\label{thm:BinaryYManyInstrumentCdep}
Consider the binary outcome, binary treatment model \eqref{eq:generalOutcomeEquation}. Suppose the joint distribution of $(Y,X,Z)$ is known. Suppose B\ref{assump:marginalIndepZ}$^{\prime \prime}$ and B\ref{assump:NonTrivialinstrument}$^\prime$ hold. For $x \in \{0,1\}$, let
\[
	p_x = \big\{ \Prob(Y_x=1 \mid Z_\ell = z_\ell^j): \ell \in \{1,\ldots,L\}, j \in \{j,\ldots,J \} \big\}.
\]
Then the identified set for $(p_0,p_1)$ is $\Theta_0(c) \times \Theta_1(c)$. Consequently, the model is falsified if and only if this set is empty. 
\end{theorem}

This result is a direct generalization of theorem \ref{thm:BinaryYoneInstrumentCdep}. All of the discussion and interpretation there applies here as well. One new issue arises: The baseline model, $c = (0,\ldots,0)$, is falsifiable even if none of the instruments are falsifiable themselves. That is, consider the model that imposes $Y_x \indep Z_\ell$ but no constraints on the statistical dependence between the other instruments $Z_{-\ell}$ and $Y_x$. We can obtain this model by setting $c = \iota - e_\ell$ where $\iota$ is a vector of 1's and $e_\ell$ is the unit vector. Suppose this model is not falsified, regardless of which instrument $\ell$ we choose to impose exogeneity on. Despite this, it is nonetheless possible that the baseline model $c = (0,\ldots,0)$ is falsified. In this case the instruments are incompatible with each other.

As in proposition \ref{prop:PropertiesOneInstBounds} in the single binary instrument case, we next state some useful properties of the identified set given in theorem \ref{thm:BinaryYManyInstrumentCdep}.

\begin{proposition}\label{prop:PropertiesManyInstBounds}
Suppose the assumptions of theorem \ref{thm:BinaryYManyInstrumentCdep} hold. Then:
\begin{enumerate}
\item There is a closed set $\mathcal{C} \subseteq [0,1]^L$ such that $\Theta_0(c)\times \Theta_1(c)$ is non-empty for all $c \in \mathcal{C}$. The set $\mathcal{C}$ is point identified.

\item For each $c \in \mathcal{C}$, the set $\Theta_0(c) \times \Theta_1(c)$ is a closed, convex polytope.

\item The correspondence $\phi: \mathcal{C} \rightrightarrows [0,1]^{2LJ}$ defined by $\phi(c) = \Theta_0(c)\times \Theta_1(c)$ is continuous at any $c \in \mathcal{C}$ such that $\text{int} (\mathcal{H}_x) \cap \text{int} (\mathcal{D}(c) ) \neq \emptyset$ for $x \in \{ 0,1 \}$. 
\end{enumerate}
\end{proposition}

The set $\mathcal{C}$ is the set of all instrument partial exogeneity assumptions which do \emph{not} refute the model. The complement of this set is the set of all instrument partial exogeneity assumptions which are refuted. The falsification frontier in this case is defined as follows. Let
\[
	S_{< c} = \{ \tilde{c} \in [0,1]^L : \tilde{c} < c \}
\]
be the set of points which are smaller than $c$. Then 
\begin{equation}\label{eq:FFforHetModel}
	\text{FF} = \{ c \in [0,1]^L : S_{<c} \subseteq ([0,1]^L \setminus \mathcal{C}) \}.
\end{equation}
This is the set of all points $c$ such that any points smaller than $c$ are not in $\mathcal{C}$. 
This frontier is identified from the data since $\mathcal{C}$ is defined implicitly by $\{ \Theta_0(c) \times \Theta_1(c): c\in [0,1]^L\}$, which is point identified.

We next show how to use theorem \ref{thm:BinaryYManyInstrumentCdep} to get identified sets for counterfactual probabilities $\Prob(Y_x=1)$ and for the average treatment effect. Define
\[
	\overline{P}_x(c)
	= \sup_{\textbf{a} \in \Theta_x(c)} \; \sum_{j = 1}^J \sum_{\ell = 1}^L a_{(\ell-1)J + j}\Prob(Z_\ell = z_\ell^j)
\]
and
\[
	\underline{P}_x(c)
	= \inf_{\textbf{a} \in \Theta_x(c)} \; \sum_{j = 1}^J \sum_{\ell = 1}^L a_{(\ell-1)J + j}\Prob(Z_\ell = z_\ell^j).
\]
As in the single binary instrument case, these are both finite dimensional linear programs and hence can be computed easily.

\begin{corollary}\label{corr:hetTrtMultipleATEbounds}
Suppose the assumptions of theorem \ref{thm:BinaryYManyInstrumentCdep} hold. Then:
\begin{enumerate}
\item Let $c \in \mathcal{C}$. The identified set for $(\Prob(Y_0=1), \Prob(Y_1=1))$ is $I_0(c) \times I_1(c)$ where $I_x =[\underline{P}_x(c), \overline{P}_x(c)]$.

\item For $x \in \{0,1\}$, the functions $\underline{P}_x(c)$ and $\overline{P}_x(c)$ are continuous at any $c \in \mathcal{C}$ such that $\text{int} (\mathcal{H}_x) \cap \text{int} (\mathcal{D}(c) ) \neq \emptyset$. Moreover, they are monotonic in each component of $c$.

\item Let $c \in \mathcal{C}$. The identified set for ATE is $[\underline{\text{ATE}}(c),\overline{\text{ATE}}(c)]$ where
\[
	\underline{\text{ATE}}(c) = \underline{P}_1(c) - \overline{P}_0(c)
	\qquad \text{and} \qquad
	\overline{\text{ATE}}(c) = \overline{P}_1(c) - \underline{P}_0(c).
\]
\end{enumerate}
\end{corollary}

In the single binary instrument case, there is a simple closed form expression for the falsification point $c^*$. Consequently, computing the falsification adaptive set for ATE can be done easily by using linear programming to compute $\overline{P}_x(c^*)$ and $\underline{P}_x(c^*)$. In the multiple discrete instrument case, there is not a simple closed form expression for the falsification frontier. Nonetheless, it can be computed by a sequence of convex optimization problems. Specifically, the sets $\mathcal{D}(c)$ and $\mathcal{H}_x$ are polytopes in $\R^{LJ}$. Define the `distance' between these polytopes as
\[
	\text{dist}(\mathcal{D}(c), \mathcal{H}_x) = \inf_{(v_1,v_2) \in \mathcal{D}(c) \times \mathcal{H}_x} \; \| v_1 - v_2 \|.
\]
The model is refuted if and only if $\text{dist}(\mathcal{D}(c), \mathcal{H}_x) > 0$ for some $x \in \{0,1 \}$. For any $c$, the distance between these polytopes can be computed via convex optimization. Thus one can do a grid search over $c$ and collect all values where the distance is zero for both $x \in \{0,1\}$, and hence the model is not refuted. This is a numerical approximation to the set $\mathcal{C}$. We can then use this set to approximate the falsification frontier.

Given this approximate falsification frontier, we can compute the falsification adaptive set for ATE as follows. For any $c$, we can compute $\overline{P}_x(c)$ and $\underline{P}_x(c)$ via linear programming, since the objective function is linear and the constraint set is a polytope in $\R^{LJ}$. Thus we can simply compute these bounds for all $c \in \text{FF}$. The union of those bounds is the falsification adaptive set. Identified sets for ATE at values of $c$ beyond the falsification frontier can be computed similarly.

\subsection{Continuous Outcomes}\label{sec:hetTrtCts}

In section \ref{sec:hetTrtBinary} we considered binary outcomes, binary treatment, and multiple discrete instruments. In this case, the joint distribution of the data is characterized by a finite dimensional vector. It is well known that with discretely distributed data identified sets can often be computed using linear programming; see our literature review in appendix \ref{sec:relatedLiterature}. There are two advantages, however, to deriving analytical results like those of theorems \ref{thm:BinaryYoneInstrumentCdep} and \ref{thm:BinaryYManyInstrumentCdep}. First, they can explain the role of identifying assumptions, as illustrated in figures \ref{fig:FFbinaryYsingleIV} and \ref{fig:FFbinaryYsingleIV_falsification}. 
Second, they allow us to generalize the results to cases where brute force computation is costly or impossible. In this section, we show that the analytical results we derived under binary outcomes generalize to continuous outcomes. This leads us to a relatively simple and feasible approach for computing identified sets, falsification frontiers, and the falsification adaptive set under relaxations of instrument exogeneity with continuous outcomes.

We begin by assuming that outcomes are continuously distributed.

\begin{het2Assump}%
\label{assump:contsupport}
For any $x, x' \in \{0,1\}$ and $z \in \{0,1\}^L$, the distribution of $Y_x \mid X=x', Z=z$ is continuous with respect to the Lebesgue measure.
\end{het2Assump}

Assumption C\ref{assump:contsupport} supposes that, conditional on the treatment and instruments, potential outcomes are continuously distributed. It implies that, conditional on the treatment and instruments, observed outcomes are also continuously distributed. We also suppose there are $L$ binary instruments $Z = (Z_1,\ldots,Z_L)$. We can allow for discrete instruments as in section \ref{sec:hetTrtBinary}, but we only consider binary instruments to simplify the notation.

Let $\mathcal{P}(A)$ denote the set of densities (with respect to the Lebesgue measure) with support contained in $A$. This set can be written as 
\[
	\mathcal{P}(A) = \left\{ p: \int_A p(y) \; dy = 1, \ p(y) \geq 0 \text{ for all } y \in \R \right\}.
\]
We begin by considering the baseline case where the instruments are exogenous. When there is just a single instrument, this setting was studied in \cite{Kitagawa2009}. The following result is essentially his proposition 3.1.

\begin{proposition}\label{prop:kitagawa3point1}
Let $L=1$. Suppose C\ref{assump:contsupport} (continuous outcomes) and B\ref{assump:marginalIndepZ} (instrument exogeneity) hold. For each $x \in \{0,1\}$, let $\mathcal{Y}_x$ be a known set and suppose $\supp(Y_x) \subseteq \mathcal{Y}_x$. Then the identified set for $(f_{Y_1},f_{Y_0})$ is 
\[
	\left\{ f_1 \in \mathcal{P}(\mathcal{Y}_1): f_1(\cdot) \geq \max_{z=0,1} \; f_{Y,X|Z}(\cdot,1 \mid z)\right\}\\
	\times \left\{ f_0 \in \mathcal{P}(\mathcal{Y}_0): f_0(\cdot) \geq \max_{z=0,1} \; f_{Y,X|Z}(\cdot,0 \mid z)\right\}.
\]
Consequently, the model is refuted if
\[
	\int_{\mathcal{Y}_x} \max_{z=0,1} \; f_{Y,X|Z}(y,x \mid z) \; dy > 1
\]
for some $x \in \{0,1 \}$.
\end{proposition}

The assumption that $\supp(Y_x) \subseteq \mathcal{Y}_x$ is not restrictive since we can let $\mathcal{Y}_x = \R$, but this notation allows us to impose bounds on the support if they are known.

We generalize proposition \ref{prop:kitagawa3point1} in two ways. First, we consider an arbitrary number $L$ of binary instruments. Second, we relax instrument exogeneity B\ref{assump:marginalIndepZ} to instrument partial exogeneity B\ref{assump:marginalIndepZ}$^{\prime \prime}$. As discussed above, our proof strategy generalizes the result in section \ref{sec:hetTrtBinary}. Independence of each instrument with each potential outcome requires that 
\[
	f_{Y_x|Z_\ell}(y \mid 1) = f_{Y_x|Z_\ell}(y \mid 0)
\]
for all $y \in \R$ and $\ell \in \{1, \ldots, L \}$. This constraint is analogous to the diagonal constraint in the left plots of figures \ref{fig:FFbinaryYsingleIV} and \ref{fig:FFbinaryYsingleIV_falsification}. The $c$-dependence assumption B\ref{assump:marginalIndepZ}$^{\prime \prime}$ does not require the densities $f_{Y_x | Z_\ell}(\cdot \mid z)$ to be the same for all $z \in \{0,1\}$, but it bounds the variation across $z$ by a magnitude determined by the vector $c$ of sensitivity parameters. We make this constraint precise below.

Analogously to section \ref{sec:hetTrtBinary}, we will derive the joint identified set for the conditional densities 
\[
	\left\{f_{Y_x|Z_\ell}(\cdot \mid z): x \in \{0,1\}, z \in\{0,1\}, \ell \in \{ 1, \ldots, L \} \right\}
\]
under instrument partial exogeneity. We define the following notation for these densities. For all $\ell \in \{1, \ldots, L\}$ and $x\in\{0,1\}$, let
\[
	\mathbf{f}_{Y_x|Z_\ell}
	=
	\begin{pmatrix}
		f_{Y_x|Z_\ell }(\cdot \mid 0) \\
		f_{Y_x|Z_\ell }(\cdot \mid 1)
	\end{pmatrix}
	\in \mathcal{P}(\supp(Y_x \mid Z_\ell = 0))\times \mathcal{P}(\supp(Y_x \mid Z_\ell = 1))
\]
and
\[
	(\mathbf{f}_{Y_x|Z_\ell})_{\text{all } \ell}
	= 
	\begin{pmatrix}
		\mathbf{f}_{Y_x|Z_1} \\
		\vdots \\
		\mathbf{f}_{Y_x|Z_L}
	\end{pmatrix}
	\in \prod_{\ell=1}^L \Big( \mathcal{P}(\supp(Y_x \mid Z_\ell = 0))\times \mathcal{P}(\supp(Y_x \mid Z_\ell = 1))\Big).
\]
Thus we will derive the joint identified set for the vectors $(\mathbf{f}_{Y_0|Z_\ell})_{\text{all } \ell}$ and $(\mathbf{f}_{Y_1|Z_\ell})_{\text{all } \ell}$.

This identified set has a structure analogous to that in section \ref{sec:hetTrtBinary}. Thus we first define the continuous outcome generalization of the set $\mathcal{D}(c)$ defined in equations \eqref{eq:diamondSet} and \eqref{eq:generalDiamondSet}. To do this, define
\[
	k_z^\ell(c_\ell) = \frac{\Prob(Z_\ell=z)\max\{\Prob(Z_\ell=1-z) - c_\ell,0\}}{\Prob(Z_\ell = 1-z)\min\{\Prob(Z_\ell=z)+c_\ell,1\}}
\]
for $\ell \in \{1, \ldots, L \}$ and $z \in \{0,1 \}$. This function generalizes equation \eqref{eq:kandK}. For each $x \in \{0,1\}$, let $\mathcal{Y}_x$ be a known set. We assume this set contains the support of $Y_x \mid X,Z$ for all values of the conditioning variables. Define
\begin{align*}
	\mathcal{D}_{x,\ell}(c_\ell) 
	&= \Big\{(f(\cdot \mid 0),f(\cdot \mid 1)) 
	\in \mathcal{P}(\mathcal{Y}_x)^2 :
	f(\cdot \mid 0) k_0^\ell(c_\ell) \leq f(\cdot \mid 1) \text{ and } f(\cdot \mid 1) k_1^\ell(c_\ell) \leq f(\cdot \mid 0) \Big\} %
\end{align*}
for $\ell \in \{1, \ldots, L \}$ and $x \in \{0,1\}$. This set generalizes equations \eqref{eq:diamondSet} and \eqref{eq:generalDiamondSet}. It depends on $x \in \{0,1\}$ only via the possibility that $\mathcal{Y}_x$
might depend on $x$. Finally, take the product to get
\[
	\mathcal{D}_x(c) = \prod_{\ell = 1}^L \mathcal{D}_{x,\ell}(c_\ell)
\]
We will show that $\mathcal{D}_x(c)$ is the set of densities $(\mathbf{f}_{Y_x|Z_\ell})_{\text{all } \ell}$ consistent with $c$-dependence (B\ref{assump:marginalIndepZ}$^{\prime \prime}$). Note that setting $c_\ell =0$ implies that $f_{Y_x|Z_\ell} = f_{Y_x}$, as required by the baseline model.

Next we generalize the set $\mathcal{H}_x$ defined in equations \eqref{eq:boxSet} and \eqref{eq:discreteInstrumentBoxSet}. With continuous outcomes this will be the set of proper densities $(\mathbf{f}_{Y_x|Z_\ell})_{\text{all } \ell}$ consistent with the observed data.
Let $\mathbf{A}_x$ be the $2L \times 2^L$ matrix such that for $\ell \in \{1, \ldots, L\}$ and $m \in \{ 1,\ldots,2^L \}$:
\[
	[\mathbf{A}_x]_{(\ell-1)2+1,m} = \frac{\Prob(X=1-x, Z = z^m)}{\Prob(Z_\ell = 0)} \indicator \big( [z^m]_\ell = 0 \big)
\]
and
\[
	[\mathbf{A}_x]_{(\ell-1)2+2,m} = \frac{\Prob(X=1-x, Z = z^m)}{\Prob(Z_\ell = 1)} \indicator \big( [z^m]_\ell = 1 \big).
\]
Here we use $z^m$ to denote the $m$th element of the support of $Z$. Let
\[
	\textbf{f}_{Y,X|Z_\ell}(\cdot,x) 
	= 
	\begin{pmatrix}
		f_{Y,X|Z_\ell}(\cdot,x \mid 0)\\
		f_{Y,X|Z_\ell}(\cdot,x \mid 1)
	\end{pmatrix}
	\qquad \text{and} \qquad
	(\mathbf{f}_{Y,X|Z_\ell}(\cdot,x) )_{\text{all } \ell} 
	=
	\begin{pmatrix}
		\mathbf{f}_{Y,X|Z_1}(\cdot,x) \\
		\vdots \\
		\mathbf{f}_{Y,X|Z_L}(\cdot,x)
	\end{pmatrix}.
\]
These are vectors of observed densities of outcomes and treatment conditional on different instruments. 
Define
\[
	\mathcal{H}_x
	= \left\{ \mathbf{f}_x \in \mathcal{P}(\mathcal{Y}_x)^{2L}
	: \mathbf{f}_x = (\mathbf{f}_{Y,X|Z_\ell}(\cdot,x))_{\text{all } \ell} + \mathbf{A}_x \mathbf{q} \ \text{ for some }  
	\mathbf{q} \in 
	\mathcal{P}(\mathcal{Y}_x)^{2^L}
	\right\}.
\]
This set generalizes equation \eqref{eq:discreteInstrumentBoxSet}. Finally, let
\[
	\Theta_x(c) = \mathcal{D}_x(c) \cap \mathcal{H}_x.
\]

\begin{theorem}\label{thm:ContYManyInstrumentCdep}
Consider the continuous outcome, binary treatment model \eqref{eq:generalOutcomeEquation}. Suppose the joint distribution of $(Y,X,Z)$ is known. Suppose C\ref{assump:contsupport}, B\ref{assump:marginalIndepZ}$^{\prime \prime}$, and B\ref{assump:NonTrivialinstrument}$^\prime$ hold. For each $x \in \{0,1\}$, let $\mathcal{Y}_x$ be a known set and suppose $\supp(Y_x \mid X=x',Z=z^m) \subseteq \mathcal{Y}_x$ for all $x' \in \{0,1\}$ and $m \in \{1,\ldots,2^L\}$. Then the identified set for $((\mathbf{f}_{Y_0|Z_\ell})_{\text{all } \ell}, (\mathbf{f}_{Y_1|Z_\ell})_{\text{all } \ell})$ is
\[
	\Theta_0(c) \times \Theta_1(c).
\]
Consequently, the model is falsified if and only if this set is empty. 
\end{theorem}

As in proposition \ref{prop:kitagawa3point1}, knowledge of the set $\mathcal{Y}_x$ is not restrictive since we can let $\mathcal{Y}_x = \R$, but this notation allows us to impose bounds on its support if they are known. We could also generalize this set to allow it to vary with $x'$ and $z^m$, at the cost of additional notation.

The identified set $\Theta_0(c) \times \Theta_1(c)$ is a convex set of functions defined by linear equality and inequality constraints. It weakly expands as components of $c$ increase. It nests the identified set under the baseline independence assumption ($c=0$) and the identified set under no assumptions on the dependence between potential outcomes and instruments ($c = \iota_L$, an $L$-vector of ones).

For a given $c$, this identified set can be empty, which implies the model at that $c$ is falsified. Define
\[
	\mathcal{C} = \{c \in [0,1]^L: \Theta_0(c) \times \Theta_1(c) \neq \emptyset\}.
\]
This is the set of all partial exogeneity assumptions which are not refuted. It has the property that $c \in \mathcal{C}$ and $c' \geq c$ implies $c'\in\mathcal{C}$. Using this set, we can compute the falsification frontier as in equation \eqref{eq:FFforHetModel}.

As in section \ref{sec:hetTrtBinary}, we can obtain the identified set for functionals of the densities $f_{Y_x \mid Z_\ell}$ via convex optimization. 
For example, the identified set for $(f_{Y_0}(\cdot), f_{Y_1}(\cdot))$ is $I_0(c) \times I_1(c)$ where
\begin{align*}
	I_x(c) = \{&f \in \mathcal{P}(\mathcal{Y}_x) : f(\cdot) = f_{Y_x|Z_1}(\cdot \mid 0) \Prob(Z_1 = 0) + f_{Y_x|Z_1}(\cdot \mid 1) \Prob(Z_1 = 1) \\
	&\quad \text{ where $(f_{Y_x | Z_1}( \cdot \mid 0), f_{Y_x | Z_1}(\cdot \mid 1))$ is the first component of some $(\mathbf{f}_{Y_x | Z_\ell})_{\text{all } \ell}$ in $\Theta_x(c)$ } \}.
\end{align*}
Given this identified set, suppose we are interested in a linear functional
\[
	\Gamma(f_1,f_0) 
	= \int_{\mathcal{Y}_1} \omega_1(y_1) f_1(y_1) \; dy_1 
	+ \int_{\mathcal{Y}_0} \omega_0(y_0) f_0(y_0) \; dy_0
\]
where $\omega_1$ and $\omega_0$ are known weight functions. For example, with $\omega_1(y) = y$ and $\omega_0(y) = -y$, $\Gamma(f_1,f_0)$ equals the average treatment effect. Let
\[
	\overline{\Gamma}(c)
	= \sup_{f_1 \in I_1(c)} \int_{\mathcal{Y}_1} \omega_1(y_1)f_1(y_1) \; dy_1 
	+ \sup_{f_0 \in I_0(c)}\int_{\mathcal{Y}_0} \omega_0(y_0)f_0(y_0) \; dy_0
\]
and
\[
	\underline{\Gamma}(c) 
	= \inf_{f_1 \in I_1(c)} \int_{\mathcal{Y}_1} \omega_1(y_1)f_1(y_1) \; dy_1 
	+ \inf_{f_0 \in I_0(c)}\int_{\mathcal{Y}_0} \omega_0(y_0)f_0(y_0) \; dy_0.
\]
Then $[\underline{\Gamma}(c), \overline{\Gamma}(c)]$ is the closure of the identified set for $\Gamma(f_{Y_1},f_{Y_0})$. This follows by convexity of the identified set in theorem \ref{thm:ContYManyInstrumentCdep}. These bounds are weakly increasing in each component of $c$. We discuss computation next.

\subsubsection*{Computation}

The identified set $\Theta_0(c) \times \Theta_1(c)$ is infinite dimensional, since it is a set of nonparametric densities. Suppose we are interested in linear functionals $\Gamma(f_1,f_0)$. Even though the corresponding identified set---when it is nonempty---is an interval, we have characterized it via optimization over the infinite dimensional spaces $\Theta_x(c)$. Here we discuss one approach to computing these identified sets in practice. When the baseline model is falsified, this will also allow us to compute approximations to the falsification frontier and the falsification adaptive set. Specifically, we will approximate the infinite dimensional space of densities of interest by a finite dimensional sieve space. Similar approximations of identified sets have been used in \cite{MogstadSantosTorgovitsky2016}, for example. Alternatively, it is likely that the computational approach developed in \cite{ChristensenConnault2019} could be adapted to our setting. Unlike the sieve based approach we consider below, the dimension of their optimization problem does not depend on the precision of the density approximation. We leave the application of their approach to our problem to future work.

For simplicity, let $\mathcal{Y}_x = [0,1]$ for $x \in \{0,1\}$. This restriction can be relaxed by transforming the outcome variable to have support in the unit interval. We also assume that for any $x, x' \in \{0,1\}$ and $z \in \{0,1\}^L$ the density $f_{Y_x|X,Z}(\cdot \mid x',z)$ is continuous on $[0,1]$. Let $\widetilde{\mathcal{P}}([0,1])$ denote this set of continuous pdfs. Let 
\begin{align*}
	\mathcal{F}^{M} 
	&= \left\{f^{M}(y) = \sum_{m=0}^M a_{m}b_m(y): (a_{1},\ldots,a_{M}) = \textbf{a}^M \in \mathcal{A}_M\right\}\\
	&= \left\{f^M(y) = \textbf{a}^{M\prime}\textbf{b}^M(y): \textbf{a}^M \in \mathcal{A}_M\right\}
\end{align*}
where $\mathcal{A}_M \subseteq \R^M$ and $\{ b_m(\cdot) \}$ are known basis functions. We assume that every element of $\widetilde{\mathcal{P}}([0,1])$ can be approximated by a sequence of elements in $\mathcal{F}_{M}$ as $M\rightarrow\infty$. For example, $\mathcal{F}^M$ could be a set of Bernstein polynomials, which can uniformly approximate any continuous function on $[0,1]$.
For Bernstein polynomials, 
\[
	\mathcal{A}_M 
	= \left\{ \textbf{a}^M : a_m \geq 0, \sum_{m=0}^M a_m = 1\right\}.
\]
Hence $\mathcal{A}_M$ is a closed, bounded, and convex subset of $\R^M$.

To approximate $\Theta_x(c) = \mathcal{D}_x(c) \cap \mathcal{H}_x$, we approximate each of the sets $\mathcal{D}_x(c)$ and $\mathcal{H}_x$. We first consider a finite dimensional approximation to the set $\mathcal{D}_x(c)$. Since $\mathcal{Y}_x = [0,1]$ does not depend on $x$, we let $\mathcal{D}(c) = \mathcal{D}_x(c)$. Let 
\begin{align*}
	\mathcal{D}_\ell^M(c_\ell) 
	&= \Bigg\{(f(\cdot \mid 0),f(\cdot \mid 1)) = \left( \textbf{a}^M(0)'\textbf{b}^M(\cdot), \ \textbf{a}^M(1)'\textbf{b}^M(\cdot)\right) : \textbf{a}^M(0),\textbf{a}^M(1)\in \mathcal{A}_M, \\
	&\hspace{5cm} \left( \textbf{a}^M(0)k_0^\ell(c_\ell) - \textbf{a}^M(1) \right)' \textbf{b}^M(y)\leq 0 \ \text{ and} \\
	&\hspace{5cm} \left( \textbf{a}^M(1)k_1^\ell(c_\ell) - \textbf{a}^M(0) \right)' \textbf{b}^M(y)\leq 0 \ \text{ for all } y \in [0,1] \Bigg\}.
\end{align*}
Define the set 
\begin{align*}
	\mathcal{A}_\ell^M(c_\ell) 
	= \Bigg\{ \left( \textbf{a}^M(0), \textbf{a}^M(1) \right) \in \mathcal{A}_M \times \mathcal{A}_M : 
	&\left(\textbf{a}^M(0)k_0^\ell(c_\ell) - \textbf{a}^M(1) \right)' \textbf{b}^M(y)\leq 0 \ \text{ and} \\
	&\left(\textbf{a}^M(1)k_1^\ell(c_\ell) - \textbf{a}^M(0) \right)' \textbf{b}^M(y)\leq 0 \ \text{ for all } y \in [0,1] \Bigg\}.
\end{align*}
Then
\[
	\mathcal{D}_\ell^M(c_\ell)
	= \left\{(f(\cdot \mid 0),f(\cdot \mid 1)) = \left( \textbf{a}^M(0)'\textbf{b}^M(\cdot), \textbf{a}^M(1)'\textbf{b}^M(\cdot)\right): \left( \textbf{a}^M(0),\textbf{a}^M(1) \right) \in \mathcal{A}^M_\ell(c_\ell)\right\}.
\]
The set $\mathcal{A}_\ell^M(c_\ell)$ is closed and convex since it is defined by a finite number of weak inequalities. There are a continuum of inequalities, however. In practice, we also approximate $\mathcal{A}_\ell^M(c_\ell)$ as follows. Pick a grid of points $\{y_1,\ldots,y_N\} \subseteq [0,1]$ that becomes dense in $[0,1]$ as $N \rightarrow \infty$. Then let $\mathcal{A}_\ell^{M,N}(c_\ell)$ be the set of coefficients $\textbf{a}^M(0)$ and $\textbf{a}^M(1)$ satisfying
\[
	\left( \textbf{a}^M(0)k_0^\ell(c_\ell) - \textbf{a}^M(1) \right)' \textbf{b}^M(y_n)\leq 0
\]
and
\[
	\left( \textbf{a}^M(1)k_1^\ell(c_\ell) - \textbf{a}^M(0) \right)' \textbf{b}^M(y_n)\leq 0
\]
for all $n \in \{1 ,\ldots,N \}$. This is a finite set of linear inequalities. We keep the $N$ implicit in the notation from here on.

Approximate the overall set $\mathcal{D}(c)$ by the product
\[
	\mathcal{D}^M(c) = \prod_{\ell = 1}^L \mathcal{D}^M_\ell(c_\ell).
\]
This set is characterized by
\[
	\mathcal{A}^M(c) = \prod_{\ell = 1}^L \mathcal{A}^M_\ell(c_\ell),
\]
a closed, convex, and bounded set in $\R^{2ML}$.

We similarly approximate $\mathcal{H}_x$ by a finite dimensional set. Let
\[
	\mathbf{f}_{Y|X,Z_\ell}(\cdot \mid x)
	= 
	\begin{pmatrix}
		f_{Y|X,Z_\ell}(\cdot \mid x,0)\\
		f_{Y|X,Z_\ell}(\cdot \mid x,1)
	\end{pmatrix}
	\qquad \text{and} \qquad
	(\mathbf{f}_{Y|X,Z}(\cdot \mid x) )_{\text{all } \ell}
	= 
	\begin{pmatrix}
		\mathbf{f}_{Y|X,Z_1}(\cdot \mid x)\\
		\vdots\\
		\mathbf{f}_{Y|X,Z_L}(\cdot \mid x)
	\end{pmatrix}.
\]
These conditional pdfs are point identified directly from the data. For each $z \in \{0,1\}$, we approximate $f_{Y|X,Z_\ell}(\cdot \mid x,z)$ by
\begin{align*}
		f^M_{Y|X,Z_\ell}(\cdot \mid x,z) 
		&= \sum_{m=0}^M \xi_\ell(x,z)_m b_m(\cdot)\\
		&= \mathbf{\xi}_\ell(x,z)^{M\prime}\textbf{b}^M(\cdot)
\end{align*}
where $\{ b_m(\cdot) \}$ are known basis functions. When these basis functions are orthonormal, so that
\[
	\int_0^1 b_m(y) b_n(y) \; dy= \indicator(m=n), 
\]
this approximation can be obtained choosing
\[
	\xi_\ell(x,z)_m
	= \int_0^1 f_{Y|X,Z_\ell}(y \mid x,z)b_m(y) \; dy.
\]
More concretely, if we use a Bernstein polynomial approximation, then
\[
	\xi_\ell(x,z)_m = f_{Y|X,Z_\ell} \left( \frac{m}{M} \mid x,z \right).
\]
Here we have approximated the conditional distribution of $Y \mid X,Z$. But $\mathcal{H}_x$ is defined in terms of distributions of $(Y,X) \mid Z$. We can write the vector $(\textbf{f}_{Y,X|Z}(y,x))_{\text{all } \ell}$ as
\[
	(\textbf{f}_{Y,X|Z}(y,x))_{\text{all } \ell}
	= \textbf{B}_x (\textbf{f}_{Y|X,Z}(y \mid x))_{\text{all } \ell}
\]
where
\[
	\mathbf{B}_x
	= 
	\begin{pmatrix}
	\Prob(X=x \mid Z_1 = 0) & 0 & \ldots & 0\\
	0 & \Prob(X=x \mid Z_1 = 1) &\ldots &  0\\
	\vdots& \vdots & \ddots & \vdots \\
	0 & 0 & \ldots  & \Prob(X=x \mid Z_L = 1)
	\end{pmatrix}.
\]
Hence its approximation can be written as
\[
	(\textbf{f}_{Y,X|Z}(y,x))_{\text{all } \ell} = \textbf{B}_x \Xi^M \textbf{b}^M(y)
\]
where
\[
	\Xi^M =
	\begin{pmatrix}
	\xi^1_1(x,0) & \xi^2_1(x,0) & \ldots & \xi^M_1(x,0) \\ 
	\xi^1_1(x,1) & \xi^2_1(x,1) & \ldots & \xi^M_1(x,1) \\ 
	\vdots & \vdots &\ddots & \vdots \\
	\xi^1_L(x,0) & \xi^2_L(x,0) & \ldots & \xi^M_L(x,0) \\ 
	\xi^1_L(x,1) & \xi^2_L(x,1) & \ldots & \xi^M_L(x,1)
	\end{pmatrix}.
\]
Next, approximate $\mathcal{P}(\mathcal{Y}_x)^{2^L}$ by
\[
	\left\{ \textbf{f}(\cdot) = \Gamma \textbf{b}^M(\cdot): \gamma_1^M, \ldots, \gamma_{2^L}^M \in \mathcal{A}^M\right\}
\]
where
\[
	\Gamma = 
	\begin{pmatrix}
		\gamma_1^{M\prime}\\
		\vdots\\
		\gamma_{2^L}^{M\prime}
	\end{pmatrix}
\]
is a $2^L \times M$ matrix of coefficients.
Thus we approximate $\mathcal{H}_x$ by
\begin{align*}
	\mathcal{H}_x^M 
	&= \left\{\textbf{f}(y) = (\textbf{B}_x \Xi^M + \textbf{A}_x \Gamma)\textbf{b}^M(y): \gamma_1, \ldots,\gamma_{2^L} \in \mathcal{A}_M\right\}\\
	&= \{\textbf{f}(y) = \textbf{E}^M \textbf{b}^M(y): \textbf{E}^M \in \textbf{B}_x \Xi^M + \textbf{A}_x \mathcal{A}_M^{2^L}\}.
\end{align*}
If $\mathcal{A}_M$ is closed, convex, and bounded, the affine mapping
\[
	\textbf{B}_x \Xi^M + \textbf{A}_x \mathcal{A}_M^{2^L}
\]
will also be closed, convex, and bounded.

Thus we approximate the identified set for $((\mathbf{f}_{Y_0|Z_\ell})_{\text{all } \ell}, (\mathbf{f}_{Y_1|Z_\ell})_{\text{all } \ell})$ by
\[
	\Theta_0^M(c) \times \Theta_1^M(c)
	= \big( \mathcal{D}^M(c) \cap \mathcal{H}^M_0 \big)
	\times
	\big( \mathcal{D}^M(c) \cap \mathcal{H}^M_1 \big)
\]
where
\[
	\mathcal{D}^M(c) \cap \mathcal{H}^M_x
	= \left\{ \textbf{f}(y) = \textbf{E}^M \textbf{b}^M(y): \textbf{E}^M \in \mathcal{A}^M(c) \cap (\textbf{B}_x \Xi^M + \textbf{A}_x \mathcal{A}_M^{2^L}) \right\}.
\]
This approximate identified set is empty---and hence the approximate model is falsified---if the set
\[
	\mathcal{A}^M(c) \cap (\textbf{B}_x \Xi^M + \textbf{A}_x \mathcal{A}_M^{2^L})
\]
is empty for either $x=0$ or $x=1$. These two sets are finite dimensional polytopes. Hence it is computationally tractable to compute the distance between them,
\[
	\text{dist}(\mathcal{A}^M(c), \ \textbf{B}_x \Xi^M + \textbf{A}_x \mathcal{A}_M^{2^L})
	= \inf_{(\textbf{v}_1, \textbf{v}_2) \in \mathcal{A}^M(c) \times (\textbf{B}_x \Xi^M + \textbf{A}_x \mathcal{A}_M^{2^L}) }\|\textbf{v}_1 - \textbf{v}_2 \|,
\]
for any given $c$. Since they are closed, convex, and bounded, the approximate model will be non-refuted if and only if
\[
	\text{dist}(\mathcal{A}^M(c), \ \textbf{B}_x \Xi^M + \textbf{A}_x \mathcal{A}_M^{2^L}) = 0.
\]
This characterization can be used to compute an approximate falsification frontier as well as an approximate falsification adaptive set. 

Although we omit a full analysis, we expect the identified sets $\Theta_0^M(c) \times \Theta_1^M(c)$ will converge to $\Theta_0(c) \times \Theta_1(c)$ as $M \rightarrow \infty$ under suitable regularity conditions. Consequently, the approximate falsification frontier and falsification adaptive sets should also converge as $M \rightarrow \infty$. Finally, the value of $M$ can be chosen large enough so that the approximate objects of interest change by less than a preset tolerance for further increases in $M$.

\section{Conclusion}

\subsection*{Summary}

In this paper we outlined a systematic and constructive answer to the question ``What should researchers do when their baseline model is refuted?'' Our answer focuses on what can be learned from falsified models, rather than treating falsification as a nuisance to be ignored or as a fatal flaw which dooms a study. 

We gave four recommendations. First: Measure the extent of falsification. Do this by defining continuous relaxations of the key identifying assumptions of interest, and then relaxing these assumptions until the model is no longer falsified. We call the set of points at which this happens the falsification frontier. Second: Present the \emph{falsification adaptive set}, the identified set for the parameter of interest under the assumption that the true model lies somewhere on the falsification frontier. This second recommendation is a generalization of standard practice for non-refuted baseline models. Moreover, it does not require the researcher to select or calibrate sensitivity parameters. Third: Present the identified set for selected points of interest on the falsification frontier. Fourth: Present identified sets for points beyond the falsification frontier, as a further sensitivity analysis.

We illustrated these four recommendations in two different overidentified instrumental variable models. The first model imposes homogeneous treatment effects but allows for continuous treatments, while the second model allows for heterogeneous treatment effects but focuses on binary treatments. In both models, multiple instruments are observed and the key identifying assumptions are exclusion and exogeneity of each instrument. We considered continuous relaxations of these assumptions. We then characterized the falsification frontier, the falsification adaptive set, and identified sets for points beyond the frontier. In the homogeneous treatment effect model, the falsification adaptive set has a particularly simple closed form expression, depending only on the value of a handful of 2SLS regression coefficients. In the heterogeneous treatment effect model, the falsification adaptive set can be quickly computed using convex optimization.

We showed how to use our results in four substantively different empirical applications. There we emphasized that that falsification adaptive set is an informative complement to traditional overidentification test $p$-values: It directly summarizes the range of estimates corresponding to non-falsified alternative models. We also emphasized the importance of controlling for the possibly invalid instruments when considering alternative models.

\subsection*{Future Research}

In this paper we outlined our general approach and then applied it to the classical linear instrumental variable model. The falsification frontier and falsification adaptive sets, however, must generally be derived separately for each baseline model and class of relaxations from those baseline assumptions. Thus in future work we plan to do this kind of analysis for different baseline models and different classes of relaxations. Here we focused on instrumental variable models. Other possible applications include various kinds of placebo tests, like those discussed in \cite{HeckmanHotz1989}. It may also be possible to derive results for a general class of models, by applying the results of \cite{ChesherRosen2017} or
\cite{Torgovitsky2018}, for example. Such extensions are important since---just like the choice of the baseline assumptions---the choice of relaxation matters. Hence comparing empirical results for different relaxations will be informative.

Finally, in this paper we focused on population level analysis. We briefly discussed estimation and inference in sections \ref{subsec:estimationInference} and \ref{subsec:FASestimationLinearIV}, but a full analysis of this remains for future work.

\singlespacing
\bibliographystyle{econometrica}
\bibliography{FF_paper}

\newpage
\appendix

\begin{center}
{\LARGE \textbf{Appendices to Accompany \\[0.5em] ``Salvaging Falsified Instrumental Variable Models''}}
\end{center}
\vspace{1mm}
\begin{center}
\begin{tabular}[h]{ccc}
    Matthew A. Masten &\hspace{5mm}   &  Alexandre Poirier  \\[1.25mm]
    {\small Duke University} & \hspace{5mm} &  {\small Georgetown University}
\end{tabular}
\end{center}

\titleformat*{\section}{\large\bfseries}
\section*{Table of Contents}

\hypersetup{
    colorlinks=true,
    linkcolor=mylinkcolor,
    citecolor=black,
    filecolor=black,
    urlcolor=black,
}

This document contains twelve appendices:
\begin{enumerate}[label=\Alph*.]
\item \hyperref[sec:relatedLiterature]{A detailed review} of how our results relate to the previous literature.

\item \hyperref[sec:AltResponses]{A discussion} of five alternative responses to falsification from the literature.

\item \hyperref[sec:MachadoEtAl]{An alternative approach} to defining falsification points.

\item \hyperref[sec:BayesianApproach]{A discussion} of Bayesian approaches to falsification, and how it relates to our results.

\item \hyperref[sec:transformedInstruments]{An analysis} of the implications of using affine transformations of one's original instruments, in the the linear model.

\item \hyperref[sec:covariatesInLinearModel]{An extension} of our linear model analysis to allow for covariates. Here we consider both linear and nonparametric controls. We then use these results to show how to use our approach to falsification to analyze subgroups.

\item \hyperref[sec:ExoExclu]{A formal explanation} that exclusion and exogeneity are technically equivalent in linear models, although they are substantively distinct.

\item \hyperref[sec:comparing2SLSandFAS]{An analysis} of the 2SLS estimand in falsified models, and a comparison of it with the falsification adaptive set.

\item \hyperref[sec:GMMobjFun]{A discussion} of how to interpret non-falsified models, and how our results relate to the previous literature on this topic.

\item \hyperref[sec:additionalEmpirics]{Additional empirical results} for the  \cite{DurantonMorrowTurner2014} application.

\item \hyperref[sec:proofsHomogModel]{Proofs} for section \ref{sec:homogModel}.

\item \hyperref[sec:proofsHetModel]{Proofs} for section \ref{sec:hetModel}.

\end{enumerate}

\hypersetup{
    colorlinks=true,
    linkcolor=black,
    citecolor=black,
    filecolor=black,
    urlcolor=black,
}

\newpage
\section{Related Literature}\label{sec:relatedLiterature}

In this section, we review the related literature.

\subsubsection*{Falsification}

We begin with the literature on falsification. Typically a model entails many different assumptions. Although we can sometimes falsify all assumptions jointly, it is often impossible to know \emph{which} specific assumption is false. This result is known in philosophy as the Duhem-Quine thesis (\citealt{Bonk2008}, \citealt{Stanford2017}). For example, \citet[section 2.1]{Ariew2018} states that ``the thesis formed from the two sub-theses, that (i) since empirical statements are interconnected, they cannot be singly disconfirmed, and (ii) if we wish to hold a particular statement true we can always adjust another statement, has become known as the Duhem-Quine thesis.'' That is, the general idea that one can consider trade-offs between assumptions to avoid falsification of any given assumption is well known in philosophy. Our paper builds on this literature by providing a specific formalization of this trade-off using identification theory.

In econometrics, three papers are most closely related. Of these, the conceptually closest is  \cite{ManskiPepper2018}. They study identification of treatment effects under combinations of different ``bounded variation'' assumptions, which are continuous relaxations of certain baseline assumptions. In an empirical analysis of right-to-carry laws, they compute identified sets for various combinations of these relaxations; see their table 2. Our recommended responses to baseline falsification can be read from this table. First, the boundary between the dark gray identified sets in that table and the rest is essentially a discrete approximation to the falsification frontier in their model. Second, the collection of the singleton identified sets on this boundary is what we call the falsification adaptive set. Third, given all of the separate presented singleton identified sets on the frontier, one can focus on any single set of interest. Finally, they also present identified sets beyond the frontier. In our paper, we formally define the falsification frontier and falsification adaptive set in a general model and relate it to the problem of responding to specification test rejection. We emphasize that these concepts and the four responses apply to any falsified model, not just the specific one considered by \cite{ManskiPepper2018}. Finally, we illustrate these responses using refutable instrumental variable models, whereas \cite{ManskiPepper2018} studied a model with cross sectional and time series variation in treatment rather than variation due to instrumental variables.

\cite{Ramsahai2012} studies a heterogeneous treatment effect model with a binary outcome, a binary treatment, and a binary instrument. He defines a continuous relaxation of the instrument exogeneity assumption and then shows how to numerically compute identified sets for a single value of this relaxation. He notes that the model can be falsified even if exogeneity does not hold exactly, and that the model is not falsified if instrument exogeneity is completely relaxed. He does not formally define and derive the falsification point, however. He also does not consider multi-dimensional relaxations. On pages 842--843, he notes that ``it is not obvious how the methods described in [his] paper can be extended to compute bounds'' as a function of his relaxation. In our analysis, we show how to compute bounds for any value of our continuous relaxation, and thus also are able to derive falsification points, falsification frontiers, and the falsification adaptive set. Unlike him, we also consider multiple instruments as well as models for continuous outcomes and continuous treatments.

\cite{MachadoShaikhVytlacil2018} study the testable implications of instrument exogeneity, sometimes combined with various monotonicity assumptions. They only consider the case where the outcome, treatment, and instrument are all binary. Their proposition 4.1 characterizes the set of distributions of unobservables (potential outcomes, potential treatments, and the instrument) which falsify instrument exogeneity. In their appendix A, they illustrate this result by parameterizing the distribution of unobservables. From the set of falsified distributions of unobservables, they pick the distribution which has the smallest relaxation of the baseline assumption, in a certain sense. They call this the ``minimal detectable violation.'' This is equivalent to what we call a falsification point. We explain this equivalence formally in our appendix \ref{sec:MachadoEtAl}. Beyond this similarity, however, there are several main differences between our paper and theirs. First, they only consider the case where outcomes are binary, whereas we also consider the continuous outcome case. Second, even in the binary outcome case their falsification point analysis is based on parametric index model assumptions. Our analysis of that case does not use such assumptions. Third, they do not consider the linear case, which allows for continuous treatments. Fourth, they do not consider the multiple instrument case. Finally, unlike their paper, we characterize identified sets under relaxations of the baseline assumptions. This allows us to characterize the falsification adaptive set, as well as identified sets beyond the falsification frontier. In particular, this allows researchers to do sensitivity analysis even when the baseline model is not falsified.

Before these three papers, a large literature in asset pricing beginning with \cite{HansenJagannathan1991, HansenJagannathan1997} developed in response to refutation of models via overidentification tests. For example, \citet[page 239]{HansenHeatonLuttmer1995} state that their goal ``is to shift the focus of statistical analysis of asset pricing models away from whether the models are correctly specified and \emph{toward measurement of the extent to which they are misspecified}'' (emphasis added). Our paper is motivated by the same concern---the falsification frontier provides one way of measuring the extent to which one's baseline model is misspecified. Besides this shared motivation, however, the technical approach in our paper is substantially different from that taken in this asset pricing literature. For an explanation and further discussion of the approach in that literature, see \citet[section 3.3]{Ludvigson2013}.

Recent work by \cite{DHaultfoeuilleEtAl2018} also has a similar motivation. They characterize the constraints that the rational expectations assumption places on the marginal distributions of realized outcomes and subjective expectations. When rational expectations is refuted, they define minimal deviations from the rational expectations assumption. They then propose using these minimal deviations as an input to structural estimation. Again, besides a shared motivation, the technical approach in our paper is substantially different. We also study a different setting, instrumental variable models.

Finally, many researchers have studied falsifiable partially identified models. As with point identified models, however, the existing literature focuses on testing the baseline model. For example, see \cite{BontempsMagnacMaurin2012} and \cite{BugniCanayShi2015}. The results in \cite{ChernozhukovLeeRosen2013} can also be used for this purpose. The falsification frontier we develop builds on such baseline tests by characterizing how much the baseline assumptions must be relaxed before the model is no longer falsified.

\subsubsection*{The Testable Implications of Instrumental Variable Models}\label{sec:IVtestingLit}

A large literature has studied the testable implications of instrumental variable models, especially for models with heterogeneous treatment effects. \cite{FloresChen2018} and \cite{SwansonEtAl2018} provide excellent surveys. Here we discuss the results most related to our paper. 

Our analysis of the heterogeneous treatment effect case starts from \cite{Manski1990}. He derived identified sets on average treatment effects assuming potential outcomes are mean independent of the instrument. When the outcome, treatment, and instrument are all binary, \cite{BalkePearl1997} characterize when Manski's bounds are empty, and hence when the model is falsified.\footnote{\cite{BalkePearl1997} assume the instrument is independent of the potential outcomes jointly, whereas \cite{Manski1990} only assumed the instrument is independent of each potential outcome separately. (Here we suppose outcomes are binary so that mean independence is equivalent to statistical independence.) This difference does not affect whether the identified set is empty or not, given any fixed distribution of observables. Hence it does not change the testable implications of the model. When the identified set is not empty, however, this difference \emph{can} affect the size of the identified set. See the second paragraph of section 3 in \cite{SwansonEtAl2018} for further discussion.} \citet[proposition 3.1]{Kitagawa2009} generalizes this characterization to allow for continuous outcomes, still requiring the treatment and instrument to be binary. As \cite{Kitagawa2009} notes, his extension is an adaptation of corollary 2.2.1 in Manski's \citeyearpar{Manski2003} analysis of missing data. \citet[proposition 2.4]{BeresteanuMolchanovMolinari2012} further generalize this characterization to allow for continuous instruments and discrete treatment, for discrete or continuous outcomes. \citet[proposition 1]{KedagniMourifie2017} provide an alternative characterization when instruments and outcomes are continuous, treatment is binary, and under the stronger assumption that the instrument is independent of the potential outcomes jointly; also see proposition 2.5 of \cite{BeresteanuMolchanovMolinari2012} for a result under this stronger independence assumption.

Our paper builds on these results in two ways. First, and most importantly: These papers only consider falsification of the baseline model with valid instruments. We consider falsification of a class of models which nests the baseline case along with continuous relaxations of the baseline assumptions. This allows us to define and derive the falsification frontier and falsification adaptive set. Second, for non-falsified models, our results allow researchers to do sensitivity analysis by relaxing the instrument exogeneity assumption. 

Our paper complements these results in one subtle way: These previous papers focus on the single instrument case whereas we study multiple instruments. If one assumes that all instruments are jointly independent of potential outcomes, and that instruments are all discrete, then the multiple instrument case can be reduced to a single instrument, and thus the previous results apply.\footnote{Another way to directly apply previous results is by looking at one instrument at a time, obtaining bounds using that instrument only, and then intersecting each of those instrument-specific bounds. This generally does not yield sharp bounds, however.} If even one of the instruments is invalid, however, this new instrument will also typically be invalid. For this reason, we purposely do not combine multiple instruments into one new instrument. This allows us to relax each instrument exogeneity assumption separately. In our baseline model, we only assume that each instrument is independent of each potential outcome, rather than assuming joint independence of all instruments with each potential outcome. Thus our baseline model with multiple instruments is slightly weaker than the model obtained by combining the instruments into a single discrete instrument and then applying one of the papers mentioned above. Hence even the testable implications we derive in our baseline case are apparently new, although this is a minor point since our primary contribution is the falsification frontier analysis rather than the baseline analysis.

Finally, a large literature on the testable implications of instrument exogeneity combined with other assumptions has developed. Most notably, many papers have studied the testable implications of the monotonicity assumption of \cite{ImbensAngrist1994}. \cite{FloresChen2018} give a comprehensive review. In this paper we focus on instrument exogeneity only.

\subsubsection*{Sensitivity Analysis in Homogeneous Effect Instrumental Variable Models}\label{sec:litReviewLinearIVsensitivity}

In a series of papers from the mid-1950's to early 1960's, Ta-Chung Liu criticized the identifying assumptions used in the Cowles Commission approach to simultaneous equations models (\citealt{Liu1955,Liu1960,Liu1963}). Among other criticisms, he argued that exclusion restrictions are unlikely to hold exactly, and thus no valid instruments actually exist. Frank Fisher called this a ``disturbing argument,'' since ``its premises apparently cannot be doubted and because its conclusions, if accepted, imply that the hope of structural estimation by any techniques whatsoever is forlorn indeed'' (\citealt{Fisher1961}, page 139). Fisher countered this argument, however, by providing perhaps the first sensitivity analysis of an instrumental variable model. He considered a linear model where the supposedly excluded variable actually has a nonzero coefficient, with the magnitude of this coefficient measuring the scale of the deviation from the baseline assumption. He showed that the bias of $k$-class estimators---which includes 2SLS---is continuous in this deviation from the exclusion restriction. He thus argued that it suffices that failures of the exclusion restriction are sufficiently small---it is not necessary for exclusion to hold exactly. Although Fisher called this paper his ``most important contribution to econometrics'' (\citealt{Fisher2005}, page 545), it apparently was mostly forgotten for over 40 years.

Around the early 2000's, however, several papers renewed the study of sensitivity analysis in linear instrumental variable models. This includes \cite{AngristKrueger1994}, \cite{AltonjiElderTaber2005}, \cite{Small2007}, \cite{ConleyHansenRossi2012}, \cite{Ashley2009}, \cite{Kraay2012}, \cite{AshleyParmeter2015}, and \cite{vanKippersluisRietveld2017,vanKippersluisRietveld2018}. In section \ref{sec:homogModel}, we study the same classical linear model as all of these papers, along with similar relaxations of the baseline assumptions. In that section, our main contribution is the analysis of falsification points, frontiers, and the falsification adaptive set. In contrast, this previous literature has focused on sensitivity analysis for statistical inference on the treatment variable coefficient. Our results build on Small's \citeyearpar{Small2007} equation (8), which characterizes the set of instrument coefficients consistent with the data. Small focuses on the implications of this characterization for \emph{non-rejection} of the overidentification test (see his proposition 1), and then uses this result to do sensitivity analysis for the treatment effect. In contrast, we focus on the implications of this characterization for falsification. 

\subsubsection*{Sensitivity Analysis in Heterogeneous Effect Instrumental Variable Models}

Unlike the literature on models with homogeneous treatment effects, there are far fewer papers on sensitivity analysis in instrumental variable models with heterogeneous treatment effects. Specifically, few papers consider continuous relaxations of the baseline instrumental variable assumptions while still allowing for heterogeneous treatment effects.

An important early paper was \cite{HotzMullinSanders1997}, who use a mixture model to allow for relaxations of the baseline assumptions. There are three main differences with our paper: First, they only study sensitivity, not falsification. Second, their mixture relaxation is different from the kind of relaxation we consider. Finally, they focus on the average effect of treatment on the treated, whereas our sensitivity analysis allows for a broader set of parameters of interest.

\label{litReview:linearProg}
When all variables are discrete, it is well known that identified sets can often be computed using linear programming. For example, see the literature review in \cite{Torgovitsky2018}. This approach has been used in several papers to do sensitivity analysis. One paper is \cite{Ramsahai2012}, which we already discussed above. \citet[section 4]{Laffers2018CE} considers continuous relaxations of instrument exogeneity. He then computes identified sets for ATE for several values of this relaxation. In \cite{Laffers2018EE} he applies this approach to various additional forms of continuous relaxations. These three papers all require all variables to be discrete. A key contribution of our paper is that our results allow for continuous outcome variables. Furthermore, these papers focus on sensitivity analysis, rather than falsification (see our discussion of \citealt{Ramsahai2012} above, however).

Finally, the paper by \cite{Huber2014} is somewhat related. He studies deviations from the baseline instrument exogeneity, exclusion, and monotonicity assumptions. These deviations are not relaxations of the baseline assumption; rather, they point identify LATE given the values of the sensitivity parameters. He does not consider parameters besides LATE. Moreover, his sensitivity analysis does not account for the possibility that the model is falsified. Hence his analysis is substantially different from ours.

\subsubsection*{Discrete Relaxations of Instrumental Variable Assumptions}

Above we focused on the previous literature which studied continuous relaxations of the baseline instrumental variables assumptions, like we do. Several papers instead consider discrete relaxations. \cite{ManskiPepper2000, ManskiPepper2009} propose and study the monotone instrumental variable assumption. \cite{BlundellEtAl2007} use this assumption in their study of wage distributions. \cite{KreiderPepper2007}, \cite{KreiderPepperGundersenJolliffe2012}, \cite{GundersenKreiderPepper2012}, and \cite{KreiderPepperRoy2016} use this assumption combined with assumptions on measurement error to do sensitivity analysis in various empirical applications. \cite{ChenFloresFloresLagunes2016} also study identification under monotonicity type assumptions on the instrument. \cite{NevoRosen2012} suppose the correlation between the instrument and the unobservable is of the same sign and smaller than the correlation between treatment and the unobservable.

As with the other papers discussed above, these papers' primary goal is to obtain identified sets on various parameters under assumptions weaker than the standard instrumental variables assumptions. While our results allow for similar analyses, our primary focus is on falsified models.

\subsubsection*{Local Versus Global Approaches to Sensitivity Analysis}

Several papers use local asymptotics to study various kinds of sensitivity to misspecification. For example, see \cite{AndrewsGentzkowShapiro2017}, \cite{BonhommeWeidner2018}, and \cite{ArmstrongKolesar2019}. The local asymptotic approach assumes the baseline model is approximately correct, in the sense that the magnitude of model misspecification is about the same size as the magnitude of sampling uncertainty. When analyzing falsifiable models, this means that the local approach is designed to handle misspecification which is difficult to detect with specification tests. For example, \cite{BonhommeWeidner2018} cite the following quote from \citet[page 294]{HuberRonchetti2009}:
\begin{quote}
``[such local] neighborhoods make eminent sense, since the standard goodness-of-fit tests are just able to detect deviations of this order. \emph{Larger deviations should be taken care of by diagnostic and modeling}, while smaller ones are difficult to detect and should be covered (in the insurance sense) by [local] robustness.'' [emphasis added]
\end{quote}
In this paper we focus on models which are clearly falsified---the ``larger deviations'' which Huber and Ronchetti refer to. In these larger deviations, it is known that the model is not approximately correct. For this reason we take a global approach to sensitivity analysis, which does not rely on linking the size of model misspecification to the size of sampling uncertainty.

Finally, note that this distinction is empirical relevant since many baseline models are strongly rejected in practice. For example, both \cite{AndrewsGentzkowShapiro2017} and \cite{ArmstrongKolesar2019} study the classic automobile data used by  \cite{BerryLevinsohnPakes1995} to estimate a baseline random coefficients logit model. This model is strongly rejected by the data, however. Specifically, \cite{ArmstrongKolesar2019} report that ``The $J$-statistic for testing the hypothesis that all moments are correctly specified equals 404.7. Consequently, the hypothesis is rejected at the usual significance levels'' (page 34). We deal with empirical findings of large misspecification like this by taking a global approach.

\section{Alternative Responses to Falsification}\label{sec:AltResponses}

In this paper we provided four constructive responses to refutation of a baseline model. In this section, we briefly discuss five alternative responses from the literature.
\begin{enumerate}
\item \emph{Present the no assumption bounds.} This response is nested as a special case of our analysis. Specifically, in the heterogeneous treatment effects model, the no assumption bounds are obtained by allowing $c_\ell$ to be sufficiently large for all $\ell \in \{1,\ldots, L \}$. In the homogeneous treatment effects model, the no assumption bounds are $(-\infty,\infty)$, which is obtained as $\delta_\ell \rightarrow \infty$ for all $\ell \in \{1,\ldots,L \}$.

\item \emph{Present identified sets under alternative, discrete relaxations of the baseline assumptions.} For example, one could use the analysis of \cite{ManskiPepper2000, ManskiPepper2009} or \cite{NevoRosen2012}. Even these relaxations are falsifiable, however. When these weaker assumption are falsified, researchers again face the question of what to do. Our identification assumptions are not nested with these alternative relaxations. Hence they are complementary and can be used together with these previous results.

\item \emph{Focus on a non-falsified subgroup.} We discuss this point formally in section \ref{sec:covariatesInLinearModel}, so here we give an informal explanation. Suppose a vector of covariates $W$ is observed. Then the model may be refuted conditional on some values $w \in \supp(W)$ but not conditional on other values. So one could present results for those subgroups where the model is not refuted. This is reasonable under either of two assumptions:
\begin{enumerate}
\item You assume treatment effects do not depend on covariates. Thus the subgroup effect equals the effect for the entire population.

\item You are specifically interested in a certain subgroup, which happens to be one which does not refute the model, and you are not interested in other subgroups.
\end{enumerate}
If, however, you care about other subgroups, or the entire population, and you want to allow for coefficients that are functions of the covariates, then it is not sufficient to focus on non-falsified subgroups. In section \ref{sec:covariatesInLinearModel} we formally discuss how you can use our results to handle the subgroups where the model is refuted and then aggregate up to get identified sets for an overall treatment effect.

\item \emph{Use a specific estimator to redefine the parameter of interest.} For example, in a fully parametric model one could declare the maximum likelihood estimand to be the parameter of interest. When the model is falsified, this parameter is called the pseudo-true parameter, while its estimator is called the quasi-MLE. Note that this approach absolves the researcher from any need to do specification testing, since interest in the parameter is motivated by the estimator rather than an underlying model. While this approach is sometimes used, it was not advocated in White's original analysis of quasi-MLE. To the contrary, he used quasi-MLE as a tool for constructing specification tests, emphasizing that
\begin{quote}
``[if] one has an indication that the parameter estimator is inconsistent for the parameters of interest'' then ``the model specification must be carefully re-examined.'' (\citealt{White1982}, pages 16--17)
\end{quote}
He reaches a similar conclusion in his 1994 book:
\begin{quote}
``In many cases, the goal of empirical economic research is to test hypotheses about parameters to which one wishes to attribute economic meaning. It is our view that this is inappropriate and unjustified without first establishing that the model within which the hypotheses are being tested is correctly specified to at least some extent. Otherwise, one may only have confidence that one is testing hypotheses about [the pseudo-true parameter] $\theta^*$; the economic interpretation desired is untenable.'' (\citealt{White1994}, pages 346--347)
\end{quote}
We concur with this sentiment and focus on methods which hold the parameter of interest fixed but vary the model assumptions to avoid falsification.

\item \emph{Drop the research and move on to a different project.} This response has three problems: (a) It creates a kind of publication bias, since we would never observe analyses with refuted baseline models. (b) It prevents us from answering certain empirical questions when there are no good alternative models or datasets.
And (c) It treats all falsified models equally; in contrast, our analysis distinguishes between models which are just barely falsified and those which are strongly falsified. 
\end{enumerate}

\section{An Alternative Definition of the Falsification Point}\label{sec:MachadoEtAl}

In this section we discuss an alternative method of defining the falsification point or frontier. This alternative approach is motivated by the analysis in appendix A of \cite{MachadoShaikhVytlacil2018}. 
For concreteness, we focus on the binary outcome heterogeneous treatment effects model of section \ref{sec:hetTrtBinary}. Let
\[
	\mathcal{Q} = \{ F_{Y_1,Y_0,X,Z} : \text{This distribution is consistent with the observed data $F_{Y,X,Z}$} \}.
\]
That is, $Q \in \mathcal{Q}$ if and only if $Q_{X,Z} = F_{X,Z}$ and
\[
	Q_{Y_x \mid X,Z}(\cdot \mid x, \cdot) = F_{Y \mid X,Z}(\cdot \mid x, \cdot)
\]
for each $x \in \{0,1 \}$, so that equation \eqref{eq:generalOutcomeEquation} holds. For each $Q \in \mathcal{Q}$, define
\[
	c_x(Q) = \sup_{y \in \{ 0,1 \}} | \Prob_Q(Z=1 \mid Y_x=y) - \Prob_Q(Z=1) |
\]
Then compute
\[
	c_x^* = \inf_{Q \in \mathcal{Q}} \; c_x(Q).
\]

\begin{proposition}\label{prop:MachadoEtAlrelationship}
Consider the model of section \ref{sec:hetTrtBinary}. The falsification point is $c^* = \max \{ c_0^*, c_1^* \}$.
\end{proposition}

In sections \ref{sec:homogModel} and \ref{sec:hetModel}, we obtain falsification frontiers by characterizing identified sets under a given relaxation of the assumptions, and then checking when they are empty. A benefit of this approach is that we can present identified sets for values of the sensitivity parameters on or beyond the falsification frontier. The approach discussed in the present section provides the falsification point but does not characterize identified sets for values $c \geq c^*$.

\begin{proof}[Proof of proposition \ref{prop:MachadoEtAlrelationship}]
Let $c \geq c^*$. Then, $c \geq c_x^*$ for $x \in \{ 0,1 \}$. Thus, by definition of $c_x^*$, there exists a distribution $Q^x \in \mathcal{Q}$ under which $Z$ is $c$-dependent with $Y_x$, for each $x \in \{0,1\}$. Define $Q_{Y_1,Y_0,X,Z}$ as follows: Let $Q_{X,Z} = F_{X,Z}$. Next, define
\[
	Q_{Y_1,Y_0\mid X,Z} = Q^1_{Y_1 \mid X,Z} \cdot Q^0_{Y_0 \mid X,Z}.
\]
Then (a) $Q \in \mathcal{Q}$, so that it is consistent with the data, and (b) under $Q$, $Z$ is $c$-dependent with $Y_x$ for each $x \in \{0,1\}$. Thus the model under partial exogeneity with $c$-dependence is not falsified.

Next suppose $c < c^*$. Then $c_x(Q) > c$ for all $Q \in \mathcal{Q}$. That is, all distributions of the unobservables that are consistent with the data have
\[
	\sup_{y \in \{ 0,1 \}} | \Prob_Q(Z=1 \mid Y_x=y) - \Prob_Q(Z=1) | > c
\]
for each $x \in \{0,1\}$, and so cannot satisfy $c$-dependence. Thus the model under partial exogeneity with $c$-dependence is falsified.
\end{proof}

\section{A Bayesian Perspective on Falsification}\label{sec:BayesianApproach}

In this section we briefly discuss how one might respond to falsification of a baseline model from a Bayesian perspective.
To do this, it is first helpful to describe our analysis from a different perspective than that used in the main paper. After that, we describe an alternative Bayesian perspective. We then compare the two approaches.

For concreteness, we focus on the classical linear instrumental variable model of section \ref{sec:homogModel}. Several papers study sensitivity analysis in this model from a Bayesian perspective, including \cite{Ashley2009}, \cite{ConleyHansenRossi2012}, \cite{Kraay2012}, and \cite{ChanTobias2015}. Also see \cite{DitragliaGarciaJimeno2019} for a detailed analysis of priors in instrumental variable models. Here we focus on the problem of falsification.

\subsubsection*{An Alternative Perspective of Our Analysis}

Suppose we impose A\ref{assump:homog:relevance:gen}, A\ref{assump:homog:nonsing:gen}, and A\ref{assump:exogeneity:gen}, but we make no assumptions about $\gamma$. Then the model is no longer falsifiable. Instead, there is always a nonempty identified set for $(\beta,\gamma)$. Proposition \ref{prop:IdentSetForGammas} in appendix \ref{sec:GMMobjFun} below shows that the identified set for $\gamma$ is a line in $\R^L$. Figure \ref{fig:GammaIdentSet} shows two examples in the two instrument case. Falsification of the baseline model occurs when this line does not pass through the origin.

For simplicity, focus on the two instrument case, as depicted in figure \ref{fig:GammaIdentSet}. Let $\gamma_1^*$ denote the horizontal intercept of this line. Let $\gamma_2^*$ denote the vertical intercept. Then $\delta_1^* = \gamma_1^*$ and $\delta_2^* = \gamma_2^*$, where $\delta_1^*$ and $\delta_2^*$ are the falsification points defined in section \ref{sec:linearModelFF}. The falsification frontier corresponds to the line segment with endpoints $(0, \gamma_2^*)$ and $(\gamma_1^*, 0)$.

Each point in the identified set for $\gamma$ corresponds to a unique value of $\beta$, namely, the value $b$ that satisfies
\[
	\gamma = \psi - b \Pi.
\]
See proposition \ref{prop:IdentSetForGammas}. Thus the falsification adaptive set is simply the values of $b$ corresponding to $\gamma$'s in the line segment between $(0, \gamma_2^*)$ and $(\gamma_1^*, 0)$. As we emphasize in our fourth recommendation, researchers may want to do sensitivity analysis beyond the falsification frontier. That is, they may want to start by presenting the falsification adaptive set, but also present larger identified sets corresponding to values of $\delta$ that lie beyond the falsification frontier. In this case, the identified set for $\beta$ simply expands to include values of $b$ that correspond to $\gamma$'s beyond the line segment between $(0, \gamma_2^*)$ and $(\gamma_1^*, 0)$.

\subsubsection*{A Bayesian Perspective}

With that as background, we next discuss a Bayesian approach. The baseline model corresponds to a dogmatic prior that places all of its mass on $\gamma = 0$. Instead, suppose we specify a prior $\pi$ on $\gamma$ which is centered at zero, but has full support. Suppose now that we observe the distribution of the data $F_{Y,X,Z}$. And we see that this distribution is incompatible with the baseline model. That is, the identified set for $\gamma$ described in proposition \ref{prop:IdentSetForGammas} does not pass through the origin. A prior which was dogmatic on zero would be refuted. But a full support prior would not. Instead, we would update our prior so that our posterior will place support only on the identified set for $\gamma$. Given that our prior had full support, our posterior will have full support on this line in $\R^L$. Since each point on this line corresponds to a unique value of $\beta$, this posterior on $\gamma$ will imply a posterior on $\beta$. This posterior will have full support on $\R$, but of course will give larger weight to some values than to others. 

Note that there is no need to specify a prior on $\beta$, since we are working with population level data. The only uncertainty that is being captured by the prior is the model uncertainty over the precise value of $\gamma$, rather than sampling uncertainty. That said, the prior on $\gamma$ implies a prior on $\beta$, which depends on the values $\psi$ and $\Pi$. These values themselves come from the data. Thus, before seeing the population data, it may still be useful to consider one's prior on $\beta$ and on the values of $\psi$ and $\Pi$ when thinking about the choice of the prior on $\gamma$. We leave a full analysis of prior selection to future work.

\subsubsection*{A Comparison and Integration of Both Approaches}

Now that we have described both approaches, let's compare them. As noted above, the Bayesian approach yields a posterior for $\beta$ that has full support on $\R$. Assuming the prior on $\gamma$ places most of its mass near zero, the posterior on $\beta$ will put more mass on the values corresponding to $\gamma$'s closer to the origin. Thus means that it will generally put large amount of mass on the falsification adaptive set. In general, researchers taking a Bayesian perspective may find it useful to compute the posterior probability that $\beta$ lies within the falsification adaptive set. If we view the falsification adaptive set as a partial summary of the identified set for $\gamma$, reporting this posterior probability can help researchers understand which aspects of the conclusions come from the data and which come from the prior. This is similar to Moon and Schorfheide's \citeyearpar{MoonSchorfheide2012} suggestion that researchers reported estimated identified sets along with Bayesian posteriors. 

As researchers relax their model beyond the falsification frontier, the identified set for $\beta$ will expand beyond the falsification adaptive set. Bayesian researchers may also find it helpful to compare the size of credible sets to the falsification adaptive set. More generally, they could compute a kind of Bayesian falsification frontier: For $\alpha \in (0,1)$, what is the set of $\delta$'s that yield a posterior for $\beta$ which places probability $1-\alpha$ on the identified set $\mathcal{B}(\delta)$?

In this section we have briefly compared and contrasted our approach with a Bayesian approach, and shown that they can complement each other. There are many open questions, however, like how to choose priors and how to analyze finite sample versions of this analysis. In particular, in finite samples we never know for sure that our model is falsified. Thus it will be interesting to see how Bayesian shrinkage methods, like those in \cite{FesslerKasy2019}, can be applied to analyze falsifiable models. We leave this and related questions to future work.

\section{Affine Transformations of the Instruments in the Linear Model}\label{sec:transformedInstruments}

In this section we discuss the implications of using affine transformations of the original instruments $Z$. Specifically, let $A$ denote an invertible $L \times L$ matrix. Let $b \in \R^L$. Then $\widetilde{Z} = A Z + b$ is an affine transformation of $Z$. Let
\[
	\widetilde{\psi} = \var(\widetilde{Z})^{-1} \cov(\widetilde{Z},Y)
	\qquad \text{and} \qquad
	\widetilde{\pi} = \var(\widetilde{Z})^{-1} \cov(\widetilde{Z},X).
\]
Here we assume $K = 1$ for simplicity.

\begin{proposition}\label{prop:scaleInvariance}
Suppose $A$ is diagonal. Then
\[
	\frac{\widetilde{\psi}_\ell}{\widetilde{\pi}_\ell} = \frac{\psi_\ell}{\pi_\ell}.
\]
Hence the FAS using $Z$ as instruments equals the FAS using $\widetilde{Z}$ as instruments.
\end{proposition}

\begin{proof}[Proof of proposition \ref{prop:scaleInvariance}]
We have
\begin{align*}
	\widetilde{\psi}_\ell
		&= e_\ell' \widetilde{\psi} \\
		&= e_\ell' \var(\widetilde{Z})^{-1} \cov(\widetilde{Z},Y) \\
		&= e_\ell' \var(AZ + b)^{-1} \cov(A Z + b,Y) \\
		&= e_\ell' \var(AZ)^{-1} \cov(A Z,Y) \\
		&= e_\ell' [ A \var(Z) A']^{-1} A \cov(Z,Y) \\
		&= e_\ell' (A')^{-1} \var(Z)^{-1} \cov(Z,Y) \\
		&= e_\ell' (A')^{-1} \psi.
\end{align*}
Similarly,
\[
	\widetilde{\pi}_\ell = e_\ell' (A')^{-1} \pi.
\]
Thus
\begin{align*}
	\frac{\widetilde\psi_\ell}{\widetilde\pi_\ell}
	&= \frac{e_\ell' (A')^{-1}\psi}{e_\ell' (A')^{-1}\pi} \\
	&= \frac{a_\ell e_\ell' \psi}{a_\ell e_\ell' \pi} \\
	&= \frac{e_\ell'\psi}{e_\ell'\pi} \\
	&= \frac{\psi_\ell}{\pi_\ell}.
\end{align*}
The second line follows since $A = \diag(a_1,\ldots,a_L)$, so that $e_\ell' (A')^{-1} = a_\ell e_\ell'$. 
\end{proof}

Proposition \ref{prop:scaleInvariance} shows that the just-identified 2SLS estimand using $\widetilde{Z}_\ell$ as an instrument and $\widetilde{Z}_{-\ell}$ as controls is equivalent to the just-identified 2SLS estimand using $Z_\ell$ as an instrument and $Z_{-\ell}$ as controls. As stated, this implies that the falsification adaptive set is the same too. When $A$ is invertible and diagonal, $\widetilde{Z}_\ell = a_\ell Z + b_\ell$ for some $a_\ell \neq 0$. So the proposition shows that the falsification adaptive set is invariant to rescaling each instrument and changing their location.

When $A$ is not diagonal, however, these just-identified 2SLS estimands $\psi_\ell / \pi_\ell$ are not invariant to affine transformations. As we show next, this is not surprising since using a non-diagonal affine transformation is equivalent to using the original instruments, but with a different way of relaxing the exclusion restriction.

To see this, consider the case where $A$ is a rotation matrix and $b = 0$. From the outcome equation we have
\begin{align*}
	Y
		&= \beta' X + \gamma' Z + U \\
		&= \beta' X + \gamma' A^{-1} A Z + U \\
		&= \beta' X + (A^{-1} \gamma)' \widetilde{Z} + U \\
		&= \beta' X + \widetilde{\gamma}' \widetilde{Z} + U.
\end{align*}
In the last line we defined $\widetilde{\gamma} = A^{-1} \gamma$. $A^{-1}$ is also a rotation matrix, and hence $\widetilde{\gamma}$ is a rotation of $\gamma$. To illustrate this, consider the two instrument case and suppose $A$ represents a 45 degree counter-clockwise rotation:
\[
	A =
	\begin{pmatrix}
		\cos (\pi/4) & -\sin (\pi/4) \\
		\sin (\pi/4) & \cos (\pi/4)
	\end{pmatrix}.
\]
Then
\[
	A^{-1} = A' = 
	\begin{pmatrix}
		\cos (-\pi/4) & -\sin (-\pi/4) \\
		\sin (-\pi/4) & \cos (-\pi/4)
	\end{pmatrix}
\]
represents a 45 degree clockwise rotation.
\begin{figure}[!t]
\centering
\includegraphics[width=0.18\linewidth]{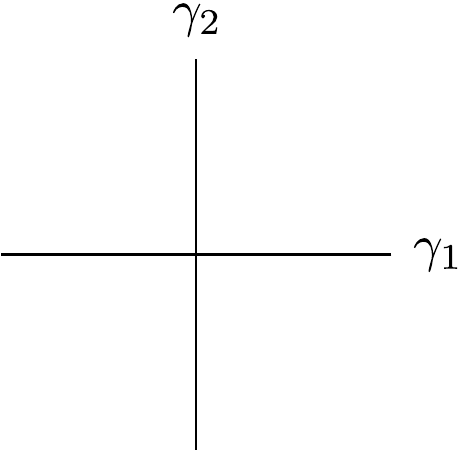}
\hspace{0.01\linewidth}
\includegraphics[width=0.18\linewidth]{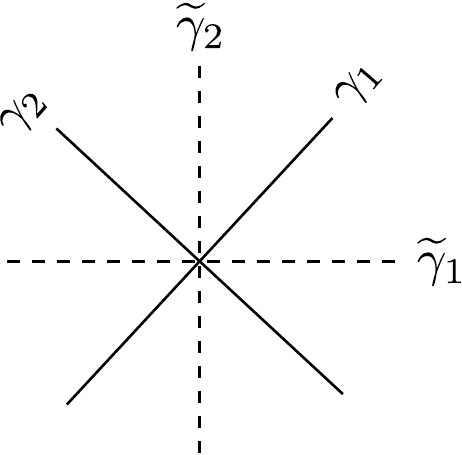}
\hspace{0.01\linewidth}
\includegraphics[width=0.18\linewidth]{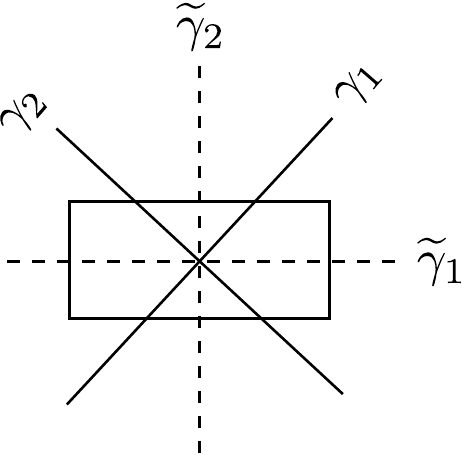}
\hspace{0.01\linewidth}
\includegraphics[width=0.18\linewidth]{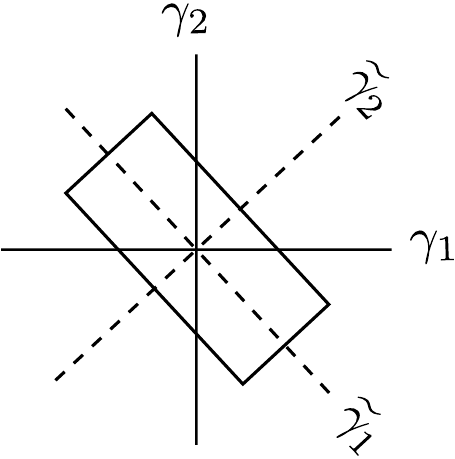}
\hspace{0.01\linewidth}
\includegraphics[width=0.18\linewidth]{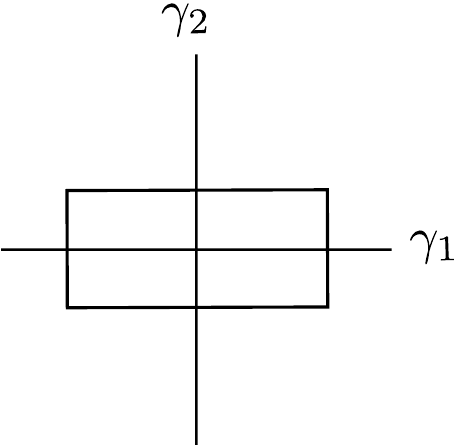}
\caption{Rotating the instruments leads to a different class of exclusion restriction relaxations. See text for explanation.}
\label{fig:rotations}
\end{figure}
The first plot in figure \ref{fig:rotations} shows the original axes for measuring $(\gamma_1, \gamma_2)$. The second plot shows the rotated axes. Suppose we now apply our results directly using $\widetilde{Z}$ as the instruments. That is, we define relaxations of exclusion $\widetilde{\gamma} = 0$ by
\[
	| \widetilde{\gamma}_\ell | \leq \delta_\ell.
\]
Then we measure relaxations of $\widetilde{\gamma} = 0$ by placing rectangles around the origin, in the rotated coordinate system. The third plot in figure \ref{fig:rotations} shows an example of this constraint set, for a given value of $\delta$. The fourth plot in figure \ref{fig:rotations} shows this same constraint set, but now in our original coordinate system. Thus we see that relaxations of the exclusion restriction $\gamma = 0$ are now measured in terms of rotated rectangles around the origin. In contrast, suppose we instead used our original instruments $Z$ and measured relaxations of exclusion by $| \gamma_\ell | \leq \delta_\ell$. The fifth and final plot in figure \ref{fig:rotations} shows this constraint set, for the same $\delta$. 

Comparing the last two plots in figure \ref{fig:rotations} thus shows why rotating the instruments will generally yield different a falsification adaptive sets: It is equivalent to using a different class of relaxations of the exclusion restriction with the original instruments. Different classes of relaxations of one's assumptions will typically lead to different falsification adaptive sets. This is exactly the same situation as one's baseline analysis, where different baseline assumptions typically lead to different baseline estimands.

\section{Including Covariates in the Linear Model and Subgroup Analysis}\label{sec:covariatesInLinearModel}

In this section we discuss how to include additional covariates in the analysis of section \ref{sec:homogModel}. We consider two approaches: In the first approach we include covariates linearly. In the second approach we condition nonparametrically on the covariates. Finally, we discuss how to use our results in combination with the subgroup analysis approach discussed in section \ref{sec:AltResponses}.

\subsubsection*{Linear Controls}

Generalize equation \eqref{eq:constantCoeffOutcomeEqGen} as follows:
\[
	Y = X'\beta + Z'\gamma + W'\alpha + U
	\tag{\ref{eq:constantCoeffOutcomeEqGen}$^\prime$}
\]
where $W\in \R^{\text{dim}(W)}$ is a $\text{dim}(W)$-vector of additional covariates and $\alpha$ is a $\text{dim}(W)$-vector of unknown constant coefficients. In addition to assumptions A\ref{assump:homog:relevance:gen}--A\ref{assump:exogeneity:gen}, we make two further assumptions:
\begin{enumerate}
\item The controls are exogenous: $\cov(W,U) = 0$.

\item The vector $\widetilde{Z} = (Z,W)$ has a nonsingular variance matrix.
\end{enumerate}
In the framework considered in section \ref{sec:homogModel}, we can think of the vector $W$ as a vector of instruments whose exclusion violation is not restricted at all. That is, we set $\delta_\ell = \infty$ for the elements of $\widetilde{Z}$ corresponding to the $W$ covariates. Hence we do not place any constraints on the magnitude of $\alpha$. It is nonetheless helpful to keep $Z$ and $W$ separate in the notation.

Under equation (\ref{eq:constantCoeffOutcomeEqGen}$^\prime$) and the exogenous controls assumption, we can obtain the identified set for $\beta$ by simply applying the analysis of section \ref{sec:homogModel}, but with an initial step where the exogenous covariates are partialled out from instruments $Z$ and endogenous regressors $X$. To see this, write
\begin{align*}
	\cov(Z,Y) &= \cov(Z,X)\beta + \var(Z)\gamma + \cov(Z,W)\alpha \\
	\cov(W,Y) &= \cov(W,X)\beta + \cov(W,Z)\gamma + \var(W)\alpha.
\end{align*}
Solve the second equation for $\alpha$ to get
\[
	\alpha = \var(W)^{-1} \cov(W,Y) - \var(W)^{-1} \cov(W,X) \beta - \var(W)^{-1} \cov(W,Z) \gamma.
\]
Substitute this into the first equation to get
\begin{align*}
	\cov(Z - \cov(Z,W)\var(W)^{-1}W,Y) 
	&= \cov(Z-\cov(Z,W)\var(W)^{-1}W,X)\beta \\
	&\quad + \var(Z - \cov(Z,W)\var(W)^{-1}W)\gamma.
\end{align*}
By the assumption that $\widetilde{Z}$ has a nonsingular variance matrix, solve for $\gamma$ to get
\begin{align*}
	\gamma = 
	\var(Z - \cov(Z,W)\var(W)^{-1}W)^{-1}
	\big( &\cov(Z - \cov(Z,W)\var(W)^{-1}W,Y) \\
	&\quad - \cov(Z-\cov(Z,W)\var(W)^{-1}W,X) \beta \big).
\end{align*}
Thus we can write the identified set for $\beta$ as 
\[
	\mathcal{B}(\delta) = \{b\in\R^K : -\delta \leq \widetilde{\Psi} - \widetilde{\Pi}\beta \leq \delta\}
\]
where
\[
	\widetilde{\Psi} = \var(Z - \cov(Z,W)\var(W)^{-1}W)^{-1}\cov(Z - \cov(Z,W)\var(W)^{-1}W,Y)
\]
is the vector of population OLS coefficients on $Z$ in a regression of $Y$ on $(1,Z,W)$ and
\[
	\widetilde{\Pi} = \var(Z - \cov(Z,W)\var(W)^{-1}W)^{-1}\cov(Z - \cov(Z,W)\var(W)^{-1}W,X)
\]
is the vector of population OLS coefficients on $Z$ in a regression of $X$ on $(1,Z,W)$.

From this characterization of the identified set for $\beta$, we can adapt the proof of theorem \ref{thm:identSetOnFFhomogTrt} to see that the falsification adaptive set is
\[
	\bigcup_{\delta \in \text{FF}} \mathcal{B}(\delta) 
	= \left[ \min_{\ell =1,\ldots,L} \; \frac{\widetilde{\psi}_\ell}{\widetilde{\pi}_{\ell}}, \
	\max_{\ell =1,\ldots,L} \; \frac{\widetilde{\psi}_\ell}{\widetilde{\pi}_{\ell}}\right]
\]
where $\widetilde{ \psi }_\ell$ is the $\ell$th component of $\widetilde{\Psi}$ and $\widetilde{\pi}_\ell$ is the $\ell$th component of $\widetilde{\Pi}$. Here we assume $\widetilde{\pi}_\ell \neq 0$ for all $\ell \in \{1,\ldots,L\}$. Moreover, as in lemma \ref{lemma:interpretingPsiOverPi}, $\widetilde{\psi}_\ell / \widetilde{\pi}_\ell$ equals the population 2SLS coefficient on $X$ using $Z_\ell$ as the excluded instrument and $(Z_{-\ell},W)$ as controls. Thus, under the above assumptions, the presence of these additional controls poses no additional technical challenges for estimation and inference since the falsification adaptive set can still be computed via standard linear 2SLS regressions.

\subsubsection*{Nonparametric Controls}

Above we included controls linearly and assumed they were exogenous. An alternative approach is to condition on them nonparametrically. This allows the value of the covariates to affect the coefficients on $X$ and $Z$, and also allows the controls to be endogenous. Thus we now generalize equation \eqref{eq:constantCoeffOutcomeEqGen} as follows:
\[
	Y = X'\beta(W) + Z'\gamma(W) + U
	\tag{\ref{eq:constantCoeffOutcomeEqGen}$^{\prime \prime}$}
\]
where $W$ is a $\dim(W)$-vector of controls, and $\beta(\cdot)$ and $\gamma(\cdot)$ are unknown functions. We replace A\ref{assump:homog:relevance:gen}--A\ref{assump:exclusion:gen} with their conditional-on-$W$ versions:
\begin{enumerate}
\item[A\ref{assump:homog:relevance:gen}$^*$.] $\cov(Z,X \mid W=w)$ has rank $K$ for almost all $w \in \supp(W)$.

\item[A\ref{assump:homog:nonsing:gen}$^*$.] $\var(Z \mid W=w)$ has full rank for almost all $w \in \supp(W)$.

\item[A\ref{assump:exogeneity:gen}$^*$.] For all $\ell \in \{1,\ldots,L\}$, $\cov(Z_\ell,U \mid W=w) = 0$ for almost all $w \in \supp(W)$.

\item[A\ref{assump:exclusion:gen}$^*$.] For all $\ell \in \{1,\ldots,L\}$, $\gamma_\ell(w) = 0$ for almost all $w \in \supp(W)$.
\end{enumerate}
We relax A\ref{assump:exclusion:gen}$^*$ by
\begin{enumerate}
\item[A\ref{assump:exclusion:gen}$^{**}$.] There are known functions $\delta_\ell(w) \geq 0$ such that $| \gamma_\ell(w) | \leq \delta_\ell(w)$ for almost all $w \in \supp(W)$, for all $\ell \in \{1,\ldots,L \}$. Let $\delta(w) = (\delta_1(w),\ldots,\delta_L(w))'$.
\end{enumerate}
Write
\begin{align*}
	\gamma(w) &= \var(Z \mid W=w)^{-1} \big( \cov(Z,Y \mid W=w) - \cov(Z,X \mid W=w)\beta(w) \big)
\end{align*}
for any $w \in \supp(W)$.

By derivations similar to those in section \ref{sec:homogModel}, the identified set for $\beta(w)$ is
\begin{multline*}
	\mathcal{B}(\delta(w) \mid w) \\
	= 
	\big\{ b \in \R^K: -\delta(w) \leq \var(Z \mid W=w)^{-1} \big( \cov(Z,Y \mid W=w) - \cov(Z,X \mid W=w) b \big) \leq \delta(w) \big\}
\end{multline*}
for any $w \in \supp(W)$. For a given $\delta(\cdot)$, the model is refuted if the set $\{ w \in \supp(W) : \mathcal{B}(\delta(w) \mid w) = \emptyset \}$ has positive probability.

\subsubsection*{Subgroup Analysis}

Suppose that the baseline model discussed in section \ref{sec:homogModel} is refuted. Or suppose the model with covariates based on equation (\ref{eq:constantCoeffOutcomeEqGen}$^\prime$) is refuted. As discussed in section \ref{sec:AltResponses}, one response is to switch to the nonparametric controls model based on equation (\ref{eq:constantCoeffOutcomeEqGen}$^{\prime \prime}$), and then find the set of conditioning covariates for which the model is not refuted. Specifically, define
\[
	\mathcal{W}(0) = \{w \in \supp(W) : \mathcal{B}(0 \mid w) \neq \emptyset \}.
\]
This set is point identified. If we impose A\ref{assump:homog:relevance:gen}$^*$--A\ref{assump:exclusion:gen}$^*$ only for $w \in \mathcal{W}(0)$, then the model is not refuted. Thus one approach to salvaging the model is to focus on the subgroups $\mathcal{W}(0)$, and to report $\beta(w)$ for $w \in \mathcal{W}(0)$.

Alternatively, one could make specific assumptions to extrapolate from $\beta(w)$ for $w \in \mathcal{W}(0)$ to $\beta(w)$ for $w$ in the complement $\mathcal{W}(0)^c$. For example, suppose $\beta(w)$ is observed to be constant across $w \in \mathcal{W}(0)$. Then one could assume that it is also constant across $w \in \mathcal{W}(0)^c$, and that this constant value is the same as the observed constant value across $w \in \mathcal{W}(0)$.

Suppose, however, that we want to allow for $\beta(w)$ to vary over $w$, and that we are interested in the entire population, not just the subgroups in $\mathcal{W}(0)$. In this case, a natural parameter is the mean coefficient, $\Exp[ \beta(W) ]$. We can construct the falsification adaptive set for this parameter as follows. For simplicity, suppose $K=1$. We have
\begin{align*}
	\Exp[ \beta(W) ]
		&= \int_{w \in \supp(W)} \beta(w) \; dF_W(w) \\
		&= \int_{w \in \mathcal{W}(0)} \beta(w) \; dF_W(w)
		+ \int_{w \notin \mathcal{W}(0)} \beta(w) \; dF_W(w).
\end{align*}
The first term is point identified. For the second term, we can construct the falsification adaptive set for $\beta(w)$ for each $w \in \mathcal{W}(0)$ by applying the results of section \ref{sec:homogModel}. Denote these bounds by $[L(w),U(w)]$. Then the falsification adaptive set for $\Exp[\beta(W)]$ is $[L,U]$ where
\[
	L = \int_{w \in \mathcal{W}(0)} \beta(w) \; dF_W(w) + \int_{w \notin \mathcal{W}(0)} L(w) \; dF_W(w)
\]
and
\[
	U = \int_{w \in \mathcal{W}(0)} \beta(w) \; dF_W(w) + \int_{w \notin \mathcal{W}(0)} U(w) \; dF_W(w).
\]
Keep in mind that this falsification adaptive set relaxes exclusion by using values $\delta_\ell(w)$ that vary across $w$. In particular, $\delta_\ell(w) = 0$ for all $w \in \mathcal{W}(0)$. If we require $\delta_\ell(w) = \delta_\ell$ to be constant across $w$, then the falsification adaptive set will be weakly larger.

\section{Technical Equivalence of Exclusion and Exogeneity in Linear Models}\label{sec:ExoExclu}

In section \ref{sec:homogModel} we explicitly modeled relaxations of the instrument exclusion restriction, holding instrument exogeneity fixed. In this section we show that the same formal results can be used to relax exogeneity alone or to relax exogeneity and exclusion simultaneously. This analysis is essentially a variation of \cite{Small2007}.

Let $A$ and $U$ denote unobserved random variables. Suppose the true model for potential outcomes is
\[
	Y(x,z,a) = x' \beta + z' \lambda + \phi a + U.
\]
Thus observed outcomes are
\begin{align*}
	Y &= Y(X,Z,A) \\
		&= X' \beta + Z' \lambda + (\phi A + U).
\end{align*}
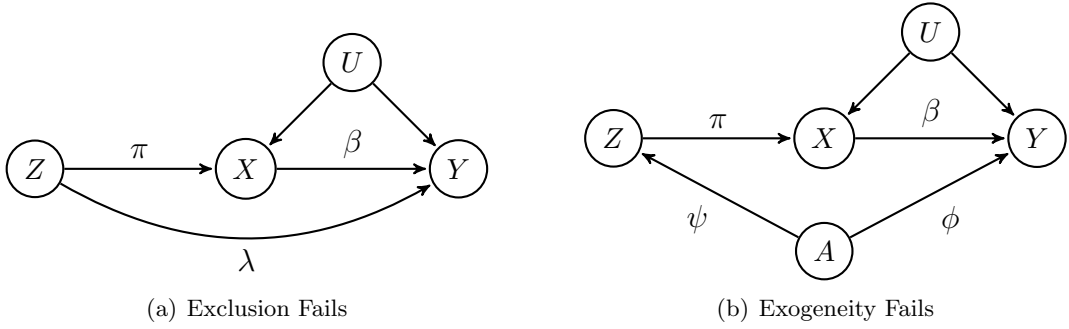
\begin{figure}[t]
\label{fig:IVgraphs}
\centering
\subfigure[Exclusion Fails]{
\begin{tikzpicture}[->,>=stealth',shorten >=1pt,auto,node distance=2cm,
  thick,main node/.style={circle,fill=white,draw,font=\sffamily\bfseries}]

\node[main node] (1) {$U$};
\node[main node] (2) [below left of = 1] {$X$};
\node[main node] (3) [below right of = 1]{$Y$};
\node[main node] (4) [left=2cm of 2]{$Z$};

\path[every node/.style={font=\sffamily\large}]
(1) edge node [above right] {} (3)
(1) edge node [above right] {} (2)
(2) edge node [above] {$\beta$} (3)
(4) edge node [above] {$\pi$} (2)
(4) edge [bend right] node [below] {$\lambda$} (3);
\end{tikzpicture}
}
\hspace{.3in}
\subfigure[Exogeneity Fails]{
\begin{tikzpicture}[->,>=stealth',shorten >=1pt,auto,node distance=2cm,
  thick,main node/.style={circle,fill=white,draw,font=\sffamily\bfseries}]

\node[main node] (1) {$U$};
\node[main node] (2) [below left of = 1] {$X$};
\node[main node] (3) [below right of = 1]{$Y$};
\node[main node] (4) [left=2cm of 2]{$Z$};
\node[main node] (5) [below=0.7cm of 2]{$A$};

\path[every node/.style={font=\sffamily\large}]
(1) edge node [above right] {} (3)
(1) edge node [above right] {} (2)
(2) edge node [above] {$\beta$} (3)
(4) edge node [above] {$\pi$} (2)
(5) edge node [below left] {$\psi$} (4)
(5) edge node [below right] {$\phi$} (3);
\end{tikzpicture}
}
\caption{Two ways for the standard instrumental variable assumptions to fail.}
\end{figure}
This equation allows for violations of exclusion if $\lambda$ is nonzero. It also allows for violations of exogeneity if $A$ and $Z$ are correlated. Figure \ref{fig:IVgraphs} shows these two kinds of failures. Let $\psi$ be the coefficients on $Z$ in a linear projection of $A$ onto $Z$ and let $\widetilde{U}$ be the projection residual. Thus
\[
	A = Z' \psi + \widetilde{U}
\]
where $\cov(Z_\ell, \widetilde{U}) = 0$ by construction. Then
\begin{align*}
	Y
		&= X' \beta + Z' \lambda + \phi A + U \\
		&= X' \beta + Z' \lambda + \phi(Z' \psi + \widetilde{V})  + U \\
		&= X' \beta + Z' (\lambda + \phi \psi) + (\widetilde{U} + U) \\
		&\equiv X' \beta + Z' \gamma + V.
\end{align*}
In the last line we defined
\[
	\gamma = \lambda + \phi \psi
\]
and $V = \widetilde{U} + U$. For all $\ell \in \{1,\ldots,L\}$, we have $\cov(Z_\ell,V) = 0$ since $\cov(Z_\ell,U) = 0$ by A\ref{assump:exogeneity:gen} and $\cov(Z_\ell, \widetilde{U}) = 0$ by construction.

Thus the formal results in section \ref{sec:homogModel} apply to relaxations of exogeneity as well. The key difference is that the \emph{interpretation} of the coefficient on $Z$ changes depending on whether we are interested in exclusion or exogeneity failure. For example, suppose $\lambda = 0$ so that we are only worried about exogeneity failures. Then the interpretation of $\gamma$ depends on the size of the causal effect of $\phi$ and the relationship between $Z$ and $A$. This is the case that \cite{Small2007} focuses on. To aid interpretation of these parameters, he adapts the ideas of \cite{Imbens2003} to phrase constraints on $\phi$ and $\| \psi \|$ in terms of $R^2$ measures.

\section{Comparing the 2SLS Estimand with the Falsification Adapative Set}\label{sec:2SLSvsFAS}\label{sec:comparing2SLSandFAS}

The following proposition gives the relationship between the 2SLS estimand
\[
	\beta_{\textsc{2SLS}} = \frac{\cov(X,Z)\var(Z)^{-1}\cov(Z,Y)}{\cov(X,Z)\var(Z)^{-1}\cov(Z,X)}
\]
and the falsification adaptive set when there is a single endogenous variable. Here we assume all instruments are relevant.

\begin{proposition}\label{prop:2SLSandFAS}
Suppose $K=1$. Suppose $\pi_\ell \neq 0$ for all $\ell =1,\ldots,L$. Then
\[
	\beta_{\textsc{2SLS}} = \sum_{\ell = 1}^L W_\ell  \frac{\psi_\ell}{\pi_\ell}
\]
where
\[
	W_\ell = \frac{\cov(\sum_{j=1}^L \pi_j Z_j,\pi_\ell Z_\ell) }{\var(\sum_{j=1}^L \pi_j Z_j)}.
\]
\end{proposition}

Recall that $\psi_\ell / \pi_\ell$ is the 2SLS estimand using $Z_\ell$ as an instrument and the remaining instruments $Z_{-\ell}$ as controls. The falsification adaptive set is the interval between the smallest and largest such estimands. Proposition \ref{prop:2SLSandFAS} shows that the 2SLS estimand is a weighted average of the estimands $\psi_\ell / \pi_\ell$. The weights sum to one, $\sum_{\ell = 1}^L W_\ell = 1$, but the weights are not necessarily positive. For example, this may happen if covariances between instruments are negative, or if some $\pi_\ell < 0$. 

Thus an important implication of this result is that the 2SLS estimand does not have to be inside the falsification adaptive set. Indeed, it can be arbitrarily far from it. Put differently, in a falsified linear IV model, the 2SLS estimand can be arbitrarily far from the set of parameters which correspond to minimally non-falsified models.

\begin{proof}[Proof of proposition \ref{prop:2SLSandFAS}]
Recall that
\[
	\underset{(L \times 1)}{\psi} \equiv \var(Z)^{-1}\cov(Z,Y)
	\qquad \text{and} \qquad
	\underset{(L \times 1)}{\pi} \equiv \var(Z)^{-1}\cov(Z,X).
\]
Thus
\[
	\beta_\textsc{2SLS} = \frac{\cov(X,Z) \psi}{\cov(X,Z) \pi}.
\]
Since $\pi$ is the vector of coefficients from a linear projection of $X$ onto the instruments, we can write
\[
	X = \sum_{\ell = 1}^L \pi_\ell Z_\ell + V
\]
where $\cov(Z_\ell,V) = 0$ for all $\ell$. Consider the numerator of $\beta_{\textsc{2SLS}}$:
\begin{align*}
	\cov(X,Z) \psi
		&= \sum_{\ell =1}^L \cov(X,Z_\ell) \psi_\ell \\
		&= \sum_{\ell =1}^L \cov \left( \sum_{j=1}^L \pi_j Z_j + V,Z_\ell \right) \psi_\ell \\
		&= \sum_{\ell =1}^L \cov \left( \sum_{j=1}^L \pi_j Z_j,Z_\ell \right) \psi_\ell \\
		&= \sum_{\ell =1}^L \frac{\psi_\ell}{\pi_\ell} \cov \left( \sum_{j=1}^L \pi_j Z_j, \pi_\ell Z_\ell \right).
\end{align*}
Next consider the denominator of $\beta_{\textsc{2SLS}}$:
\begin{align*}
	\cov(X,Z) \pi
		&= \sum_{\ell =1}^L \cov(X,Z_\ell) \pi_\ell \\
		&= \sum_{\ell =1}^L \cov \left(\sum_{s = 1}^L \pi_s Z_s + V,Z_\ell \right) \pi_\ell \\
		&= \sum_{\ell =1}^L \sum_{s = 1}^L \cov ( \pi_s Z_s,\pi_\ell Z_\ell) \\
		&= \var \left( \sum_{j=1}^L \pi_j Z_j \right).
\end{align*}
Putting these two results together yields the result.
\end{proof}

\section{Further Remarks on Interpreting Non-Falsified Models}\label{sec:GMMobjFun}

In this section we elaborate on the discussion in section \ref{subsec:FFinterp}. We only consider the linear instrumental variable model of section \ref{sec:homogModel} with one endogenous variable and two instruments. \cite{Newey1985} analyzes the problem in a much broader class of moment equality models.

With just two instruments, the outcome equation is
\begin{equation}\label{eq:outcomeEqForNonFalseInterpSec}
	Y = \alpha + \beta X + \gamma_1 Z_1 + \gamma_2 Z_2 + U.
\end{equation}
Here we include a constant intercept $\alpha$, unlike in section \ref{sec:homogModel}. So we also impose the normalization $\Exp(U) = 0$. This normalization combined with the exogeneity assumption A\ref{assump:exogeneity:gen} gives us the following three moment conditions:
\begin{align*}
	\Exp[U] &= 0 \\
	\Exp[Z_1 U] &= 0 \\
	\Exp[Z_2 U] &= 0.
\end{align*}
We can use these moments to define a GMM objective function
\[
	Q(b) = m(b)' W m(b)
\]
where $W$ is a symmetric positive definite weighting matrix, $m(b) = (m_1(b),m_2(b))'$, and 
\[
	m_\ell(b) = \cov(Z_\ell,Y) - b \cov(Z_\ell,X).
\]
for $\ell \in \{1,2\}$. This leads to the following standard result.

\begin{proposition}\label{prop:GMMobjFun}
Consider the outcome equation \eqref{eq:outcomeEqForNonFalseInterpSec}. Suppose $\Exp(U) = 0$. Suppose A\ref{assump:homog:relevance:gen}, A\ref{assump:exogeneity:gen}, and A\ref{assump:exclusion:gen} hold; that is, suppose the baseline model is true. Then $Q$ is uniquely minimized at $\beta$. Moreover, $Q(\beta) = 0$.
\end{proposition}

This result immediately provides a way of checking whether the model is falsified. Let
\[
	Q^* = \min_{b \in \R} \; Q(b).
\]

\begin{corollary}\label{corr:checkGMMobjFunZeroForFalse}
Suppose $Q^* > 0$. Then the model is false.
\end{corollary}

As we discussed in section \ref{subsec:FFinterp}, however, the inverse statement is \emph{not} true. That is, observing $Q^* = 0$ does not imply the model is true. It merely implies that there exist values of the unobservables that are consistent with the assumptions and the data. But, as is well known, there \emph{also} exist values of the unobservables that are consistent with the data but which are \emph{not} consistent with the assumptions. To see this, we will study the value of $Q^*$ as the true coefficients $(\gamma_1,\gamma_2)$ on the instruments vary. For simplicity, suppose we use the identity weight matrix. Then 
\begin{equation}\label{eq:equalWeightGMMobjFn}
	Q(b) = \big( \cov(Z_1,Y) - b \cov(Z_1,X) \big)^2 + \big( \cov(Z_2,Y) - b \cov(Z_2,X) \big)^2.
\end{equation}

\begin{proposition}\label{prop:valueFunctionExpression}
Suppose the true model is given by \eqref{eq:outcomeEqForNonFalseInterpSec} with coefficients $(\beta,\gamma_1,\gamma_2)$. That is, the exclusion restriction does not necessarily hold. Suppose A\ref{assump:homog:relevance:gen} and A\ref{assump:exogeneity:gen} hold. Then the value of $Q$ evaluated at its optimizer is
\[
	Q^*(\gamma_1,\gamma_2)
	= f(\gamma_1,\gamma_2)^2 \left[ (1-\omega)^2 + \omega^2 \left( \frac{\cov(Z_2,X)}{\cov(Z_1,X)} \right)^2 \right]
\]
where
\[
	f(\gamma_1,\gamma_2) = \gamma_1 \left[ \var(Z_1) - \cov(Z_1,Z_2) \frac{\cov(Z_1,X)}{\cov(Z_2,X)} \right] + \gamma_2 \left[ \cov(Z_1,Z_2) - \var(Z_2) \frac{\cov(Z_1,X)}{\cov(Z_2,X)} \right]
\]
and
\[
	\omega = \frac{\cov(X,Z_1)^2}{\cov(X,Z_1)^2 + \cov(X,Z_2)^2}.
\]
\end{proposition}

Notice that $Q^*(\gamma_1,\gamma_2)$ does not depend on the value of $\beta$. The left plot in figure \ref{fig:GMMobjFnBaseline} shows an example of $Q^*(\gamma_1,\gamma_2)$.\footnote{This figure plots the function $Q_\textsc{mt}^*(\delta, \gamma_1, \gamma_2)$ which we discuss below. When evaluated at $\delta = 0$, and when $\var(Z_1) = \var(Z_2) = 1$ as used in the figure, we have $ Q_\textsc{mt}^*(0, \gamma_1,\gamma_2) = Q^*(\gamma_1,\gamma_2)$. In general, $Q_\textsc{mt}^*(0, \gamma_1,\gamma_2)$ equals the GMM objective function with weight matrix $W = \text{diag}(1/\var(Z_1)^2, 1/\var(Z_2)^2)$.} 
From this proposition, we see that $Q^*(\gamma_1,\gamma_2) = 0$ if and only if $f(\gamma_1,\gamma_2) = 0$. This equation defines a line in $\R^2$. Thus the model is not refuted if and only if $(\gamma_1,\gamma_2)$ falls on this line. Only one of these points---the origin---corresponds to the baseline model. All other points on this line are false models which are observationally equivalent to a set of structural parameters which satisfy the model assumptions.

\begin{figure}[t]
\centering
\includegraphics[width=52mm]{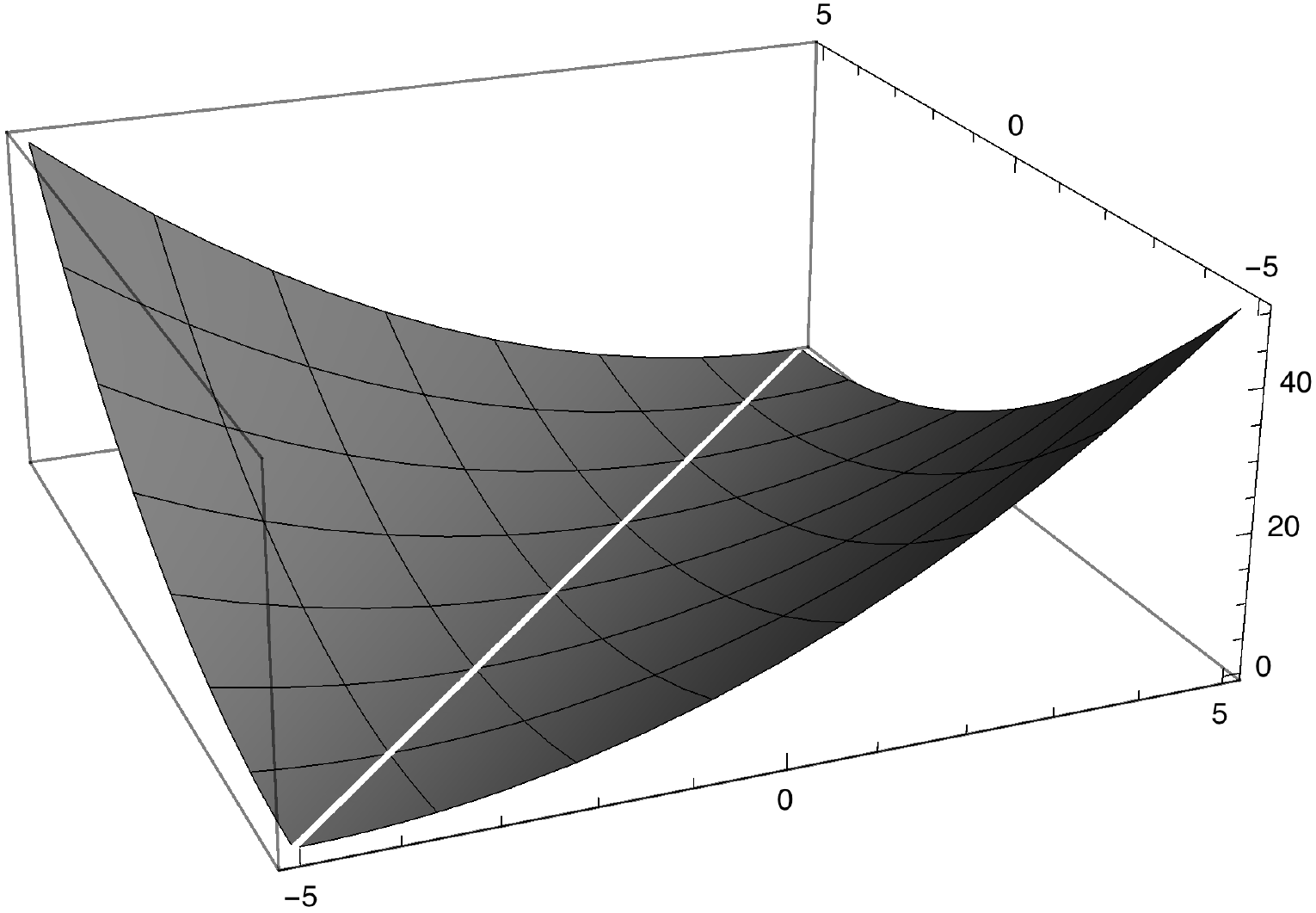}
\hspace{1mm}
\includegraphics[width=52mm]{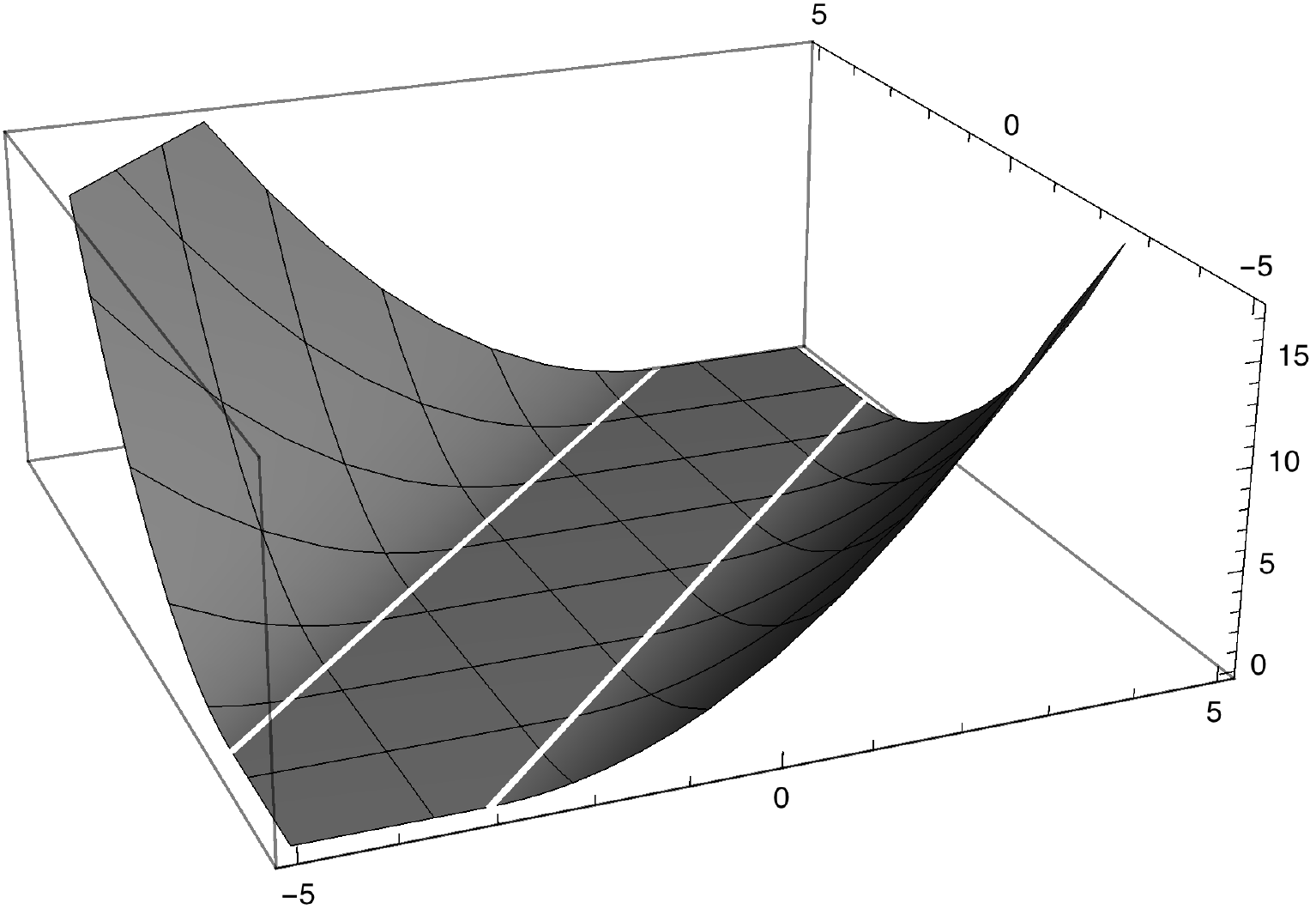}
\hspace{1mm}
\includegraphics[width=52mm]{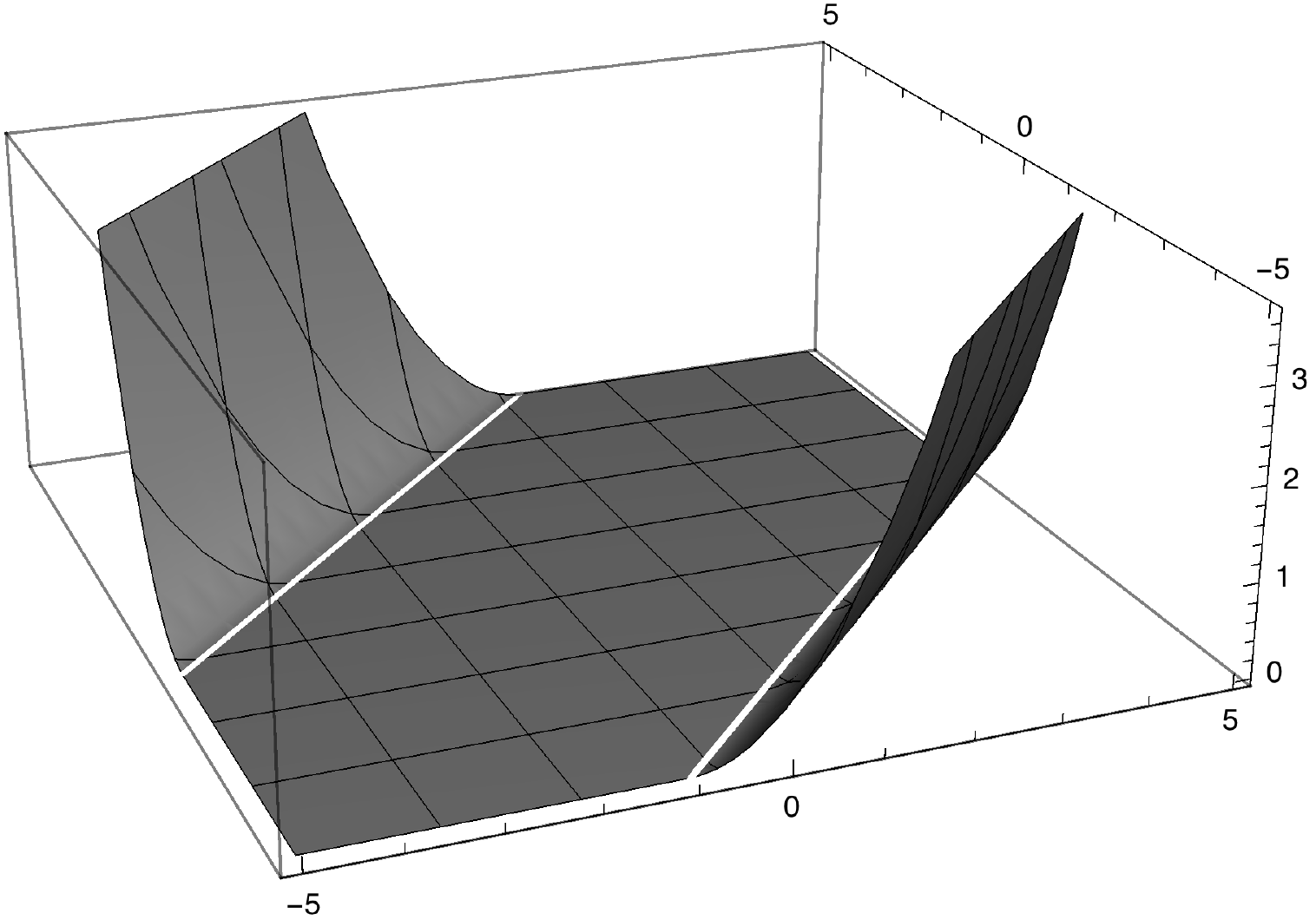}
\caption{Plots of $Q_\textsc{mt}^*(\delta,\gamma_1,\gamma_2)$ as functions of $(\gamma_1,\gamma_2)$ for three different values of $\delta$. Left: The baseline case, $\delta_1 = \delta_2 = 0$. Middle: A relaxed model with $\delta_1 = \delta_2 = 1$. Right: A relaxed model with $\delta_1 = \delta_2 = 2$. All three plots set $\var(Z_1) = \var(Z_2) = 1$, $\cov(Z_1,Z_2) = 0$, and $\cov(Z_1,X) = \cov(Z_2,X) = 0.5$. }
\label{fig:GMMobjFnBaseline}
\end{figure}

Next we consider the relaxed model which replaces exclusion A\ref{assump:exclusion:gen} with partial exclusion A\ref{assump:exclusion:gen}$^\prime$. Following \cite{ManskiTamer2002} and \cite{ChernozhukovHongTamer2007}, we can again characterize falsification of the model in terms of a certain objective function. Under partial exclusion, theorem \ref{thm:idset:homog:gen} shows that the identified set $\mathcal{B}(\delta)$ is characterized by four moment inequalities:
\begin{align*}
	(\psi_\ell - \pi_\ell \beta) + \delta_\ell &\geq 0 \\
	-(\psi_\ell - \pi_\ell \beta) + \delta_\ell &\geq 0
\end{align*}
for $\ell \in \{1,2\}$. Recall that $\psi$ and $\pi$ were defined on page \pageref{eq:defOfPsiPi}. Define
\[
	Q_\textsc{mt}(b,\delta) = \sum_{\ell=1}^2 \Big( \max \big\{ - [(\psi_\ell - \pi_\ell b) + \delta_\ell], 0 \big\}^2 + \max \big\{ - [-(\psi_\ell - \pi_\ell b) + \delta_\ell], 0 \big\}^2 \Big).
\]
Then $Q_\textsc{mt}(b,\delta) = 0$ if and only if $b \in \mathcal{B}(\delta)$. This objective function generalizes equation \eqref{eq:equalWeightGMMobjFn}. Recall the notation $B_\ell(\delta) = [\underline{b}_\ell(\delta), \overline{b}_\ell(\delta)]$ from equation \eqref{eq:K1identBetaSetLinearIV}. The following proposition characterizes this objective function at its minimal value.

\begin{proposition}\label{prop:QoptimalValueIdentSet}
Let $Q_\textsc{mt}^*(\delta) = \min_{b \in \R} Q_\textsc{mt}(b,\delta)$. Then
\[
	Q_\textsc{mt}^*(\delta)
	=
	\begin{cases}
		0 
		&\text{if $\max \{ \underline{b}_1(\delta), \underline{b}_2(\delta) \} \leq \min \{ \overline{b}_1(\delta), \overline{b}_2(\delta) \}$} \\[1em]
		\dfrac{1}{2} \big( \underline{b}_2(\delta) - \overline{b}_1(\delta) \big)^2 
		&\text{if $\overline{b}_1(\delta) < \underline{b}_2(\delta)$} \\[1em]
		\dfrac{1}{2} \big( \underline{b}_1(\delta) - \overline{b}_2(\delta) \big)^2 
		&\text{if $\overline{b}_2(\delta) < \underline{b}_1(\delta)$}.
	\end{cases}
\]
\end{proposition}

The identified set $\mathcal{B}(\delta)$ is the intersection of two intervals, $B_1(\delta)$ and $B_2(\delta)$. When these two intervals overlap, the model is not refuted and $Q_\textsc{mt}^*(\delta) = 0$. This is the first case. When the intervals do not overlap, $Q_\textsc{mt}^*(\delta)$ is strictly positive. The second case above occurs when $B_1(\delta)$ is completely below $B_2(\delta)$. The third case above is the reverse. That is, if $Q_\textsc{mt}^*(\delta) > 0$, then the model is false.

As in proposition \ref{prop:valueFunctionExpression} above, we can view $Q_\textsc{mt}^*(\delta)$ as a function of the true underlying parameters, which we write as
\[
	Q_\textsc{mt}^*(\delta, \gamma_1, \gamma_2).
\]
We give an explicit expression for this function on page \pageref{proof:derivationsOfPartialIdentQstar}. The middle and right plots in figure \ref{fig:GMMobjFnBaseline} show two examples of $Q_\textsc{mt}^*(\delta,\gamma_1,\gamma_2)$, for two different values of $\delta$. Like the baseline case, $(\delta_1,\delta_2) = (0,0)$, there is a set of values $(\gamma_1,\gamma_2)$ where the objective function is zero. This set contains the values in the square $[-\delta_1,\delta_1] \times [-\delta_2,\delta_2]$ which are consistent with partial exclusion A\ref{assump:exclusion:gen}$^\prime$. But this set also contains values outside of this square. Each of these values is a false model which is observationally equivalent to a set of structural parameters which satisfy the model assumptions. As the components of $\delta$ increase, this set of $(\gamma_1,\gamma_2)$ values such that  $Q_\textsc{mt}^*(\delta,\gamma_1,\gamma_2) = 0$ grows.

Finally, we relate these results to the falsification frontier. In figure \ref{fig:GMMobjFnBaseline}, as we vary $(\gamma_1,\gamma_2)$ in the domain, the distribution of the observed data changes. This follows since $\cov(Z_1,Y)$ and $\cov(Z_2,Y)$ change when $(\gamma_1,\gamma_2)$ change while holding all else fixed. In our analysis in sections \ref{sec:generalFF} and \ref{sec:homogModel} we fix a single distribution of the observed data $(Y,X,Z_1,Z_2)$. We then check whether the model is falsified by that distribution. For example, we check whether the objective function in equation \eqref{eq:equalWeightGMMobjFn} is zero at its optimizer. Suppose it is not. Then the model is falsified. This means that, in the left plot of figure \ref{fig:GMMobjFnBaseline}, we observed a distribution of $(Y,X,Z_1,Z_2)$ off the diagonal defined by $f(\gamma_1,\gamma_2) = 0$. As we increase the components of $\delta$, the value function $Q_\textsc{mt}^*(\delta)$ will flatten, until eventually we have $Q_\textsc{mt}^*(\delta) = 0$. The smallest set of $\delta$'s at which this happens is what we call the falsification frontier.

\subsubsection*{The Identified Set for $(\beta,\gamma)$}

An alternative way to view these arguments is to study the identified set for $(\beta,\gamma)$ when no constraints on the values of $\gamma$ are imposed. Specifically, we have the following result.

\begin{proposition}\label{prop:IdentSetForGammas}
Suppose A\ref{assump:homog:relevance:gen}, A\ref{assump:homog:nonsing:gen}, and A\ref{assump:exogeneity:gen} hold. Suppose $K=1$. Then
\[
	\Gamma = \left\{ \gamma \in \R^L : \gamma = \psi - b \Pi \text{ for some $b \in \R$} \right\}
\]
is the identified set for $\gamma$. Moreover, the identified set for $(\beta,\gamma)$ is
\[
	\{ (\beta,\gamma) \in \R^{1+L} : \psi - \beta \Pi = \gamma \text{ for some $\gamma \in \Gamma$} \}.
\]
\end{proposition}

Recall from page \pageref{eq:defOfPsiPi} that $\psi$ is the vector of reduced form regression coefficients of $Y$ on $Z$. Similarly, $\Pi$ is the vector of first stage regression coefficients of $X$ on $Z$. Proposition \ref{prop:IdentSetForGammas} is essentially a restatement of equation (8) in \cite{Small2007}. 

The identified set $\Gamma$ is never empty, which means that without any constraints on the direct causal effects of the instruments, the model is not refutable. With a single instrument ($L = 1$), the identified set $\Gamma$ always contains 0. That is, the exclusion restriction $\gamma = 0$ is always consistent with the data and assumptions A\ref{assump:homog:relevance:gen}, A\ref{assump:homog:nonsing:gen}, and A\ref{assump:exogeneity:gen}. This is simply the well known result that the classical linear IV model with one instrument is not refutable.

With $L \geq 2$ instruments, however, the identified set $\Gamma$ does not necessarily contain the origin. Thus the baseline model (the outcome equation \eqref{eq:constantCoeffOutcomeEqGen} plus assumptions A\ref{assump:homog:relevance:gen}--A\ref{assump:exclusion:gen}) with two or more instruments is refutable: The baseline model is refuted if and only if the origin is not an element of the identified set $\Gamma$. Figure \ref{fig:GammaIdentSet} shows two examples of the identified sets $\Gamma$, for the two instrument case. On the left, the baseline model is not refuted since the set $\Gamma$ passes through the origin. On the right, the baseline model is refuted since the set $\Gamma$ does not pass through the origin. %

\begin{figure}[t]
\centering
\includegraphics[width=80mm]{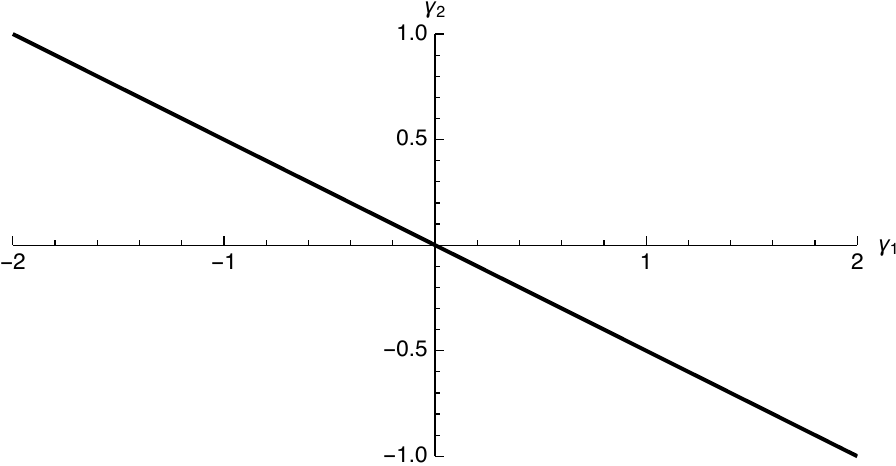}
\hspace{1mm}
\includegraphics[width=80mm]{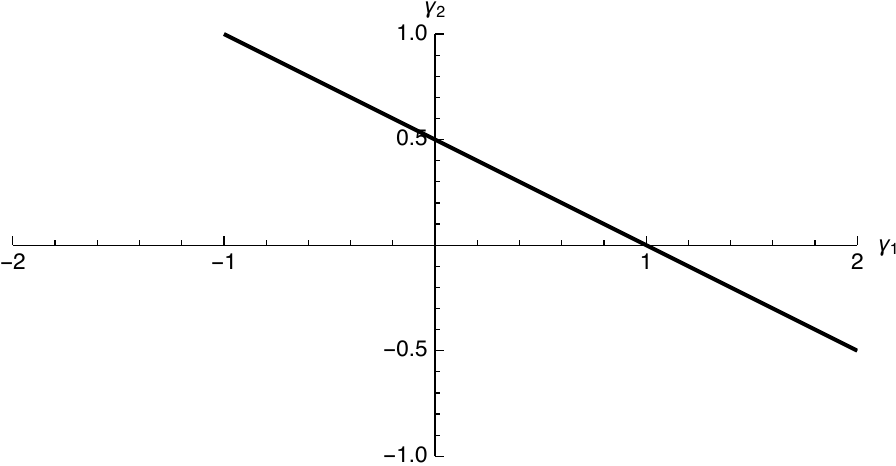}
\caption{Example identified sets for $(\gamma_1,\gamma_2)$. Left: A distribution of $(Y,X,Z_1,Z_2)$ which does not refute the baseline model. Right: A distribution of $(Y,X,Z_1,Z_2)$ which does refute the baseline model.}
\label{fig:GammaIdentSet}
\end{figure}

Importantly, even when the baseline model is not refuted, the identified set for $\gamma$ always contains elements which are inconsistent with the exclusion restriction. This is the same point we discussed above, by analyzing the GMM objective function.

To relate this discussion to the falsification frontier, fix a distribution of the observed data $(Y,X,Z)$. We then check whether the baseline model is falsified by that distribution. For example, we check whether $\Gamma$ contains the origin. Suppose it does not. Then the baseline model is falsified. We then compute the falsification frontier. While this set was given by proposition \ref{prop:K1FFhomogTrt}, it can alternatively be obtained as the set
\[
	\text{FF} = \left\{ \delta \in [0,\infty)^L : \left( \prod_{\ell=1}^L [-\delta_\ell, \delta_\ell] \right) \cap \Gamma \text{ is a singleton} \right\}.
\]
Essentially, we consider boxes around the origin. We then find all boxes that have a singleton intersection with the line $\Gamma$.

At any point $\delta$ on the falsification frontier, the relaxed model (the outcome equation \eqref{eq:constantCoeffOutcomeEqGen} plus assumptions A\ref{assump:homog:relevance:gen}--A\ref{assump:exogeneity:gen} and A\ref{assump:exclusion:gen}$^\prime$) is not refuted and there is a unique $\gamma$ value consistent with the relaxed model assumptions. If we drop the partial exclusion assumption A\ref{assump:exclusion:gen}$^\prime$, however, there are infinitely many values of $\gamma$ consistent with the remaining assumptions and the data---all the elements of the set $\Gamma$.

\subsubsection*{Proofs for section \ref{sec:GMMobjFun}}

\begin{proof}[Proof of proposition \ref{prop:GMMobjFun}]
This result is standard, but we include the proof for completeness. Substituting the outcome equation \eqref{eq:outcomeEqForNonFalseInterpSec} into our system of three moments and using exclusion yields the moments
\begin{align*}
	\Exp[ Y - \alpha - \beta X ] &= 0 \\
	\Exp[Z_1 (Y - \alpha - \beta X)] &= 0 \\
	\Exp[Z_2 (Y - \alpha - \beta X)] &= 0.
\end{align*}
Next we profile out the intercept. Solving the first equation for $\alpha$ yields $\alpha = \Exp(Y) - \beta \Exp(X)$. Substituting this into the second two equations gives
\begin{align*}
	\Exp[Z_1 (Y - \Exp(Y) + \beta \Exp(X) - \beta X)] &= 0 \\
	\Exp[Z_2 (Y - \Exp(Y) + \beta \Exp(X) - \beta X)] &= 0.
\end{align*}
We can rewrite this system as
\begin{align*}
	\cov(Z_1,Y) - \beta \cov(Z_1,X) &= 0 \\
	\cov(Z_2, Y) - \beta \cov(Z_2,X) &= 0.
\end{align*}
This is a system of two moment equalities with one unknown, $\beta$. The left hand side is the vector $m(\beta)$. This explains the definition of $Q$ and shows that $Q(\beta) = 0$. The fact that $Q$ is uniquely minimized at $\beta$ follows by the relevance assumption A\ref{assump:homog:relevance:gen} and since the moments are linear in $\beta$.
\end{proof}

\begin{proof}[Proof of corollary \ref{corr:checkGMMobjFunZeroForFalse}]
This is simply the contrapositive of proposition \ref{prop:GMMobjFun}. For concreteness, however, notice that the model is falsified if and only if both instruments are relevant, so that $\cov(Z_\ell,X) \neq 0$ for $\ell \in \{1,2\}$, and 
\[
	\frac{\cov(Z_1,Y)}{\cov(Z_1,X)} \neq \frac{\cov(Z_2,Y)}{\cov(Z_2,X)}.
\]
This happens when there does not exist a solution to the population moment equations.
\end{proof}

\begin{proof}[Proof of proposition \ref{prop:valueFunctionExpression}]
Solving the first order conditions for minimizing $Q$ yields
\begin{align*}
	b^*
		&= \frac{\cov(Z_1,Y) \cov(Z_1,X) + \cov(Z_2, Y) \cov(Z_2,X)}{\cov(Z_1,X)^2 + \cov(Z_2,X)^2} \\
		&= \frac{\cov(Y,Z_1)}{\cov(X,Z_1)} \frac{\cov(X,Z_1)^2}{\cov(X,Z_1)^2 + \cov(X,Z_2)^2} + \frac{ \cov(Y,Z_2)}{ \cov(X,Z_2)} \frac{\cov(X,Z_2)^2}{\cov(X,Z_1)^2 + \cov(X,Z_2)^2} \\
		&\equiv \beta_{Z_1}^\textsc{2sls} \omega + \beta_{Z_2}^\textsc{2sls} (1-\omega)
\end{align*}
where $\omega \in [0,1]$. That is, the optimizer is a convex combination of the two individual 2SLS estimands. Thus the value function is
\begin{align*}
	Q^* 
		&= Q(b^*) \\
		&= \big( \cov(Z_1,Y) - b^* \cov(Z_1,X) \big)^2 + \big( \cov(Z_2,Y) - b^* \cov(Z_2,X) \big)^2 \\
		&= \big( \cov(Z_1,Y) - [\beta_{Z_1}^\textsc{2sls} \omega + \beta_{Z_2}^\textsc{2sls} (1-\omega)] \cov(Z_1,X) \big)^2 \\
		&\qquad + \big( \cov(Z_2,Y) - [\beta_{Z_1}^\textsc{2sls} \omega + \beta_{Z_2}^\textsc{2sls} (1-\omega)] \cov(Z_2,X) \big)^2 \\
		&= \left( \cov(Z_1,Y) - \cov(Z_1,Y) \omega - (1-\omega) \beta_{Z_2}^\textsc{2sls} \cov(Z_1,X) \right)^2 \\
		&\qquad + \left( \cov(Z_2,Y) - (1-\omega) \cov(Z_2,Y) - \omega \beta_{Z_1}^\textsc{2sls} \cov(Z_2,X) \right)^2 \\
		&= (1-\omega)^2 \left( \cov(Z_1,Y) - \cov(Z_1,X) \beta_{Z_2}^\textsc{2sls} \right)^2 +\omega^2 \left( \cov(Z_2,Y) - \cov(Z_2,X) \beta_{Z_1}^\textsc{2sls} \right)^2 \\
		&= (1-\omega)^2 \left( \cov(Z_1,Y) - \cov(Z_2,Y) \frac{\cov(Z_1,X)}{\cov(Z_2,X)} \right)^2 \\
		&\quad +\omega^2 \left( \cov(Z_2,Y) - \cov(Z_1,Y) \frac{\cov(Z_2,X)}{\cov(Z_1,X)} \right)^2.
\end{align*}
To understand how this value function changes as the underlying dgp changes, we substitute the outcome equation \eqref{eq:outcomeEqForNonFalseInterpSec} for $Y$. Since $Y$ only enters via its covariance with $Z_\ell$, this yields
\begin{align*}
	\cov(Z_\ell, Y)
		&= \cov \left( Z_\ell, \alpha + \beta X + \gamma_1 Z_1 + \gamma_2 Z_2 + U \right) \\
		&= \beta \cov(Z_\ell, X) + \sum_{j \in \{1,2\}} \gamma_j \cov(Z_\ell, Z_j).
\end{align*}
Thus
\begin{align*}
	&Q^*(\gamma_1,\gamma_2) \\
		&= (1-\omega)^2 \Bigg( \beta \cov(Z_1,X) + \gamma_1 \var(Z_1) + \gamma_2 \cov(Z_1,Z_2) \\
		&\qquad\qquad\qquad - [\beta \cov(Z_2,X) + \gamma_1 \cov(Z_1,Z_2) + \gamma_2 \var(Z_2)] \frac{\cov(Z_1,X)}{\cov(Z_2,X)} \Bigg)^2 \\
		&\quad +\omega^2 \Bigg( \beta \cov(Z_2,X) + \gamma_1 \cov(Z_1,Z_2) + \gamma_2 \var(Z_2) \\
		&\quad\qquad\qquad -[\beta \cov(Z_1,X) + \gamma_1 \var(Z_1) + \gamma_2 \cov(Z_1,Z_2) ] 
		  \frac{\cov(Z_2,X)}{\cov(Z_1,X)} \Bigg)^2 \\
		&= (1-\omega)^2 \left( \gamma_1 \left[ \var(Z_1) - \cov(Z_1,Z_2) \frac{\cov(Z_1,X)}{\cov(Z_2,X)} \right] + \gamma_2 \left[ \cov(Z_1,Z_2) - \var(Z_2) \frac{\cov(Z_1,X)}{\cov(Z_2,X)} \right] \right)^2 \\
		&\quad +\omega^2 \left( 
		\gamma_1 \left[ \cov(Z_1,Z_2) - \var(Z_1) \frac{\cov(Z_2,X)}{\cov(Z_1,X)} \right]
		+ \gamma_2 \left[ \var(Z_2) - \cov(Z_1,Z_2) \frac{\cov(Z_2,X)}{\cov(Z_1,X)} \right] \right)^2 \\
		&= (1-\omega)^2 \left( \gamma_1 \left[ \var(Z_1) - \cov(Z_1,Z_2) \frac{\cov(Z_1,X)}{\cov(Z_2,X)} \right] + \gamma_2 \left[ \cov(Z_1,Z_2) - \var(Z_2) \frac{\cov(Z_1,X)}{\cov(Z_2,X)} \right] \right)^2 \\
		&\quad +\omega^2 \left( \frac{\cov(Z_2,X)}{\cov(Z_1,X)} \right)^2 \left( 
		\gamma_1 \left[ \var(Z_1) - \cov(Z_1,Z_2) \frac{\cov(Z_1,X)}{\cov(Z_2,X)} \right]
		+ \gamma_2 \left[ \cov(Z_1,Z_2) - \var(Z_2) \frac{\cov(Z_1,X)}{\cov(Z_2,X)} \right] \right)^2.
\end{align*}
\end{proof}

\begin{proof}[Proof of proposition \ref{prop:QoptimalValueIdentSet}]
When the model is not refuted, $Q_\textsc{mt}(b,\delta) = 0$ by definition. This gives us the first case in the definition of $Q_\textsc{mt}^*(\delta)$. So suppose the model is refuted. Then the intervals $[\underline{b}_1(\delta), \overline{b}_1(\delta)]$ and $[\underline{b}_2(\delta), \overline{b}_2(\delta)]$ do not intersect. Hence one of them lies strictly to the left of the other. Suppose that 
\[
	[L_1,U_1] \equiv [\underline{b}_1(\delta), \overline{b}_1(\delta)]
\]
lies to the left of
\[
	[L_2,U_2] \equiv [\underline{b}_2(\delta), \overline{b}_2(\delta)],
\]
meaning that $U_1 < L_2$. First we show that $Q_\textsc{mt}(b,\delta)$ is uniquely minimized at the midpoint between $U_1$ and $L_2$,
\[
	b^* = \frac{U_1 + L_2}{2}.
\]
First note that any value $b$ outside of $[L_1,U_2]$ can never be a minimizer. Furthermore, $b^*$ must lie inside $[U_1,L_2]$. For $b$'s in this range, the objective function is
\begin{align*}
	Q_\textsc{mt}(b,\delta)
		&= \max \{ b-U_1, 0 \}^2 + \max \{ -(b - L_2), 0 \}^2 \\
		&= (b - U_1)^2 + (b - L_2)^2.
\end{align*}
The FOC is
\[
	2(b - U_1)(-1) + 2(b - L_2)(-1) = 0.
\]
Solve for $b$ to get the midpoint. This essentially follows since our objective function gives equal weight to violations of each of the two inequalities. Next, to obtain the value of the objective function at the optimizer, simply plug $b^*$ back in to get
\begin{align*}
	Q_\textsc{mt}(b^*)
		&= \left( \frac{U_1+L_2}{2} - U_1 \right)^2 + \left( \frac{U_1+L_2}{2} - L_2 \right)^2 \\
		&= \frac{(L_2 - U_1)^2}{4} + \frac{(U_1-L_2)^2}{4} \\
		&= \frac{1}{2} (U_1-L_2)^2.
\end{align*}
This gives us the second case in the definition of $Q_\textsc{mt}^*(\delta)$. The third case occurs when the interval $[\underline{b}_2(\delta), \overline{b}_2(\delta)]$ lies to the left of $[\underline{b}_1(\delta), \overline{b}_1(\delta)]$. The result follows by following the derivations of the second case, but reversing the role of the two intervals.
\end{proof}

\begin{proof}[Derivation of $Q_\textsc{mt}^*(\delta,\gamma_1,\gamma_2)$]\label{proof:derivationsOfPartialIdentQstar}
First substitute the definitions of the bounds $\underline{b}_\ell(\delta)$ and $\overline{b}_\ell(\delta)$ into the expression for $Q_\textsc{mt}^*(\delta)$. This yields $Q_\textsc{mt}^*(\delta) = 0$ if 
\[
	 \max \left\{ \dfrac{\psi_1}{\pi_1} - \dfrac{\delta_1}{|\pi_1|}, \dfrac{\psi_2}{\pi_2} - \dfrac{\delta_2}{| \pi_2 |} \right\} \leq \min \left\{ \dfrac{\psi_1}{\pi_1} + \dfrac{\delta_1}{|\pi_1|}, \dfrac{\psi_2}{\pi_2} + \dfrac{\delta_2}{| \pi_2 |}  \right\}
\]
and
\[
	Q_\textsc{mt}^*(\delta)
	=
	\begin{cases}
		\dfrac{1}{2} \left[ \left( \dfrac{\psi_2}{\pi_2} - \dfrac{\delta_2}{| \pi_2 |} \right) - \left( \dfrac{\psi_1}{\pi_1} + \dfrac{\delta_1}{| \pi_1 |} \right) \right]^2
		&\text{if $\left( \dfrac{\psi_2}{\pi_2} - \dfrac{\delta_2}{| \pi_2 |} \right) > \left( \dfrac{\psi_1}{\pi_1} + \dfrac{\delta_1}{| \pi_1 |} \right)$} \\[1em]
		\dfrac{1}{2} 
		\left[ \left( \dfrac{\psi_1}{\pi_1} - \dfrac{\delta_1}{|\pi_1|} \right) -
		\left( \dfrac{\psi_2}{\pi_2} + \dfrac{\delta_2}{|\pi_2|} \right) \right]^2
		&\text{if $\left( \dfrac{\psi_1}{\pi_1} - \dfrac{\delta_1}{|\pi_1|} \right) >
		\left( \dfrac{\psi_2}{\pi_2} + \dfrac{\delta_2}{|\pi_2|} \right)$}.
	\end{cases}
\]
Of all of these terms, only $\psi_\ell$ depends on $(\beta,\gamma_1,\gamma_2)$. Specifically,
\begin{align*}
	\psi
		&= \var(Z)^{-1} \cov(Z,Y) \\
		&=
		\begin{pmatrix}
			\var(Z_1) & \cov(Z_1,Z_2) \\
			\cov(Z_1,Z_2) & \var(Z_2)
		\end{pmatrix}^{-1}
		\begin{pmatrix}
			\cov(Z_1,Y) \\
			\cov(Z_2,Y)
		\end{pmatrix} \\
		&= \frac{1}{\var(Z_1) \var(Z_2) - \cov(Z_1,Z_2)^2}
		\begin{pmatrix}
			\var(Z_2) & -\cov(Z_1,Z_2) \\
			-\cov(Z_1,Z_2) & \var(Z_1)
		\end{pmatrix}
		\begin{pmatrix}
			\cov(Z_1,Y) \\
			\cov(Z_2,Y)
		\end{pmatrix}.
\end{align*}
So
\[
	\psi_1
	=
	\frac{\cov(Z_1,Y) \var(Z_2) - \cov(Z_2,Y) \cov(Z_1,Z_2)}{\var(Z_1) \var(Z_2) - \cov(Z_1,Z_2)^2}
\]
and
\[
	\psi_2
	=
	\frac{\cov(Z_2,Y) \var(Z_1) - \cov(Z_1,Y) \cov(Z_1,Z_2)}{\var(Z_1) \var(Z_2) - \cov(Z_1,Z_2)^2}.
\]
Now substitute
\[
	\cov(Z_\ell,Y) = \beta^* \cov(Z_\ell,X) + \sum_{j \in \{1,2\}} \gamma_j^* \cov(Z_\ell,Z_j)
\]
to obtain the desired function $Q_\textsc{mt}^*(\delta,\gamma_1,\gamma_2)$. As in the baseline case, notice that this function does not actually depend on the value of $\beta$.
\end{proof}

\begin{proof}[Proof of proposition \ref{prop:IdentSetForGammas}]
This follows immediately from the proof of theorem \ref{thm:idset:homog:gen}.
\end{proof}

\section{\cite{DurantonMorrowTurner2014}: Analysis of export value}\label{sec:additionalEmpirics}

\cite{DurantonMorrowTurner2014} consider two different outcome variables: Propensity to export weight and propensity to export value. In section \ref{sec:DMT2014} we focused on export weight. Here we briefly discuss the results for export value. Panel A of table \ref{DMTtable3} reproduces columns 5--8 of table 5 in \cite{DurantonMorrowTurner2014}. These are their main results for export value. As with our analysis of export weight, we add the falsification adaptive set to the last row of panel A. In panel B we present results using plan and exploration as instruments, controlling for railroads. Compared to their baseline results, our modified baseline results in panel B suggest that the effect of within city highways on propensity to export value is potentially large and negative. Their theoretical model, however, predicts that this effect should be small and positive (their comparative static 3). That said, there is substantial model uncertainty: The falsification adaptive set in column 4 is $[-0.19, 0.047]$, which includes small positive values along with the large negative values. There is substantial statistical uncertainty as well. Thus the data does not allow us to draw a definitive conclusion about the validity of their model's predicted comparative static.

\begin{table}[!thtp]
\scriptsize
\caption[]{\label{DMTtable3} Baseline 2SLS results for \cite{DurantonMorrowTurner2014}: The effect of highways on export value. Non-highlighted parts reproduce results from their paper. Highlighted parts are new. Panel A reproduces columns 1--4 of their table 5. It also shows the estimated falsification adaptive set. Panel B uses only two of their instruments, controlling for the other. See text for discussion.}
\vspace{2mm}

\begin{adjustwidth}{-0.25in}{-0.25in}

\setlength{\linewidth}{.1cm}
\newcommand{\contents}{
\centering

\begin{tabular}{lcccc}
\hline
\mystrut
& \multicolumn{4}{c}{Dependent variable: Export value} \\[0.3em]
 & (1) & (2) & (3) & (4) \\[0.4em]
 \hline
\multicolumn{5}{l}{Panel A. Plan, exploration, and railroads used as instruments} \mystrut \\[0.5em]
\hline
\mystrut
log highway km & 1.10*** & 0.17 & 0.070 & -0.026 \\
 & (0.17) & (0.16) & (0.14) & (0.12) \\[0.4em]
log employment &  & 0.91*** & 1.19** & 0.90** \\
 &  & (0.091) & (0.58) & (0.42) \\[0.4em]
Market access (export) &  & -0.19 & -0.38*** & -0.36*** \\
 &  & (0.12) & (0.14) & (0.11) \\[0.4em]
log 1920 population &  &  & -0.34 & -0.23 \\
 &  &  & (0.30) & (0.30) \\[0.4em]
log 1950 population &  &  & 0.95* & 0.49 \\
 &  &  & (0.48) & (0.51) \\[0.4em]
log 2000 population &  &  & -0.84 & -0.14 \\
 &  &  & (0.72) & (0.58) \\[0.4em]
log \% manuf. emp. &  &  &  & 0.83*** \\
 &  &  &  & (0.16) \\[0.4em]
 First-stage $F$ stat. & 97.5 & 90.3 & 80 & 84.8 \\[0.4em]
Overid. $p$-value & 0.081 & 0.071 & 0.28 & 0.55 \\[0.4em]
\rowcolor{lightgray}  FAS & [0.24, 0.78] & [-0.85, -0.085] & [-0.85, -0.01] & [-0.19, 0.047] \\[0.4em]
 \hline
\rowcolor{lightgray} \multicolumn{5}{l}{Panel B. Plan and exploration used as instruments, controlling for railroads} \mystrut \\[0.5em]
\hline
\rowcolor{lightgray}  \mystrut
 log highway km & 0.69** & -0.22 & -0.16 & -0.14 \\
\rowcolor{lightgray}  & (0.30) & (0.23) & (0.18) & (0.15) \\[0.4em]
\rowcolor{lightgray} log 1898 railroad km & 0.41** & 0.32** & 0.22* & 0.11 \\
\rowcolor{lightgray}  & (0.20) & (0.14) & (0.13) & (0.11) \\[0.4em]
\rowcolor{lightgray}  First-stage $F$ stat. & 61.1 & 65.4 & 77.8 & 82.2 \\[0.4em]
\rowcolor{lightgray} Overid. $p$-value & 0.60 & 0.47 & 0.41 & 0.74 \\[0.4em]
\hline
\multicolumn{5}{p{\linewidth}}{\emph{Notes}: 66 observations per column. All specifications include a constant. Heteroskedasticity robust standard errors in parentheses. ***, **, *: statistically significant at 1\%, 5\%, 10\%.}
\end{tabular}
}

\setbox0=\hbox{\contents}
\setlength{\linewidth}{\wd0-2\tabcolsep-.25em}
\contents

\end{adjustwidth}
\end{table}

\section{Proofs for section \ref{sec:homogModel}}\label{sec:proofsHomogModel}

\begin{proof}[Proof of proposition \ref{prop:testableImplicationsOfLinearIV}]
This result is well known, but we include a proof for completeness. Suppose equation \eqref{eq:classicalSarganEqs} holds for all $m, \ell \in \{ 1,\ldots,L\}$. We will construct a joint distribution $(Y,X,Z,\widetilde{U})$ and a parameter $\widetilde{\beta}$ consistent with the data and assumptions A\ref{assump:homog:relevance:gen}--A\ref{assump:exclusion:gen}.

By the relevance assumption A\ref{assump:homog:relevance:gen}, there exists an $\ell$ such that $\cov(X,Z_\ell) \neq 0$. Let
\[
	\widetilde{\beta} = \frac{\cov(Y,Z_\ell)}{\cov(X,Z_\ell)}.
\]
Let $\widetilde{U} = Y - X \widetilde{\beta}$. For every $m \in \{1,\ldots,L\}$,
\begin{align*}
	\cov(\widetilde{U}, Z_m) 
	&= \cov(Y,Z_m) - \cov(X,Z_m) \widetilde{\beta} \\
	&= \cov(Y,Z_m) - \cov(X,Z_m) \frac{\cov(Y,Z_\ell)}{\cov(X,Z_\ell)}\\
	&= \frac{\cov(Y,Z_m) \cov(X,Z_\ell) - \cov(Y,Z_\ell)\cov(X,Z_m)}{\cov(X,Z_\ell)}\\
	&= 0.
\end{align*}
Thus A\ref{assump:exogeneity:gen} holds. A\ref{assump:exclusion:gen} holds by definition of $\widetilde{U}$. A\ref{assump:homog:nonsing:gen} holds automatically. Thus the model is not refuted.

Next suppose the model is not refuted. Then there exists a joint distribution of $(Y,X,Z,U)$ and a value $\beta$ consistent with the model assumptions and the data. By A\ref{assump:exogeneity:gen} we have
\begin{align*}
	0
		&= \cov(U,Z_\ell) \\
		&= \cov(Y - X \beta, Z_\ell) \\
		&= \cov(Y,Z_\ell) - \beta \cov(X,Z_\ell)
\end{align*}
for all $\ell \in \{1,\ldots,L \}$. Suppose $\beta = 0$. Then $\cov(Y,Z_\ell) = 0$ for all $\ell$, and hence equation \eqref{eq:classicalSarganEqs} holds for all $m, \ell \in \{1,\ldots,L\}$. 

Suppose $\beta \neq 0$. Suppose $\cov(X,Z_\ell) = 0$. Then the above equation implies that we must have $\cov(Y,Z_\ell) = 0$ as well. Hence equation \eqref{eq:classicalSarganEqs} holds for this $\ell$ and any other $m \in \{ 1,\ldots, L \}$.

Finally, consider any pair $m$ and $\ell$ such that $\cov(X,Z_m) \neq 0$ and $\cov(X,Z_\ell) \neq 0$. Then the above equation implies
\[
	\frac{\cov(Y,Z_\ell)}{\cov(X,Z_\ell)} = \beta = \frac{\cov(Y,Z_m)}{\cov(X,Z_m)}.
\]
Thus we have shown that equation \eqref{eq:classicalSarganEqs} holds for all $m, \ell \in \{1,\ldots,L\}$. 
\end{proof}

\begin{proof}[Proof of theorem \ref{thm:idset:homog:gen}]
First we show that any value of $\beta$ consistent with the model must lie in $\mathcal{B}(\delta)$. By the outcome equation \eqref{eq:constantCoeffOutcomeEqGen} and the instrument exogeneity A\ref{assump:exogeneity:gen},
\begin{align*}
	\cov(Z,Y)
		&= \cov(Z, X' \beta + Z' \gamma + U) \\
		&= \cov(Z,X)\beta + \var(Z)\gamma.
\end{align*}
By A\ref{assump:homog:nonsing:gen},
\[
	\gamma = \var(Z)^{-1} ( \cov(Z,Y) - \cov(Z,X) \beta ).
\]
Since $-\delta \leq \gamma\leq \delta$ (component-wise) by A\ref{assump:exclusion:gen}$^\prime$, we have $\beta \in \mathcal{B}(\delta)$.

Next we show that $\mathcal{B}(\delta)$ is sharp. Let $b \in \mathcal{B}(\delta)$. Define
\[
	\gamma = \var(Z)^{-1}(\cov(Z,Y) - \cov(Z,X)b).
\]
Then $\gamma$ satisfies A\ref{assump:exclusion:gen}$^\prime$ by definition of $\mathcal{B}(\delta)$. Next, define $\widetilde{U} \equiv Y - X'b - Z'\gamma$. Then
\begin{align*}
	\cov(Z, \widetilde{U})
		&= \cov(Z,Y - X'b - Z' \gamma) \\
		&= \cov(Z,Y - X'b - Z'\var(Z)^{-1}(\cov(Z,Y) - \cov(Z,X)b))\\
		&= \cov(Z,Y) - \cov(Z,X)b - \var(Z)\var(Z)^{-1}(\cov(Z,Y) - \cov(Z,X)b)\\
		&= 0.
\end{align*}
Thus A\ref{assump:exogeneity:gen} holds. Hence $\mathcal{B}(\delta)$ is sharp. That the model is refuted if and only if this set is empty follows by the definition of the (sharp) identified set.
\end{proof}

\begin{proof}[Proof of corollary \ref{corr:K1identBetaSetLinearIV}]
Write the identified set from theorem \ref{thm:idset:homog:gen} as
\begin{align*}
	\mathcal{B}(\delta)
	&= \left\{ b\in\R: -\delta \leq \psi - b \pi \leq \delta\right\}\\
	&= \left\{ b \in \R: -\delta_\ell \leq \psi_\ell - b \pi_\ell \leq \delta_\ell, \ \ell = 1,\ldots, L \right\}\\
	&= \left\{ b \in \R: \psi_\ell-\delta_\ell \leq b \pi_\ell \leq \psi_\ell + \delta_\ell, \ \ell = 1,\ldots, L \right\}.
\end{align*}
Equation \eqref{eq:K1identBetaSetLinearIV} follows immediately by considering the three cases $\pi_\ell = 0$, $\pi_\ell < 0$, and $\pi_\ell > 0$ separately.
\end{proof}

\begin{proof}[Proof of lemma \ref{lemma:interpretingPsiOverPi}]
Without loss of generality, let $\ell = 1$. The result for $\ell \neq 1$ can be obtained by permuting the components of the vector $Z$. Then $\widetilde{X} = (X,Z_2,\ldots,Z_L)$. Hence
\[
	\cov(Z,\widetilde{X}_1)
	= 
	\begin{pmatrix}
	\cov(Z_1,X) & \cov(Z_1, Z_{-1})\\
	\cov(Z_{-1},X) & \var(Z_{-1})
	\end{pmatrix}.
\]
By block matrix inversion, the first row of $\cov(Z,\widetilde{X}_1)^{-1}$ is
\begin{multline*}
	e_1'\cov(Z,\widetilde{X}_1)^{-1}
	= \\
	\begin{pmatrix}
		(\cov(Z_1,X) - \cov(Z_1,Z_{-1}) \var(Z_{-1})^{-1} \cov(Z_{-1},X))^{-1}  \\[0.5em]
	 	-(\cov(Z_1,X) - \cov(Z_1,Z_{-1}) \var(Z_{-1})^{-1} \cov(Z_{-1},X))^{-1} \cov(Z_1,Z_{-1})\var(Z_{-1})^{-1}
	\end{pmatrix}'.
\end{multline*}
Hence
\begin{align*}
	&e_1'\cov(Z,\widetilde{X}_1)^{-1} \cov(Z,Y) \\
	&= \frac{\cov(Z_1,Y) - \cov(Z_1,Z_{-1})\var(Z_{-1})^{-1}\cov(Z_{-1},Y)}{\cov(Z_1,X) - \cov(Z_1,Z_{-1}) \var(Z_{-1})^{-1} \cov(Z_{-1},X)}\\
	&= \frac{\cov(Z_1 -  \cov(Z_1,Z_{-1})\var(Z_{-1})^{-1} Z_{-1},Y)}{\cov(Z_1 - \cov(Z_1,Z_{-1}) \var(Z_{-1})^{-1} Z_{-1},X)}\\
	&= \frac{\cov(\widetilde{Z}_1,Y)}{\cov(\widetilde{Z}_1,X)}\\
	&= \frac{\cov(\widetilde{Z}_1,Y)}{\var(\widetilde{Z}_1)} \Bigg/ \frac{\cov(\widetilde{Z}_1,X)}{\var(\widetilde{Z}_1)} \\
	&= \frac{\psi_1}{\pi_1}.
\end{align*}
Here we defined $\widetilde{Z}_1 = Z_1 -  \cov(Z_1,Z_{-1})\var(Z_{-1})^{-1} Z_{-1}$. This is the population residual of $Z_1$ after removing the projection of $Z_1$ onto $Z_{-1}$. The last line thus follows by the partitioned regression formula.
\end{proof}

\begin{proof}[Proof of equation \eqref{eq:2SLSwithAndWithoutControls} on page \pageref{eq:2SLSwithAndWithoutControls}]
As in the proof of lemma \ref{lemma:interpretingPsiOverPi}, without loss of generality suppose $\ell = 1$. Suppose the baseline model ($\delta = 0$) holds. Then
\begin{align*}
	\frac{\psi_1}{\pi_1}
	&= \frac{\cov(\widetilde{Z}_1,Y)}{\cov(\widetilde{Z}_1,X)} \\
	&= \frac{\cov(\widetilde{Z}_1, \beta X + U)}{\cov(\widetilde{Z}_1,X)} \\
	&= \beta \frac{\cov(\widetilde{Z}_1,X) }{\cov(\widetilde{Z}_1,X)} \\
	&= \beta.
\end{align*}
Similarly, $\cov(Z_1,Y) / \cov(Z_1,X) = \beta$. Equation \eqref{eq:2SLSwithAndWithoutControls} follows.
\end{proof}

We use the following lemma in the proofs of proposition \ref{prop:K1FFhomogTrt} and theorem \ref{thm:identSetOnFFhomogTrt}. It says that the identified set for $\beta$ is a singleton at any point $\delta$ in the set FF defined in equation \eqref{eq:linearIVgeneralFF}.

\begin{lemma}\label{lemma:FFisSingleton}
Suppose A\ref{assump:homog:relevance:gen}--A\ref{assump:exogeneity:gen} hold. Suppose $K=1$. Let
\[
	b \in \left[\min_{\ell=1,\ldots,L:\pi_\ell \neq 0} \frac{\psi_\ell}{\pi_\ell}, \max_{\ell=1,\ldots,L:\pi_\ell \neq 0} \frac{\psi_\ell}{\pi_\ell}\right].
\]
Define $\delta(b) = (|\psi_1 - b\pi_1|, \ldots, |\psi_L - b \pi_L|)$. Then $\mathcal{B}(\delta(b)) = \{ b \}$.
\end{lemma}

\begin{proof}[Proof of lemma \ref{lemma:FFisSingleton}]
We have
\begin{align*}
		\mathcal{B}(\delta(b))
		&= \bigcap_{\ell =1,\ldots,L:\pi_\ell \neq 0} \left[\dfrac{\psi_\ell}{\pi_\ell} - \dfrac{\delta_\ell(b)}{|\pi_\ell|}, \dfrac{\psi_\ell}{\pi_\ell} + \dfrac{\delta_\ell(b)}{|\pi_\ell|} \right]\\
		&= \bigcap_{\ell =1,\ldots,L:\pi_\ell \neq 0} 
		\left[\dfrac{\psi_\ell}{\pi_\ell} - \dfrac{|\psi_\ell - b \pi_\ell|}{|\pi_\ell|}, \ \dfrac{\psi_\ell}{\pi_\ell} + \dfrac{|\psi_\ell - b\pi_\ell|}{|\pi_\ell|} \right]\\
		&= \bigcap_{\ell =1,\ldots,L:\pi_\ell \neq 0}  \left[\dfrac{\psi_\ell}{\pi_\ell} - \left|\frac{\psi_\ell}{\pi_\ell} -b \right|, \dfrac{\psi_\ell}{\pi_\ell} + \left|\frac{\psi_\ell}{\pi_\ell} -b \right| \right]\\
		&= \left(\bigcap_{\ell =1,\ldots,L: \psi_\ell \geq b \pi_\ell, \pi_\ell \neq 0}
		\left[\dfrac{\psi_\ell}{\pi_\ell} - \left|\frac{\psi_\ell}{\pi_\ell} -b \right|, \dfrac{\psi_\ell}{\pi_\ell} + \left|\frac{\psi_\ell}{\pi_\ell} -b \right| \right]\right) \\
		&\qquad \bigcap 
		\left(\bigcap_{\ell =1,\ldots,L: \psi_\ell < b \pi_\ell, \pi_\ell \neq 0}
		\left[\dfrac{\psi_\ell}{\pi_\ell} - \left|\frac{\psi_\ell}{\pi_\ell} -b \right|, \dfrac{\psi_\ell}{\pi_\ell} + \left|\frac{\psi_\ell}{\pi_\ell} -b \right| \right]\right) \\
		&= \left(\bigcap_{\ell =1,\ldots,L: \psi_\ell \geq b \pi_\ell, \pi_\ell \neq 0}
		\left[b, 2\frac{\psi_\ell}{\pi_\ell} - b \right]\right)
		\bigcap 
		\left(\bigcap_{\ell =1,\ldots,L: \psi_\ell < b \pi_\ell, \pi_\ell \neq 0}
		\left[2\dfrac{\psi_\ell}{\pi_\ell} - b, b\right]\right)\\
		&= \{b\}.
\end{align*}
The first line follows by equation \eqref{eq:K1identBetaSetLinearIV}, and by the definition of $\delta(b)$. The second line also uses the definition of $\delta(b)$. The remaining lines follow by considering two cases so that we can eliminate the absolute values.
\end{proof}

\begin{proof}[Proof of proposition \ref{prop:K1FFhomogTrt}]
Let $\text{FF}$ denote the true falsification frontier. Let
\[
	\text{FF}^\text{guess} = \left\{ \delta \in \R^L_{\geq 0}: \delta_\ell = | \psi_\ell - b \pi_\ell |, \ \ell=1,\ldots,L, \ b \in \left[\min_{\ell=1,\ldots,L:\pi_\ell \neq 0} \frac{\psi_\ell}{\pi_\ell}, \max_{\ell=1,\ldots,L:\pi_\ell \neq 0} \frac{\psi_\ell}{\pi_\ell}\right]\right\}.
\]
We will show $\text{FF} = \text{FF}^\text{guess}$. We split the proof in three parts. The first two parts together show that $\text{FF}^\text{guess} \subseteq \text{FF}$. The third part shows that $\text{FF}^\text{guess} \supseteq \text{FF}$.
\begin{enumerate}
\item We first show that if $\delta \in \text{FF}^\text{guess}$, then the identified set $\mathcal{B}(\delta)$ is not empty. This follows immediately from lemma \ref{lemma:FFisSingleton}.

\item We next show that $\delta'<\delta$ for $\delta \in \text{FF}^\text{guess}$ implies that $\mathcal{B}(\delta')$ is empty. So let $\delta' < \delta$ where $\delta \in \text{FF}^\text{guess}$ and $\delta' \geq 0$. Consider two cases.
\begin{enumerate}
\item First suppose $\delta'_\ell < \delta_\ell$ for some $\ell$ such that $\pi_\ell = 0$. By the definition of $\text{FF}^\text{guess}$, $\delta_\ell = | \psi_\ell |$.
\begin{itemize}
\item Suppose $\psi_\ell > 0$. Then $\delta_\ell = \psi_\ell$ and hence $[\psi_\ell - \delta_\ell, \psi_\ell + \delta_\ell] = [0, 2 \psi_\ell]$. Since $\delta_\ell' < \delta_\ell$, $\psi_\ell - \delta_\ell' > 0$ and so 0 is not an element of $[\psi_\ell - \delta_\ell', \psi_\ell + \delta_\ell']$. Hence $B_\ell(\delta'_\ell) = \emptyset$ by equation \eqref{eq:K1identBetaSetLinearIV}.

\item Suppose $\psi_\ell < 0$. Then $\delta_\ell = -\psi_\ell$ and hence $[\psi_\ell - \delta_\ell, \psi_\ell + \delta_\ell] = [2 \psi_\ell, 0]$. Since $\delta_\ell' < \delta_\ell$, $\psi_\ell + \delta_\ell' < 0$ and so 0 is not an element of $[\psi_\ell - \delta_\ell', \psi_\ell + \delta_\ell' ]$. Hence $B_\ell(\delta'_\ell) = \emptyset$ by equation \eqref{eq:K1identBetaSetLinearIV}.

\item Suppose $\psi_\ell = 0$. Then $\delta_\ell = 0$ and so we cannot have $\delta_\ell' < \delta_\ell$, since $\delta_\ell' \geq 0$.
\end{itemize}
Thus in this case we must have $\mathcal{B}(\delta') = \emptyset$.

\item Next suppose $\delta'_\ell < \delta_\ell$ for some $\ell$ such that $\pi_\ell \neq 0$. $\delta' < \delta$ implies that $\mathcal{B}(\delta') \subseteq \mathcal{B}(\delta)$. By lemma \ref{lemma:FFisSingleton}, $\mathcal{B}(\delta) = \{ b^* \}$ for some value $b^*$. Thus it suffices to show that $b^* \notin \mathcal{B}(\delta')$. That will imply that $\mathcal{B}(\delta') = \emptyset$.

\medskip

To show that $b^* \notin \mathcal{B}(\delta')$ it suffices to show that $b^* \notin B_\ell(\delta')$ for some $\ell$, since $\mathcal{B}(\delta')$ is the intersection of these sets over all $\ell$'s, by corollary \ref{corr:K1identBetaSetLinearIV}. From that corollary we have
\[
	B_\ell(\delta') = \left[\dfrac{\psi_\ell}{\pi_\ell} - \dfrac{\delta_\ell'}{|\pi_\ell|}, \dfrac{\psi_\ell}{\pi_\ell} + \dfrac{\delta_\ell'}{|\pi_\ell|} \right].
\]
Consider two cases:
\begin{enumerate}
\item Suppose $b^* \leq \psi_\ell / \pi_\ell$. Then
\begin{align*}
	\dfrac{\psi_\ell}{\pi_\ell} - \dfrac{\delta_\ell}{|\pi_\ell|}
	&= \dfrac{\psi_\ell}{\pi_\ell} - \dfrac{|\psi_\ell - b^*\pi_\ell|}{|\pi_\ell|} \\
	&= \dfrac{\psi_\ell}{\pi_\ell} - \left|\frac{\psi_\ell}{\pi_\ell} -b^*\right| \\
	&= b^*.
\end{align*}
Hence
\begin{align*}
	\dfrac{\psi_\ell}{\pi_\ell} - \dfrac{\delta_\ell'}{|\pi_\ell|} 
		&> \dfrac{\psi_\ell}{\pi_\ell} - \dfrac{\delta_\ell}{|\pi_\ell|} \\
		&= b^*.
\end{align*}
Thus $b^* \notin B_\ell(\delta')$.

\medskip

\item Suppose $b^* > \psi_\ell / \pi_\ell$. Then
\begin{align*}
	\dfrac{\psi_\ell}{\pi_\ell} + \dfrac{\delta_\ell}{|\pi_\ell|}
		&= \dfrac{\psi_\ell}{\pi_\ell} + \dfrac{|\psi_\ell - b^*\pi_\ell|}{|\pi_\ell|} \\
		&= \dfrac{\psi_\ell}{\pi_\ell} + \left|\frac{\psi_\ell}{\pi_\ell} -b^*\right| \\
		&= b^*.
\end{align*}
Hence
\begin{align*}
	\dfrac{\psi_\ell}{\pi_\ell} + \dfrac{\delta_\ell'}{|\pi_\ell|} 
		&< \dfrac{\psi_\ell}{\pi_\ell} + \dfrac{\delta_\ell}{|\pi_\ell|} \\
		&= b^*.
\end{align*}
Thus $b^* \notin B_\ell(\delta')$.
\end{enumerate}
\end{enumerate}
Steps 1 and 2 together imply that $\text{FF}^\text{guess} \subseteq \text{FF}$.

\item Finally, we show that $\text{FF}^\text{guess} \supseteq \text{FF}$. We show the contrapositive: $\delta \notin \text{FF}^\text{guess}$ implies that $\delta \notin \text{FF}$. So let $\delta \notin \text{FF}^\text{guess}$. There are three cases to consider.
\begin{enumerate}
\item Suppose $\mathcal{B}(\delta) = \emptyset$. Then $\delta \notin \text{FF}$ by definition.

\item Denote
\[
	b_\text{min} = \min_{\ell=1,\ldots,L:\pi_\ell \neq 0} \frac{\psi_\ell}{\pi_\ell}
	\qquad \text{and} \qquad
	b_\text{max} = \max_{\ell=1,\ldots,L:\pi_\ell \neq 0} \frac{\psi_\ell}{\pi_\ell}.
\]
Suppose $\mathcal{B}(\delta) \neq \emptyset$ and $\mathcal{B}(\delta) \subseteq [b_\text{min}, b_\text{max}]$. We'll show that we can find a $\delta' < \delta$ such that $\mathcal{B}(\delta') \neq \emptyset$, and hence $\delta \notin \text{FF}$. We use two observations:

\begin{enumerate}
\item Since $\delta \notin \text{FF}^\text{guess}$ we must have $\delta \neq \delta(b)$ for all $b \in [b_\text{min}, b_\text{max}]$. In particular, $\delta \neq \delta(b)$ for all $b \in \mathcal{B}(\delta)$, since we're considering the case where $\mathcal{B}(\delta) \subseteq [b_\text{min}, b_\text{max}]$.

\medskip

\item Let $b \in \mathcal{B}(\delta)$. By the characterization of $\mathcal{B}(\delta)$ in theorem \ref{thm:idset:homog:gen} (specifically, see equation \eqref{eq:CalBisIntersectionOfHalfspaces} in our discussion below), this implies that $| \psi_\ell - b \pi_\ell | \leq \delta_\ell$ for all $\ell =1,\ldots,L$. That is, $\delta(b) \leq \delta$.
\end{enumerate}
Let $b'$ be any element of $\mathcal{B}(\delta)$. This exists by assumption. Let $\delta' = \delta(b')$. Observation (i) implies $\delta' \neq \delta$. Observation (ii) implies $\delta' \leq \delta$. Hence $\delta' < \delta$. Moreover, we have $\mathcal{B}(\delta') = \mathcal{B}(\delta(b')) = \{ b' \} \neq \emptyset$ by lemma \ref{lemma:FFisSingleton}. Thus $\delta \notin \text{FF}$, by definition of the falsification frontier.

\item Suppose $\mathcal{B}(\delta)$ contains an element $b \notin [b_\text{min}, b_\text{max}]$. Suppose $b > b_\text{max}$. Let $\delta' = \delta(b_\text{max})$. 
\begin{itemize}
\item Observation (i) above implies $\delta' \neq \delta$.

\medskip

\item Next we show $\delta' \leq \delta$. Observation (ii) above implies $\delta_\ell(b) = | \psi_\ell - b \pi_\ell | \leq \delta_\ell$ for all $\ell=1,\ldots,L$. So it suffices to show
\[
	\delta_\ell(b_\text{max}) = | \psi_\ell - b_\text{max} \pi_\ell | \leq | \psi_\ell - b \pi_\ell | = \delta_\ell(b)
\]
for all $\ell$. This holds immediately for all $\ell$ with $\pi_\ell = 0$. For the other $\ell$'s, it follows by the structure of the $\delta_\ell(\cdot)$ function, the definition
\[
	b_\text{max} = \max_{\ell=1,\ldots,L:\pi_\ell \neq 0} \frac{\psi_\ell}{\pi_\ell},
\]
and the fact that $b > b_\text{max}$. Specifically, notice that for $\ell$ such that $\pi_\ell \neq 0$, $\delta_\ell(\cdot)$ is zero at $\psi_\ell / \pi_\ell$ and otherwise is strictly increasing as the distance between $\psi_\ell / \pi_\ell$ and its argument increases. (Specifically, it is a V-shaped function, whose slope is determined by $\pi_\ell$, but this specific shape is not important here.) Since $b > b_\text{max}$, $b > \psi_s / \pi_s$ for all $s$ with $\pi_s \neq 0$. Thus we must have $\delta_\ell(b) > \delta_\ell(\psi_s/\pi_s)$ for all $\ell$ and $s$ with $\pi_\ell \neq 0$ and $\pi_s \neq 0$. In particular, it holds for the $s$ such that $b_\text{max} = \psi_s / \pi_s$.
\end{itemize}
Hence $\delta' < \delta$. Moreover, $\mathcal{B}(\delta') = \mathcal{B}(\delta(b_\text{max})) = \{ b_\text{max} \} \neq \emptyset$ by lemma \ref{lemma:FFisSingleton}. So $\delta \notin \text{FF}$, by definition of the falsification frontier. A similar argument applies if instead we have $b < b_\text{min}$.
\end{enumerate} 
\end{enumerate}
\end{proof}

\begin{proof}[Proof of corollary \ref{corr:K1L2FFhomogTrt}]
The falsification frontier is given by
\begin{align*}
		\text{FF}
	&= \left\{ \delta \in \R^2_{\geq 0}: \delta_\ell = | \psi_\ell - b \pi_\ell |, \ \ell=1,2, \ b \in \left[\min_{\ell=1,2} \frac{\psi_\ell}{\pi_\ell}, \max_{\ell=1,2} \frac{\psi_\ell}{\pi_\ell}\right]\right\},
\end{align*}
a line segment in $\R^2$. We can directly see that $\text{FF}$ is the line segment between
\[
	\left( 0, \left|\psi_2 - \frac{\psi_1}{\pi_1}\pi_2\right| \right)
	\qquad \text{and} \qquad
	\left( \left|\psi_1 - \frac{\psi_2}{\pi_2}\pi_1\right|, 0 \right)
\]
which corresponds the equation
\[
	| \pi_2 | \left| \frac{\psi_1}{\pi_1} - \frac{\psi_2}{\pi_2} \right| - \left| \frac{\pi_2}{\pi_1} \right| \delta_1 = \delta_2
\]
for nonnegative values of $\delta_1$ and $\delta_2$.
\end{proof}

\begin{proof}[Proof of theorem \ref{thm:identSetOnFFhomogTrt}]
We have
\begin{align*}
\bigcup_{\delta\in\text{FF}}\mathcal{B}(\delta) &= \bigcup_{b \in \left[\min_{\ell=1,\ldots,L:\pi_\ell \neq 0} \frac{\psi_\ell}{\pi_\ell}, \max_{\ell=1,\ldots,L:\pi_\ell \neq 0} \frac{\psi_\ell}{\pi_\ell}\right]} \mathcal{B}(\delta(b))\\
&= \bigcup_{b \in \left[\min_{\ell=1,\ldots,L:\pi_\ell \neq 0} \frac{\psi_\ell}{\pi_\ell}, \max_{\ell=1,\ldots,L:\pi_\ell \neq 0} \frac{\psi_\ell}{\pi_\ell}\right]} \{b\}\\
 &= \left[\min_{\ell=1,\ldots,L:\pi_\ell \neq 0} \frac{\psi_\ell}{\pi_\ell}, \max_{\ell=1,\ldots,L:\pi_\ell \neq 0} \frac{\psi_\ell}{\pi_\ell}\right].
\end{align*}
The first equality follows by the characterization of the falsification frontier in proposition \ref{prop:K1FFhomogTrt}. The second equality follows by lemma \ref{lemma:FFisSingleton}.
\end{proof}

\begin{proof}[Proof of proposition \ref{prop:directionalFalsificationPoint}]
Let
\[
	m^* =
	\max_{\ell,\ell' \in \{1,\ldots,L\}} \;
	\frac{\dfrac{\psi_\ell}{\pi_\ell} - \dfrac{\psi_{\ell'}}{\pi_{\ell'}}
	}{
	\dfrac{d_\ell}{|\pi_\ell|} + \dfrac{d_{\ell'}}{|\pi_{\ell'}|}
	}.
\]
Suppose $m < m^*$. Then there is an $s$ and an $s'$ such that
\[
	m <
	\frac{
	\dfrac{\psi_s}{\pi_s} - \dfrac{\psi_{s'}}{\pi_{s'}} 
	}{
	\dfrac{d_s}{| \pi_s |} - \dfrac{d_{s'}}{| \pi_{s'} |}
	}.
\]
Rearranging yields
\[
	\frac{\psi_s}{\pi_s} - \frac{m d_s}{| \pi_s |}
	>
	\frac{\psi_{s'}}{\pi_{s'}} + \frac{m d_{s'}}{| \pi_{s'} |}.
\]
Hence
\[
	\max_{\ell = 1,\ldots,L} \left( \frac{\psi_\ell}{\pi_\ell} - \frac{m d_\ell}{| \pi_\ell |} \right)
	>
	\min_{\ell' = 1,\ldots,L} \left( \frac{\psi_{\ell'}}{\pi_{\ell'}} + \frac{m d_{\ell'}}{| \pi_{\ell'} |} \right).
\]
But recall from corollary \ref{eq:K1identBetaSetLinearIV} that for any $\delta$,
\begin{align*}
	\mathcal{B}(\delta)
		&= \bigcap_{\ell=1,\ldots,L} \left[\dfrac{\psi_\ell}{\pi_\ell} - \dfrac{\delta_\ell}{|\pi_\ell|}, \dfrac{\psi_\ell}{\pi_\ell} + \dfrac{\delta_\ell}{|\pi_\ell|} \right] \\
		&= \left[ \max_{\ell = 1,\ldots,L} \left( \frac{\psi_\ell}{\pi_\ell} - \frac{m d_\ell}{| \pi_\ell |} \right), \min_{\ell' = 1,\ldots,L} \left( \frac{\psi_{\ell'}}{\pi_{\ell'}} + \frac{m d_{\ell'}}{| \pi_{\ell'} |} \right) \right].
\end{align*}
We just showed, however, that for $\delta = m \cdot d$, the left endpoint is larger than the right endpoint. Hence $\mathcal{B}(m \cdot d)$ is empty. 

Now, let $\delta = m^*\cdot d$. Let $s, s'$ be such that 
\[
m^* = 
	\frac{
	\dfrac{\psi_s}{\pi_s} - \dfrac{\psi_{s'}}{\pi_{s'}} 
	}{
	\dfrac{d_s}{| \pi_s |} - \dfrac{d_{s'}}{| \pi_{s'} |}
	}.
	\]
Rearranging yields
\[
	\frac{\psi_s}{\pi_s} - \frac{m^* d_s}{| \pi_s |}
	=
	\frac{\psi_{s'}}{\pi_{s'}} + \frac{m^* d_{s'}}{| \pi_{s'} |}.
\]
Let $b^*$ denote this common value.

Next let $\ell \in \{1,\ldots,L\}$ be arbitrary. We will show that $b^* \in B_\ell(\delta)$. We prove this by contradiction. Consider two cases.
\begin{enumerate}
\item Suppose
\[
	\frac{\psi_\ell}{\pi_\ell} + \frac{m^* d_\ell}{|\pi_\ell|} < b^*.
\]
Rearranging this equation yields
\[
	\frac{
	\dfrac{\psi_s}{\pi_s} - \dfrac{\psi_{\ell}}{\pi_{\ell}} 
	}{
	\dfrac{d_s}{| \pi_s |} + \dfrac{d_{\ell}}{| \pi_{\ell} |}
	} > m^*.
\]
This contradicts the definition of $s$.

\item Suppose
\[
	b^* < \frac{\psi_\ell}{\pi_\ell} + \frac{m^* d_\ell}{|\pi_\ell|}.
\]
Rearranging this equation yields
\[
	\frac{
	\dfrac{\psi_\ell}{\pi_\ell} - \dfrac{\psi_{s'}}{\pi_{s'}} 
	}{
	\dfrac{d_\ell}{| \pi_\ell |} + \dfrac{d_{s'}}{| \pi_{s'} |}
	} > m^*.
\]
This contradicts the definition of $s'$.
\end{enumerate}
Thus $b^* \in B_\ell(\delta)$ for any $\ell\in\{1,\ldots,L\}$. Hence $b^* \in \mathcal{B}(\delta)$. Thus implies that $\mathcal{B}(\delta)$ is not empty, as desired.
\end{proof}

\begin{proof}[Proof of corollary \ref{corr:identSetDirectionalFP}]
This result follows immediately from corollary \ref{corr:K1identBetaSetLinearIV}, proposition \ref{prop:K1FFhomogTrt}, and proposition \ref{prop:directionalFalsificationPoint}.
\end{proof}

\newpage
\subsubsection*{Proofs for section \ref{sec:linearModelGeneralK}: $K$ endogenous variables}

To prove the results in this section, we first review the structure of the identified set for $\beta$. From theorem \ref{thm:idset:homog:gen}, this set is
\begin{align*}
	\mathcal{B}(\delta)
	&= \left\{b\in \R^K: - \delta \leq \psi - \Pi b \leq \delta \right\}\\
	&= \bigcap_{\ell = 1}^L \left\{b\in \R^K: - \delta_\ell \leq \psi_\ell - \pi_\ell' b \leq \delta_\ell \right\}
\end{align*}
where $\pi_\ell'$ is the $\ell$th row of the matrix $\Pi$. So if we let
\[
	B_\ell(\delta_\ell) = \left\{b\in \R^K: - \delta_\ell \leq \psi_\ell - \pi_\ell' b \leq \delta_\ell \right\}
\]
then
\begin{equation}\label{eq:CalBisIntersectionOfHalfspaces}
	\mathcal{B}(\delta) = \bigcap_{\ell = 1}^L B_\ell(\delta_\ell).
\end{equation}
In our baseline model with $\delta_\ell = 0$, $B_\ell(0)$ is a hyperplane in $\R^K$:
\[
	B_\ell(0) = \{b\in\R^K: \psi_\ell = \pi_\ell' b \}.
\]
As we relax exclusion, this set goes from being a hyperplane to the space between two hyperplanes. Figures \ref{fig:LisKplus1FAS_inside}, \ref{fig:LisKplus1FAS_outside}, and \ref{fig:generalLinearFAS} plot examples of these spaces. The overall identified set $\mathcal{B}(\delta)$ is just the intersection of these spaces.

With that as background, we next prove our main results: proposition \ref{prop:generalLinearFF}, our characterization of the falsification frontier, and theorem \ref{thm:generalLinearFAS}, our characterization of the falsification adaptive set. We do this using a sequence of lemmas. 

We begin by showing a basic geometric fact about the set $\text{FAS}^*$ when $L = K+1$. Here and elsewhere we use the notation $\beta^\textsc{2sls}_{-\ell} = \beta_{\{1,\ldots,L\} \setminus \{ \ell \}}^\textsc{2sls}$.

\begin{lemma}\label{lemma:singletonFAS}
Suppose A\ref{assump:homog:relevance:gen}$^\prime$, A\ref{assump:homog:nonsing:gen}, and A\ref{assump:exogeneity:gen} hold. Suppose $L = K +1$. Then exactly one of the following holds:
\begin{enumerate}
\item $\beta^\textsc{2sls}_{-\ell} = \beta^\textsc{2sls}_{-\ell'}$ for all $\ell,\ell'\in\{1,\ldots,L\}$.

\item $\pi_\ell'\beta^\textsc{2sls}_{-\ell} \neq \psi_\ell$ for all $\ell \in \{1,\ldots,L\}$.
\end{enumerate}
\end{lemma}

\begin{proof}[Proof of lemma \ref{lemma:singletonFAS}]
We will show that (1) holds if and only if (2) does not hold.
\begin{enumerate}
\item $(\Rightarrow$) Suppose statement (1) holds. So $\beta^\textsc{2sls}_{-\ell} = \beta^\textsc{2sls}_{-\ell'}$ for all $\ell,\ell'\in\{1,\ldots,L\}$. Then $\pi_\ell'\beta^\textsc{2sls}_{-\ell} = \pi_\ell'\beta^\textsc{2sls}_{-\ell'} = \psi_\ell$ for any $\ell$ and any $\ell'\neq \ell$. Therefore statement (2) does not hold.

\item $(\Leftarrow)$. Suppose (2) does not hold. Then there exists an $\ell$ such that $\pi_\ell'\beta^\textsc{2sls}_{-\ell} = \psi_\ell$. By definition of $\beta_{-\ell}^\textsc{2sls}$, we have $\pi_{\ell'}'\beta^\textsc{2sls}_{-\ell} = \psi_{\ell'}$ for all $\ell' \neq \ell$. But now we know it also holds for $\ell' = \ell$. Thus $\beta^\textsc{2sls}_{-\ell}$ is an element of \emph{every} hyperplane $B_{\ell'}(0)$.  This implies that
\[
	\left(\Pi_{\{1,\ldots,L\}\setminus\{\ell'\}} \right)b = \psi_{\{1,\ldots,L\}\setminus\{\ell'\}}
\]
is satisfied by $b = \beta^\textsc{2sls}_{-\ell}$. But this equation is also satisfied by $b = \beta^\textsc{2sls}_{-\ell'}$, by definition of this estimand. By A\ref{assump:homog:relevance:gen}$^\prime$, this equation has a unique solution. Thus $\beta^\textsc{2sls}_{-\ell'} = \beta^\textsc{2sls}_{-\ell}$. This argument holds for all $\ell' \in \{1,\ldots,L \} \setminus \{ \ell \}$. Thus (1) holds.
\end{enumerate}
\end{proof}

The next lemma shows that, when $L=K+1$ and $\text{FAS}^*$ is not a singleton, we can write any element of $\R^K$ as a weighted sum of our just-identified 2SLS estimands.

\begin{lemma}\label{lemma:weightedAverageOf2SLS}
Suppose A\ref{assump:homog:relevance:gen}$^\prime$ holds. Suppose $L = K+1$. Assume that $\text{FAS}^*$ is not a singleton. Then for any $b \in \R^K$ there exist weights $w_1(b),\ldots,w_L(b)$ such that
\[
	b = \sum_{\ell =1}^L w_\ell(b) \beta_{-\ell}^\textsc{2sls}
\]
and the weights that sum to one, $\sum_{\ell =1}^L w_\ell (b) = 1$.
\end{lemma}

\begin{proof}[Proof of lemma \ref{lemma:weightedAverageOf2SLS}]
Since we want the weights to sum to one, the last weight is always defined by
\[
	1 - \sum_{\ell=1}^{L=1} w_\ell(b)
\]
once we've chosen the first $L-1$ weights. Thus we want to find $(w_1(b),\ldots,w_{L-1}(b))$ such that
\[
	b =
	\begin{pmatrix}
		\beta_{-1}^\textsc{2sls} & \cdots & \beta_{-L}^\textsc{2sls}
	\end{pmatrix}
	\begin{pmatrix}
		w_1(b)\\
		\vdots\\
		1 - \sum_{\ell =1}^{L-1} w_\ell(b)
	\end{pmatrix}
\]
holds. Multiplying this out we get
\[
	b = \sum_{s=1}^{L-1} \beta_{-s}^\textsc{2sls} w_s(b) + \left( \beta_{-L}^\textsc{2sls} - \beta_{-L}^\textsc{2sls} \sum_{s=1}^{L-1} w_s(b) \right)
\]
or
\[
	b - \beta_{-L}^\textsc{2sls} =
	\sum_{s=1}^{L-1} (\beta_{-s}^\textsc{2sls} - \beta_{-L}^\textsc{2sls}) w_s(b).
\]
Thus in matrix notation we have
\[
	b - \beta_{-L}^\textsc{2sls}
	= \underbrace{
	\begin{pmatrix}
		\beta_{-1}^\textsc{2sls} - \beta_{-L}^\textsc{2sls} &
		\cdots &
		\beta_{-(L-1)}^\textsc{2sls} - \beta_{-L}^\textsc{2sls}
	\end{pmatrix}
	}_{(L-1)\times(L-1)}
	\begin{pmatrix}
		w_1(b)\\
		\vdots\\
		w_{L-1}(b)
	\end{pmatrix}.
\]
This system has a unique solution if the matrix
\[
	\begin{pmatrix}
	\beta_{-1}^\textsc{2sls} - \beta_{-L}^\textsc{2sls} & \cdots & \beta_{-(L-1)}^\textsc{2sls} - \beta_{-L}^\textsc{2sls}
	\end{pmatrix}
\]
has full rank. So that's what we'll show.

This matrix doesn't have full rank if we can find a $c \in \R^{L-1}$ with $c \neq \0_{L-1}$ such that
\begin{equation}\label{eq:checkingRankForWeightedAvgLemma}
	\begin{pmatrix}
	\beta_{-1}^\textsc{2sls} - \beta_{-L}^\textsc{2sls} & \cdots & \beta_{-(L-1)}^\textsc{2sls} - \beta_{-L}^\textsc{2sls}
	\end{pmatrix}
	c
	= \mathbf{0}_{L-1}.
\end{equation}
Multiplying out we see that this equality is equivalent to
\[
	\sum_{\ell =1}^{L-1} c_\ell \beta_{-\ell}^\textsc{2sls}
	=
	\beta_{-L}^\textsc{2sls} \sum_{\ell =1}^{L-1} c_\ell.
\]
Premultiplying by $\pi_L'$ gives
\begin{align*}
	\pi_L' \beta_{-L}^\textsc{2sls} \sum_{\ell =1}^{L-1} c_\ell 
	&= \sum_{\ell =1}^{L-1} c_\ell \pi_L' \beta_{-\ell}^\textsc{2sls} \\
	&= \psi_L \sum_{\ell =1}^{L-1} c_\ell.
\end{align*}
The second line follows since $\beta_{-\ell}^\textsc{2sls}$ lies on the line $\psi_L = \pi_L' \tilde{b}$ for all $\ell \neq L$, by definition of these 2SLS estimands. Thus
\[
	(\psi_L - \pi_L' \beta_{-L}^\textsc{2sls}) \sum_{\ell =1}^{L-1} c_\ell
	= 0.
\]
By lemma \ref{lemma:singletonFAS}, $\psi_L \neq \pi_L' \beta_{-L}^\textsc{2sls}$. So we must have
\[
	\sum_{\ell =1}^{L-1} c_\ell =0.
\]
We can now go back and substitute this into equation \eqref{eq:checkingRankForWeightedAvgLemma} to get
\[
	\begin{pmatrix}
	\beta_{-1}^\textsc{2sls} - \beta_{-L}^\textsc{2sls} & \cdots & \beta_{-(L-1)}^\textsc{2sls} - \beta_{-L}^\textsc{2sls}
	\end{pmatrix}
	\begin{pmatrix}
		c_1 \\
		\vdots \\
		- \sum_{\ell-1}^{L-2} c_\ell
	\end{pmatrix}
	= \mathbf{0}_{L-1}.
\]
Multiplying this out yields
\[
	\sum_{\ell=1}^{L-2} c_\ell \beta_{-\ell}^\textsc{2sls} - \beta_{-L}^\textsc{2sls} \sum_{\ell=1}^{L-2} c_\ell
	+ (\beta_{-(L-1)} - \beta_{-L}) \left( - \sum_{\ell=1}^{L-2} c_\ell \right) = \mathbf{0}_{L-1}.
\]
Simplifying gives
\[
	  \beta_{-(L-1)}^\textsc{2sls} \sum_{\ell =1}^{L-2}c_\ell
	  =  \sum_{\ell =1}^{L-2}c_\ell \beta_{-\ell}^\textsc{2sls}.
\]
Premultiply both sides by $\pi_{L-1}$ and repeat the arguments above to get
\[
	\sum_{\ell =1}^{L-2}c_\ell = 0.
\]
Repeat this argument iteratively to see that $c = \textbf{0}_{L-1}$. Thus the matrix
\[
	\begin{pmatrix}
	\beta_{-1}^\textsc{2sls} - \beta_{-L}^\textsc{2sls} & \cdots & \beta_{-(L-1)}^\textsc{2sls} - \beta_{-L}^\textsc{2sls}
	\end{pmatrix}
\]
is nonsingular.
\end{proof}

Now recall our general definition of the falsification frontier: It is the set of values $\delta \in \R^L_{\geq 0}$ such that $\mathcal{B}(\delta) \neq \emptyset$ and $\mathcal{B}(\delta') = \emptyset$ for any $\delta' < \delta$. We first show that, in the $L=K+1$ case, the identified set for $\beta$ is a singleton at any point $\delta$ in the conjectured falsification frontier defined in equation \eqref{eq:KisLplus1FF}. For this we define
\[
	\delta_\ell(b) = | \psi_\ell - \pi_\ell' b |
\]
for all $\ell =1,\ldots,L$. Let $\delta(b) = (\delta_1(b),\ldots,\delta_L(b))$.

\begin{lemma}\label{lemma:generalLinearSingletonOnFF}
Suppose A\ref{assump:homog:relevance:gen}$^\prime$, A\ref{assump:homog:nonsing:gen}, and A\ref{assump:exogeneity:gen} hold. Suppose $L = K +1 $. Let $b \in \text{FAS}^*$. Then $\mathcal{B}(\delta(b)) = \{ b \}$.
\end{lemma}

Figure \ref{fig:LisKplus1FAS_inside} illustrates the geometric intuition for this lemma.

\begin{proof}[Proof of lemma \ref{lemma:generalLinearSingletonOnFF}]
By lemma \ref{lemma:singletonFAS}, there are two cases to consider: $\text{FAS}^*$ is either a singleton or a nondegenerate simplex in $\R^K$.

\bigskip

\textbf{Case 1}. Suppose $\text{FAS}^* =\{b\}$ is a singleton. By the definition of $\text{FAS}^*$, this implies that $b = \beta_\mathcal{L}^\textsc{2SLS}$ for any $\mathcal{L} \subseteq \{ 1,\ldots, L \}$ with $| \mathcal{L} | = K$. Moreover, for any such $\mathcal{L}$,
\begin{align*}
	\bigcap_{\ell \in \mathcal{L}} B_\ell(0) 
		&= \{ \tilde{b} \in \R^K : \Pi_\mathcal{L} \tilde{b} = \psi_\mathcal{L} \} \\
		&= \{ b \}.
\end{align*}
Let $\mathcal{L}_1$ and $\mathcal{L}_2$ be subsets of $\{1,\ldots,L\}$ such that $\mathcal{L}_1 \cup \mathcal{L}_2 = \mathcal{L}$. Then
\begin{align*}
	\mathcal{B}(0)
		&= \bigcap_{\ell=1}^L B_\ell(0) \\
		&= \left( \bigcap_{\ell \in \mathcal{L}_1} B_\ell(0)  \right) \bigcap \left( \bigcap_{\ell \in \mathcal{L}_2} B_\ell(0)  \right) \\
		&= \{ b \}.
\end{align*}
The first line holds by equation \eqref{eq:CalBisIntersectionOfHalfspaces}.
Thus $\mathcal{B}(\delta(b)) = \mathcal{B}(0) = \{ b \}$.

\bigskip

\textbf{Case 2}. Suppose $\text{FAS}^*$ is not a singleton. Then $\pi_\ell'\beta^\textsc{2sls}_{-\ell} \neq \psi_\ell$ for all $\ell \in \{1,\ldots,L\}$. We prove equality of sets by showing that both directions of set inclusion hold.

\medskip

\textbf{Step 1 ($\supseteq$)}. First we show that $\mathcal{B}(\delta(b)) \supseteq \{ b \}$. By definition of $\delta_\ell(\cdot)$,
\[
	\psi_\ell - \pi_\ell' b \in [-\delta_\ell(b),\delta_\ell(b)]
\]
for all $\ell$. Thus, by the characterization of $\mathcal{B}(\cdot)$ in theorem \ref{thm:idset:homog:gen}, $b \in \mathcal{B}(\delta(b))$.

\bigskip

\textbf{Step 2 ($\subseteq$)}. Next we show that $\mathcal{B}(\delta(b)) \subseteq \{ b \}$. First suppose $\delta_\ell(b) = 0$ for all $\ell$. In this case the baseline model is not refuted. $b$ is then the unique value that satisfies all of the instrument constraints. So $\text{FAS}^*$ is a singleton. This is a contradiction. So we must have $\delta_\ell(b) > 0$ for some $\ell$. 

We will show that any element $b^* \neq b$ is not in $\mathcal{B}(\delta(b))$. The set $\text{FAS}^*$ is a polytope. Consider its alternative half-space representation. The half-spaces correspond to one side of the hyperplanes $B_\ell(0)$. For example, the shaded areas in figure \ref{FAS_K2L3} can be written as the intersection of half-spaces formed by the plotted lines. Formally, write
\begin{equation}\label{eq:halfSpaceFAS}
	\text{FAS}^* = \bigcap_{\ell =1}^L \{ \tilde{b} \in \R^K: \psi_\ell - \pi_\ell' \tilde{b} \leq 0\}.
\end{equation}
Here we assume without loss of generality that all of the inequalities go in the same direction ($\leq 0$).
\begin{itemize}
\item If one of the half-spaces originally has the form $\psi_\ell - \pi_\ell' \tilde{b} \geq 0$, it can be rewritten as $-\psi_\ell - (-\pi_\ell' \tilde{b}) \leq 0$. This is equivalent to replacing $Z_\ell$ with $- Z_\ell$. As shown in appendix \ref{sec:transformedInstruments}, changing the sign of an instrument does not affect the estimands $\beta_\mathcal{L}^\textsc{2sls}$ and therefore does not affect the set $\text{FAS}^*$.
\end{itemize}
Now recall that
\begin{align*}
	\mathcal{B}(\delta)
	&= \bigcap_{\ell = 1}^L \left\{ \tilde{b} \in \R^K: - \delta_\ell \leq \psi_\ell - \pi_\ell' \tilde{b} \leq \delta_\ell \right\} \\
	&= \left(\bigcap_{\ell =1}^L
	\{ \tilde{b} \in \R^K: \psi_\ell - \pi_\ell' \tilde{b} \geq - \delta_\ell \}\right) 
	\bigcap 
	\left(\bigcap_{\ell =1}^L
	\{\tilde{b} \in \R^K: \psi_\ell - \pi_\ell' \tilde{b} \leq \delta_\ell \}\right).
\end{align*}
Evaluating this expression at $\delta(b)$ yields
\begin{align*}
	\mathcal{B}(\delta(b))
	&= \left(\bigcap_{\ell =1}^L \{ \tilde{b} \in \R^K: \psi_\ell - \pi_\ell' \tilde{b} \geq - | \psi_\ell - \pi_\ell' b |\}\right) 
	\bigcap 
	\left(\bigcap_{\ell =1}^L\{\tilde{b} \in \R^K: \psi_\ell - \pi_\ell' \tilde{b} \leq |\psi_\ell - \pi_\ell' b |\}\right) \\
	&\equiv \mathcal{P}_1(b) \cap \mathcal{P}_2(b).
\end{align*}
By lemma \ref{lemma:weightedAverageOf2SLS}, any element in $\R^K$ can be written as a linear combination of our $L$ different just-identified 2SLS estimands. In particular, we can write $b^*$ in this way:
\[
	b^* = \sum_{\ell =1}^L w_\ell(b^*) \beta^\textsc{2sls}_{-\ell}
\]
where $\beta^\textsc{2sls}_{-\ell} = \beta_{\{1,\ldots,L\} \setminus \{ \ell \}}^\textsc{2sls}$ and where $w_\ell(b^*)$ are weights that sum to one, $\sum_{\ell =1}^L w_\ell (b^*) = 1$. (These weights are not necessarily non-negative.) Remember that our goal is to show that $b^* \neq b$ implies $b^* \notin \mathcal{B}(\delta(b))$. By our decomposition above, it suffices to show that $b^* \notin \mathcal{P}_1(b)$. 

Since $b \in \text{FAS}^*$,
\[
	\psi_\ell - \pi_\ell' b \leq 0
\]
for all $\ell$. This follows directly from our half-space representation of $\text{FAS}^*$. Thus
\[
	-|\psi_\ell - \pi_\ell' b | = \psi_\ell - \pi_\ell' b
\]
for all $\ell$. Hence 
\begin{align*}
	 \mathcal{P}_1(b)
	 	&= \bigcap_{\ell =1}^L \{ \tilde{b} \in \R^K: \psi_\ell - \pi_\ell' \tilde{b} \geq - | \psi_\ell - \pi_\ell' b |\}\ \\
	 	&= \bigcap_{\ell =1}^L \{ \tilde{b} \in \R^K: \psi_\ell - \pi_\ell' \tilde{b} \geq \psi_\ell - \pi_\ell' b \} \\
	 	&= \bigcap_{\ell =1}^L \{ \tilde{b} \in \R^K: \pi_\ell' (\tilde{b} - b) \leq 0 \}.
\end{align*}
So $b^* \in \mathcal{P}_1(b)$ if and only if $\pi_\ell' (b^* - b) \leq 0$ for all $\ell$. Focus on just one $\ell$ for a moment. Then
\begin{align*}
	\pi_{\ell}' (b^* - b) 
	&= \sum_{s =1}^{L} (w_s(b^*) - w_s(b))\pi_{\ell}' \beta^\textsc{2sls}_{-s}\\
	&= \sum_{s =1}^{L} (w_s(b^*) - w_s(b))\pi_{\ell}' \beta^\textsc{2sls}_{-s} - \pi_{\ell}' \beta^\textsc{2sls}_{-\ell} \sum_{s=1}^L (w_s(b^*) - w_s(b)) \\
	&= \sum_{s \neq \ell} (w_s(b^*) - w_s(b))\pi_{\ell}' \beta^\textsc{2sls}_{-s}
	- \pi_{\ell}' \beta^\textsc{2sls}_{-\ell} \sum_{s \neq \ell} (w_s(b^*) - w_s(b)) \\
	&= \sum_{s \neq \ell} (w_s(b^*) - w_s(b)) \psi_\ell
	- \pi_{\ell}' \beta^\textsc{2sls}_{-\ell} \sum_{s \neq \ell} (w_s(b^*) - w_s(b)) \\
	&= (\psi_{\ell}  - \pi_{\ell}' \beta^\textsc{2sls}_{-\ell})
	\sum_{s \neq \ell} (w_s(b^*) - w_s(b)).
\end{align*}
The first line follows by substituting in our linear combination representation of any element in $\R^K$ for both $b^*$ and $b$. The second line follows since the weights sum to one, so the sum of the difference in weights is zero. The third line comes from cancelling the $\ell$th term in the two summations. The fourth line follows since, for all $s \neq \ell$, $\beta_{-s}^\textsc{2sls}$ lies on the line $\psi_\ell = \pi_\ell' \tilde{b}$. This holds by the definition of these 2SLS estimands.

Next notice that $\psi_{\ell}  - \pi_{\ell}' \beta^\textsc{2sls}_{-\ell} < 0$. This follows from $\beta^\textsc{2sls}_{-\ell}\in \text{FAS}^*$, the half-space representation of $\text{FAS}^*$, and the fact that $\text{FAS}^*$ is a nondegenerate simplex. So suppose by way of contradiction that $b^* \in \mathcal{P}_1(b)$. Then $\pi_\ell' (b^* - b ) \leq 0$ for all $\ell$. We've just seen that this implies
\[
	\sum_{s \neq \ell} (w_s(b^*) - w_s(b)) \geq 0
\]
for all $\ell$. But now we have
\begin{align*}
	0
		&= \sum_{s=1}^L (w_s(b^*) - w_s(b)) \\
		&= \sum_{s \neq \ell} (w_s(b^*) - w_s(b)) + (w_\ell(b^*) - w_\ell(b)).
\end{align*}
Thus
\begin{align*}
	w_\ell(b^*) - w_\ell(b)
		&= - \sum_{s \neq \ell} (w_s(b^*) - w_s(b)) \\
		&\leq 0.
\end{align*}
This holds for all $\ell$. But the only way for the sum of non-positive numbers to equal zero is if they're all zero. Thus $w_\ell(b^*) = w_\ell(b)$ for all $\ell$. This implies that $b^* = b$, a contradiction. Therefore $b^* \notin \mathcal{P}_1(b)$. As discussed above, this suffices to see that $b^* \notin \mathcal{B}(\delta(b))$. 
\end{proof}

\begin{figure}[!t]
\centering
\includegraphics[width=0.32\linewidth]{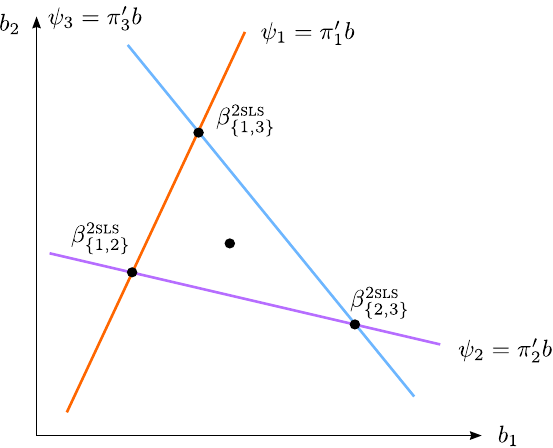}
\includegraphics[width=0.32\linewidth]{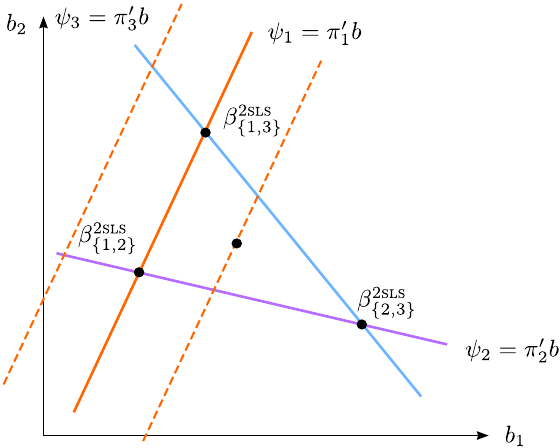}
\includegraphics[width=0.32\linewidth]{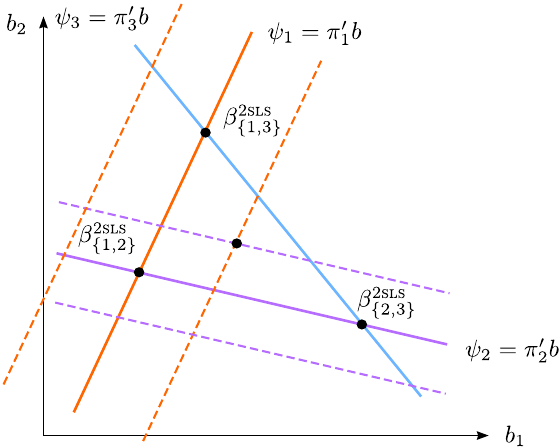} \\[1em]
\includegraphics[width=0.32\linewidth]{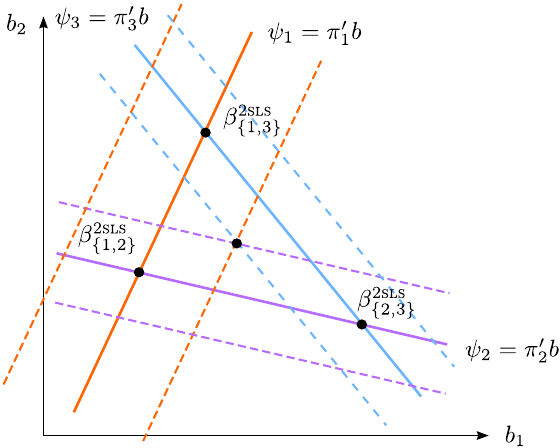}
\hspace{0.01\linewidth}
\includegraphics[width=0.32\linewidth]{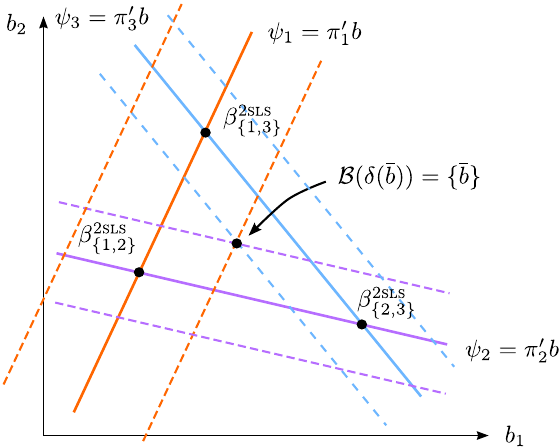}
\caption{Geometric intuition for a key step in the proof of theorem \ref{thm:LisKplus1FAS}. Start at the top left. Here we have $K=2$ and $L=3$, as in the left plot of figure \ref{FAS_K2L3}. We pick a point inside the convex hull of the three just identified 2SLS estimands. Call it $\overline{b}$. We'll now illustrate the intuition for lemma \ref{lemma:generalLinearSingletonOnFF}, which shows that $\mathcal{B}(\delta(\overline{b})) = \{ \overline{b} \}$. We next plot the sets $B_\ell(\delta_\ell(\overline{b})) = \{ b \in \R^K : -\delta_\ell(\overline{b}) \leq \psi_\ell - \pi_\ell' b \leq \delta_\ell(\overline{b}) \}$, adding one at a time. We look at relaxations $\delta_\ell(\overline{b}) = | \psi_\ell - \pi_\ell' \overline{b} |$ which are chosen to be as small as possible, while still making sure that $\overline{b}$ is included in each of the tubes centered at the different hyperplanes $B_\ell(0)$. The result is that the identified set for $\beta$ given the relaxation $\delta(\overline{b})$---the intersection of these three tubes---is the singleton $\{ \overline{b} \}$.}
\label{fig:LisKplus1FAS_inside}
\end{figure}

\begin{figure}[!t]
\centering
\includegraphics[width=0.275\linewidth]{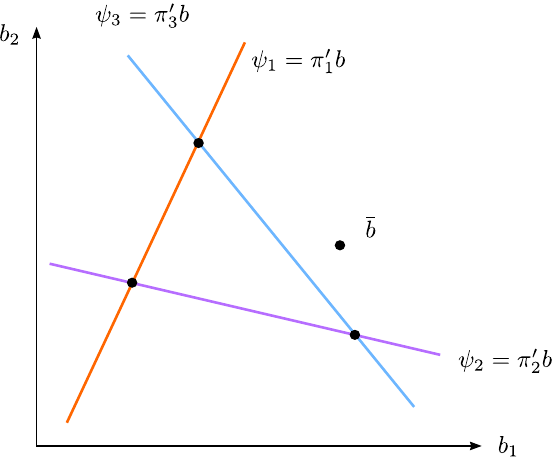}
\hspace{0.01\linewidth}
\includegraphics[width=0.32\linewidth]{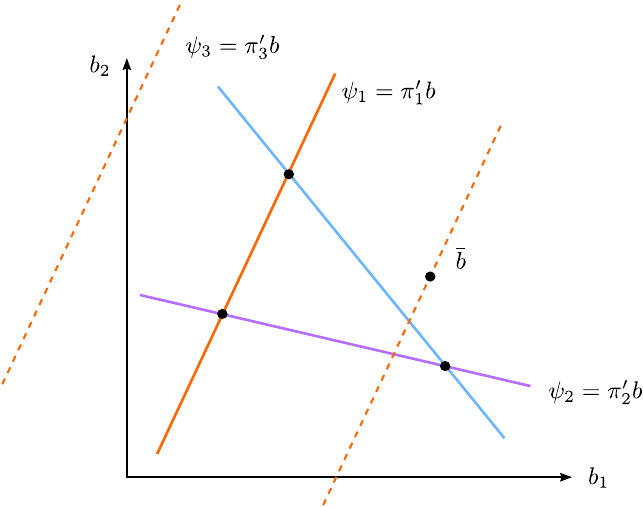}
\hspace{0.01\linewidth}
\includegraphics[width=0.32\linewidth]{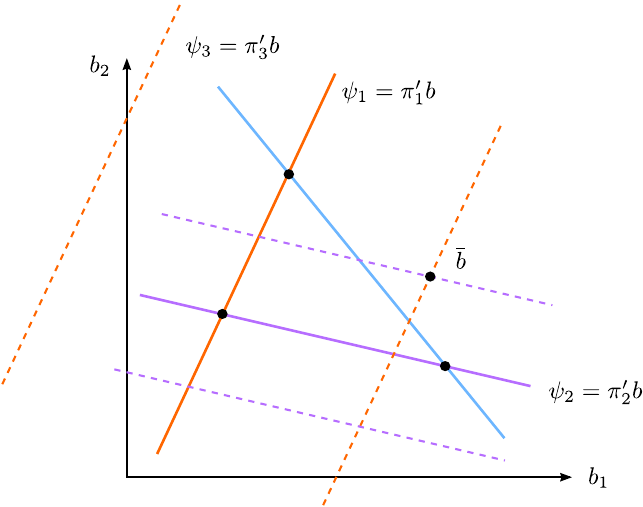} \\
\includegraphics[width=0.32\linewidth]{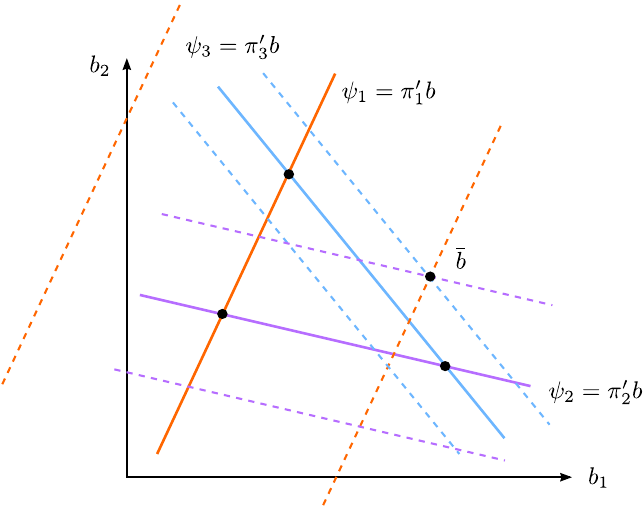}
\hspace{0.01\linewidth}
\includegraphics[width=0.32\linewidth]{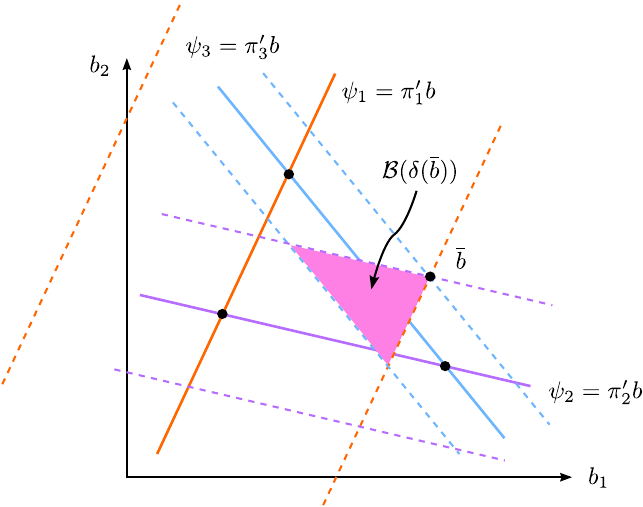}
\caption{Geometric intuition for a key step in the proof of theorem \ref{thm:LisKplus1FAS}, continued from figure \ref{fig:LisKplus1FAS_inside}. Here we instead pick a point \emph{outside} the convex hull of the three just identified 2SLS estimands. We again relax each exclusion restriction just enough so that $\overline{b}$ is inside each of the tubes centered at the different hyperplanes $B_\ell(0)$. But now when we look at the intersection of these tubes we get a non-singleton set. Because this set has an interior, we've relaxed exclusion too much---we can strengthen some of the exclusion restrictions without obtaining a falsified model. Thus the point $\overline{b}$ is not in the falsification adaptive set.}
\label{fig:LisKplus1FAS_outside}
\end{figure}

\begin{lemma}\label{lemma:singletonOnFASgeneralCase}
Suppose A\ref{assump:homog:relevance:gen}$^\prime$, A\ref{assump:homog:nonsing:gen}, and A\ref{assump:exogeneity:gen} hold. Suppose $L \geq K+1$. Let $b \in \mathcal{P}$. Then $\mathcal{B}(\delta(b)) = \{ b \}$.
\end{lemma}
 
\begin{proof}[Proof of lemma \ref{lemma:singletonOnFASgeneralCase}]
As in the proof of lemma \ref{lemma:generalLinearSingletonOnFF}, we prove set equality by showing that both directions of set inclusion hold.

\bigskip

\textbf{Step 1 ($\supseteq$)}. The proof of this step from lemma \ref{lemma:generalLinearSingletonOnFF} applies without modification.

\bigskip

\textbf{Step 2 ($\subseteq$)}. Since $b \in \mathcal{P}$ there is some $\mathcal{L} \subseteq \{1,\ldots,L \}$ with $| \mathcal{L} | = K$ such that $b \in \mathcal{P}_\mathcal{L}$. Let
\[
	\mathcal{B}_{\mathcal{L}}(\delta) = \bigcap_{\ell \in \mathcal{L}} B_\ell(\delta).
\]
This is the identified set for $\beta$ when we impose partial exclusion for only the instruments with indices in $\mathcal{L}$. By lemma \ref{lemma:generalLinearSingletonOnFF}, $\mathcal{B}_\mathcal{L}(\delta(b)) = \{ b \}$. Moreover, $\mathcal{B}_{\mathcal{L}}(\delta) \supseteq \mathcal{B}(\delta)$. Thus $\mathcal{B}(\delta(b)) \subseteq \{ b \}$.
\end{proof}

The following variation on Farkas' lemma (for example, corollary 22.3.1 on page 200 of \citealt{Rockafellar1970}; \citealt{Border2019} provides an extensive discussion) is helpful.

\begin{lemma}[Variation on Farkas' Lemma]\label{lemma:Farkas}
Let $x_1,\ldots,x_n \in \R^K$. Then $\0 \notin \text{conv}(\{x_1,\ldots,x_n\})$ if and only if there exists a point $p \in \R^K$ such that $p'x_i >0$ for all $i=1,\ldots,n$.
\end{lemma}

In particular, this lemma shows that if you pick a point outside a convex hull of $\{ x_1,\ldots,x_n \}$, then you can find a vector $p$ pointing towards all of the vectors $x_1,\ldots,x_n$.

\begin{proof}[Proof of lemma \ref{lemma:Farkas}]
$\0 \notin \text{conv}(\{x_1,\ldots,x_n\})$ if and only if the system 
\begin{align*}
	\begin{pmatrix}
		1 & \cdots & 1\\
		x_1 & \cdots & x_n
	\end{pmatrix}\lambda &= \begin{pmatrix}
	1\\
	\0
	\end{pmatrix} \\
	\lambda &\geq 0
\end{align*}
has no solutions. By Farkas' lemma (for example, theorem 6 in \citealt{Border2019}), this system does not have a solution if and only if there exists a point $\tilde{p} = (p_0, p) \in \R^{K+1}$ such that
\begin{align*}
	\tilde{p}'\begin{pmatrix}
		1 & \cdots & 1\\
		x_1 & \cdots & x_n
	\end{pmatrix} &\geq 0\\
	\tilde{p}' \begin{pmatrix}
	1\\
	\0
	\end{pmatrix} <0.
\end{align*}
That is,
\begin{align*}
	\begin{pmatrix}
		p_0 & p'
	\end{pmatrix}
	\begin{pmatrix}
		1 & \cdots & 1\\
		x_1 & \cdots & x_n
	\end{pmatrix} &\geq 0\\
	\begin{pmatrix}
		p_0 & p'
	\end{pmatrix}
	\begin{pmatrix}
	1\\
	\0
	\end{pmatrix} <0
\end{align*}
which is equivalent to
\begin{align}\label{eq:farkasLastSystem}
	p_0 + p'x_i &\geq 0 \qquad \text{for all $i=1,\ldots,n$} \notag \\
	p_0 &<0.
\end{align}
For the $\Rightarrow$ direction, this immediately implies that $p'x_i > 0$ for all $i=1,\ldots,n$. For the $\Leftarrow$ direction, we take $p' x_i > 0$ for all $i=1,\ldots,n$ as our assumption. We can then always pick a $p_0 < 0$ sufficiently close to zero that the system \eqref{eq:farkasLastSystem} holds.
\end{proof}

\begin{lemma}\label{lemma:substitutingWeightedAverage}
Suppose A\ref{assump:homog:relevance:gen}$^\prime$ holds. Suppose $L = K+1$. Suppose $\text{FAS}^*$ is not a singleton. Then $\psi_\ell - \pi_\ell' b = w_\ell(b) (\psi_\ell - \pi_\ell' \beta_{-\ell}^\textsc{2sls})$ for all $\ell =1,\ldots,L$.
\end{lemma}
{
\allowdisplaybreaks
\begin{proof}[Proof of lemma \ref{lemma:substitutingWeightedAverage}]
For $\ell' \in\{1,\ldots,L\}$ we have
\begin{align*}
	\psi_{\ell'} - \pi_{\ell'}' b
		&= \psi_{\ell'} - \pi_{\ell'}' \left( \sum_{\ell=1}^L w_\ell(b) \beta_{-\ell}^\textsc{2sls} \right) \\
		&= \psi_{\ell'} - \pi_{\ell'}' \left( \sum_{\ell \neq {\ell'}} w_\ell(b) \beta_{-\ell}^\textsc{2sls} + w_{\ell'}(b) \beta_{-{\ell'}}^\textsc{2sls} \right) \\
		&= \psi_{\ell'} - \pi_{\ell'}' \left( \sum_{\ell \neq {\ell'}} w_\ell(b) \beta_{-\ell}^\textsc{2sls} + \left( 1 - \sum_{\ell \neq {\ell'}} w_\ell(b) \right) \beta_{-{\ell'}}^\textsc{2sls} \right) \\
		&= \psi_{\ell'} - \sum_{\ell \neq {\ell'}} w_\ell(b) \pi_{\ell'}' \beta_{-\ell}^\textsc{2sls} - \left( 1 - \sum_{\ell \neq {\ell'}} w_\ell(b) \right) \pi_{\ell'}' \beta_{-{\ell'}}^\textsc{2sls} \\
		&= \psi_{\ell'} - \sum_{\ell \neq {\ell'}} w_\ell(b) \psi_{\ell'} - \left( 1 - \sum_{\ell \neq {\ell'}} w_\ell(b) \right) \pi_{\ell'}' \beta_{-{\ell'}}^\textsc{2sls} \\
		&= \left( 1 - \sum_{\ell \neq {\ell'}} w_\ell(b) \right) \psi_{\ell'} - \left( 1 - \sum_{\ell \neq {\ell'}} w_\ell(b) \right) \pi_{\ell'}' \beta_{-{\ell'}}^\textsc{2sls} \\
		&= \left( 1 - \sum_{\ell \neq {\ell'}} w_\ell(b) \right) ( \psi_{\ell'} - \pi_{\ell'}' \beta_{-{\ell'}}^\textsc{2sls} ) \\
		&= w_{\ell'}(b) ( \psi_{\ell'} - \pi_{\ell'}' \beta_{-{\ell'}}^\textsc{2sls} ).
\end{align*}
In the first line we used the representation $b = \sum_{\ell =1}^L w_\ell(b) \beta_{-\ell}^\textsc{2sls}$, by lemma \ref{lemma:weightedAverageOf2SLS}. The last equality follows since the weights sum to one.
\end{proof}
}

\begin{lemma}\label{lemma:EisEmpty}
Suppose A\ref{assump:homog:relevance:gen}$^\prime$ holds. Suppose $L = K+1$. Suppose $\text{FAS}^*$ is not a singleton. Without loss of generality, write
\[
	\text{FAS}^* = \bigcap_{\ell =1}^L \{ b \in\R^K: \psi_\ell - \pi_\ell' b \leq 0\}.
\]
(See equation \eqref{eq:halfSpaceFAS} in the proof of lemma \ref{lemma:generalLinearSingletonOnFF}, and the surrounding discussion.) Then there are no $b \in \R^K$ such that $\psi_\ell - \pi_\ell' b \geq 0$ for all $\ell =1,\ldots,L$.
\end{lemma}

\begin{proof}[Proof of lemma \ref{lemma:EisEmpty}]
Let
\[
	\mathcal{E} = \bigcap_{\ell =1}^L \{ \tilde{b} \in\R^K: \psi_\ell - \pi_\ell' \tilde{b} \geq 0\}.
\]
Suppose by way of contradiction that there is some element $b \in \mathcal{E}$. Write
\[
	b = \sum_{\ell =1}^L w_\ell(b) \beta_{-\ell}^\textsc{2sls}
\]
where $\sum_{\ell =1}^L w_\ell(b) = 1$. By lemma \ref{lemma:substitutingWeightedAverage} (which we can use since $L=K+1$ and $\text{FAS}^*$ is not a singleton),
\[
	\psi_\ell - \pi_\ell' b = w_\ell(b) (\psi_\ell - \pi_\ell' \beta_{-\ell}^\textsc{2sls}).
\]
This term is $\geq 0$ since $b \in \mathcal{E}$. Since $\beta_{-\ell}^\textsc{2sls} \in \text{FAS}^*$ and by lemma \ref{lemma:singletonFAS} we have $\psi_\ell - \pi_\ell' \beta_{-\ell}^\textsc{2sls} < 0$. This implies $w_\ell(b) \leq 0$. This must hold for all $\ell$. That is a contradiction to $\sum_{\ell =1}^L w_\ell(b) = 1$. Hence $\mathcal{E} = \emptyset$.
\end{proof}

\begin{lemma}\label{lemma:convexHullsOfBandPis}
Suppose A\ref{assump:homog:relevance:gen}$^\prime$, A\ref{assump:homog:nonsing:gen}, and A\ref{assump:exogeneity:gen} hold. Suppose $L = K+1$. Then there exists a normalization of the hyperplanes $\{b \in \R^K : \pi_\ell'b = \psi_\ell\}$ with $\psi_\ell \geq 0$ and such that
\[
	\0 \in \text{conv}(\{\pi_\ell: \ell = 1,\ldots,L\})
	\qquad \Rightarrow \qquad
	 \0 \in \text{conv}(\{\beta^\textsc{2sls}_{-\ell}: \ell=1,\ldots,L\}).
\]
\end{lemma}

\begin{proof}[Proof of lemma \ref{lemma:convexHullsOfBandPis}]
We prove this by contrapositive. So suppose $\0 \notin \text{conv}(\{\beta^\textsc{2sls}_{-\ell}: \ell =1,\ldots,L\})$. By $L = K+1$ and lemma \ref{lemma:singletonFAS}, there are two cases to consider.
\begin{enumerate}
\item Suppose $\text{FAS}^* = \text{conv}(\{\beta^\textsc{2sls}_{-\ell}: \ell =1,\ldots,L\})$ is a singleton. Let $b^*$ denote this common value. By definition of $\beta_{-\ell}^\textsc{2sls}$, $\pi_\ell' \beta_{-\ell}^\textsc{2sls} = \psi_\ell$ for all $\ell$. Thus $\pi_\ell' b^* = \psi_\ell$ for all $\ell$. Since $\0$ is not in the convex hull of the $\beta_{-\ell}^\textsc{2sls}$ points by assumption, we have $b^* \neq \0$. Recall also that we normalized $\psi_\ell \geq 0$ for all $\ell$. So we have $\pi_\ell' b^* \geq 0$ for all $\ell$.

\medskip

Let $\mathcal{L}_0$ denote the indices $\ell$ such that $\psi_\ell = 0$. Note that $\pi_\ell' b^* = 0$ for $\ell \in \mathcal{L}_0$. The set
\[
	\bigcup_{\ell \in \mathcal{L}_0} \{b\in\R^K:\pi_\ell'b = 0\}
\]
does not equal $\R^K$. This follows since it is a finite union of $K-1$ dimensional hyperplanes. So there is a point $v \in \R^K$ which is not in this set. That is, there exists a vector $v \in \R^K$ such that $\pi_\ell'v \neq 0$ for all $\ell \in \mathcal{L}_0$.

\medskip

Note that hyperplanes corresponding to $\ell \in \mathcal{L}_0$ are invariant to sign normalizations of $\pi_\ell$ since
\[
	\{b\in\R^K:\pi_\ell'b = 0\} = \{b\in\R^K:(-\pi_\ell)'b = 0\}.
\]
The element $v$ is independent of the normalization since if $\pi_\ell'v \neq 0$, then $(-\pi_\ell)'v \neq 0$ as well. Let
\[
	b^*(t) = b^* + tv
\]
for $t \geq 0$. At $t = 0$ we have
\[
	\pi_\ell'b^*(0) = \psi_\ell > 0
\]
for all $\ell \in \mathcal{L}\setminus\mathcal{L}_0$. By continuity, for small enough $t > 0$ we also have
\[
	\pi_\ell'b^*(t) = \pi_\ell'b^* + t\pi_{\ell}'v > 0
\]
for all $\ell \in \mathcal{L} \setminus \mathcal{L}_0$. Denote this point by $\bar{t} > 0$. Next, for all $\ell \in \mathcal{L}_0$ we have
\begin{align*}
	\pi_\ell'b^*(\bar{t})
	&= \pi_\ell'b^* + \bar{t} \pi_\ell'v \\
	&= \bar{t} \pi_\ell'v \\
	&\neq 0.
\end{align*}
Therefore, given that the sign of $\pi_{\ell}$ for $\ell \in \mathcal{L}_0$ can be renormalized, we can normalize $\pi_\ell$ for $\ell \in \mathcal{L}_0$ such that $\pi_\ell'b^*(\bar{t}) >0$ for all $\ell \in \mathcal{L}_0$. So now we have found a normalization and a point $b^*(\bar{t})$ such that $\pi_\ell' b^*(\bar{t}) > 0$ for all $\ell$. Lemma \ref{lemma:Farkas} thus implies that $\0 \notin \text{conv}(\{\pi_\ell:\ell =1,\ldots,L\})$.

\item Suppose $\text{conv}(\{\beta^\textsc{2sls}_{-\ell}: \ell =1,\ldots,L\})$ is not a singleton. By A\ref{assump:homog:relevance:gen}.2$^\prime$, the equation
\begin{equation}\label{eq:mat_tilde}
	\begin{pmatrix}
		1 & \ldots &  1\\
		\pi_1 & \ldots & \pi_L
	\end{pmatrix}\begin{pmatrix}
		\lambda_1\\
		\vdots\\
		\lambda_L
	\end{pmatrix}
	=
	\begin{pmatrix}
		1\\
		\0_K
	\end{pmatrix}
\end{equation}
has a unique solution $\lambda \in \R^L$. We want to show that this solution does \emph{not} satisfy $\lambda \geq 0$. 

Recall that we can write $\text{FAS}^* = \text{conv}(\{\beta^\textsc{2sls}_{-\ell}: \ell =1,\ldots,L\})$ as an intersection of halfspaces $\{ \tilde{b} \in \R^K : \pi_\ell' \tilde{b} \leq \psi_\ell \}$. Without loss of generality, write
\[
	\text{FAS}^* = \left( \bigcap_{\ell \in \mathcal{L}_1} \{ \tilde{b} \in \R^K : \psi_\ell - \pi_\ell' \tilde{b} \leq 0 \} \right) \bigcap \left( \bigcap_{\ell \in \mathcal{L}_2} \{ \tilde{b} \in \R^K : \psi_\ell - \pi_\ell' \tilde{b} \geq 0 \} \right)
\]
for some partition $\mathcal{L}_1 \cup \mathcal{L}_2 = \mathcal{L}$. Since we have normalized $\psi_\ell \geq 0$, we \emph{cannot} have $\mathcal{L}_1 = \emptyset$ because this would imply that $\0 \in \text{FAS}^*$. So this implies that there is some index $k_1$ such that
\[
	\text{conv}(\{\beta^\textsc{2sls}_{-\ell}: \ell =1,\ldots,L\})
	\subseteq
	\{b \in \R^K : \pi_{k_1}'b \geq \psi_{k_1}\}.
\]
However, we also cannot have $\mathcal{L}_1 = \{1,\ldots,L\}$.
\begin{itemize}
\item If $\mathcal{L}_1 = \{1,\ldots,L \}$ then
\[
	\text{FAS}^* =  \bigcap_{\ell \in \{1,\ldots,L\}} \{ \tilde{b} \in \R^K : \psi_\ell - \pi_\ell' \tilde{b} \leq 0 \}.
\]
By lemma \ref{lemma:EisEmpty}, this implies that
\[
	\bigcap_{\ell \in \{1,\ldots,L\}} \{ \tilde{b} \in \R^K : \psi_\ell - \pi_\ell' \tilde{b} \geq 0 \}
\]
is empty. This is a contradiction since
\[
	\0 \in \bigcap_{\ell \in \{1,\ldots,L\}} \{ \tilde{b} \in \R^K : \psi_\ell - \pi_\ell' \tilde{b} \geq 0 \}
\]
by the normalization $\psi_\ell \geq 0$.
\end{itemize}
Hence there is some index $k_2 \in \mathcal{L}_2$. That is, there is a $k_2$ such that
\[
	\{b \in \R^K : \pi_{k_2}'b \leq \psi_{k_2}\} \supseteq \text{conv}(\{\beta^\textsc{2sls}_{-\ell}: \ell =1,\ldots,L\}).
\]
Now, from the last row of equation \eqref{eq:mat_tilde} we have
\[
	\sum_{\ell=1}^L \lambda_\ell \pi_\ell = \0_K.
\]
Multiply this equation by $(\beta^\textsc{2sls}_{-k})'$ to get
\[
	(\beta^\textsc{2sls}_{-k})' \sum_{\ell =1}^L \pi_\ell \lambda_\ell = 0.
\]
Separate the $\ell = k$ from the sum to get
\[
	\sum_{\ell \neq k} \lambda_\ell \pi_\ell' \beta_{-k}^\textsc{2sls} + \lambda_k \pi_k' \beta_{-k}^\textsc{2sls} = 0.
\]
Recall that, by the definition of these 2SLS estimands, $\pi_k' \beta_{-\ell}^\textsc{2sls} = \psi_k$ for $k \neq \ell$. Hence
\begin{equation}\label{eq:lin_tilde}
	\sum_{\ell \neq k} \lambda_\ell \psi_k + \lambda_k \pi_k' \beta_{-k}^\textsc{2sls} = 0.
\end{equation}
Consider this equation for our two different indices $k_1, k_2 \in \{1,\ldots,L\}$,
\[
	\sum_{\ell \neq k_1} \lambda_\ell \psi_{k_1} + \lambda_{k_1} \pi_{k_1}' \beta_{-k_1}^\textsc{2sls} = 0
	\qquad \text{and} \qquad
	\sum_{\ell \neq k_2} \lambda_\ell \psi_{k_2} + \lambda_{k_2} \pi_{k_2}' \beta_{-k_2}^\textsc{2sls} = 0.
\]
Subtract the left equation from the right to get
\[
	0 = \lambda_{k_2} (\pi_{k_2}'\beta^\textsc{2sls}_{-k_2} - \psi_{k_2}) + \lambda_{k_1} (\psi_{k_1} - \pi_{k_1}'\beta^\textsc{2sls}_{-k_1}).
\]
As mentioned earlier, lemma \ref{lemma:singletonFAS} implies that $\pi_{k_2}'\beta^\textsc{2sls}_{-k_2} - \psi_{k_2} < 0$ and $\psi_{k_1} - \pi_{k_1}'\beta^\textsc{2sls}_{-k_1} < 0$. This equation implies that $\lambda_{k_2}$ and $\lambda_{k_1}$ cannot both have the same signs. If one of them is $> 0$ then the other must be $< 0$. In this case, we're done. Here we use the fact that $\lambda$ is the \emph{unique} solution to equation \eqref{eq:mat_tilde}.

So suppose both of them are zero. Let $k_3 \in \{1,\ldots,L\}\setminus\{k_1,k_2\}$. Then by similar derivations we have
\[
	0 = \lambda_{k_2} (\pi_{k_2}'\beta^\textsc{2sls}_{-k_2} - \psi_{k_2}) + \lambda_{k_3} (\psi_{k_1} - \pi_{k_3}'\beta^\textsc{2sls}_{-k_3}).
\]
We know that $\lambda_{k_2} = 0$ already and thus
\[
	0 = \lambda_{k_3} (\psi_{k_1} - \pi_{k_3}'\beta^\textsc{2sls}_{-k_3}).
\]
We know that $\psi_{k_1} - \pi_{k_3}'\beta^\textsc{2sls}_{-k_3} \neq 0$ and thus we must have $\lambda_{k_3} = 0$. Continuing in this way we obtain $\lambda = \0$. This contradicts $\sum_{\ell = 1}^L \lambda_\ell = 1$. Thus $\lambda$ cannot be nonnegative. So equation \eqref{eq:mat_tilde} then implies that $\0 \notin \text{conv}(\{\pi_\ell:\ell =1,\ldots,L\})$.
\end{enumerate}
\end{proof}

\begin{lemma}\label{lemma:case2:FFguessGeneralSupersetFF}
Suppose A\ref{assump:homog:relevance:gen}$^\prime$, A\ref{assump:homog:nonsing:gen}, and A\ref{assump:exogeneity:gen} hold. Suppose $L \geq K + 1$. Let $b \notin \mathcal{P}$. Then there exists a $\delta' < \delta(b)$ such that $\mathcal{B}(\delta') \neq \emptyset$.
\end{lemma}

Figures \ref{fig:LisKplus1FAS_outside} and \ref{fig:generalLinearFAS} illustrate the geometric intuition for this lemma.

\begin{proof}[Proof of lemma \ref{lemma:case2:FFguessGeneralSupersetFF}]
Without loss of generality, suppose $b = \0$. This follows since we can simply translate our coordinate system so that the origin is at $b$. Put differently, we map all $x \in \R^K$ to $x - b$. In particular, the hyperplanes
\[
	B_\ell(0) = \{ \tilde{b} \in \R^K: \psi_\ell = \pi_\ell' \tilde{b} \}
\]
get mapped to
\[
	B_\ell(0) - b
	= \{ \tilde{b} \in \R^K : \psi_\ell - \pi_\ell' b = \pi_\ell' \tilde{b} \}.
\]
Thus we let $b = \0$. Throughout this proof, pick one of the normalizations given by lemma \ref{lemma:convexHullsOfBandPis}. In these normalizations we have $\psi_\ell \geq 0$ for all $\ell$. Next, there are two cases to consider.
\begin{enumerate}
\item Suppose $\delta_\ell(b) >0$ for all $\ell$. Since $b \notin \mathcal{P}$, $b \notin \mathcal{P}_{\mathcal{L}} = \text{conv}(\{\beta^\textsc{2sls}_{\mathcal{L}\setminus\{l\}}:\ell \in\mathcal{L}\})$ for any $\mathcal{L}$ with $|\mathcal{L}| = K+1$. Since $b = \0$, lemma \ref{lemma:convexHullsOfBandPis} implies that $\0 \notin \text{conv}(\{\pi_\ell: \ell \in \mathcal{L}\})$. This holds for any set $\mathcal{L}$ such that $|\mathcal{L}| = K+1$. This implies that $\0 \notin \text{conv}(\{\pi_\ell:\ell =1,\ldots,L\})$.
\begin{itemize}
\item If $\0 \in \text{conv}(\{\pi_\ell:\ell =1,\ldots,L\})$, then Caratheodory's theorem (e.g., chapter 17 of \citealt{Rockafellar1970}) implies that $\0$ is in the convex hull of a $(K+1)$-element subset of $\{ \pi_\ell : \ell =1,\ldots,L \}$, That is, $\0 \in \text{conv}(\{\pi_\ell: \ell \in \mathcal{L}\})$ for some $\mathcal{L}$ with $| \mathcal{L} | = K+1$. This is a contradiction.
\end{itemize}
Since $\0 \notin \text{conv}(\{\pi_\ell:\ell =1,\ldots,L\})$, lemma \ref{lemma:Farkas} implies that there exists a vector $\bar{b}$ such that $\bar{b}' \pi_\ell > 0$ for all $\ell = 1,\ldots,L$. Define
\begin{align*}
	b(\varepsilon)
		&= b + \varepsilon \bar{b} \\
		&= \0 + \varepsilon \bar{b}.
\end{align*}
This is a point that starts at $b$ and moves in the direction of $\overline{b}$. Recall that we normalized $\psi_\ell \geq 0$ for all $\ell$. Moreover, note that we must actually have $\psi_\ell > 0$ since $0<\delta_\ell(\0) = | \psi_\ell |$. Thus, since $\psi_\ell > 0$ and $\pi_\ell' \bar{b} > 0$ for all $\ell$, there exists an $\bar{\varepsilon} > 0$ such that
\[
	\psi_\ell - \bar{\varepsilon} \pi_\ell' \bar{b} > 0
\]
for all $\ell$. This implies that
\begin{align*}
	| \psi_\ell - \pi_\ell' b(\bar{\varepsilon}) |
		&= | \psi_\ell - \pi_\ell' (b + \bar{\varepsilon} \bar{b}) | \\
		&= | (\psi_\ell - \pi_\ell' b) - \bar{\varepsilon} \pi_\ell' \bar{b} | \\
		&= (\psi_\ell - \pi_\ell'b) - \bar{\varepsilon} \pi_\ell' \bar{b} \\
		&= \delta_\ell(b) - \bar{\varepsilon} \pi_\ell' \bar{b} \\
		&< \delta_\ell(b).
\end{align*}
The third line uses $b = \0$ and our specific choice of $\bar{\varepsilon}$ to ensure that this term is positive, and hence we can drop the absolute value. The fourth line uses $\psi_\ell \geq 0$ and $b = \0$. The fifth line uses $\pi_\ell'\bar{b} > 0$ and $\bar{\varepsilon} > 0$.

Let
\[
	\delta_\ell' = | \psi_\ell - \pi_\ell' b(\bar{\varepsilon}) |.
\]
We have shown that $\delta' < \delta(b)$. Finally, by our characterization of $\mathcal{B}(\cdot)$ and the definition of $\delta'$, we have $b(\bar{\varepsilon}) \in \mathcal{B}(\delta')$. Hence $\mathcal{B}(\delta') \neq \emptyset$.

\item Suppose $\delta_\ell(b) = 0$ for some $\ell$'s. There can be at most $K-1$ such indices, since otherwise we would have $b \in \mathcal{P}$. Let $\mathcal{L}_0$ denote the set of these indices. As in the proof of lemma \ref{lemma:convexHullsOfBandPis}, we must have $\psi_\ell = 0$ for $\ell \in \mathcal{L}_0$. In this case, consider the subspace
\begin{align*}
	\{ \tilde{b} \in \R^K : 0 = \pi_\ell' \tilde{b}, \ \ell \in \mathcal{L}_0\}.
\end{align*}
This is a linear subspace of dimension at least $1$ (by assumption that $\delta_\ell(b) = 0$ for some $\ell$'s) and at most $K-1$ (as noted earlier). Within this subspace, we can look at the remaining indices $\{ 1,\ldots, L \} \setminus \mathcal{L}_0$. We have $\delta_\ell(b) > 0$ for all of these indices. By restricting attention to this subspace we can thus immediately apply the analysis of case 1.
\end{enumerate}
\end{proof}

For the next two lemmas, let
\[
	\text{FF}^\text{guess} = \{\delta \in \R^L_{\geq 0}: \delta_\ell = | \psi_\ell - \pi_\ell' b |, \ell=1,\ldots,L, b \in \mathcal{P} \}
\]
and let $\text{FF}$ denote the true falsification frontier.

\begin{lemma}\label{lemma:FFguessGeneralSubsetFF}
Suppose A\ref{assump:homog:relevance:gen}$^\prime$, A\ref{assump:homog:nonsing:gen}, and A\ref{assump:exogeneity:gen} hold. Suppose $L \geq K+1$. Then $\text{FF}^\text{guess} \subseteq \text{FF}$.
\end{lemma}

\begin{proof}[Proof of lemma \ref{lemma:FFguessGeneralSubsetFF}]
Recall the definition $\delta_\ell(b) = | \psi_\ell - \pi_\ell' b |$. Let $\delta \in \text{FF}^\text{guess}$. Then, by definition, there is a $b \in \mathcal{P}$ such that $\delta_\ell(b) = \delta_\ell$ for all $\ell$. Thus $\mathcal{B}(\delta) = \{b\}$ by lemma \ref{lemma:singletonOnFASgeneralCase}.

Let $\delta' < \delta(b)$. Then there is some index $\ell$ such that $0 < \delta_\ell' < \delta_\ell(b)$. So $\delta_\ell(b) \notin [-\delta_\ell', \delta_\ell' ]$. That is,
\[
	\psi_\ell - \pi_\ell' b \notin [-\delta'_\ell, \delta'_\ell].
\]
So $b \notin B_\ell(\delta')$. This implies that $b \notin \mathcal{B}(\delta')$. But since $\mathcal{B}(\delta') \subseteq \mathcal{B}(\delta) = \{ b \}$, we must have $\mathcal{B}(\delta') =\emptyset$. Hence, by the definition of the falsification frontier, $\text{FF}^\text{guess} \subseteq \text{FF}$.
\end{proof}

\begin{lemma}\label{lemma:FFguessGeneralSupersetFF}
Suppose A\ref{assump:homog:relevance:gen}$^\prime$, A\ref{assump:homog:nonsing:gen}, and A\ref{assump:exogeneity:gen} hold. Suppose $L \geq K+1$. Then $\text{FF}^\text{guess} \supseteq \text{FF}$.
\end{lemma}

\begin{proof}[Proof of lemma \ref{lemma:FFguessGeneralSupersetFF}]
We will show the contrapositive: $\delta \notin \text{FF}^\text{guess}$ implies $\delta \notin \text{FF}$. So let $\delta \notin \text{FF}^\text{guess}$. There are three cases to consider.
\begin{enumerate}
\item Suppose $\delta$ is such that $\mathcal{B}(\delta) = \emptyset$.
\begin{itemize}
\item In this case the result follows immediately since $\mathcal{B}(\delta) = \emptyset$ implies $\delta\notin \text{FF}$ by definition.
\end{itemize}

\item Suppose $\delta$ is such that $\mathcal{B}(\delta)$ contains an element $b \notin \mathcal{P}$.
\begin{itemize}
\item By lemma \ref{lemma:case2:FFguessGeneralSupersetFF} there exists a $\delta' < \delta(b)$ with $\mathcal{B}(\delta') \neq \emptyset$. Moreover, $\delta(b) \leq \delta$ by the characterization of $\mathcal{B}(\delta)$ in theorem \ref{thm:idset:homog:gen}. Thus $\delta \notin \text{FF}$ by the definition of the falsification frontier.
\end{itemize}

\item Suppose $\delta$ is such that $\mathcal{B}(\delta) \neq \emptyset$ and $\mathcal{B}(\delta) \subseteq \mathcal{P}$. 
\begin{itemize}
\item We use two observations:

\medskip

\begin{enumerate}
\item Since $\delta \notin \text{FF}^\text{guess}$ we must have $\delta \neq \delta(b)$ for all $b \in \mathcal{P}$. In particular, $\delta \neq \delta(b)$ for all $b \in \mathcal{B}(\delta)$, since we are considering the case where $\mathcal{B}(\delta) \subseteq \mathcal{P}$. 

\medskip

\item Let $b \in \mathcal{B}(\delta)$. By the characterization of $\mathcal{B}(\delta)$ in theorem \ref{thm:idset:homog:gen} (see equation \eqref{eq:CalBisIntersectionOfHalfspaces}), this implies that $| \psi_\ell - \pi_\ell' b | \leq \delta_\ell$. That is, $\delta(b) \leq \delta$.
\end{enumerate}

\medskip

Let $b'$ be any element of $\mathcal{B}(\delta)$. This exists by assumption. Let $\delta' = \delta(b')$. The previous two observations imply that $\delta' \neq \delta$ and $\delta' \leq \delta$. This implies $\delta' < \delta$. Moreover, we have $\mathcal{B}(\delta') = \mathcal{B}(\delta(b')) = \{ b' \} \neq \emptyset$ by lemma \ref{lemma:singletonOnFASgeneralCase}. Thus $\delta \notin \text{FF}$, by definition of the falsification frontier.

\end{itemize}
\end{enumerate}
All values of $\delta$ must fall in one of these three cases. Therefore $\text{FF}^\text{guess} \supseteq \text{FF}$.
\end{proof}

\begin{proof}[Proof of proposition \ref{prop:generalLinearFF}]
This follows directly from lemmas \ref{lemma:FFguessGeneralSubsetFF} and \ref{lemma:FFguessGeneralSupersetFF}.
\end{proof}

\begin{proof}[Proof of proposition \ref{prop:KisLplus1FF}]
This is a special case of proposition \ref{prop:generalLinearFF}.
\end{proof}

\begin{proof}[Proof of theorem \ref{thm:generalLinearFAS}]
We have
\begin{align*}
	\bigcup_{\delta \in \text{FF}} \mathcal{B}(\delta)
		&= \bigcup_{b \in \mathcal{P}} \mathcal{B}(\delta(b)) \\
		&= \bigcup_{b \in \mathcal{P}} \{ b \} \\
		&= \mathcal{P}.
\end{align*}
The first line follows by proposition \ref{prop:generalLinearFF}. The second line follows by lemma \ref{lemma:singletonOnFASgeneralCase}.
\end{proof}

\begin{proof}[Proof of theorem \ref{thm:LisKplus1FAS}]
This is a special case of theorem \ref{thm:generalLinearFAS}. In particular, note that $\mathcal{P} = \text{FAS}^*$ in this case.
\end{proof}

To prove corollary \ref{corr:FASprojection}, we use the following definition: Let $P$ be a $K \times K$ idempotent matrix. Define the linear operator $p : \R^K \rightarrow \R^K$ by $p(a) = Pa$. $p$ is called a \emph{projection}. For $A \subseteq \R^K$, define the projection of the set $A$ as
\[
	\text{proj}(A) = \{ p(a) \in \R^K : a \in A \}.
\]
We use the following lemma.

\begin{lemma}\label{lemma:projAndconvCommute}
$\text{proj}(\text{conv}(A)) = \text{conv}(\text{proj}(A))$.
\end{lemma}

\begin{proof}[Proof of lemma \ref{lemma:projAndconvCommute}]
\textbf{Step 1 ($\supseteq$).} Since $A \subseteq \text{conv}(A)$, $\text{proj}(A) \subseteq \text{proj}(\text{conv}(A))$. Thus
\[
	\text{conv}(\text{proj}(A))
	\subseteq \text{conv}(\text{proj}(\text{conv}(A))).
\]
Since $p(\cdot)$ is a linear operator, $\text{proj}(\text{conv}(A))$ is convex. Hence
\[
	\text{conv}(\text{proj}(\text{conv}(A)))
	=
	\text{proj}(\text{conv}(A)).
\]
This with our previous result shows that $\text{conv}(\text{proj}(A)) \subseteq \text{proj}(\text{conv}(A))$.

\bigskip

\textbf{Step 2 ($\subseteq$).} Let $x \in \text{proj}(\text{conv}(A))$. Then $x = p(y)$ for some $y \in \text{conv}(A)$. Since $y \in \text{conv}(A)$, we can write
\[
	y = \sum_{i=1}^n \theta_i y_i
\]
where $ \sum_{i=1}^n \theta_i = 1$,$ \theta_i \geq 0$, and $y_i \in A$ for $i=1,\ldots,n$, where $n$ is a finite integer. Since $p(\cdot)$ is a linear operator,
\[
	p(y) = \sum_{i=1}^n \theta_i p(y_i).
\]
Thus we've written $x$ as a convex combination of elements $p(y_i) \in \text{proj}(A)$. Hence $x \in \text{conv}(\text{proj}(A))$. Thus we've shown that $\text{proj}(\text{conv}(A)) \subseteq \text{conv}(\text{proj}(A))$.
\end{proof}

\begin{proof}[Proof of corollary \ref{corr:FASprojection}]
The first part of this result states that $\mathcal{P} \subseteq \text{FAS}^*$. To see this, let $b \in \mathcal{P}$. Then $b \in \mathcal{P}_\mathcal{L}$ for some set of indices $\mathcal{L}$ with $| \mathcal{L} | = K+1$. But if $b$ is a convex combination of $\{ \beta^\textsc{2sls}_{\mathcal{L} \setminus \{ \ell \} } : \ell \in \mathcal{L} \}$, then it is also a convex combination of the larger set of elements $\{ \beta_{\mathcal{S}}^\textsc{2sls} : \mathcal{S} \subseteq \{ 1,\ldots, L \}, | \mathcal{S} | = K \}$. Hence $b \in \text{FAS}^*$.

\medskip

To prove the second part, consider projections defined by
\[
	P = 
	\begin{pmatrix}
		\alpha' \\
		0 \\
		\vdots \\
		0
	\end{pmatrix}
\]
where $\alpha \in \R^{K}$. For any set $A \subseteq \R^K$, let $[A]_1 = \{ a_1 \in \R : a = (a_1,\ldots,a_K) \in \R^K \}$. Since $P$ maps the $2,\ldots,K$ components of any vector $a \in \R^K$ to zero, it suffices to show that
\[
	[\text{proj}(\mathcal{P})]_1 = [\text{proj}(\text{FAS}^*) ]_1
\]
and
\begin{equation}\label{eq:projectionOfFASStar}
	[\text{proj}(\text{FAS}^*)]_1
	= \left[\min_{\mathcal{L}\subseteq \{1,\ldots,L\}, |\mathcal{L}| = K} \alpha' \beta_{\mathcal{L}}^\textsc{2sls}, \
	\max_{\mathcal{L}\subseteq \{1,\ldots,L\}, |\mathcal{L}| = K} \alpha' \beta_{\mathcal{L}}^\textsc{2sls}\right].
\end{equation}
We have
\begin{align*}
	\text{proj}(\text{FAS}^*)
		&= \text{proj}(\text{conv}(\mathcal{P})) \\
		&= \text{proj}(\text{conv}(\{\beta_{\mathcal{L}}^\textsc{2sls}:\mathcal{L}\subseteq \{1,\ldots,L\}, |\mathcal{L}| = K\}))\\
	&= \text{conv}(\text{proj}(\{\beta_{\mathcal{L}}^\textsc{2sls}:\mathcal{L}\subseteq \{1,\ldots,L\}, |\mathcal{L}| = K\}))\\
	&= \text{conv}(\{P\beta_{\mathcal{L}}^\textsc{2sls}:\mathcal{L}\subseteq \{1,\ldots,L\}, |\mathcal{L}| = K\}).
\end{align*}
The third line follows from lemma \ref{lemma:projAndconvCommute}. The fourth by the definition of the projection of a set. Using the specific form of the projection $P$ now gives equation \eqref{eq:projectionOfFASStar}.

Similarly,
\begin{align*}
	\text{proj}(\mathcal{P}) 
	&= \text{proj} \left(\bigcup_{\mathcal{L} \subseteq \{ 1,\ldots, L \} : | \mathcal{L} | = K+1}  \text{conv} \big( \big\{ \beta_{\mathcal{L} \setminus \{ \ell \}}^\textsc{2sls} : \ell \in \mathcal{L} \big\} \big)\right)\\
	&= \bigcup_{\mathcal{L} \subseteq \{ 1,\ldots, L \} : | \mathcal{L} | = K+1} \text{proj} \left( \text{conv} \big( \big\{ \beta_{\mathcal{L} \setminus \{ \ell \}}^\textsc{2sls} : \ell \in \mathcal{L} \big\} \big)\right)\\
	&= \bigcup_{\mathcal{L} \subseteq \{ 1,\ldots, L \} : | \mathcal{L} | = K+1}  \text{conv}\left( \text{proj} \big( \big\{ \beta_{\mathcal{L} \setminus \{ \ell \}}^\textsc{2sls} : \ell \in \mathcal{L} \big\} \big)\right).
\end{align*}
The first line follows by definition of $\mathcal{P}$. The second since projections and unions commute. The third from lemma \ref{lemma:projAndconvCommute}. Hence
\[
	[\text{proj}(\mathcal{P})]_1
	= \bigcup_{\mathcal{L} \subseteq \{ 1,\ldots, L \} : | \mathcal{L} | = K+1}
	\left[\min_{\ell\in\mathcal{L}} \; \alpha' \beta_{\mathcal{L} \setminus \{ \ell \}}^\textsc{2sls}, \
	\max_{\ell\in\mathcal{L}} \; \alpha' \beta_{\mathcal{L} \setminus \{ \ell \}}^\textsc{2sls}\right]
\]
Thus we see that the first component of $\text{proj}(\mathcal{P})$ is a union of closed intervals.

If $\mathcal{L}$ and $\mathcal{L}'$ differ by at most one element, then the intervals
\[
	\left[\min_{\ell\in\mathcal{L}} \; \alpha' \beta_{\mathcal{L} \setminus \{ \ell \}}^\textsc{2sls}, \
	\max_{\ell\in\mathcal{L}} \; \alpha' \beta_{\mathcal{L} \setminus \{ \ell \}}^\textsc{2sls}\right]
\]
and
\[
	\left[\min_{\ell\in\mathcal{L}'} \; \alpha' \beta_{\mathcal{L}' \setminus \{ \ell \}}^\textsc{2sls}, \
	\max_{\ell\in\mathcal{L}'} \; \alpha' \beta_{\mathcal{L}' \setminus \{ \ell \}}^\textsc{2sls}\right]
\]
will overlap (have non-empty intersection). Their union is therefore also a closed interval. Because we take the union over all $\mathcal{L} \subseteq \{ 1,\ldots, L \} : | \mathcal{L} | = K+1$, we can find a sequence $(\mathcal{L}_1,\ldots,\mathcal{L}_N)$ such that
\[
	\left[\min_{\ell\in\mathcal{L}_n} \; \alpha' \beta_{\mathcal{L}_n \setminus \{ \ell \}}^\textsc{2sls}, \
	\max_{\ell\in\mathcal{L}_n} \; \alpha' \beta_{\mathcal{L}_n \setminus \{ \ell \}}^\textsc{2sls}\right]
\]
and
\[
	\left[\min_{\ell\in\mathcal{L}_{n+1}} \; \alpha' \beta_{\mathcal{L}_{n+1} \setminus \{ \ell \}}^\textsc{2sls}, \
	\max_{\ell\in\mathcal{L}_{n+1}} \; \alpha' \beta_{\mathcal{L}_{n+1} \setminus \{ \ell \}}^\textsc{2sls}\right]
\]
overlap and such that
\[
	\bigcup_{n=1}^N \mathcal{L}_n = \{1,\ldots,L\}
\]
Thus
\begin{align*}
	\bigcup_{\mathcal{L} \subseteq \{ 1,\ldots, L \} : | \mathcal{L} | = K+1}  \left[\min_{\ell\in\mathcal{L}} \; \alpha' \beta_{\mathcal{L} \setminus \{ \ell \}}^\textsc{2sls}, \max_{\ell\in\mathcal{L}} \; \alpha' \beta_{\mathcal{L} \setminus \{ \ell \}}^\textsc{2sls}\right] 
	&= \left[\min_{\mathcal{L}\subseteq \{1,\ldots,L\}, |\mathcal{L}| = K} \alpha' \beta_{\mathcal{L}}^\textsc{2sls}, \max_{\mathcal{L}\subseteq \{1,\ldots,L\}, |\mathcal{L}| = K} \alpha' \beta_{\mathcal{L}}^\textsc{2sls}\right].
\end{align*}
Putting everything together yields $[\text{proj}(\mathcal{P})]_1 = [\text{proj}(\text{FAS}^*)]_1$ as desired.
\end{proof}

\begin{figure}[!t]
\centering
\includegraphics[width=0.31\linewidth]{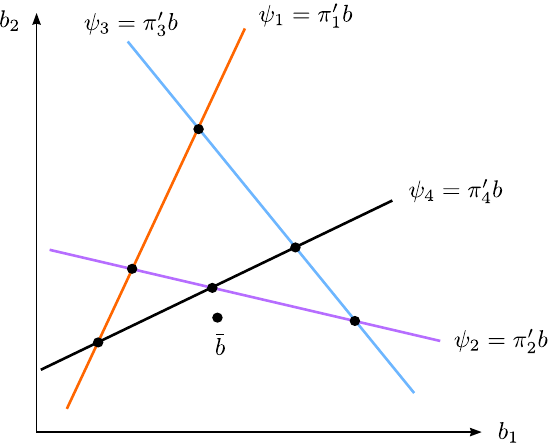}
\includegraphics[width=0.31\linewidth]{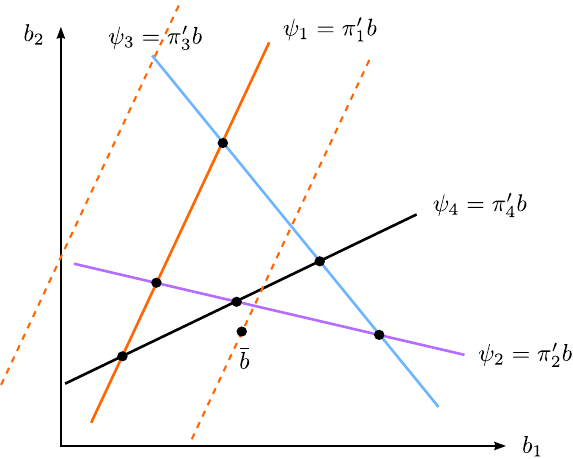}
\includegraphics[width=0.31\linewidth]{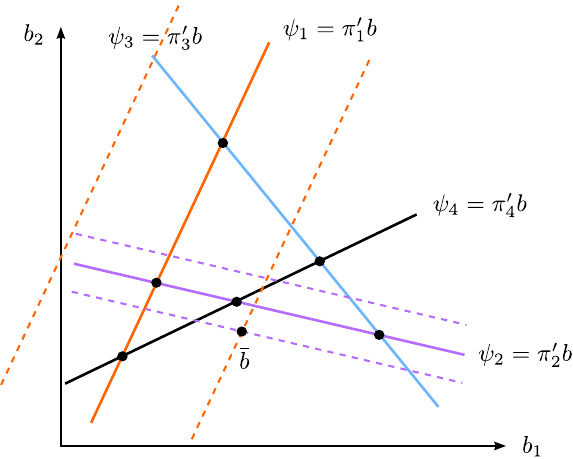} \\
\includegraphics[width=0.31\linewidth]{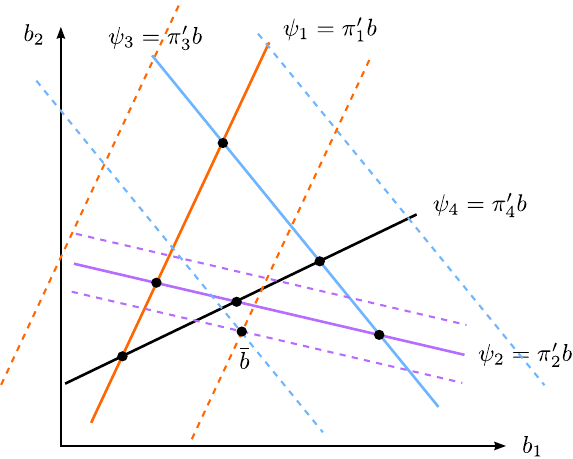}
\includegraphics[width=0.31\linewidth]{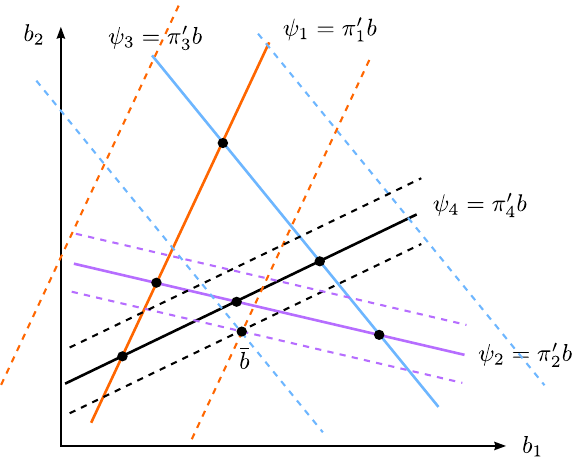}
\includegraphics[width=0.31\linewidth]{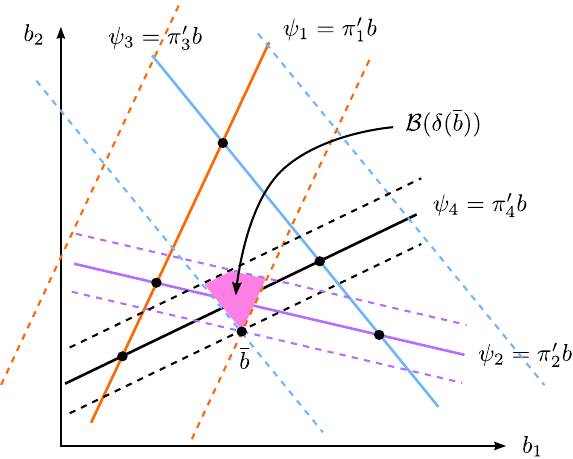}
\caption{Geometric intuition for a key step in the proof of theorem \ref{thm:generalLinearFAS}. This figure is analogous to figure \ref{fig:LisKplus1FAS_outside}, except now we have $K=2$ and $L=4$. Moreover, here we pick $\overline{b}$ to be outside of $\mathcal{P}$ but inside $\text{FAS}^*$. Thus here we illustrate why $\text{FAS}^*$ is not the falsification adaptive set when $L > K+1$. The argument is the same as before: We relax each exclusion restriction just enough to include $\overline{b}$ in the tubes. But when we do this their intersection is non-singleton. Thus we've relaxed exclusion too much. The point $\overline{b}$ is not in the falsification adaptive set.}
\label{fig:generalLinearFAS}
\end{figure}

\section{Proofs for section \ref{sec:hetModel}}

\subsubsection*{Proofs for section \ref{sec:hetTrtBinary}: Single binary instrument}

\begin{proof}[Proof of theorem \ref{thm:BinaryYoneInstrumentCdep}]
This result follows by applying theorem \ref{thm:BinaryYManyInstrumentCdep} to the case with $J = 2$ and $L =1$. Here we will show that the general expressions for $\mathcal{D}(c)$ and $\mathcal{H}_x$ used in theorem \ref{thm:BinaryYManyInstrumentCdep} simplify to the ones used in theorem \ref{thm:BinaryYoneInstrumentCdep}.

We start with $\mathcal{D}(c)$. From theorem \ref{thm:BinaryYManyInstrumentCdep}, it is the set of $(a_0,a_1)\in[0,1]^2$ such that the following eight inequalities hold
\begin{align*}
	& (1-p_Z) a_0 - \min\{1-p_Z + c,1\} (p_Z a_1 + (1-p_Z)a_0) \leq 0 \\
 	& (1-p_Z) a_0 - \max\{1-p_Z - c,0\} (p_Z a_1 + (1-p_Z)a_0) \geq 0 \\
	& (1-p_Z) (1-a_0) - \min\{1-p_Z + c,1\} [1-(p_Z a_1 + (1-p_Z)a_0)] \leq 0 \\
	& (1-p_Z) (1-a_0) - \max\{1-p_Z - c,0\} [1-(p_Z a_1 + (1-p_Z)a_0)] \geq 0 \\
	& p_Z a_1 - \min\{p_Z + c,1\} (p_Z a_1 + (1-p_Z)a_0) \leq 0 \\
 	& p_Z a_1 - \max\{p_Z - c,0\} (p_Z a_1 + (1-p_Z)a_0) \geq 0 \\
	& p_Z (1-a_1) - \min\{p_Z + c,1\} [1-(p_Z a_1 + (1-p_Z)a_0)] \leq 0 \\
	& p_Z (1-a_1) - \max\{p_Z - c,0\} [1-(p_Z a_1 + (1-p_Z)a_0)] \geq 0.
\end{align*}
Note that we indexed the elements of $\mathcal{D}(c)$ by $j \in \{0,1\}$ rather than $j \in \{1,2\}$. Also recall that $p_Z = \Prob(Z=1)$. Using the function $k_z(c)$ and removing redundant equations, we can rearrange these 8 inequalities to obtain
\begin{align*}
	a_1  &\geq a_0 k_0(c)\\
	a_0 &\geq a_1 k_1(c)\\
	(1-a_1) &\geq (1-a_0)k_0(c)\\
	(1-a_0) &\geq (1-a_1)k_1(c).
\end{align*}
This is the set of inequalities described in equation \eqref{eq:diamondSet} and used in theorem \ref{thm:BinaryYoneInstrumentCdep}.

Next consider $\mathcal{H}_x$. From theorem \ref{thm:BinaryYManyInstrumentCdep}, we have
\begin{align*}
	\mathcal{H}_x 
	= \{(a_0,a_1)\in[0,1]^2: \ &\text{For some $(q_0,q_1) \in [0,1]^2$,} \\
		& a_0 = \Prob(Y=1, X=x \mid Z=0) + \Prob(X=1-x \mid Z=0) q_0,\\
		& a_1 = \Prob(Y=1, X=x \mid Z=1) + \Prob(X=1-x \mid Z=1) q_1 \}.
\end{align*}
This is simply
\begin{align*}
	\mathcal{H}_x 
	&= [\Prob(Y=1,X=x \mid Z=0), \ \Prob(Y=1,X=x \mid Z=0) + \Prob(X=1-x \mid Z=0)]\\
	&\qquad \times [\Prob(Y=1,X=x \mid Z=1), \ \Prob(Y=1,X=x \mid Z=1) + \Prob(X=1-x \mid Z=1)]
\end{align*}
as in theorem \ref{thm:BinaryYoneInstrumentCdep}.
\end{proof}

\begin{proof}[Proof of proposition \ref{prop:PropertiesOneInstBounds}]
Parts 1 and 2 are immediate corollaries of proposition \ref{prop:PropertiesManyInstBounds}. Consider part 3. By proposition \ref{prop:PropertiesManyInstBounds}, the correspondence $\phi$ is continuous on
\[
	(c^*,1] = \left\{ c\in [c^*,1]: \big( \text{int} (\mathcal{D}(c)) \cap \text{int} (\mathcal{H}_x) \big) \neq \emptyset, x \in \{0,1\} \right\}.
\]
Here we show that $\phi$ is also continuous at $c = c^*$. Again by proposition \ref{prop:PropertiesManyInstBounds}, this correspondence is uhc at $c=c^*$. To show it is also lhc at $c=c^*$, let $U = [c^*,1]$. $U$ is an open neighborhood of $c^*$ relative to the domain $[c^*,1]$. Let $E$ be any open set such that
\[
	\big( \Theta_0(c^*)\times\Theta_1(c^*) \big) \cap E \neq \emptyset.
\]
Since the set $\Theta_0(c)\times\Theta_1(c)$ is weakly increasing in $c$ (with respect to the set inclusion order $\subseteq$),
\[
	\big( \Theta_0(c^*)\times\Theta_1(c^*) \big) \subseteq \big( \Theta_0(c)\times\Theta_1(c) \big)
\]
for any $c \in U$. Hence 
\[
	\left[ \big( \Theta_0(c^*)\times\Theta_1(c^*) \big) \cap E \right] \subseteq \left[ \big( \Theta_0(c)\times\Theta_1(c) \big) \cap E \right]
\]
for any $c \in U$. Since the set on the left is not empty, the set on the right is not empty:
\[
	\big( \Theta_0(c)\times\Theta_1(c) \big) \cap E \neq \emptyset
\]
for any $c \in U$. Therefore $\phi$ is continuous at $c=c^*$.
\end{proof}

\begin{proof}[Proof of corollary \ref{corr:hetTrtBinATEbounds}]
This follows immediately from the more general result, corollary \ref{corr:hetTrtMultipleATEbounds}.
\end{proof}

\subsubsection*{Proofs for section \ref{sec:hetTrtBinary}: Multiple discrete instruments}

\begin{proof}[Proof of theorem \ref{thm:BinaryYManyInstrumentCdep}]
There are two steps. First we show that $\Theta_0(c) \times \Theta_1(c)$ is an outer identified set, in the sense that it contains the true parameters. Then we show that it is sharp, in the sense that any element of this set is consistent with the data and the assumptions. Finally, the fact that the model is refuted if and only if this set is empty follows by the definition of a (sharp) identified set.

\bigskip

\textbf{Step 1.} We first show that the true conditional probabilities $(p_0,p_1)$ are in $\Theta_0(c) \times \Theta_1(c)$. Fix $x \in \{0,1\}$. By $c_\ell$-dependence,
\begin{align*}
	\Prob(Z_\ell = z_\ell^j \mid Y_x = 1) &\leq \min\{\Prob(Z_\ell = z_\ell^j) + c_\ell,1\} \\
	\Prob(Z_\ell = z_\ell^j \mid Y_x = 1) &\geq \max\{\Prob(Z_\ell = z_\ell^j) - c_\ell,0\} \\
	\Prob(Z_\ell = z_\ell^j \mid Y_x = 0) &\leq \min\{\Prob(Z_\ell = z_\ell^j) + c_\ell,1\} \\
	\Prob(Z_\ell = z_\ell^j \mid Y_x = 0) &\geq \max\{\Prob(Z_\ell = z_\ell^j) - c_\ell,0\}
\end{align*}
for all $j \in \{1,\ldots,J\}$. Consider the first inequality. Use Bayes' rule to rewrite the left hand side. This yields
\begin{align*}
	\Prob(Y_x = 1 \mid  Z_\ell = z_\ell^j) \Prob(Z_\ell = z_\ell^j) 
	&\leq \min\{\Prob(Z= z_\ell^j) + c_\ell,1\} \Prob(Y_x = 1)\\
	&= \min\{\Prob(Z= z_\ell^j) + c_\ell,1\} \sum_{k=1}^J \Prob(Y_x = 1 \mid  Z_\ell = z_\ell^k)\Prob(Z_\ell = z_\ell^k).
\end{align*}
In the last line we applied the law of total probability. Hence
\[
		\Prob(Y_x = 1 \mid  Z_\ell = z_\ell^j) \Prob(Z_\ell = z_\ell^j)
		- \min\{\Prob(Z= z_\ell^j) + c_\ell,1\} \sum_{k=1}^J \Prob(Y_x = 1 \mid  Z_\ell = z_\ell^k)\Prob(Z_\ell = z_\ell^k) \leq 0.
\]
Repeat these derivations for the other three inequalities. This yields the definition of $\mathcal{D}_\ell(c_\ell)$. Thus we have shown that the true conditional probabilities $\Prob(Y_x=1 \mid Z_\ell=z_\ell^j)$ are in $\mathcal{D}_\ell(c_\ell)$ for each $\ell \in \{1,\ldots,L\}$. $p_x$ is simply the vector of these conditional probabilities. Thus we have shown $p_x \in \mathcal{D}(c)$.

\bigskip

Next we show that $p_x \in \mathcal{H}_x$. Element $(\ell-1)J + j$ of $p_x$ can be written as
\begin{align*}
	[p_x]_{(\ell-1)J + j}
	&= \Prob(Y_x = 1 \mid Z_\ell = z_\ell^j) \\
	&= \Prob(Y_x = 1, X=x \mid Z_\ell = z_\ell^j) + \Prob(Y_x = 1, X=1-x \mid Z_\ell = z_\ell^j)\\
	&= \Prob(Y = 1, X=x \mid Z_\ell = z_\ell^j) + \Prob(Y_x = 1, X=1-x \mid Z_\ell = z_\ell^j)\\
	&= \Prob(Y = 1, X=x \mid Z_\ell = z_\ell^j) \\
	&\quad + \sum_{k=1}^{J^{L-1}} \Prob(Y_x = 1 \mid  X=1-x,Z_\ell = z_\ell^j, Z_{-\ell} = z_{-\ell}^k) \Prob(X=1-x, Z_{-\ell} = z_{-\ell}^k \mid Z_\ell = z_\ell^j).
\end{align*}
This is simply equation \eqref{eq:manyIVlawTotProb}. Since
\[
	\Prob(Y_x = 1 \mid X=1-x, Z_\ell = z_\ell^j, Z_{-\ell} = z_{-\ell}^k) \in [0,1],
\]
these derivations imply that
\[
	p_x = \textbf{b}_x + \textbf{A}_x q_x
\]
where $q_x \in [0,1]^{J^L}$ is a vector with elements
\[
	\Prob(Y_x = 1 \mid  X=1-x,Z_\ell = z_\ell^j, Z_{-\ell} = z_{-\ell}^k).
\]
Thus $p_x \in \mathcal{H}_x$. 

\bigskip

Thus we have shown that $p_x \in \mathcal{D}(c)\cap\mathcal{H}_x$. Since this is true for each $x \in \{0,1\}$, we have $p \in \Theta_1(c)\times \Theta_0(c)$.

\bigskip

\textbf{Step 2.} Next we show sharpness. Let $(p_0,p_1) \in \Theta_1(c)\times \Theta_0(c)$. First, $p_x \in \mathcal{D}(c)$ implies
\[
	\Prob(Z = z_\ell^j \mid Y_x = y) \in [\max\{\Prob(Z_\ell = z_\ell^j) -c_\ell, 0\}, \min\{\Prob(Z_\ell = z_\ell^j) + c_\ell, 1\}]
\]
for all $\ell \in \{1,\ldots,L \}$ and all $j \in \{1,\ldots,J\}$. This follows by reversing the arguments at the beginning of step 1. Hence $c_\ell$-dependence is satisfied for all $\ell \in \{1,\ldots,L\}$. 

Next, since $p_x \in \mathcal{H}_x$ there is a vector $q_x \in [0,1]^{J^L}$ such that
\[
	p_x = \textbf{b}_x + \textbf{A}_x q_x.
\]
This vector $q_x$ consists of elements of the form
\[
	\Prob(Y_x=1 \mid X=1-x,Z=z).
\]
Take these elements and combine them with the observed probabilities
\[
	\Prob(Y=1 \mid X=x,Z=z) = \Prob(Y_x=1 \mid X=x,Z=z)
\]
to get a distribution of $Y_x \mid (X,Z)$. Do this for both $x \in \{0,1\}$. Combine these conditional marginal distributions into a joint distribution $(Y_0,Y_1) \mid (X,Z)$ using any copula, and finally combine these with the known marginals $(X,Z)$ to get a joint distribution of $(Y_0,Y_1,X,Z)$. By construction, this joint distribution is consistent with $c$-dependence, the distribution of the observed data $(Y,X,Z)$, and yields the point $(p_0,p_1)$ that we started with.
\end{proof}

\begin{lemma}\label{lemma:ContCorresp}
Suppose B\ref{assump:NonTrivialinstrument}$^\prime$ holds. Then the correspondence $\mathcal{D}: [0,1]^L \rightrightarrows [0,1]^{LJ}$ defined by $\mathcal{D}(c)$ in equation \eqref{eq:generalDiamondSet} is continuous.
\end{lemma}

\begin{proof}[Proof of lemma \ref{lemma:ContCorresp}]
To prove this result, we show that $\mathcal{D}_\ell$ is both upper hemicontinuous (uhc) and lower hemicontinuous (lhc).

\bigskip

\textbf{Step 1: Upper hemicontinuity.} Pick any $c_\ell \in [0,1]$ and $a \in [0,1]^J$. Consider a sequence $\{ c_\ell^n \}$ converging to $c_\ell$ and a sequence $\{ a^n \}$ converging to $a$ such that $a^n \in \mathcal{D}(c^n)$ for all $n$. From the first inequality in equation \eqref{eq:generalDiamondSet},
\[
	\Prob(Z_\ell = z_\ell^j) a_j^n - \min\{\Prob(Z_\ell = z_\ell^j) + c_\ell^n,1\} \sum_{k=1}^J \Prob(Z_\ell = z_\ell^k) a_k^n \leq 0.
\]
Taking limits as $n \rightarrow \infty$ yields
\[
	\Prob(Z_\ell = z_\ell^j) a_j - \min\{\Prob(Z_\ell = z_\ell^j) + c_\ell,1\} \sum_{k=1}^J \Prob(Z_\ell = z_\ell^k) a_k \leq 0.
\]
Repeating this for the other three inequalities shows that $a \in \mathcal{D}_\ell(c_\ell)$. Thus $\mathcal{D}_\ell$ has a closed graph. Hence $\mathcal{D}_\ell$ is uhc by theorem 17.11 in \cite{AliprantisBorder2006}, since the range $[0,1]^J$ is closed.

\bigskip

\textbf{Step 2: Lower hemicontinuity.} We next show that $\mathcal{D}_\ell$ is lhc. We do this in two parts: First for $\bar{c} \in (0,1)$ and second for $\bar{c} = 0$.

\medskip

\emph{Part A.} Let $\bar{c} \in (0,1]$. We will show $\mathcal{D}_\ell$ is lhc at $\bar{c}$. Let $E$ be an open set in $[0,1]^J$ such that
\[
	\mathcal{D}_\ell(\bar{c}) \cap E \neq \emptyset.
\]
By the definition of lhc, we must find an open set $U$ (open relative to the domain $[0,1]$) such that 
\begin{enumerate}
\item $\bar{c} \in U$, and 

\item $c \in U$ implies that $\mathcal{D}_\ell(c) \cap E \neq \emptyset$.
\end{enumerate}
We consider two cases.

\bigskip

\texttt{Case 1}. Suppose there is an element $\bar{a} = (\bar{a}_1,\ldots,\bar{a}_J) \in \mathcal{D}_\ell(\bar{c}) \cap E$ such that $\bar{a}_1 = \cdots = \bar{a}_J$. In this case, we can simply let $U = [0,1]$. Why? $U$ is open (relative to $[0,1]$). It contains $\bar{c}$ since it has to be in $[0,1]$. Moreover, we know that $\bar{a} \in E$. But since $\bar{a}_1 = \cdots = \bar{a}_J$, we also know that $\bar{a} \in \mathcal{D}_\ell(0)$. (This is simply our baseline model where all these conditional probabilities equal each other, by statistical independence.) Hence $\bar{a} \in \mathcal{D}_\ell(c)$ for any $c \in [0,1]$, since $\mathcal{D}_\ell(c)$ is weakly increasing in $c$. Thus $\mathcal{D}_\ell(c) \cap E \neq \emptyset$ for all $c \in [0,1] = U$. This is the second requirement we needed of $U$.

\bigskip

\texttt{Case 2}. Suppose none of the elements in $\mathcal{D}_\ell(\bar{c}) \cap E$ have all components equal. Let $\bar{a}$ be an arbitrary element of this intersection. Consider
\[
	a(\epsilon) = (1-\epsilon)\bar{a} + \epsilon \frac{1}{2} \iota_J
\]
where $\iota_J$ is a $J$-vector of ones. $\frac{1}{2} \iota_J \in \mathcal{D}_\ell(c)$ for all $c$ since $\mathcal{D}_\ell(c)$ always contains the diagonal. This combined with convexity of the set $\mathcal{D}_\ell(c)$ implies that $a(\epsilon) \in \mathcal{D}_\ell(\bar{c})$ for all $\epsilon \in [0,1]$. Finally, since $E$ is open, there exists an $\epsilon^* \in [0,1]$ such that $a(\epsilon^*) \in \mathcal{D}_\ell(\bar{c}) \cap E$. Let $a^* = a(\epsilon^*)$. Let 
\[
	U = \left\{c \in[0,1] : a^* \in \text{int} (\mathcal{D}_\ell(c))\right\}.
\]
Equivalently,
\begin{align*}
	U 
	= \Bigg\{ c\in [0,1] : \ &\text{ For all  $j=1,\ldots,J$,} \\
	&\Prob(Z_\ell = z_\ell^j) a_{j}^* - \min\{\Prob(Z_\ell = z_\ell^j) + c,1\} \sum_{k=1}^J \Prob(Z_\ell = z_\ell^k) a_{k}^* < 0, \\
	&\Prob(Z_\ell = z_\ell^j) a_{j}^* - \max\{\Prob(Z_\ell = z_\ell^j) - c,0\} \sum_{k=1}^J \Prob(Z_\ell = z_\ell^k) a_{k}^* > 0, \\
	&\Prob(Z_\ell = z_\ell^j) (1-a_{j}^*) - \min\{\Prob(Z_\ell = z_\ell^j) + c,1\} \sum_{k=1}^J \Prob(Z_\ell = z_\ell^k) (1-a_{k}^*) < 0,\\
	& \Prob(Z_\ell = z_\ell^j) (1-a_{j}^*) - \max\{\Prob(Z_\ell = z_\ell^j) - c,0\} \sum_{k=1}^J \Prob(Z_\ell = z_\ell^k) (1-a_{k}^*) > 0 \Bigg\}.
\end{align*}
This set $U$ can be written as
\begin{align*}
	U 
	= \bigcap_{j=1}^J \Bigg[ &\left\{c \in [0,1]: \Prob(Z_\ell = z_\ell^j) a_{j}^* - \min\{\Prob(Z_\ell = z_\ell^j) + c,1\} \sum_{k=1}^J \Prob(Z_\ell = z_\ell^k) a_{k}^* < 0 \right\} \\
	&\cap \left\{c \in [0,1]: \Prob(Z_\ell = z_\ell^j) a_{j}^* - \max\{\Prob(Z_\ell = z_\ell^j) - c,0\} \sum_{k=1}^J \Prob(Z_\ell = z_\ell^k) a_{k}^* > 0 \right\} \\
	&\cap \left\{c \in [0,1]: \Prob(Z_\ell = z_\ell^j) (1-a_{j}^*) - \min\{\Prob(Z_\ell = z_\ell^j) + c,1\} \sum_{k=1}^J \Prob(Z_\ell = z_\ell^k) (1-a_{k}^*) < 0 \right\} \\
	&\cap \left\{c \in [0,1]: \Prob(Z_\ell = z_\ell^j) (1-a_{j}^*) - \max\{\Prob(Z_\ell = z_\ell^j) - c,0\} \sum_{k=1}^J \Prob(Z_\ell = z_\ell^k) (1-a_{k}^*) > 0 \right\} \Bigg].
\end{align*}
By continuity in $c$ of the inequality components, this is a finite intersection of open sets. Hence $U$ is open. $U$ is not empty since $\bar{c} \in U$, which follows since $a^* \in \text{int}(\mathcal{D}_\ell(\bar{c}))$.

Finally, we show that $U$ satisfies the second property we need: Suppose $c \in U$. By definition of $U$, $a^* \in \text{int}(\mathcal{D}_\ell(c))$. Moreover, $a^* \in E$ by construction. Thus $a^* \in \big( \mathcal{D}_\ell(c) \cap E \big)$ for any $c \in U$. Hence
\[
	\mathcal{D}_\ell(c) \cap E \neq \emptyset	
\]
for any $c \in U$.

\medskip

\emph{Part B.} Let $\overline{c} = 0$. Let $E$ be an open set in $[0,1]^J$ such that $\mathcal{D}_\ell(0) \cap E \neq \emptyset$. Let $U = [0,1]$. $U$ is open (relative to $[0,1]$). Moreover, 
\[
	\big( \mathcal{D}_\ell(0) \cap E \big) \subseteq \big( \mathcal{D}_\ell(c) \cap E \big)
\]
since $\mathcal{D}_\ell$ is weakly increasing. Since the set on the left is not empty, the set on the right is not empty:
\[
	\mathcal{D}_\ell(c) \cap E \neq \emptyset
\]
for any $c \in U$.

\bigskip

\textbf{Step 3: Take products.} Finally, theorem 17.28 in \cite{AliprantisBorder2006} shows that the product of continuous correspondences is continuous. Hence the correspondence $\mathcal{D}(\cdot)$ defined by
\[
	\mathcal{D}(c) = \prod_{\ell=1}^L \mathcal{D}_\ell(c_\ell)
\]
is continuous.
\end{proof}

\begin{proof}[Proof of proposition \ref{prop:PropertiesManyInstBounds}]
\hfill
\begin{enumerate}
\item By lemma \ref{lemma:ContCorresp}, the correspondence $\mathcal{D}(\cdot)$ is continuous. Moreover $\mathcal{D}(\cdot)$ is closed-valued since it is defined by a finite number of linear weak inequalities. The correspondence that maps $c$ into $\mathcal{H}_x$ is continuous since it is constant-valued. It is also closed- and compact-valued since $[0,1]^{J^L}$ is compact, and $\mathcal{H}_x$ is the image of $[0,1]^{J^L}$ under a continuous (affine) mapping. Hence the set
\[
	\mathcal{C}_x = \{c \in [0,1]^L : \mathcal{D}(c)\cap \mathcal{H}_x \neq \emptyset\}
\]
is closed, by exercise 11.18(b) on page 58 of \cite{Border1985}. Finally, let
\[
	\mathcal{C} = \mathcal{C}_0 \cap \mathcal{C}_1.
\]
This set is the intersection of two closed sets and hence is closed. It depends only on $\mathcal{D}(\cdot)$ and $\mathcal{H}_x$, which are point identified. Hence $\mathcal{C}$ is point identified.

\item By definition of $\Theta_x(c)$,
\[
	\Theta_0(c) \times \Theta_1(c) = \big( \mathcal{D}(c) \times \mathcal{H}_0 \big) \times \big( \mathcal{D}(c) \times \mathcal{H}_1 \big).
\]
Each of these sets is a closed, convex polytope. To see this, consider $\mathcal{D}(c)$ and $\mathcal{H}_x$ separately.
\begin{enumerate}
\item $\mathcal{D}(c)$ is closed, as shown in part 1. It is convex since it is defined by a collection of linear inequalities.

\item $\mathcal{H}_x$ is compact, as shown in part 1. Therefore it is closed and bounded. It is convex valued since it is an affine transformation of the convex set $[0,1]^{J^L}$, and convexity is preserved by affine transformations.
\end{enumerate}
Since both $\mathcal{D}(c)$ and $\mathcal{H}_x$ are closed, convex, bounded by $[0,1]^{LJ}$, and defined by linear inequalities, they are both polytopes. Finally, all three properties---closed, convex, polytope---are preserved by taking finite Cartesian products.

\item Next we show continuity of the correspondence $\phi : \mathcal{C} \rightrightarrows [0,1]^{2LJ}$ defined by $\phi(c) = \Theta_0(c)\times\Theta_1(c)$. First consider the correspondence $\phi_x : \mathcal{C} \rightrightarrows [0,1]^{LJ}$ defined by
\[
	\phi_x(c) = \Theta_x(c) = \mathcal{D}(c) \cap \mathcal{H}_x.
\]
It is convex-valued and closed-valued since its values are the intersection of two closed and convex sets. It is uhc by theorem 17.25 in \cite{AliprantisBorder2006}, since $\mathcal{D}(c)$ and $\mathcal{H}_x$ are both uhc. It is lhc for all $c$ such that $\text{int}(\mathcal{D}(c)) \cap \text{int}(\mathcal{H}_x) \neq \emptyset$, by applying theorem B in \cite{LechickiSpakowski1985}. Finally, $\phi$ is just the product of the correspondences $\phi_0$ and $\phi_1$. The product of continuous correspondences is continuous by theorem 17.28 in \cite{AliprantisBorder2006}.

\end{enumerate}

\end{proof}

\begin{proof}[Proof of corollary \ref{corr:hetTrtMultipleATEbounds}]
\hfill
\begin{enumerate}
\item We can write the identified set for $(\Prob(Y_0 = 1), \Prob(Y_1 = 1))$ as
\begin{multline*}
	\left\{ \left( \sum_{j = 1}^J \sum_{\ell = 1}^L a_{(\ell-1)J + j}\Prob(Z_\ell = z_\ell^j), \ \sum_{j = 1}^J \sum_{\ell = 1}^L \widetilde{a}_{(\ell-1)J + j}\Prob(Z_\ell = z_\ell^j) \right) : \mathbf{a} \in \Theta_0(c), \widetilde{\mathbf{a}} \in \Theta_1(c)\right\}\\
	= \left\{\sum_{j = 1}^J \sum_{\ell = 1}^L a_{(\ell-1)J + j}\Prob(Z_\ell = z_\ell^j): \mathbf{a} \in \Theta_0(c)\right\} \times  \left\{\sum_{j = 1}^J \sum_{\ell = 1}^L \widetilde{a}_{(\ell-1)J + j}\Prob(Z_\ell = z_\ell^j): \widetilde{\mathbf{a}} \in \Theta_1(c)\right\}.
\end{multline*}
By proposition \ref{prop:PropertiesManyInstBounds}, $\Theta_x(c)$ is closed, non-empty, and bounded for $x\in\{0,1\}$. Therefore, suprema and infima of these sets are attained on $\Theta_x(c)$. Moreover, since $\Theta_x(c)$ is convex, the set
\[
	\left\{\sum_{j = 1}^J \sum_{\ell = 1}^L a_{(\ell-1)J + j}\Prob(Z_\ell = z_\ell^j): \mathbf{a} \in \Theta_x(c)\right\}
\]
equals a closed interval ranging from the infimum to the supremum.

\item The mapping $f(\cdot)$ defined by
\[
	f(\mathbf{a}) = \sum_{j = 1}^J \sum_{\ell = 1}^L a_{(\ell-1)J + j}\Prob(Z_\ell = z_\ell^j)
\]
is continuous at all $\mathbf{a} \in \R^{LJ}$. The correspondence $\Theta_x(\cdot)$ is continuous and compact-valued for all $c$ such that $\text{int}(\mathcal{D}(c)) \cap \text{int}(\mathcal{H}_x) \neq \emptyset$ by proposition \ref{prop:PropertiesManyInstBounds}. Therefore
\[
	\overline{P}_x(c) = \max\{f(\mathbf{a}):\mathbf{a} \in \Theta_x(c)\}
\]
is a continuous function of $c$ on $\{c\in\mathcal{C}: \text{int}(\mathcal{D}(c)) \cap \text{int}(\mathcal{H}_x) \neq \emptyset\}$ by the Maximum Theorem (for example, see theorem 9.14 in \citealt{Sundaram1996}). By continuity of $-f(\mathbf{a})$, $\underline{P}_x(c)$ is continuous over the same domain.

\item This follows immediately from the identified set $\Theta_0(c)\times\Theta_1(c)$ being a Cartesian product, which implies that maxima and minima for each component can be simultaneously attained.
\end{enumerate}
\end{proof}

\subsubsection*{Proofs for section \ref{sec:hetTrtCts}}

\begin{proof}[Proof of proposition \ref{prop:kitagawa3point1}]
See proposition 3.1 of \cite{Kitagawa2009}.
\end{proof}

\begin{proof}[Proof of theorem \ref{thm:ContYManyInstrumentCdep}]
This proof follows the structure of the proof of theorem \ref{thm:BinaryYManyInstrumentCdep}. 

\bigskip

\textbf{Step 1.} We first show that the true densities $((\mathbf{f}_{Y_0|Z_\ell})_{\text{all } \ell}, (\mathbf{f}_{Y_1|Z_\ell})_{\text{all } \ell})$ are in $\Theta_0(c)\times\Theta_1(c)$. Fix $x \in \{0,1\}$. By $c_\ell$-dependence, 
\begin{align*}
	\Prob(Z_\ell= z \mid Y_x = y)  &\leq \min\{\Prob(Z_\ell= z) + c_\ell,1\}\\
	\Prob(Z_\ell = z \mid Y_x = y)  &\geq \max\{\Prob(Z_\ell= z) - c_\ell,0\}.
\end{align*}
all $y \in \supp(Y_x)$ and each $z \in \{0,1\}$. These inequalities for $z=0$ are equivalent to the inequalities for $z=1$, so there are only two non-redundant inequalities here. Use Bayes' rule to rewrite the left hand side. This yields
\begin{align*}
	f_{Y_x|Z_\ell}(y \mid z) \Prob(Z_\ell = z) 
	&\leq \min\{\Prob(Z_\ell= z) + c_\ell,1\} f_{Y_x}(y)\\
	&= \min\{\Prob(Z_\ell= z) + c_\ell,1\} \sum_{\widetilde{z}=0}^1 f_{Y_x|Z_\ell}(y \mid \widetilde{z})\Prob(Z_\ell = \widetilde{z}).
\end{align*}
Hence
\[
	f_{Y_x|Z_\ell}(y \mid z) \Prob(Z_\ell = z) 
	- \min\{\Prob(Z_\ell= z) + c_\ell,1\} \sum_{\widetilde{z}=0}^1 f_{Y_x|Z_\ell}(y \mid \widetilde{z})\Prob(Z_\ell = \widetilde{z}) \leq 0.
\]
Evaluating at $z=0$ yields
\[
	f_{Y_x| Z_\ell}(y \mid 0) \Prob(Z_\ell=0) - \min \{ \Prob(Z_\ell=0) + c_\ell, 1 \} \big( f_{Y_x | Z_\ell}(y \mid 0) \Prob(Z_\ell=0) + f_{Y_x | Z_\ell}(y \mid 1) \Prob(Z_\ell=1) \big) \leq 0.
\]
Rearranging yields
\[
	f_{Y_x | Z_\ell}(y \mid 0) \Prob(Z_\ell=0) \big( 1 - \min \{ \Prob(Z_\ell=0) + c_\ell,1 \} \big)
	\leq \Prob(Z_\ell=1) \min \{ \Prob(Z_\ell=0) + c_\ell,1 \} f_{Y_x | Z_\ell}(y \mid 1)
\]
or
\[
	f_{Y_x | Z_\ell}(y \mid 0) \frac{\Prob(Z_\ell=0) \big( 1 - \min \{ \Prob(Z_\ell=0) + c_\ell,1 \} \big)}{ \Prob(Z_\ell=1) \min \{ \Prob(Z_\ell=0) + c_\ell,1 \} }
	\leq f_{Y_x | Z_\ell}(y \mid 1).
\]
In the numerator,
\[
	1 - \min \{ \Prob(Z_\ell=0) + c_\ell,1 \}
	= \max \{ \Prob(Z_\ell=1) - c_\ell, 0 \}.
\]
Hence we have
\[
	f_{Y_x | Z_\ell}(y \mid 0) k_0^\ell(c_\ell) \leq f_{Y_x | Z_\ell}(y \mid 1).
\]
Repeat these derivations for the other inequality evaluated at $z=0$ to get the second inequality in the definition of $\mathcal{D}_{x,\ell}(c_\ell)$. Thus we have shown that the true conditional densities $\mathbf{f}_{Y_x|Z_\ell}$ are in $\mathcal{D}_{x,\ell}(c_\ell)$ for all $\ell$. Since $(\mathbf{f}_{Y_x|Z_\ell})_{\text{all } \ell}$ is just the vector of these conditional densities, we have shown that it is in $\mathcal{D}_x(c)$.

\bigskip

Next we show that $(\mathbf{f}_{Y_x|Z_\ell})_{\text{all } \ell} \in \mathcal{H}_x$. For $z \in \{0,1\}$, element $(\ell-1)2 + 1 + z$ of $(\mathbf{f}_{Y_x|Z_\ell})_{\text{all } \ell}$ can be written as
\begin{align*}
	&[(\mathbf{f}_{Y_x|Z_\ell})_{\text{all } \ell}]_{(\ell-1)2 + z} \\
	&= f_{Y_x|Z_\ell}(y \mid z) \\
	&= f_{Y_x,X|Z_\ell}(y,x \mid z)+ f_{Y_x,X|Z_\ell}(y,1-x \mid z)\\
	&= f_{Y,X|Z_\ell}(y,x \mid z)+ f_{Y_x,X|Z_\ell}(y,1-x \mid z)\\
	&= f_{Y,X|Z_\ell}(y,x \mid z) + \sum_{k=1}^{2^L} f_{Y_x|X,Z_\ell,Z_{-\ell}}(y \mid 1-x,z, z_{-\ell}^k) \Prob(X=1-x, Z_{-\ell} = z_{-\ell}^k \mid Z_\ell = z).
\end{align*}
Since
\[
	f_{Y_x|X,Z_\ell,Z_{-\ell}}(y,1-x \mid z, z_{-\ell}^k)
	\in \mathcal{P}(\mathcal{Y}_x)
\]
these derivations imply that
\[
	(\mathbf{f}_{Y_x|Z_\ell})_{\text{all } \ell} = (\textbf{f}_{Y,X|Z_\ell}(\cdot,x))_{\text{all } \ell} + \textbf{A}_x \textbf{q}
\]
where $\mathbf{q} \in \mathcal{P}(\mathcal{Y}_x)^{2^L}$ is a vector with elements $f_{Y_x|X,Z_\ell,Z_{-\ell}}(y \mid 1-x,z, z_{-\ell}^k)$. Thus $(\mathbf{f}_{Y_x|Z_\ell})_{\text{all } \ell} \in \mathcal{H}_x$. 

\bigskip

Thus we have shown that $(\mathbf{f}_{Y_x|Z_\ell})_{\text{all } \ell} \in \mathcal{D}_x(c)\cap\mathcal{H}_x$. Since this is true for each $x\in\{0,1\}$, we have $((\mathbf{f}_{Y_0|Z_\ell})_{\text{all } \ell}, (\mathbf{f}_{Y_1|Z_\ell})_{\text{all } \ell}) \in \Theta_0(c)\times\Theta_1(c)$.

\bigskip

\textbf{Step 2.} Next we show sharpness. Let $(\mathbf{f}_0,\mathbf{f}_1) \in \Theta_1(c)\times \Theta_0(c)$. First, $\mathbf{f}_x \in \mathcal{D}_x(c)$ implies
\[
	\Prob(Z_\ell = z \mid Y_x = y) \in [\max\{\Prob(Z_\ell = z) -c_\ell, 0\}, \min\{\Prob(Z_\ell = z) + c_\ell, 1\}]
\]
for all $\ell \in \{1,\ldots,L \}$ and all $z \in \{0,1\}$. This follows by reversing the arguments at the beginning of step 1. Hence $c_\ell$-dependence is satisfied for all $\ell \in \{1,\ldots,L\}$.

Next, since $\mathbf{f}_x \in \mathcal{H}_x$, there is a vector of counterfactual densities $\mathbf{q}_x \in \mathcal{P}(\mathcal{Y}_x)^{2^L}$ such that
\[
	\mathbf{f}_x = (\textbf{f}_{Y,X|Z_\ell}(\cdot,x))_{\text{all } \ell} + \textbf{A}_x \mathbf{q}_x.
\]
This vector $\mathbf{q}_x$ consists of elements of the form $f_{Y_x | X,Z}(\cdot \mid 1-x,z)$. Take these densities and combine them with the observed densities
\[
	f_{Y \mid X,Z}(\cdot \mid x,z) = f_{Y_x \mid X,Z}(\cdot \mid x,z)
\]
to get a distribution of $Y_x \mid (X,Z)$. Do this for both $x \in \{0,1\}$. Combine these conditional marginal distributions into a joint distribution $(Y_0,Y_1) \mid (X,Z)$ using any copula, and finally combine these with the known marginals $(X,Z)$ to get a joint distribution of $(Y_0,Y_1,X,Z)$. By construction, this joint distribution is consistent with $c$-dependence, the distribution of the observed data $(Y,X,Z)$, and yields the point $(\mathbf{f}_0,\mathbf{f}_1)$ that we started with.
\end{proof}

\nocite{MastenPoirierFFarxiv}
\nocite{Bonk2008}
\nocite{Stanford2017}
\nocite{Ariew2018}
\nocite{ManskiPepper2018}
\nocite{ManskiPepper2018}
\nocite{ManskiPepper2018}
\nocite{Ramsahai2012}
\nocite{MachadoShaikhVytlacil2018}
\nocite{MastenPoirierFFarxiv}
\nocite{MastenPoirierFFarxiv}
\nocite{Ramsahai2012}
\nocite{MachadoShaikhVytlacil2018}
\nocite{HansenJagannathan1991,HansenJagannathan1997}
\nocite{HansenHeatonLuttmer1995}
\nocite{Ludvigson2013}
\nocite{DHaultfoeuilleEtAl2018}
\nocite{BontempsMagnacMaurin2012}
\nocite{BugniCanayShi2015}
\nocite{ChernozhukovLeeRosen2013}
\nocite{Liu1955,Liu1960,Liu1963}
\nocite{Fisher1961}
\nocite{Fisher2005}
\nocite{AngristKrueger1994}
\nocite{AltonjiElderTaber2005}
\nocite{Small2007}
\nocite{ConleyHansenRossi2012}
\nocite{Ashley2009}
\nocite{Kraay2012}
\nocite{AshleyParmeter2015}
\nocite{vanKippersluisRietveld2017,vanKippersluisRietveld2018}
\nocite{Small2007}
\nocite{ManskiPepper2000,ManskiPepper2009}
\nocite{BlundellEtAl2007}
\nocite{KreiderPepper2007}
\nocite{KreiderPepperGundersenJolliffe2012}
\nocite{GundersenKreiderPepper2012}
\nocite{KreiderPepperRoy2016}
\nocite{ChenFloresFloresLagunes2016}
\nocite{NevoRosen2012}
\nocite{ManskiPepper2000,ManskiPepper2009}
\nocite{NevoRosen2012}
\nocite{White1982}
\nocite{White1994}
\nocite{Small2007}
\nocite{Small2007}
\nocite{Imbens2003}
\nocite{Newey1985}
\nocite{ManskiTamer2002}
\nocite{ChernozhukovHongTamer2007}
\nocite{Small2007}

\end{document}